\documentclass[letterpaper,12pt,fleqn]{article}

\usepackage{amsfonts}
\usepackage{bm}
\usepackage{geometry}
\usepackage{amsmath}
\usepackage{amssymb}
\usepackage{amsthm}
\usepackage{a4wide}
\usepackage{lscape}
\usepackage{rotating}
\usepackage{mathpazo}
\usepackage{enumerate}
\usepackage{comment}
\usepackage{filecontents}
\usepackage{float}
\usepackage{epstopdf}
\usepackage{xcolor,colortbl}
\usepackage{booktabs}
\usepackage{url}
\usepackage{authblk}
\usepackage{natbib}
\usepackage{subcaption}
\usepackage{xr}
\usepackage{setspace}
\input{ee.sty}

\theoremstyle{plain}
\newtheorem{theorem}{Theorem}
\newtheorem{corollary}{Corollary}
\theoremstyle{definition}
\newtheorem{remark}{Remark}
\newcounter{saveeqn}

\theoremstyle{plain}

\newtheorem{proposition}{Proposition}
\newtheorem{assumption}{Assumption}

\newcommand{\I}{\mI}

\definecolor{Grey}{gray}{0.9}

\makeatletter
\newcommand\Label[1]{&\refstepcounter{equation}(\theequation)\ltx@label{#1}&}
\makeatother

\title{\vspace{-2cm}The Incidental Parameters Problem in Testing for Remaining Cross-section Correlation
\thanks{We would like to thank the Co-Editor (Jianqing Fan), an anonymous Associate Editor, two anonymous referees for numerous suggestions that greatly improved this paper. We would also like to thank participants at AMES (Seoul, 2018 and Xiamen, 2019), IPDC (Seoul, 2018 and Vilnius, 2019), IAAE (Montreal, 2018), Nordic Econometric Meeting (Stockholm, 2019), and seminar participants at Copenhagen Business School, Copenhagen University, Erasmus University Rotterdam, KU Leuven,  Monash University, New Economic School, Queen Mary University London, University of Amsterdam, University of Bern, University of Birmingham, University of Groningen, University of Southern California, and Uppsala University.  We express special thanks to Geert Dhaene, Jiti Gao, Cheng Hsiao, Frank Kleibergen, Blaise Melly, Roger Moon, Andreas Pick, Geert Ridder, Vasilis Sarafidis, Peter Schmidt, Yongcheol Shin, Vanessa Smith, Bas Werker, and Joakim Westerlund for helpful comments.}
}
\author[1]{Art\={u}ras Juodis \thanks{a.juodis@uva.nl. Amsterdam School of Economics, University of Amsterdam, The Netherlands. Financial support from the Netherlands Organization for Scientific Research (NWO) under research grant number 451--17--002 is gratefully acknowledged.}}
\author[2]{Simon Reese \thanks{Corresponding author. simon.reese@nek.lu.se. Department of Economics, Lund University, Sweden. Simon Reese was a Postdoctoral Scholar Research Associate at USC Dornsife INET while working on this project.}}
\affil[1]{University of Amsterdam and Tinbergen Institute}
\affil[2]{Lund University}
\begin{document}
	
\singlespace
\bibpunct{(}{)}{,}{a}{}{,}
\bibliographystyle{ecta}

\maketitle

\begin{abstract}
In this paper we consider the properties of the \citet{CDtest2004, doi:10.1080/07474938.2014.956623} CD test for cross-section correlation when applied to residuals obtained from panel data models with many estimated parameters. We show that the presence of period-specific parameters leads the CD test statistic to diverge as the time dimension of the sample grows. This result holds even if cross-section dependence is correctly accounted for and hence constitutes an example of the Incidental Parameters Problem. The relevance of this problem is investigated both for the classical Two-way Fixed Effects estimator as well as the Common Correlated Effects estimator of \citet{ECTA:ECTA692}. We suggest a weighted CD test statistic which re-establishes standard normal inference under the null hypothesis. Given the widespread use of the CD test statistic to test for remaining cross-section correlation, our results have far reaching implications for empirical researchers.


\end{abstract}

\onehalfspace

\section{Introduction}
\label{section::introduction}
Given that economic agents rarely act entirely independently of each other, modeling cross-section dependence plays a prominent role in panel data econometrics. Time fixed effects are probably the simplest way of addressing this issue and allow controlling for a common trend whose effect is homogeneous across cross-section units. During the last decade, interactive fixed effects models have become a popular, more general alternative. They assume the presence of a factor error structure, i.e. a small number of unobserved common trends interacted with entity-specific slope coefficients. Using either of the two modeling possibilities begs the question whether cross-section dependence is adequately accounted for.

In this paper, we show that the application of tests for cross-section dependence to regression residuals obtained from two-way fixed effects models or interactive effects models is problematic. 
%
We use the popular CD test statistic of \citet{CDtest2004, doi:10.1080/07474938.2014.956623} and show that the inclusion of period-specific parameters introduces a bias term of order $\sqrt{T}$. In order to avoid erroneous rejection of the null hypothesis of no unaccounted cross-section dependence, we suggest a modified test statistic that re-weights cross-section covariances with Rademacher distributed weights. This weighted CD test statistic is asymptotically standard normal and has very good size under appropriate regularity conditions on the chosen weights.

The use of the CD test statistic as a misspecification test for regression models that already account for cross-section dependence has repeatedly been observed in empirical studies (see e.g. \citealp{HOLLY2010160}, \citet{JAE:JAE2338}, \citet{JAE:JAE2468} and \citet{EBERHARDT201545} among others). In some applications, the CD test statistic is explicitly used as a model-selection tool, interpreting a reduction in the absolute value of the CD test statistic  as an indication of a \emph{better} model. In other cases, only specifications not rejected by the CD test statistic (given some significance level) were considered.
However, to the best of our knowledge, there are currently no theoretical results in the literature that could justify this practice.
Furthermore, the scenario of testing for cross-section dependence in cross-sectionally demeaned or defactored data is usually completely ignored in any of the large scale Monte Carlo studies integral to theoretical and empirical papers in the field. Only very recently, \citet{mao2018testing} reported the size of three tests for cross-section dependence in a subset of his simulation experiments. However, despite clear evidence of excessive over-rejections, these results are neither discussed nor investigated theoretically. Therefore, this study is the first one to investigate the properties of cross-section dependence tests applied to residuals of models that control for sources of common variation across cross-sections. In particular, given its popularity in the applied panel data literature, we restrict our attention to the CD test statistic of \citet{CDtest2004, doi:10.1080/07474938.2014.956623}. Our interest lies in residuals of models that characterize cross-section dependence as driven by a small number of unobservable factors. This comprises both the two-way fixed effects (2WFE) model as well as models with a multifactor error structure where factors are interacted with unit-specific slope coefficients.
The results that we obtain are summarized as follows:
\begin{enumerate}
	\item The application of the CD test to residuals obtained from a model where common factors enter either as time fixed effects or through a multifactor error structure renders the test statistic biased for any fixed $T$, and divergent as $T\to\infty$.
	\item In addition to the mean of the CD test statistic, even its variance may be affected. This can result in an asymptotically degenerate distribution of the test statistic.
	\item A simple way of eliminating bias is to construct the CD test statistic from specifically weighted cross-section covariances rather than correlations. This leads to a valid test statistic for remaining cross-section correlation with good small sample properties in simulations. 
\end{enumerate}

The degeneracy of the CD statistics can be seen as a manifestation of the Incidental Parameters Problem (IPP) of \citet{neyman1948}. In this respect, this paper contributes to this branch of the literature. So far, major focus in the panel data literature has been related to the IPP stemming from estimated individual specific effects. Our paper is the first one to document an asymptotically non-negligible impact of estimated time specific common parameters which cannot be ruled out by restrictions on the relative expansion rate of $N$ and $T$. Furthermore, since the CD statistic can be seen as a time-series average of second degree (\emph{degenerate}) U-statistics, our results shed some light on the potential impact of the IPP beyond simple cross-sectional averages.
Lastly, while this article only considers linear models, the average correlation approach to testing for cross-section dependence was extended to nonlinear and nonparametric panel data models by \citet{OBES:OBES646} and \citet{chen_gao_li_2012}, respectively. Hence, problems documented in this paper carry over to post-estimation properties of non-linear models discussed in e.g. \citet{fernandez_weidner_2014}, \citet{BonevaLinton2016}, or \citet{chenvalweidner2014}.

The remainder of this paper is structured as follows: Section \ref{section:testing} introduces the testing problem. In Sections \ref{section::asymptotics} we present asymptotic results for models with two-way fixed effects and multifactor error structures. In Section \ref{section::solutions} we discuss standard approaches for bias correction and propose a weighted CD test statistic that achieves this goal. Sections \ref{section::MC} and \ref{section:empirical} illustrate the problem documented in this paper by means of simulated and real data. Section \ref{section:conclusions} concludes. Additional technical discussions, Monte Carlo studies, and all proofs are relegated to the Supplementary Appendix.

\textbf{Notation:} $\mI_m$ denotes an $m \times m$ identity matrix and the subscript is sometimes disregarded from for the sake of simplicity. $\vzeros$ denotes a vector of zeros while $\mZeros$ stands for a matrix of zeros. $\vs_m$ denotes a selection vector all of whose elements are zero except for element $m$ which is one. $\viota$ is a vector entirely consisting of ones. The dimension of these latter vectors and matrices is generally suppressed for the sake of simplicity and needs to be inferred from context. For a generic $m \times n$ matrix $\mA$, $\mP_{\mA} = \mA(\mA'\mA)^{-1}\mA'$ projects onto the space spanned by the columns of $\mA$ and $\mM_{\mA} = \mI_m - \mP_{\mA}$. Furthermore, $\rk(\mA)$ denotes the rank of $\mA$, $\tr(\mA)$ its trace and $\| \mA \| =\left( \tr (\mA'\mA)\right)^{1/2}$ the Frobenius norm of $\mA$. For a set of $m \times n$ matrices $\{ \mA_1, \ldots, \mA_N \}$, $\overline{\mA} = N^{-1}\sum_{i=1}^N \mA_i$.  $\delta$ and $M$ stand for a small and large positive real number, respectively. For two real numbers $a$ and $b$, $a \vee b = \max\{ a,b\}$. Lastly,  $\largeO_P(\cdot)$ and $\smallO_P(\cdot)$ express stochastic order relations.

\section{The testing problem}	
\label{section:testing}
Let $\vz_{i}$ be a $T$-dimensional data vector observed over $N$ cross-sectional units indexed by $i$. Combining all $\vz_{i}$ we obtain a two-dimensional data array of panel data (or longitudinal data). In empirical research it is common to investigate whether $\vz_{i}$ can be regarded as independent over $i$ in order to select a model that can properly characterize the statistical properties of the data. In particular, researchers might be interested in the  statistical hypothesis
\begin{equation}
H_{0}: \vz_{i}\bot \vz_{j},\quad \text{for all }i,j=1,\ldots,N,
\end{equation}
where we use the $\bot$ notation to denote independence. Most often $\vz_{1}, \ldots, \vz_N$ contain residuals obtained from a regression model that does not allow for cross-section dependence, e.g. an entity fixed effects model or linear regression.

By far the most widely used test for cross-section dependence (correlation, to be precise) is the CD test of \citet{CDtest2004, doi:10.1080/07474938.2014.956623}, which is based on a simple rescaled sum of all pairwise cross-section correlation coefficients, formally denoted
\begin{equation}
\label{eq::cdtest}
CD = \sqrt{\frac{2T}{N(N-1)}}\sum_{i=2}^{N} \sum_{j=1}^{i-1} \widehat{\rho}_{ij}=\sqrt{\frac{TN(N-1)}{2}}\overline{\widehat{\rho}}.
\end{equation}	
Here,
\begin{align}
\label{eq:CScorrcoef}
\widehat{\rho}_{ij}=\frac{T^{-1}\sum_{t=1}^{T}(z_{i,t}-\overline{z}_{i})(z_{jt}-\overline{z}_{j})}{\sqrt{T^{-1}\sum_{t=1}^{T}(z_{i,t}-\overline{z}_{i})^{2}}\sqrt{T^{-1}\sum_{t=1}^{T}(z_{jt}-\overline{z}_{j})^{2}}}
\end{align}
is the pairwise sample correlation coefficient between units $i$ and $j$. Obviously, computing the CD test statistic involves obtaining $N(N-1)/2$ parameter estimates, each of which converges to the true parameter value at rate $\sqrt{T}$ only. These circumstances are reminiscent of the panel data setup considered in e.g. \citet{phillipsMoon1999}, \citet{Hahn2002}, where estimation of many (incidental) parameters in linear regression models turns out to have distributional effects on the asymptotic properties of common parameters. That is, they cause the incidental parameter problem.
By contrast, the asymptotic distribution of the CD test statistic is unaffected by the estimation of all $N(N-1)/2$ cross-section correlation coefficients involved in its construction. In fact, applied to the residuals of a linear regression model with strictly exogenous regressors, and individual specific means, the CD test statistic is asymptotically normal as long as $N,T\to\infty$ (see \citealp{PesaranBook}, Theorem 2), under $H_{0}$ of independence (or even limited local dependence). However, as shown below, this result does not hold when the model specification includes period-specific parameters.


\subsection{Heuristic discussion of the main result}
\label{section:heuristics_additive}
Consider a linear model where cross-section dependence is due to time fixed effects, so that both unobserved heterogeneity over cross-sections and time enters the model additively. That is the relation between the $T\times 1$ vector $\vy_i$ and the $T \times m$ matrix $\mX_i$ is formally denoted
\begin{equation}
\label{eq:themodel_additive}
\vy_i = \mX_i \vbeta +\vtau+ \viota \mu_i +  \vepsi_i, \quad i = 1, \ldots, N.
\end{equation}
Here, $\mu_i$ and $\vtau$ denote an entity-specific intercept and a $T \times 1$ vector of time-specific common parameters $\vtau$, respectively. $\vepsi_i$ is a vector of idiosyncratic error terms, independent across cross-sectional units. In the example of a Difference-in-Differences framework, $\vtau$ is the common trend affecting both treated and untreated individuals while the treatment indicator as well as other covariates are contained in $\mX_{i}$.
For the sake of simplicity, we assume $\vbeta$ and $\mu_i$ to be known, so that $\vbeta=\vzeros$ and $\mu_i = 0 \; \forall i$ can be imposed without loss of generality. This highly restrictive assumption leaves the leading terms in the analysis below unaffected and is hence innocuous for the expository purpose of this section.
Model \eqref{eq:themodel_additive} reduces to
\begin{equation}
\label{eq:themodel_stylized}
y_{i,t}=\tau_{t} + \epsi_{i,t}, \quad i= 1, \ldots, N,\quad t = 1, \ldots, T.
\end{equation}
We additionally assume that the variance of $\epsi_{i,t}$ is known and fixed to $\sigma_{i}^{2}=1$. In this setup, $y_{i,t}$ are clearly cross-sectionally dependent because of $\tau_{t}$; however we are interested in testing whether $\epsi_{i,t}$ are cross-sectionally uncorrelated. Given that $\epsi_{i,t}$ are unobserved, the common effects $\tau_{t}$ need to be estimated. The most natural approach is to estimate $\tau_{t}$ by OLS so that $ \widehat{\tau}_{t} = \overline{y}_{t} = \frac{1}{N}\sum_{i=1}^{N}y_{i,t}$ .
Using the residuals
\begin{equation}
\label{eq:resids_additive}
	\widehat{\epsi}_{i,t} = y_{i,t}-\widehat{\tau}_{t},
\end{equation}
the CD test statistic is given by
\begin{align}
\label{eq:CD_simplecase}
	CD &= \sqrt{\frac{2}{TN(N-1)}}\sum_{t=1}^{T}\sum_{i=2}^{N} \sum_{j=1}^{i-1}\widehat{\epsi}_{i,t} \widehat{\epsi}_{j,t} \notag \\
	&= \sqrt{\frac{1}{2TN(N-1)}}\sum_{t=1}^{T} \left(\sum_{i=1}^{N}\widehat{\epsi}_{i,t}\right)^2 - \sqrt{\frac{1}{2TN(N-1)}}\sum_{t=1}^{T}\sum_{i=1}^{N}\widehat{\epsi}_{i,t}^2.
\end{align}
Given this definition,
\begin{align}
	\sum_{i=1}^{N}\widehat{\epsi}_{i,t} = \sum_{i=1}^{N}y_{i,t}  - \sum_{i=1}^{N}y_{i,t} = 0,
\end{align}
implying that the first term in \eqref{eq:CD_simplecase} cancels out. The expression $\sum_{t=1}^{T}\sum_{i=1}^{N}\widehat{\epsi}_{i,t}^2$ in the second term is nothing more than the Sum of Squared Residuals (SSR) of the estimated model, a statistic that is of order $\largeO_P(NT)$. Consequently,
\begin{align}
	CD = \largeO_P(\sqrt{T}),
\end{align}
even though the error-terms $\epsi_{i,t}$ are cross-sectionally independent. Hence, the procedure commonly used by practitioners is prone to finding spurious cross-sectional dependence in the data.

The results shown for a model with additive time effects and homoscedastic errors are not coincidental. In fact, as we show later in this paper, they carry over to a more complex model where cross-section dependence is generated by a multifactor error structure and/or the errors are heteroscedastic.

\section{Asymptotic Results}
\label{section::asymptotics}
In this section we provide formal asymptotic results for CD statistic based on the residuals obtained from a 2WFE model or a model with multifactor error structure. For the sake of simplicity, we continue to assume that $\vbeta$ is known throughout this section. If the vector of slope coefficients $\vbeta$ were unknown, its estimator would converge (under correct model specification) to the true parameter value at the conventional rate $\sqrt{NT}$. Hence, convergence to the true parameter values is fast enough to ensure that any terms present in CD test statistic involving $\widehat{\vbeta}- \vbeta_{0}$ are asymptotically negligible. We shall not attempt to prove this claim formally but rely on Monte Carlo results of Section \ref{section::MC} to support this conjecture. By contrast, the documented effects on the first two moments of the CD test statistic resulted from the fact that estimates of common, period-specific parameters converge at a rate \emph{slower} than $\sqrt{NT}$.

\subsection{Two-way error component structure}
\label{ssection::asymptotics_additive}
As in Section \ref{section:heuristics_additive} we assume that the true model is given by equation  \eqref{eq:themodel_additive}. Additionally, we make the following assumptions on the model errors.
\begin{assumption}[Errors]\
	\label{ass:errorsAddHetero}
\begin{enumerate}	
\item
Let $\epsi_{i,t}=\sigma_{i}\eta_{i,t}$ where $\eta_{i,t}$ is independently and identically distributed across both $i$ and $t$ with $\E\left[\eta_{i,t}\right]=0$, $\E\left[ \eta_{i,t}^2\right]= 1$ and $\E\left[ |\eta_{i,t}|^8\right]< M$ for some $M < \infty$.
\item $\sigma_{i}$ is defined over an interval $[\delta;M]$ with $0<\delta<M<\infty$. It is independently and identically distributed across $i$, and $\sigma_{i}$ is independent of $\eta_{i,t}$ for all $i$ and $t$.
\end{enumerate}	
\end{assumption}
For technical reasons, we assume that all stochastic variables in this paper have finite eighth moments. This is a sufficient condition, which facilitates proving joint convergence of the test statistics considered in this paper, see also \citet{JAE:JAE2496}. Assumption \eqref{ass:errorsAddHetero} is general enough to cover several models of conditional heteroscedasticity in $\epsi_{i,t}$. Alternatively this assumption can be formulated in terms of unconditional variances, where $\sigma_{i}$ are fixed numbers. However, in addition to having severe conceptual shortcomings (as discussed in \citealp{ECTA:ECTA1599}) such an approach would lead to incorrect conclusions concerning the power of a CD test statistic for remaining cross-section correlation. See Section \ref{sec:suppl_theory} in the Supplementary Appendix for a detailed discussion on power. 

For example, natural examples for $\sigma_{i}$ are either a standard exponential skedastic function
\begin{equation}
\sigma_{i}=\exp(\alpha+\gamma\mu_{i}),
\end{equation}
or a location-scale model with
\begin{equation}
\sigma_{i}=\alpha+\gamma\mu_{i}.
\end{equation}
Both satisfy the required restrictions, as long as $\mu_{i}$ has a bounded support.
Furthermore, we denote different cross-sectional averages of $\sigma_{i}$ by $\overline{\sigma^{k}}=N^{-1}\sum_{i=1}^{N}\sigma_{i}^{k}$ for $k\in \mathbb{Z}$, and the corresponding population quantities by $\E[\sigma_{i}^{k}]$. Assumption \ref{ass:errorsAddHetero} guarantees that these quantities are well defined for all finite $k$.

Using these definitions the CD test statistic obtained from a model with unknown period-specific effects can be characterized as follows:
\begin{theorem}
\label{theorem::additiveHetero}
	Suppose that $\vbeta$ is known. Under Assumption \ref{ass:errorsAddHetero},
	\begin{align}
\label{eq::theorem_additive_hetero}
		CD &= \sqrt{\frac{2}{TN(N-1)}}\sum_{i=2}^{N} \sum_{j=1}^{i-1}\sum_{t=1}^{T}\epsi_{i,t}\epsi_{j,t}\left(\sigma_{i}^{-1}-\overline{\sigma^{-1}}\right)\left(\sigma_{j}^{-1}-\overline{\sigma^{-1}}\right) +\sqrt{T} \Xi \\
		&+ \largeO_{P}(\sqrt{N}T^{-1})+ \largeO_{P}(T^{-1/2})+ \largeO_{P}(N^{-1/2})+ \largeO_{P}(N^{-1}\sqrt{T}), \notag
	\end{align}
	where
\begin{align}
\Xi &= \sqrt{\frac{N}{2(N-1)}} \frac{1}{NT}\sum_{i=1}^{N}\sum_{t=1}^{T} \epsi_{i,t}^{2}\left(\left(\overline{\sigma^{-1}}\right)^{2} - 2\overline{\sigma^{-1}}\sigma_{i}^{-1} \right).
\end{align}
Furthermore, let $\Omega = \left(1-2\E[\sigma_{i}]\E[\sigma_{i}^{-1}] + \E[\sigma_{i}^{2}](\E[\sigma_{i}^{-1}])^{2}\right)^2$. Then,
\begin{equation}
\label{eq:CDthm1}
CD - \sqrt{T}\Xi \dto N\left(0, \Omega \right)
\end{equation}
as $N,T\to \infty$ jointly provided that $\sqrt{T}N^{-1} \to 0$ and $\sqrt{N}T^{-1}\to 0$.
\end{theorem}
Theorem \ref{theorem::additiveHetero} shows that the inclusion of time fixed effects into the model specification has an asymptotically non-negligible effect on the first two moments of the CD test statistic.
Via expansion with appropriate functions of $\{\sigma_i\}_{i=1}^{N}$, it can be shown that $\Xi$ indicates the presence of a deterministic bias of order $\sqrt{T}$. Abstracting from this bias, \eqref{eq:CDthm1} is dominated by an expression that reflects the CD test statistic obtained from the true model errors but imposing an incorrect normalization. 

Restrictions on the expansion rates of $N$ and $T$, i.e. $\sqrt{N}T^{-1} \to 0$ and$\sqrt{T}N^{-1} \to 0$,  are satisfied if one considers diagonal asymptotic expansion schemes as in \citet{fernandez_weidner_2014}, where $N T^{-1}\to \kappa^{2}$ with $ \kappa\in (0;\infty)$. This way $N/T + T/N$ remains a bounded constant asymptotically. The intuition for the need to restriction both $\sqrt{T}N^{-1} \to 0$ and $\sqrt{N}T^{-1} \to 0$ comes from the fact that the feasible CD statistic is a non-linear function in both $\widehat{\sigma}_{i}$ \emph{and} $\widehat{\tau}_{t}$.  Hence the rate restrictions derived for general non-linear panel data models, as studied by \citet{fernandez_weidner_2014}, are naturally applicable.

The behavior of the CD test statistic in a model with time fixed effects can be seen as an example of the incidental parameters problem (IPP) of \citet{neyman1948}, since the bias of the CD test statistic is due to the estimation of $T$ period-specific intercepts $\tau_{t}$. Interestingly, in the context of estimating linear dynamic panel data models, estimation of the time effects $\tau_t$ does not introduce any asymptotic bias into the FE estimator with strictly or weakly exogenous regressors \citep[see e.g.][]{hahnmoon2006}. In non-linear models, estimation of the time effects $\tau_{t}$ affects the asymptotic mean of the estimator for slope parameters associated with explanatory variables \citep[see e.g.][]{fernandez_weidner_2014} with the corresponding bias being proportional to $\sqrt{T N^{-1}}$. In this sense, our result adds new insights into the literature in that it highlights a scenario where the inclusion of time fixed effects into a \emph{linear} model has non-standard implications for the asymptotic distribution of the statistic of interest.

The results of Theorem \ref{theorem::additiveHetero} suggest that asymptotically standard normal inference can be recovered by bias-correcting and rescaling \eqref{eq::theorem_additive_hetero}. Before considering this remedy in Section \ref{section::solutions}, it is important to emphasize the following special case:

\begin{corollary}
	\label{corollary::homosigmas}
	Under Assumption \ref{ass:errorsAddHetero} and given $P(\sigma_{i}=\sigma)=1$,
	\begin{align}
		CD &= -\sqrt{\left(T-\frac{T}{N}\right)/2} +\largeO_{P}(R_{N,T}),
	\end{align}
	where $R_{N,T}=(N^{-1/2}\vee T^{-1/2}\vee N^{-1}\sqrt{T}\vee T^{-1}\sqrt{N})$.
\end{corollary}

Corollary \ref{corollary::homosigmas} provides more intuition about the approximate value of the bias term $\sqrt{T}\Xi$, suggesting that it should be reasonably close to $-\sqrt{T/2}$. More importantly, the result reveals that the leading stochastic component in the CD test statistic, the first term on the right-hand side of equation \eqref{eq::theorem_additive_hetero}, cancels out when error variances are homogeneous across $i$. Instead, random variation around the bias term $-\sqrt{T/2}$ is of order $\largeO_P(N^{-1/2}\wedge T^{-1/2})$, rendering the distribution of a modified version of the test statistic asymptotically degenerate.

The special case of homogeneous error variances hence entails consequences for the CD test statistic that are qualitatively different from those in the more general case where $\sigma_{i}$ may differ across $i$. Again, it would be possible to allow for asymptotically normal inference by adequately rescaling the modified test statistic. However, the resulting statistic would be of little practical use since the main source of variation is not related to error covariances across cross-sections, but is simply driven by variance of the idiosyncratic components.

\subsection{Multifactor error structure}
\label{ssection::asymptotics_multifactor}
Following \cite{ECTA:ECTA692} we consider a model where cross-section dependence enters the model via a multifactor error structure. This model is described by the $T\times 1$ vector $\vy_i$ and the $T\times m$ matrix $\mX_i$ defined as
\begin{align}
\label{eq:themodelFactor}
\vy_i &= \mX_i \vbeta + \mF \vlambda_{i} + \vepsi_i  \\
\mX_i &= \mF\mLambda_i + \mE_i,
\end{align}
where $\mF = \left[\vf_1, \vf_2, \ldots, \vf_T \right]'$ denotes a $T\times r$ matrix of unobserved common factors and where $\mE_i = \left[ \ve_{i,1}, \ve_{i,2}, \ldots, \ve_{i,T} \right]'$ is a $T \times m$ matrix of idiosyncratic variation in $\mX_{i}$. The most popular estimator designed to estimate the parameter vector $\vbeta$ in this specific model is the Common Correlated Effects (CCE) estimator of \citet{ECTA:ECTA692} which amounts to augmenting a linear regression model with cross-section averages of potentially all variables available to the researcher in order to account for the effect of unobservable common factors. In a model with homogeneous slope coefficients, we have
\begin{align}
\widehat{\vbeta}^{CCE} &= \left(\sum_{i=1}^N \mX_{i}^{\prime} \mM_{\widehat{\mF}}\mX_i\right)^{-1}\left(\sum_{i=1}^N \mX_{i}^{\prime} \mM_{\widehat{\mF}}\vy_i\right)
\intertext{where}
\label{eq:CAfactor}
\widehat{\mF} &= \begin{bmatrix} \overline{\vy}, & \overline{\mX} \end{bmatrix}
= \left(\mF  \begin{bmatrix} \overline{\vlambda}, & \overline{\mLambda} \end{bmatrix} + \begin{bmatrix} \overline{\vepsi}, & \overline{\mE} \end{bmatrix}\right) \begin{bmatrix} 1, & \vzeros \\ \vbeta, & \I  \end{bmatrix}
= \mF \overline{\mC} + \overline{\mU}.
\end{align}
Here $\overline{\mC}$ ($\overline{\mU}$) are defined implicitly in terms of $\vbeta$ and the corresponding average factor loadings (average errors). While the CCE estimator is agnostic about the true number of factors that affect the data, it has been shown that consistent estimation of the parameters of interest requires that the number of cross-section averages is at least as large as the true number of factors that drive the data \citep{Westerlund2013247}.

Given that the factor estimator $\widehat{\mF}$, defined in Eq. \eqref{eq:CAfactor}, is constructed based on observed data, restrictions on the DGP of both $y_{i,t}$ and $\vx_{i,t}$ need to be imposed.
\begin{assumption}\
	\label{ass:errors}
	\begin{enumerate}
\item
Let $\epsi_{i,t}=\sigma_{i}\eta_{i,t}$ where $\eta_{i,t}$ is independently and identically distributed across both $i$ and $t$ with $\E\left[\eta_{i,t}\right]=0$, $\E\left[ \eta_{i,t}^2\right]= 1$ and $\E\left[ |\eta_{i,t}|^8\right]< M$ for some $M < \infty$.
\item $\sigma_{i}$ is defined over an interval $[\delta;M]$ with $0<\delta<M<\infty$. It is independently and identically distributed across $i$, and $\sigma_{i}$ is independent of $\eta_{i,t}$ for all $i$ and $t$.
		\item The $m\times 1$ random vector $\ve_{i,t}$ is independently distributed across both $i$ and $t$ with $\E\left[\ve_{i,t}\right]= \vzeros$, $\E\left[ \ve_{i,t}\ve_{i,t}'\right]= \mSigma$ with the latter being a positive definite matrix and $\E\left[ ||\ve_{i,t}||^8\right]< M$.
	\end{enumerate}
\end{assumption}		
\begin{assumption}
	\label{ass:factors}
	$\vf_t$ is a covariance stationary $r\times 1$ random vector with positive definite covariance matrix $\mSigma_{\mF}$, absolutely summable autocovariances
	and
	$\E\left[ \| \vf_t \|^4 \right]< M$ .
\end{assumption}	
\begin{assumption}
	\label{ass:loadings}
	$\vlambda_{i}$ is iid over $i$ with $\E[\vlambda_{i}]=\vmu_{\vlambda}$ and $\E\left[ || \vlambda_{i} ||^4\right]< M$. Furthermore, $\mLambda_i$ is iid over $i$ with $\E\left[\mLambda_i\right] = \vmu_{\mLambda}$ and $\E\left[ || \mLambda_i ||^4\right]< M$.
\end{assumption}	
\begin{assumption}
	\label{ass:indep}	
	$\vf_t$, $\{\vlambda_{i}, \mLambda_i, \sigma_{i} \}$, $\eta_{i',t'}$ and $\ve_{i'',t''}$ are mutually independent for all $i,i',i'',t,t'$ and $t''$
\end{assumption}	
\begin{assumption}
	\label{ass:rank} $\rk\left(\left[ \vmu_{\vlambda} , \vmu_{\mLambda} \right]\right) = r = m+1$.
\end{assumption}
This set of assumptions above are a slightly more restrictive version of the framework considered in \citet{ECTA:ECTA692}. For example, the assumption of common $\mSigma$ can be straightforwardly relaxed. However, unlike $\sigma_{i}$, $\mSigma$ plays no major role for asymptotic results of this paper. Most importantly, we rule out the presence of serial correlation as this is in line with the assumptions made for the CD test to work. Furthermore, we restrict ourselves to a classical panel data regression model instead of considering heterogeneous slope coefficients. Moreover, any dependence between $\epsi_{i,t}$ and $\ve_{i,t}$ is assumed away in order to allow for a tractable proof of the main theoretical result in this paper. Lastly, the fact that we assume $\rk\left( \left[\vmu_{\vlambda} , \vmu_{\mLambda}\right] \right) = r = m+1$ to hold, suggest that our we consider an ideal setup where none of the rank condition related problems documented in \citet{KaraReeseWesterlund2015}, or \citet{JuodisEtal2017} apply.

In analogy with the previous sections, we assume that $\vbeta_{0}$ is known. Thus we disregard the estimation error $\widehat{\vbeta}^{CCE} - \vbeta_{0}$, which is generally of order $\largeO_{P}((NT)^{-1/2}) + \largeO_P(N^{-1})$, see e.g.  \citet{ECTA:ECTA692} and \citet{JuodisEtal2017} .

We begin our asymptotic analysis, by noting that in the model with assumed (known) homogeneous $\sigma$ the result follows directly as in the model with time effects only. In particular, while it is not generally emphasized in the CCE literature, the residuals from CCE estimation which are formally given by
\begin{align}
\label{eq::residuals_CCE}
\widehat{\vepsi}_i &= \mM_{\widehat{\mF}}\vy_{i}=\vepsi_i - \mP_{\widehat{\mF}}\vepsi_i - (\mP_{\widehat{\mF}} - \mP_{\mF}) \mF \vlambda_{i}
\end{align}
satisfy
\begin{equation}
\sum_{i=1}^{N}\widehat{\epsi}_{i,t}=0.
\end{equation}
In that respect, the standard 2WFE estimator is similar to the CCE estimator. More formally we formulate the following result
\begin{proposition}
\label{theorem::CCE_sigma}
	Under Assumptions \ref{ass:errors}--\ref{ass:rank} and $P(\sigma_{i}=\sigma)=1$:
	\begin{align}
		CD &= -\sqrt{\left(T-\frac{T}{N}\right)/2} +\largeO_{P}(R_{N,T}),
	\end{align}
	where $R_{N,T}=(N^{-1/2}\vee T^{-1/2}\vee N^{-1}\sqrt{T}\vee T^{-1}\sqrt{N})$.
\end{proposition}
Proposition \ref{theorem::CCE_sigma} shows that the result we derived previously for the 2WFE estimator in Corollary \ref{corollary::homosigmas} continues to hold for models with a factor error structure, as long as $N\to\infty$ and $T\to\infty$. 

It is worth mentioning that the above order effect is only valid if $\widehat{\mF}$ contains cross-sectional averages of the regressand as well as all regressors. If either of those variables is omitted (without affecting the rank condition in Assumption \ref{ass:rank}), the zero mean residual condition in \eqref{eq::residuals_CCE} is violated, and consequently the $CD$ test will have a $\largeO_{P}(1)$ term. However, this result is of limited empirical importance as in most cases researchers include all available cross-sectional averages. While this practice ensures that the estimator is invariant to $\vbeta_{0}$, inclusion of too many cross-sectional averages can potentially have detrimental effects on the asymptotic properties of the estimator, see the corresponding discussion in \citet{JuodisEtal2017} and \citet{JuodisRCCE2020}.

In the homoscedastic case two-way fixed effects and multifactor error models have similar asymptotic effects on CD test statistic. The conclusions of the next theorem, which is the main result of this paper, indicate that this equivalence does not hold in the heteroscedastic case. In order to proceed, we introduce some useful notation. For $t=1,\ldots,T$, let $\vu_{i,t}=[\epsi_{i,t}+\vbeta'\ve_{i,t},\ve_{i,t}']'$, and 
\begin{equation}
\label{eq::influence}
	\left(\overline{\mC}'\right)^{-1}\widehat{\vf}_t -  \vf_t  = \frac{1}{N}\sum_{i=1}^{N}\vpsi_{i,t},
\end{equation}
where $\vpsi_{i,t}=\left(\overline{\mC}'\right)^{-1}\vu_{i,t}$ is the influence function of the corresponding factor estimator, in this case cross-section averages of $y_{i,t}$ and $\vx_{i,t}$. Generally, the influence function depends on the joint process $\left[y_{i,t},\vx_{i,t}\right]'$, as long as all observed variables are used to form cross-sectional averages. Equipped with this notation we formulate the main result of this paper.
\begin{theorem}
\label{theorem::CCE}
	Suppose that $\vbeta$ is known. Under Assumptions \ref{ass:errors}--\ref{ass:rank},
	\begin{align}
CD &= \sqrt{\frac{2}{TN(N-1)}}\sum_{i=2}^{N}\sum_{j=1}^{i-1}\sum_{t=1}^{T}\xi_{iN,t}\xi_{jN,t} + \sqrt{T} \Phi_1 -2\sqrt{T} \Phi_2+ \largeO_{P}(R_{N,T}),
	\end{align}
	where
\begin{equation}
\xi_{iN,t}=\sigma_{i}^{-1}\epsi_{i,t}-\left(\frac{1}{N}\sum_{i=1}^{N}\sigma_{i}^{-1}\vlambda_{i}\right)'\vpsi_{i,t},
\end{equation}
and
	\begin{align}
		\Phi_1 &= \sqrt{\frac{N}{2(N-1)}} \left(\frac{1}{N}\sum_{i=1}^{N}\sigma_{i}^{-1}\vlambda_{i}\right)'  \left( \frac{1}{NT} \sum_{i=1}^{N}\sum_{t=1}^{T}\vpsi_{i,t}\vpsi_{i,t}'\right) \left(\frac{1}{N}\sum_{i=1}^{N}\sigma_{i}^{-1}\vlambda_{i}\right), \\
		\Phi_2 &= \sqrt{\frac{N}{2(N-1)}}  \left(\frac{1}{N}\sum_{i=1}^{N}\sigma_{i}^{-1}\vlambda_{i}\right)'  \left(\frac{1}{NT} \sum_{i=1}^{N}\sum_{t=1}^{T}\vpsi_{i,t}\sigma_{i}^{-1}\epsi_{i,t} \right).
	\end{align}
\end{theorem}
Here, as previously, $R_{N,T}=(N^{-1/2}\vee T^{-1/2}\vee N^{-1}\sqrt{T} \vee T^{-1}\sqrt{N})$. In line with all previous results, the CD test statistic in this setup has two diverging components. However, unlike all previous results, in particular Theorem \ref{theorem::additiveHetero}, these bias terms are not solely non-linear functions of $\sigma_{i}$. Instead, they are also influenced by the first (rescaled) moments of factor loadings in $y_{i,t}$ and $\vx_{i,t}$, as well as corresponding variances of the idiosyncratic components in $\vx_{i,t}$. Thus the influence function $\vpsi_{i,t}$ directly alters distributional properties of the CD statistic. The result is qualitatively similar to any parametric two-step estimation procedure with a plug-in first-step estimator. Thus by including cross-sectional averages as factor proxies, one is implicitly testing that both $\epsi_{i,t}$ and $\vpsi_{i,t}$ are jointly cross-sectionally uncorrelated. Notice, that while $\epsi_{i,t}$ is usually assumed to be uncorrelated over $t$ (e.g. for models with pre-determined regressors), the same is not true for $\vu_{i,t}$, e.g. if $x_{i,t}=y_{i,t-1}$. Thus if Assumption \ref{ass:errors} is appropriately relaxed, then $\xi_{iN,t}$ will be serially correlated.
\begin{remark}
One can easily see that expressions for $\xi_{iN,t},\Phi_{1},\Phi_{2}$ are the same if one assumes that $\lambda_{i}$ is known . Thus it is only the influence function of common factors, and not those of factor loadings, that has an impact on the asymptotic properties of the test statistic. This conclusion is the same as in the two-way error components model.
\end{remark}
As alluded in Proposition \ref{theorem::CCE_sigma}, the leading $\largeO_{P}(1)$ term is degenerate if $\sigma_{i}=\sigma$, as in this case:
\begin{equation}
\sum_{i=1}^{N}\xi_{iN,t}=0,\quad \forall t=1,\ldots,T.
\end{equation}
One can easily see that the expressions for $\Phi_{1}$ and $\Phi_{2}$ derived for CCE coincide with the corresponding terms of Theorem \ref{theorem::additiveHetero} (up to a negligible remainder term), upon setting $\vlambda_{i}=1$ (thus $\overline{\mC}^{-1}=1$), and $\vpsi_{i,t}=\epsi_{i,t}$. This situation is equivalent to using $\widehat{f}_{t}=\overline{y}_{t}-\overline{\vx}_{t}'\vbeta_{0}$ as factor proxies, e.g. as suggested by \citet{JAE:JAE2535} in the context of predictability testing with cross-section dependence.

Recall that result in Theorem \ref{theorem::CCE} considers only the CCE setup where the number of observable factor proxies equals the number of the true factors. If Assumption \ref{ass:rank} is relaxed and there are more observables than factors, then following \citet{KaraReeseWesterlund2015} and \citet{JuodisEtal2017}, one can show that the expressions for $\Phi_1$ and $\Phi_2$ will contain additional terms related to this discrepancy. In particular, $\Phi_{1}$, $\Phi_{2}$, and $\xi_{iN,t}$ are functions of an unknown rotation matrix, which cannot be consistently estimated from the data. However, as already emphasized by \citet{JuodisEtal2017}, this problem can be circumvented via the use of a non-parametric bias-correction method.  We will come back to this issue in Section \ref{section:BC}.


\section{Re-establishing standard normal inference}

\label{section::solutions}
The diverging bias in the CD test statistic applied to residual from models with common, period-specific parameters is fundamentally problematic for its use in the context of testing for remaining cross-section correlation. Still, the literature review in Section \ref{section::introduction} suggests that the underlying question of correct model specification is of high relevance for empirical researchers. For this reason it is relevant not to discard completely the CD test statistic but instead to discuss possible modifications aimed at ensuring an asymptotically standard normal inference under null hypothesis. Thus, methods aimed at addressing the IPP detailed in Sections \ref{section:testing} and \ref{section::asymptotics} need to be considered.

\subsection{Analytic bias correction}
\label{section:BC}
Given that parametric expressions for the bias of $CD$ have been derived in Theorems \ref{theorem::additiveHetero} and \ref{theorem::CCE}, analytic bias correction is feasible. As detailed in Section \ref{section:CDBC_details} in the Supplementary Appendix, estimates of the unknown model components that constitute either $\Xi$ or $\Phi_1$ and $\Phi_2$ can be obtained in order to eliminate the diverging component in $CD$. In addition, Propositions \ref{prop:CD_2WFE_H1} and \ref{prop:CD_CCE_H1} in the Supplementary Appendix suggest that these estimated bias terms would eliminate equivalent terms $\Xi_{\mathbb{H}_1}$ or $\Phi_{\mathbb{H}_1,1}$ and $\Phi_{\mathbb{H}_1,2}$ under the alternative hypothesis without necessarily reducing the rate at which $CD$ diverges.

Still, problems with its implementation in practice lead us to discard analytic bias correction and to consider it merely as a benchmark approach when evaluating our weighted CD test statistic, as introduced below, in Monte Carlo experiments. In particular, the construction of plug-in estimates of $\Phi_1$ and $\Phi_2$ proves to be tedious since both terms are constructed from estimates of $\sigma_{1}^2, \ldots, \sigma_{N}^{2}$, $\vlambda_1, \ldots, \vlambda_N$ and $\vbeta$. Additionally, estimation of the cross-section sum $N^{-1} \sum_{i=1}^N  \sigma_{i}^{-1}  \vlambda_{i}$ poses a trade-off between generality and accuracy of the plug-in estimate. To be specific, independence between factor loadings $\vlambda_i$ and error variances $\sigma_{i}^2$ needs to be imposed to ensure that estimation error in $\widehat{\vlambda}_i$ and $\widehat{\sigma}_{i}^2$ does not dominate the small-sample distribution of the bias-corrected version of $CD$. 

The need to impose further assumptions in order to improve the accuracy of the bias estimates is effectively a consequence of the fact that the unknown bias itself diverges at rate $\sqrt{T}$. This means that any error in the estimation of the two aforementioned components is scaled by a factor that increases as $N,T \to \infty$. Still, despite improving the approximation of plug-in estimates for $\Phi_1$ and $\Phi_2$ the assumption of independence between factor loadings and error variances can lead to size distortions in DGPs where it is violated, see Section \ref{section::MC} for a numerical illustration of the problem.

\subsection{Weighted cross-section covariances}
\label{section:weightedTest}
The method for bias correction favored in this article is to construct a CD test statistic from estimated cross-section covariances that are weighted with individual-specific, random draws $X$ from a Rademacher distribution, i.e. $P(X=1)=P(X=-1)=0.5$. This approach is based on noticing that the cross-section correlation estimator $\widehat{\rho}_{ij}$, as defined in \eqref{eq:CScorrcoef} with $z_{i,t} = \widehat{\varepsilon}_{i,t}$,
is merely a unit-specific weighting of $\widehat{\cov} \left[\widehat{\epsi}_{i,t}\widehat{\epsi}_{j,t}\right]$ with advantageous properties: Under the assumptions made in \citet{CDtest2004}, it ensures that the CD test statistic has unit variance under the null hypothesis of no cross-section correlation, allowing for standard normal inference without the need to obtain a variance estimate. However, the case is different in models with unknown, period-specific parameters $\tau_{t}$ or $\vf_{t}$. As shown in Theorems \ref{theorem::additiveHetero} and \ref{theorem::CCE}, studentization of the model residuals fails to ensure a test statistic with a variance of one, even when the true error variances are known. This justifies the use of an alternative weighting scheme which we construct as a weighted average of individual specific covariances,
\begin{align}
\label{eq:CD_simplecase2}
\sqrt{\frac{2}{TN(N-1)}}\sum_{t=1}^{T}\sum_{i=2}^{N} \sum_{j=1}^{i-1}(w_{i}\widehat{\epsi}_{i,t})( w_{j}\widehat{\epsi}_{j,t}).
\end{align}
for some set of weights $w_1, \ldots, w_N$. Notice that Eq. \eqref{eq:CD_simplecase2} 
would coincide with the CD test statistic of \citet{CDtest2004} for $w_i = \widehat{\sigma}_{i}^{-1} \; \forall i$.

This is not the case we consider here. Instead we make the following additional assumption:
\begin{assumption}[Weights]
	\label{ass:weights}
	$w_{1}, \ldots, w_{N}$ are identically and independently Rademacher distributed. Furthermore, $w_{i}$ is independent of $\vlambda_{i}$, $\sigma_{i}$, $\eta_{i,t}$ and $\ve_{i,t}$ for all $i$ and $t$.
\end{assumption}
Rademacher distributed weights amount to random sample splitting, a method for breaking dependence that is not new to econometrics. For example, \cite[see p. 358]{altonji1996} already considered it an \emph{old} concept. Its effects are most obvious in a simple modification of Theorem \ref{theorem::CCE} where we replace $\sigma_{i}^{-1}$ with $w_i$. Under Assumption \ref{ass:weights}, the expected value of $N^{-1}\sum_{i=1}^N \vlambda_{i} w_{i}$ can be conveniently split into the expectations of its two constituents. Given the zero expected value of $w_i$, an appropriate LLN applies and leads the cross-section average of $ \vlambda_{i} w_{i}$ to converge to zero. In an even simpler form the same reasoning applies to the cross-section averages involving $w_i$ in an adapted version of Theorem \ref{theorem::additiveHetero}. In both cases, this reduces the order of the leading bias components $\Xi$, $\Phi_{1}$ and $\Phi_{2}$ by a factor of $N$ and entails asymptotic unbiasedness of the weighted CD test statistic, subject to the restriction $\sqrt{T} N^{-1} \to 0$. Additional rescaling with the asymptotic variance of expression \eqref{eq:CD_simplecase2} results in a weighted CD test statistic which, as formally stated in Theorem \ref{thm:CD_W_combined} below, converges to a standard normal distribution.
\begin{theorem}
\label{thm:CD_W_combined}
	Consider the weighted CD test statistic
	\begin{align}
	CD_W &= \left( \frac{1}{NT} \sum_{i=1}^N \sum_{t=1}^T \widehat{\epsi}_{i,t}^2 w_i^2 \right)^{-1}  \left(\sqrt{\frac{2}{TN(N-1)}}\sum_{t=1}^{T}\sum_{i=2}^{N} \sum_{j=1}^{i-1}w_{i}\widehat{\epsi}_{i,t} w_{j}\widehat{\epsi}_{j,t} \right)
	\end{align}
	and assume that either of the following two points hold.
	\begin{enumerate}
		\item The data are generated by the time fixed effects model \eqref{eq:themodel_additive} such that Assumption \ref{ass:errorsAddHetero} holds. $\widehat{\epsi}_{i,t}$ is given by \eqref{eq:resids_additive}.
		\item The data are generated by the latent common factor model \eqref{eq:themodelFactor} such that Assumptions \ref{ass:errors}-\ref{ass:rank} hold. $\widehat{\epsi}_{i,t}$ is defined by \eqref{eq::residuals_CCE}.
	\end{enumerate}	
	\noindent Under either of the two sets of assumptions above as well as Assumption \ref{ass:weights}, it holds that
\begin{equation}
	CD_{W}\dto N(0,1),
\end{equation}
as $N,T\to \infty$ jointly subject to the restriction $\sqrt{T}N^{-1} \to 0$.
\end{theorem}
As stated by Theorem \ref{thm:CD_W_combined}, the use of independent Rademacher distributed weights, analogous to weights drawn from many other distributions with zero mean, re-establishes asymptotic standard normal inference under the null hypothesis of the CD test statistic. However, asymptotic unbiasedness of $CD_W$ comes at the cost of power. More specifically, our approach to bias correction centers the leading components of $CD_W$ around zero, irrespective of whether cross-section correlation in the data is completely controlled for or not. As a consequence, only increases in the variance of $CD_W$ under its alternative hypothesis lead to power against the null hypothesis of this test. We further discuss this point in Section \ref{sec:suppl_theory} of the Supplementary Appendix.

To improve the power properties of $CD_W$ we suggest a refinement of this test which follows the power enhancement approach of \citet{fan2015}. The authors suggest improving the power of high-dimensional cross-sectional tests by adding to the test statistic of interest a screening statistic. This screening statistic is equal to zero with probability approaching one under the null hypothesis of the test, but diverges at a fast rate under the alternative. In our case we choose the absolute sum of thresholded cross-section correlation coefficients. This results in a power-enhanced weighted CD test statistic which is defined as
\begin{align}
\label{eq:CD_Wplus}
CD_{W+} = CD_W + \sum_{i=2}^N \sum_{j=1}^{i-1}\left| \widehat{\rho}_{ij} \right| \mathbf{1}\left( \left|\widehat{\rho}_{ij} \right| > 2\sqrt{\ln(N)/T}  \right),
\end{align}
where $\widehat{\rho}_{ij}$ is as in \eqref{eq:CScorrcoef} with $z_{i,t}=\widehat{\varepsilon}_{i,t}$ and where $\mathbf{1}(A)$ is the indicator function for event $A$.
The screening statistic on the right-hand side of \eqref{eq:CD_Wplus} has an asymptotically negligible effect on the size of $CD_{W+}$ because individual cross-section correlation coefficients are still consistent for their true value under $\mathbb{H}_0$. For additional discussion on $CD_{W+}$ we refer to Section \ref{sec:suppl_theory} of the Supplementary Appendix.

\begin{remark}
	An approach to reducing the dependence of $CD_W$ on a specific set of random weights would consist of averaging several weighted CD test statistics. Denote by $CD_W^{(g)}$ the CD test statistic obtained for a given draw of $N$ Rademacher distributed weights, the latter being indexed by $g$. For a total of $G$ different draws, an averaged weighted CD test can be constructed as
	\begin{align}
		\overline{CD}_W = \frac{1}{\sqrt{G}} \sum_{g=1}^{G} CD_W^{(g)},
	\end{align}
	where the total number of draws $G$ should be chosen sufficiently small (e.g. $G=30$) to avoid size distortions that may arise from scaling up lower-order terms in $CD_W^{(g)}$.
\end{remark}
\begin{remark}
An initially appealing alternative to external random numbers as weights for a bias-corrected CD test statistic would be a set of $N$ statistics derived from the dataset available to the researcher. When considering such \emph{internal (sample-specific) weights}, it is desirable to opt for functions of the data that are optimal in the sense that they maximize the rate at which a weighted CD test statistic diverges for a wide number of alternatives. In Section \ref{ssection::optimal_weights} of the Supplementary Appendix we argue that such improved weights suggest the use of a different test statistic for each alternative. Hence, as far as a bias-corrected version of $CD$ is preferred to the use of a different test, obtaining the former by weighting cross-section covariances with external random weights is the best possible alternative.
\end{remark}

\section{Monte Carlo study}	
\label{section::MC}
We investigate the properties of the CD test as well as its weighted alternatives in a small set of simulation experiments. We consider the common factor model
\begin{eqnarray}
y_{i,t} &=&\beta x_{i,t} + \vlambda_{i}^{\prime }\vf_{t} + \sqrt{\sigma _{i,y}^{2}}\epsi_{i,t}  \label{eq:MCmodel1} \\
x_{i,t} &=&\mLambda_{i}^{\prime}\vf_{t} +e_{i,t}.  \label{eq:MCmodel2}
\end{eqnarray}

We consider two alternative specifications for the factor loadings on the regressand $y_{i,t}$. We define $\vlambda_{i} = \viota_r + \tilde{\vlambda}_{i}$ where the $r$ elements of $\tilde{\vlambda}_{i}$ are drawn (I) from $U(-0.75, 0.75)$ or (II) from a standardized $\chi ^{2}\left(2\right) $ distribution that has zero mean and a variance of 1/6. The latter case is designed to match the first two moments of $\tilde{\vlambda}_{i}$ in case (I). Concerning the factor loadings on the regressors $x_{i,t}$ we draw element $\Lambda^{(r')}$ of the $r\times 1$ vector $\mLambda_{i}$ from $U(-0.5,0.5)$ if $r'=1$ and from $U(0.5,1,5)$ otherwise. Concerning the latent factors, we let $\vf_{t}\sim N\left( \vzeros_{r}, \mI_r\right) $. When considering the two-way fixed effects model instead of the common latent factor model, we set $r=2$ and define $\vf_t = \left[f_t^{(1)}, 1 \right]'$ as well as $\vlambda_{i}= \left[1, \lambda_{i}^{(2)} \right]'$ and $\mLambda_{i}= \left[1, \Lambda_{i}^{(2)} \right]'$ with the construction of $f_t^{(1)}, \lambda_{i}^{(2)}$ and $\Lambda_{i}^{(2)}$ being unchanged.

Analogous to our treatment of factor loadings we consider two cases for $\epsi_{i,t}$, namely that they are drawn independently over $i$ and $t$ from (i) a standard normal distribution or (ii) a standardized $\chi ^{2}\left(2\right) $ distribution that has zero mean and unit variance. The scalar random variable $e_{i,t}$ is generated in the same way as $\epsi_{i,t}$.
Lastly, error variances are obtained in two different fashions: We set (a) $\sigma _{i}^{2}=c_{\sigma }\left( \varsigma _{i}^{2}-2\right) / \sqrt{24} + 1$ where $\varsigma _{i}^{2}\sim\chi ^{2}\left( 2\right)$ or (b) $\sigma _{i}^{2} = 0.5 + d_{\sigma} T^{-1} \sum_{t=1}^T (\vlambda_{i}' \vf_t)^2$. The scaling $c_{\sigma }$, which is only included in setup (a), is chosen manually as discussed below. This construction entails that $\E\left[ \sigma_{i}^{2}\right] =1$ and $\var\left[ \sigma _{i}^{2}\right] =c_{\sigma }^{2}/6$. The normalizing constant $d_{\sigma}$, which appears in case (b) only, is $d_{\sigma} = \sqrt{2}$ in model error case (i) and $d_{\sigma} = (1/3)^{-1/2}$ in case (b). This scaling ensures that the variance of $\sigma _{i}^{2}$ across cross-sections is approximately of the same magnitude as in case (a) with $c_{\sigma}=1$.

In all following simulations we use feasible CD statistics, where, unlike the theory part of this paper, we treat $\beta$ as an unknown parameter. 

\subsection{Standard CD statistic}
In a first instance, we consider the properties of the original CD test statistic when applied to either 2WFE or CCE residuals in a scenario where all sources of cross-section correlation are adequately accounted for. We do this for the setup (i)-(a)-(I) as described above, meaning that errors are normally distributed, error variances independent of factor loadings (or individual fixed effects) and that the latter are symmetrically distributed around their mean. Results for a model with errors from a standardized $\chi^2$ distribution are almost identical and reported in Section \ref{sec:suppl_MC} of the Supplementary Appendix. Furthermore, we consider different degrees of heterogeneity among the individual-specific error variances by reporting results for the values $c_{\sigma }\in \left\{ 0.1,0.5,1,1.5\right\}$.

Table \ref{tab:CDmoments_H0} reports the first two moments of $CD$ when applied either to 2WFE residuals when the true model is one with two-way fixed effects or to CCE residuals in a model with multifactor error structure and two factors. The results for both cases are identical and show that $CD$ has a bias that diverges towards $-\infty $ as $T\rightarrow \infty $ and a variance that is considerably below $1$. The bias term is reasonably close to the benchmark value of $-\sqrt{T/2}$ indicated by Proposition \ref{theorem::CCE_sigma} but tends towards zero as heterogeneity among the individual-specific error variances is amplified.

\begin{table}[htbp]
	\footnotesize
	\centering
	\caption{Sample moments of original CD test statistic for remaining CSD under $\mathbb{H}_0$}
	\begin{tabular}{rr|rrrr|rrrr}
		\multicolumn{10}{l}{\textbf{Part A: Application of $CD$ to 2WFE residuals}} \\
    \hline
\hline
& \multicolumn{1}{l|}{$c_{\sigma}$} & 0.1   & 0.5   & 1     & 1.5   & 0.1   & 0.5   & 1     & 1.5 \\
\multicolumn{1}{l}{N} & \multicolumn{1}{l|}{T} & \multicolumn{4}{c|}{\textbf{Mean}}     & \multicolumn{4}{c}{\textbf{Variance $\times 100$ }}\\
\hline
25    & 25    & -3.53 & -3.50 & -3.41 & -3.24 & 0.09  & 0.20  & 0.82  & 2.96 \\
25    & 50    & -5.05 & -5.01 & -4.88 & -4.63 & 0.03  & 0.14  & 0.94  & 4.07 \\
25    & 100   & -7.18 & -7.12 & -6.94 & -6.57 & 0.02  & 0.15  & 1.33  & 6.10 \\
25    & 200   & -10.18 & -10.09 & -9.83 & -9.30 & 0.01  & 0.21  & 2.12  & 10.48 \\
200   & 25    & -3.47 & -3.44 & -3.34 & -3.14 & 0.05  & 0.11  & 0.43  & 1.67 \\
200   & 50    & -4.96 & -4.92 & -4.78 & -4.49 & 0.01  & 0.05  & 0.32  & 1.48 \\
200   & 100   & -7.05 & -6.99 & -6.79 & -6.39 & 0.00  & 0.04  & 0.37  & 1.85 \\
200   & 200   & -10.00 & -9.91 & -9.63 & -9.06 & 0.00  & 0.04  & 0.42  & 2.25 \\
\hline
\hline
& \multicolumn{1}{r}{} &       &       &       & \multicolumn{1}{r}{} &       &       &       &  \\
\multicolumn{10}{l}{\textbf{Part B: Application of $CD$ to CCE residuals}} \\
\hline
\hline
& \multicolumn{1}{l|}{$c_{\sigma}$} & 0.1   & 0.5   & 1     & 1.5   & 0.1   & 0.5   & 1     & 1.5 \\
\multicolumn{1}{l}{N} & \multicolumn{1}{l|}{T} & \multicolumn{4}{c|}{\textbf{Mean}}     & \multicolumn{4}{c}{\textbf{Variance $\times 100$ }}\\
\hline
25    & 25    & -3.53 & -3.50 & -3.41 & -3.25 & 0.11  & 0.23  & 0.87  & 3.00 \\
25    & 50    & -5.05 & -5.01 & -4.88 & -4.64 & 0.04  & 0.15  & 0.89  & 3.71 \\
25    & 100   & -7.18 & -7.12 & -6.95 & -6.61 & 0.02  & 0.15  & 1.25  & 5.58 \\
25    & 200   & -10.18 & -10.10 & -9.85 & -9.37 & 0.01  & 0.21  & 2.10  & 9.77 \\
200   & 25    & -3.46 & -3.43 & -3.33 & -3.14 & 0.06  & 0.13  & 0.48  & 1.77 \\
200   & 50    & -4.96 & -4.91 & -4.77 & -4.50 & 0.02  & 0.06  & 0.35  & 1.63 \\
200   & 100   & -7.05 & -6.99 & -6.79 & -6.40 & 0.00  & 0.04  & 0.34  & 1.79 \\
200   & 200   & -10.00 & -9.91 & -9.63 & -9.08 & 0.00  & 0.04  & 0.42  & 2.27 \\
\hline
\hline
	\end{tabular}%
	\par
\begin{minipage}{0.67\textwidth}
	\smallskip
	\scriptsize Notes. \itshape $\epsi_{i,t}$ and  $e_{i,t}$ are $N(0,1)$. $\sigma _{i,y}^{2}=c_{\sigma }\left( \varsigma _{i,y}^{2}-2\right) /4+1$; $\varsigma _{i,y}^{2}$ is $\chi^2(2)$. In part B, $r=2$ and loadings $\vlambda_{i}$ are drawn from $U(0.5, 1.5)$.  $\mLambda_{i}$ has first element from $U(0.5, 1.5)$ and second from $U(-0.5,0.5)$. Factors $\vf_t$ drawn from $N(0,1)$. In part A, we restrict $\vf_t = \left[f_t^{(1)}, 1 \right]'$, $\vlambda_{i}= \left[1, \lambda_{i}^{(2)} \right]'$ and $\mLambda_{i}= \left[1, \Lambda_{i}^{(2)} \right]'$.
\end{minipage}
	\label{tab:CDmoments_H0}%
\end{table}%

Given that the theoretical results of Sections \ref{ssection::asymptotics_additive} and \ref{ssection::asymptotics_multifactor} are confirmed by Table \ref{tab:CDmoments_H0}, we turn to the properties of $CD$ under its alternative hypothesis. For this purpose, we simulate a model with multifactor error structure and three factors, implying that neither 2WFE nor CCE estimation completely accounts for all sources of cross-section correlation in the simulated data. Heterogeneity among unit-specific error variances is kept constant by only considering $c_{\sigma}=1$. Instead we consider cases (I) and (II) for factor loadings as well as cases (a) and (b) for  error variances. Again, results are presented only for normally distributed model errors. The moments of $CD$ are very similar when model errors are drawn from a standardized $\chi^2(2)$ distribution and corresponding tables can be found in Section \ref{sec:suppl_MC} of the Supplementary Appendix.

Tables \ref{tab:CDmoments_H1_FE} and \ref{tab:CDmoments_H1_CCE} report the mean and variance of $CD$ in the presence of remaining cross-section correlation. Both are very similar to the numbers reported in Table \ref{tab:CDmoments_H0} as long as factor loadings are symmetrically distributed, an observation which is in line with the theoretical results provided in Section \ref{sec:suppl_theory} of the Supplementary Appendix. Since symmetric loadings center the leading stochastic component of $CD$ around zero, the latter is dominated by a \emph{bias equivalent} that diverges towards $-\infty$ as $T\to \infty$. Qualitatively very different results can be observed when factor loadings are drawn from a skewed distribution, particularly when applying the CD test statistic to 2WFE residuals. As one can observe in columns two and four of Table \ref{tab:CDmoments_H1_FE}, the strong negative divergence of $CD$ as $T$ increases is mitigated and it is reasonable to assume that it will be turned into positive divergence for $N>200$. This pattern, which is amplified if error variances are a function of factor loadings, reflects the presence of a nonzero mean in the leading stochastic component of $CD$ that diverges to $+\infty$ at rate $N\sqrt{T}$. In the case of 2WFE residuals, its magnitude is sufficiently large to counteract the negative bias equivalent already for $N=200$. By contrast, columns two and four in Table \ref{tab:CDmoments_H1_CCE} suggest that this term is considerably smaller when the CD test is applied to CCE residuals since the mean of $CD$ continues to diverge towards $-\infty$ for $N=200$. We conjecture that the same sign reversal will eventually be obtained even if this case, but that the required cross-section dimension for this to happen is higher.

\begin{table}[htbp]
	\footnotesize
	\centering
	\caption{Sample moments of original CD test statistic applied to 2WFE residuals under $\mathbb{H}_1$}
	\begin{tabular}{rr|rr|rr|rr|rr}
    \hline
	\hline
	&       & \multicolumn{4}{c|}{\textbf{Mean}}     & \multicolumn{4}{c}{\textbf{Variance}}\\
	& \multicolumn{1}{l|}{$\vlambda{i}$:} & \multicolumn{2}{c|}{symmetric} & \multicolumn{2}{c|}{skewed} & \multicolumn{2}{c|}{symmetric} & \multicolumn{2}{c}{skewed} \\
	& \multicolumn{1}{l|}{$\sigma^2_i$:} & \multicolumn{1}{l}{$\perp \vlambda{i}$} & \multicolumn{1}{l|}{$f(\vlambda{i})$} & \multicolumn{1}{l}{$\perp \vlambda{i}$} & \multicolumn{1}{l|}{$f(\vlambda{i})$} & \multicolumn{1}{l}{$\perp \vlambda{i}$} & \multicolumn{1}{l|}{$f(\vlambda{i})$} & \multicolumn{1}{l}{$\perp \vlambda{i}$} & \multicolumn{1}{l}{$f(\vlambda{i})$} \\
	\multicolumn{1}{l}{N} & \multicolumn{1}{l|}{T} &       &       &       &       &       &       &       &  \\
	\hline
	25    & 25    & -3.45 & -3.23 & -3.41 & -2.83 & 0.01  & 0.07  & 0.01  & 0.35 \\
	25    & 50    & -4.93 & -4.62 & -4.87 & -4.03 & 0.01  & 0.11  & 0.01  & 0.59 \\
	25    & 100   & -7.00 & -6.56 & -6.92 & -5.72 & 0.01  & 0.21  & 0.01  & 1.12 \\
	25    & 200   & -9.92 & -9.30 & -9.80 & -8.12 & 0.01  & 0.36  & 0.02  & 2.14 \\
	50    & 25    & -3.41 & -3.05 & -3.37 & -2.41 & 0.00  & 0.09  & 0.01  & 0.46 \\
	50    & 50    & -4.87 & -4.36 & -4.81 & -3.44 & 0.00  & 0.15  & 0.01  & 0.78 \\
	50    & 100   & -6.93 & -6.19 & -6.83 & -4.89 & 0.01  & 0.27  & 0.01  & 1.45 \\
	50    & 200   & -9.82 & -8.80 & -9.69 & -6.95 & 0.01  & 0.48  & 0.01  & 2.61 \\
	100   & 25    & -3.40 & -2.75 & -3.35 & -1.68 & 0.00  & 0.17  & 0.01  & 0.73 \\
	100   & 50    & -4.85 & -3.95 & -4.78 & -2.39 & 0.00  & 0.26  & 0.01  & 1.19 \\
	100   & 100   & -6.89 & -5.61 & -6.80 & -3.39 & 0.01  & 0.42  & 0.01  & 2.17 \\
	100   & 200   & -9.76 & -7.95 & -9.63 & -4.78 & 0.01  & 0.76  & 0.01  & 4.01 \\
	200   & 25    & -3.39 & -2.20 & -3.34 & -0.23 & 0.00  & 0.36  & 0.01  & 1.52 \\
	200   & 50    & -4.83 & -3.16 & -4.77 & -0.34 & 0.00  & 0.49  & 0.01  & 2.23 \\
	200   & 100   & -6.87 & -4.48 & -6.78 & -0.45 & 0.00  & 0.82  & 0.01  & 3.78 \\
	200   & 200   & -9.74 & -6.38 & -9.60 & -0.66 & 0.01  & 1.35  & 0.01  & 6.92 \\
	\hline
	\hline
	\end{tabular}%
\par
\begin{minipage}{0.67\textwidth}
	\smallskip
	\scriptsize Notes. \itshape The model has a factor error structure with 3 factors. $\epsi_{i,t}$ and  $e_{i,t}$ are $N(0,1)$. ``$\sigma_{i}^2$: $\perp \vlambda_{i}$'' means that $\sigma _{i,y}^{2}=\left( \varsigma _{i,y}^{2}-2\right) /4+1$ where $\varsigma _{i,y}^{2}$ is $\chi^2(2)$. For ``$\sigma_{i}^2$: $f(\vlambda_{i})$'', we let $\sigma_{i,y}^{2} = d_{\sigma} T^{-1} \sum_{t=1}^T (\vlambda_{i}' \vf_t)^2$.
	
	The case ``$\vlambda_{i}$: symmetric'' corresponds to drawing $\vlambda_{i}$ from $U(0.5, 1.5)$. ``$\vlambda_{i}$: skewed'' loadings are from a standardized $\chi ^{2}\left(2\right)$ distribution with mean and variance equal to 1.
	
	$\mLambda_{i}$ has first element from $U(0.5, 1.5)$ and all others from $U(-0.5,0.5)$. Factors $\vf_t$ are drawn from $N(0,1)$.
\end{minipage}
	\label{tab:CDmoments_H1_FE}%
\end{table}%

\begin{table}[htbp]
	\footnotesize
	\centering
	\caption{Sample moments of original CD test statistic applied to CCE residuals under $\mathbb{H}_1$}
	\begin{tabular}{rr|rr|rr|rr|rr}
    \hline
\hline
&       & \multicolumn{4}{c|}{\textbf{Mean}}     & \multicolumn{4}{c}{\textbf{Variance}}\\
& \multicolumn{1}{l|}{$\vlambda{i}$:} & \multicolumn{2}{c|}{symmetric} & \multicolumn{2}{c|}{skewed} & \multicolumn{2}{c|}{symmetric} & \multicolumn{2}{c}{skewed} \\
& \multicolumn{1}{l|}{$\sigma^2_i$:} & \multicolumn{1}{l}{$\perp \vlambda{i}$} & \multicolumn{1}{l|}{$f(\vlambda{i})$} & \multicolumn{1}{l}{$\perp \vlambda{i}$} & \multicolumn{1}{l|}{$f(\vlambda{i})$} & \multicolumn{1}{l}{$\perp \vlambda{i}$} & \multicolumn{1}{l|}{$f(\vlambda{i})$} & \multicolumn{1}{l}{$\perp \vlambda{i}$} & \multicolumn{1}{l}{$f(\vlambda{i})$} \\
\multicolumn{1}{l}{N} & \multicolumn{1}{l|}{T} &       &       &       &       &       &       &       &  \\
\hline
25    & 25    & -3.42 & -3.38 & -3.42 & -3.21 & 0.01  & 0.02  & 0.01  & 0.11 \\
25    & 50    & -4.89 & -4.85 & -4.89 & -4.58 & 0.01  & 0.02  & 0.01  & 0.18 \\
25    & 100   & -6.95 & -6.89 & -6.95 & -6.52 & 0.01  & 0.03  & 0.01  & 0.35 \\
25    & 200   & -9.86 & -9.77 & -9.85 & -9.23 & 0.02  & 0.05  & 0.01  & 0.68 \\
50    & 25    & -3.38 & -3.34 & -3.38 & -3.17 & 0.01  & 0.01  & 0.01  & 0.07 \\
50    & 50    & -4.83 & -4.78 & -4.83 & -4.52 & 0.01  & 0.01  & 0.01  & 0.12 \\
50    & 100   & -6.87 & -6.80 & -6.86 & -6.43 & 0.01  & 0.02  & 0.01  & 0.23 \\
50    & 200   & -9.74 & -9.65 & -9.74 & -9.12 & 0.01  & 0.03  & 0.01  & 0.39 \\
100   & 25    & -3.36 & -3.32 & -3.35 & -3.15 & 0.01  & 0.01  & 0.01  & 0.05 \\
100   & 50    & -4.80 & -4.76 & -4.80 & -4.50 & 0.00  & 0.01  & 0.00  & 0.08 \\
100   & 100   & -6.83 & -6.76 & -6.83 & -6.40 & 0.01  & 0.01  & 0.00  & 0.14 \\
100   & 200   & -9.68 & -9.59 & -9.68 & -9.07 & 0.01  & 0.02  & 0.01  & 0.25 \\
200   & 25    & -3.35 & -3.31 & -3.35 & -3.14 & 0.00  & 0.01  & 0.00  & 0.05 \\
200   & 50    & -4.79 & -4.74 & -4.79 & -4.48 & 0.00  & 0.01  & 0.00  & 0.06 \\
200   & 100   & -6.81 & -6.74 & -6.80 & -6.38 & 0.00  & 0.01  & 0.00  & 0.10 \\
200   & 200   & -9.65 & -9.56 & -9.65 & -9.05 & 0.01  & 0.02  & 0.01  & 0.19 \\
\hline
\hline
	\end{tabular}%
\par
\begin{minipage}{0.7\textwidth}
	\smallskip
	\scriptsize Notes. \itshape See Table \ref{tab:CDmoments_H1_FE}.
\end{minipage}
	\label{tab:CDmoments_H1_CCE}%
\end{table}%

\subsection{Weighted CD statistic}
Next, we proceed with investigating the properties of our weighted CD test statistic. As previously, we keep heterogeneity across unit-specific error variances fixed at $c_{\sigma}=1$ and consider two different specifications each for factor loadings and error variances. We report size and power for our weighted test statistic $CD_W$ as well as the power-enhanced refinement $CD_{W+}$. As a benchmark test, we include a CD test statistic with analytic bias correction which corrects $CD$ with sample equivalents of the asymptotic bias terms in Theorems \ref{theorem::additiveHetero} and \ref{theorem::CCE}. Details on its implementation can be found in Section \ref{section:discussion_heuristics} of the Supplementary Appendix. The original CD test is left out for the sake of saving space and in particular since rejection rates are $100\%$ for most cases.

Table \ref{tab:CDW_rejections_FE} reports the size and power of all three bias corrected test statistic when applied to 2WFE residuals. The size of $CD_W$ and $CD_{W+}$ is very close to the nominal level of $5\%$ as long as $N\geq T$. Size distortions are given in cases where $T$ is considerably large than $N$ and the effect of power enhancement on size is generally negligible. The analytically bias-corrected CD test statistic $CD_{BC}$ exhibits hardly any tendency to over-reject but is rather conservative, in particular when $T$ is small.

Panel B in Table \ref{tab:CDW_rejections_FE} report results on power. Here, we see that the power of $CD_W$ is in general low and increases only in $T$. This is improved upon considerably by power enhancement, leading the refined test statistic $CD_{W+}$ to reliably reject when it should as long as the number of time periods is large enough. The performance of $CD_{W+}$ is considerably above that of the benchmark statistic $CD_{BC}$ when factor loadings are drawn from a symmetric distribution and is on par in most other cases. An exception is the case of skewed loadings with small $T$ and large $N$ where $CD_{BC}$ has markedly higher rejection rates.
It can also be noted that the performance of $CD_{BC}$ largely depends on whether factor loadings are drawn from a symmetric distribution or not. If this is the case, analytic bias correction leads to rejection rates under the alternative hypothesis that are worst among all three tests considered.

When testing for remaining cross-section correlation in CCE residuals we observe similar results as in Table \ref{tab:CDW_rejections_FE}. Interestingly, it can be observed that the benchmark statistic $CD_{BC}$ exhibits size distortions in samples with large $T$ when loadings are drawn from a skewed distribution and if factor loadings and error variances are dependent. This results from the nature of our bias correction which assumes independence between $\sigma^{2}_i$ and $\vlambda{i}$ to considerably improve the accuracy of the bias estimate.

The power properties of $CD_W$ and $CD_{W+}$ when either of these tests is applied to CCE residuals mirror those seen in the 2WFE case, even though rejection rates are generally somewhat lower. However, the power-enhanced statistic $CD_{W+}$ now performs best in all cases without ever being inferior to $CD_{BC}$.

\begin{table}[htbp]
	\scriptsize
	\caption{Rejection rates for weighted CD test statistic when applied to 2WFE residuals.}
	\begin{tabular}{rr|rrr|rrr|rrr|rrr}
		\multicolumn{14}{l}{\textbf{Part A: Size}}\\
    \hline
\hline
& \multicolumn{1}{l|}{$\vlambda{i}$:} & \multicolumn{6}{c|}{symmetric}                & \multicolumn{6}{c}{skewed} \\
& \multicolumn{1}{l|}{$\sigma^2_i$:} & \multicolumn{3}{c|}{$\perp \vlambda{i}$} & \multicolumn{3}{c|}{$f(\vlambda{i})$} & \multicolumn{3}{c|}{$\perp \vlambda{i}$} & \multicolumn{3}{c}{$f(\vlambda{i})$} \\
\multicolumn{1}{l}{N} & \multicolumn{1}{l|}{T} & \multicolumn{1}{l}{$CD_{W}$} & \multicolumn{1}{l}{$CD_{W+}$} & \multicolumn{1}{l|}{$CD_{BC}$} & \multicolumn{1}{l}{$CD_{W}$} & \multicolumn{1}{l}{$CD_{W+}$} & \multicolumn{1}{l|}{$CD_{BC}$} & \multicolumn{1}{l}{$CD_{W}$} & \multicolumn{1}{l}{$CD_{W+}$} & \multicolumn{1}{l|}{$CD_{BC}$} & \multicolumn{1}{l}{$CD_{W}$} & \multicolumn{1}{l}{$CD_{W+}$} & \multicolumn{1}{l}{$CD_{BC}$} \\
\hline
25    & 25    & 5.3   & 5.5   & 2.7   & 6.2   & 6.4   & 2.8   & 5.3   & 5.5   & 3.1   & 5.2   & 5.3   & 2.5 \\
25    & 50    & 6.6   & 6.9   & 4.3   & 5.8   & 6.1   & 4.1   & 6.2   & 6.8   & 4.0   & 6.3   & 6.4   & 2.7 \\
25    & 100   & 7.8   & 8.5   & 3.7   & 7.0   & 7.1   & 4.6   & 6.5   & 6.9   & 4.5   & 6.4   & 6.6   & 2.2 \\
25    & 200   & 9.8   & 10.3  & 5.6   & 9.0   & 10.0  & 6.1   & 9.3   & 9.6   & 6.1   & 8.9   & 9.4   & 3.3 \\
50    & 25    & 4.6   & 4.7   & 2.8   & 4.7   & 4.7   & 3.7   & 4.6   & 4.6   & 2.9   & 5.9   & 5.9   & 2.9 \\
50    & 50    & 6.1   & 6.2   & 3.4   & 5.2   & 5.4   & 5.0   & 5.4   & 5.3   & 4.4   & 4.5   & 4.6   & 2.7 \\
50    & 100   & 5.5   & 5.9   & 4.7   & 5.4   & 5.8   & 4.8   & 5.8   & 6.0   & 4.2   & 5.1   & 5.6   & 3.1 \\
50    & 200   & 6.0   & 6.5   & 4.6   & 6.1   & 6.3   & 4.9   & 6.4   & 6.2   & 4.7   & 5.7   & 5.7   & 3.4 \\
100   & 25    & 5.0   & 5.0   & 2.8   & 5.2   & 5.2   & 3.4   & 5.9   & 5.9   & 3.8   & 4.7   & 4.7   & 2.9 \\
100   & 50    & 5.0   & 5.1   & 3.5   & 4.6   & 4.8   & 3.9   & 4.4   & 4.5   & 4.0   & 5.0   & 5.1   & 2.9 \\
100   & 100   & 5.3   & 5.3   & 5.1   & 5.1   & 5.1   & 3.9   & 4.8   & 5.0   & 3.7   & 5.5   & 5.5   & 3.5 \\
100   & 200   & 4.5   & 4.5   & 4.7   & 5.4   & 5.3   & 4.5   & 5.8   & 5.6   & 4.0   & 6.2   & 6.3   & 3.1 \\
200   & 25    & 5.5   & 5.5   & 3.3   & 5.5   & 5.5   & 3.6   & 4.7   & 4.7   & 3.4   & 4.7   & 4.7   & 3.2 \\
200   & 50    & 5.4   & 5.5   & 4.4   & 5.5   & 5.6   & 3.2   & 5.5   & 5.6   & 4.2   & 4.7   & 4.8   & 3.6 \\
200   & 100   & 4.3   & 4.3   & 4.6   & 6.1   & 6.2   & 4.4   & 4.5   & 4.5   & 4.6   & 4.9   & 4.9   & 3.2 \\
200   & 200   & 4.5   & 4.8   & 4.3   & 5.1   & 5.3   & 4.6   & 5.1   & 5.1   & 4.7   & 6.2   & 6.3   & 3.4 \\
\hline
\hline
& \multicolumn{1}{r}{} &       &       & \multicolumn{1}{r}{} &       &       & \multicolumn{1}{r}{} &       &       & \multicolumn{1}{r}{} &       &       &  \\
\multicolumn{14}{l}{\textbf{Part B: Power}} \\
\hline
\hline
& \multicolumn{1}{l|}{$\vlambda{i}$:} & \multicolumn{6}{c|}{symmetric}                & \multicolumn{6}{c}{skewed} \\
& \multicolumn{1}{l|}{$\sigma^2_i$:} & \multicolumn{3}{c|}{$\perp \vlambda{i}$} & \multicolumn{3}{c|}{$f(\vlambda{i})$} & \multicolumn{3}{c|}{$\perp \vlambda{i}$} & \multicolumn{3}{c}{$f(\vlambda{i})$} \\
\multicolumn{1}{l}{N} & \multicolumn{1}{l|}{T} & \multicolumn{1}{l}{$CD_{W}$} & \multicolumn{1}{l}{$CD_{W+}$} & \multicolumn{1}{l|}{$CD_{BC}$} & \multicolumn{1}{l}{$CD_{W}$} & \multicolumn{1}{l}{$CD_{W+}$} & \multicolumn{1}{l|}{$CD_{BC}$} & \multicolumn{1}{l}{$CD_{W}$} & \multicolumn{1}{l}{$CD_{W+}$} & \multicolumn{1}{l|}{$CD_{BC}$} & \multicolumn{1}{l}{$CD_{W}$} & \multicolumn{1}{l}{$CD_{W+}$} & \multicolumn{1}{l}{$CD_{BC}$} \\
\hline
25    & 25    & 12.4  & 18.5  & 8.6   & 10.3  & 11.4  & 7.7   & 11.5  & 16.3  & 17.8  & 12.5  & 14.9  & 27.4 \\
25    & 50    & 21.9  & 81.7  & 16.3  & 17.0  & 38.9  & 11.1  & 18.9  & 52.9  & 35.1  & 24.2  & 56.9  & 50.7 \\
25    & 100   & 36.2  & 100 & 25.9  & 27.2  & 92.5  & 22.1  & 31.0  & 90.6  & 50.7  & 38.5  & 95.8  & 71.1 \\
25    & 200   & 51.4  & 100 & 41.7  & 43.9  & 100 & 33.5  & 44.3  & 99.4  & 67.0  & 52.8  & 99.9  & 83.1 \\
50    & 25    & 12.3  & 15.5  & 8.5   & 10.8  & 11.1  & 6.7   & 12.3  & 17.7  & 42.3  & 14.2  & 14.5  & 65.1 \\
50    & 50    & 20.3  & 96.3  & 15.9  & 16.7  & 39.9  & 12.5  & 19.8  & 79.1  & 65.6  & 22.8  & 66.6  & 89.1 \\
50    & 100   & 34.7  & 100 & 26.7  & 27.9  & 99.2  & 24.2  & 32.2  & 99.8  & 85.0  & 37.1  & 99.9  & 97.2 \\
50    & 200   & 52.9  & 100 & 41.4  & 43.1  & 100 & 36.9  & 47.0  & 100 & 92.0  & 54.7  & 100 & 99.2 \\
100   & 25    & 10.7  & 11.0  & 8.4   & 9.3   & 9.4   & 6.7   & 11.3  & 14.9  & 81.1  & 13.3  & 13.5  & 96.8 \\
100   & 50    & 20.3  & 98.9  & 14.9  & 17.2  & 40.5  & 12.8  & 18.1  & 97.1  & 96.6  & 20.8  & 81.8  & 99.9 \\
100   & 100   & 33.4  & 100 & 26.4  & 27.9  & 100 & 22.8  & 29.1  & 100 & 99.6  & 34.2  & 100 & 100 \\
100   & 200   & 50.2  & 100 & 41.5  & 41.9  & 100 & 39.2  & 47.3  & 100 & 100 & 52.4  & 100 & 100 \\
200   & 25    & 11.6  & 11.6  & 8.8   & 11.3  & 11.3  & 6.9   & 11.5  & 12.2  & 98.6  & 12.4  & 12.4  & 100 \\
200   & 50    & 19.2  & 99.6  & 16.2  & 15.2  & 34.2  & 13.4  & 18.5  & 99.9  & 100 & 23.8  & 87.0  & 100 \\
200   & 100   & 34.9  & 100 & 26.4  & 25.5  & 100 & 22.8  & 30.3  & 100 & 100 & 37.8  & 100 & 100 \\
200   & 200   & 49.8  & 100 & 44.3  & 40.7  & 100 & 38.5  & 45.9  & 100 & 100 & 52.3  & 100 & 100 \\
\hline
\hline
	\end{tabular}%
\par
\begin{minipage}{1.05\textwidth}
	\smallskip
	\scriptsize Notes. \itshape In Part A the model has two factors which are restricted as noted in Table \ref{tab:CDmoments_H0}. Part B corresponds to a model with factor error structure and 3 factors. For details on all other model parameters, see Table \ref{tab:CDmoments_H1_FE}. $CD_W$ is the weighted CD test statistic introduced in Theorem \ref{thm:CD_W_combined}. $CD_{W+}$ is its power-enhanced refinement. $CD_{BC}$ is a CD test statistic with analytic bias correction.
\end{minipage}
	\label{tab:CDW_rejections_FE}%
\end{table}%

\begin{table}[htbp]
	\scriptsize
	\centering
	\caption{Rejection rates for weighted CD test statistic when applied to CCE residuals.}
	\begin{tabular}{rr|rrr|rrr|rrr|rrr}
		\multicolumn{14}{l}{\textbf{Part A: Size}} \\
    \hline
\hline
& \multicolumn{1}{l|}{$\vlambda{i}$:} & \multicolumn{6}{c|}{symmetric}                & \multicolumn{6}{c}{skewed} \\
& \multicolumn{1}{l|}{$\sigma^2_i$:} & \multicolumn{3}{c|}{$\perp \vlambda{i}$} & \multicolumn{3}{c|}{$f(\vlambda{i})$} & \multicolumn{3}{c|}{$\perp \vlambda{i}$} & \multicolumn{3}{c}{$f(\vlambda{i})$} \\
\multicolumn{1}{l}{N} & \multicolumn{1}{l|}{T} & \multicolumn{1}{l}{$CD_{W}$} & \multicolumn{1}{l}{$CD_{W+}$} & \multicolumn{1}{l|}{$CD_{BC}$} & \multicolumn{1}{l}{$CD_{W}$} & \multicolumn{1}{l}{$CD_{W+}$} & \multicolumn{1}{l|}{$CD_{BC}$} & \multicolumn{1}{l}{$CD_{W}$} & \multicolumn{1}{l}{$CD_{W+}$} & \multicolumn{1}{l|}{$CD_{BC}$} & \multicolumn{1}{l}{$CD_{W}$} & \multicolumn{1}{l}{$CD_{W+}$} & \multicolumn{1}{l}{$CD_{BC}$} \\
\hline
25    & 25    & 5.9   & 5.9   & 4.2   & 5.3   & 5.6   & 4.6   & 5.3   & 5.4   & 4.3   & 5.9   & 6.3   & 5.8 \\
25    & 50    & 6.4   & 6.9   & 4.5   & 5.7   & 6.2   & 5.3   & 6.2   & 6.8   & 4.9   & 7.6   & 8.1   & 7.3 \\
25    & 100   & 8.3   & 9.0   & 6.1   & 7.9   & 8.4   & 6.6   & 7.9   & 9.0   & 7.0   & 9.2   & 9.9   & 12.9 \\
25    & 200   & 11.2  & 12.2  & 6.8   & 10.3  & 11.3  & 8.2   & 11.4  & 11.7  & 9.5   & 12.4  & 14.0  & 22.0 \\
50    & 25    & 5.9   & 6.1   & 4.1   & 6.2   & 6.4   & 5.0   & 6.3   & 6.4   & 4.4   & 6.6   & 6.6   & 4.3 \\
50    & 50    & 6.4   & 6.7   & 4.4   & 5.8   & 5.9   & 4.0   & 5.5   & 5.7   & 3.9   & 6.2   & 6.4   & 5.3 \\
50    & 100   & 6.1   & 6.3   & 5.4   & 5.9   & 6.0   & 5.9   & 6.1   & 6.1   & 4.7   & 5.7   & 5.9   & 10.4 \\
50    & 200   & 7.3   & 7.2   & 5.0   & 6.4   & 6.3   & 6.5   & 7.2   & 7.6   & 6.3   & 7.6   & 7.9   & 19.5 \\
100   & 25    & 5.2   & 5.2   & 4.5   & 5.5   & 5.5   & 4.2   & 5.4   & 5.4   & 4.2   & 5.3   & 5.3   & 4.2 \\
100   & 50    & 5.0   & 5.1   & 3.8   & 4.8   & 4.9   & 4.9   & 5.3   & 5.5   & 5.9   & 5.2   & 5.2   & 5.6 \\
100   & 100   & 5.6   & 5.8   & 4.9   & 6.0   & 5.9   & 5.7   & 4.6   & 4.8   & 4.0   & 7.0   & 7.0   & 8.6 \\
100   & 200   & 5.4   & 5.5   & 4.9   & 5.9   & 5.9   & 6.2   & 5.7   & 6.0   & 5.1   & 5.8   & 5.8   & 18.6 \\
200   & 25    & 5.3   & 5.3   & 3.8   & 5.1   & 5.1   & 3.9   & 5.3   & 5.3   & 4.5   & 5.4   & 5.4   & 5.1 \\
200   & 50    & 4.3   & 4.4   & 4.4   & 4.9   & 4.9   & 4.9   & 5.3   & 5.3   & 4.2   & 5.0   & 5.0   & 5.3 \\
200   & 100   & 4.5   & 4.6   & 4.5   & 5.6   & 5.6   & 4.4   & 5.2   & 5.4   & 5.5   & 6.3   & 6.4   & 9.7 \\
200   & 200   & 4.4   & 4.7   & 5.1   & 5.8   & 5.8   & 5.1   & 5.2   & 5.3   & 5.2   & 5.0   & 5.1   & 16.5 \\
\hline
\hline
& \multicolumn{1}{r}{} &       &       & \multicolumn{1}{r}{} &       &       & \multicolumn{1}{r}{} &       &       & \multicolumn{1}{r}{} &       &       &  \\
\multicolumn{14}{l}{\textbf{Part B: Power}}  \\
\hline
\hline
& \multicolumn{1}{l|}{$\vlambda{i}$:} & \multicolumn{6}{c|}{symmetric}                & \multicolumn{6}{c}{skewed} \\
& \multicolumn{1}{l|}{$\sigma^2_i$:} & \multicolumn{3}{c|}{$\perp \vlambda{i}$} & \multicolumn{3}{c|}{$f(\vlambda{i})$} & \multicolumn{3}{c|}{$\perp \vlambda{i}$} & \multicolumn{3}{c}{$f(\vlambda{i})$} \\
\multicolumn{1}{l}{N} & \multicolumn{1}{l|}{T} & \multicolumn{1}{l}{$CD_{W}$} & \multicolumn{1}{l}{$CD_{W+}$} & \multicolumn{1}{l|}{$CD_{BC}$} & \multicolumn{1}{l}{$CD_{W}$} & \multicolumn{1}{l}{$CD_{W+}$} & \multicolumn{1}{l|}{$CD_{BC}$} & \multicolumn{1}{l}{$CD_{W}$} & \multicolumn{1}{l}{$CD_{W+}$} & \multicolumn{1}{l|}{$CD_{BC}$} & \multicolumn{1}{l}{$CD_{W}$} & \multicolumn{1}{l}{$CD_{W+}$} & \multicolumn{1}{l}{$CD_{BC}$} \\
\hline
25    & 25    & 10.0  & 11.8  & 5.4   & 8.9   & 9.5   & 6.2   & 9.0   & 12.2  & 7.4   & 10.5  & 11.4  & 7.4 \\
25    & 50    & 13.8  & 39.6  & 10.4  & 12.5  & 18.9  & 9.4   & 13.4  & 31.4  & 13.6  & 16.5  & 29.6  & 16.1 \\
25    & 100   & 21.3  & 81.9  & 15.4  & 17.1  & 44.4  & 14.6  & 19.0  & 62.2  & 20.2  & 24.5  & 63.1  & 28.1 \\
25    & 200   & 31.8  & 97.5  & 22.9  & 26.1  & 85.1  & 20.8  & 30.3  & 84.6  & 28.8  & 36.5  & 90.2  & 43.2 \\
50    & 25    & 9.6   & 11.2  & 5.9   & 8.0   & 8.3   & 6.9   & 9.0   & 11.2  & 7.8   & 10.6  & 10.8  & 9.3 \\
50    & 50    & 14.3  & 54.2  & 8.6   & 10.3  & 16.0  & 8.4   & 11.7  & 52.8  & 12.7  & 15.2  & 34.1  & 20.0 \\
50    & 100   & 19.6  & 98.5  & 13.7  & 16.5  & 59.5  & 15.3  & 17.8  & 90.5  & 23.1  & 22.2  & 86.4  & 29.7 \\
50    & 200   & 29.4  & 100 & 23.0  & 20.8  & 98.3  & 24.5  & 27.8  & 99.2  & 34.5  & 32.7  & 100 & 45.2 \\
100   & 25    & 8.2   & 8.1   & 5.1   & 8.4   & 8.4   & 5.3   & 8.1   & 9.7   & 7.0   & 9.7   & 9.7   & 10.9 \\
100   & 50    & 12.3  & 64.8  & 9.4   & 10.9  & 15.5  & 10.0  & 11.2  & 74.7  & 15.3  & 13.4  & 38.6  & 20.5 \\
100   & 100   & 20.2  & 99.9  & 15.7  & 15.5  & 76.4  & 15.3  & 16.6  & 99.4  & 24.9  & 20.2  & 97.3  & 35.6 \\
100   & 200   & 29.3  & 100 & 23.8  & 22.8  & 100 & 24.8  & 26.0  & 100 & 40.3  & 32.0  & 100 & 50.4 \\
200   & 25    & 8.9   & 8.9   & 7.0   & 8.2   & 8.2   & 5.9   & 8.8   & 9.0   & 8.6   & 9.9   & 10.0  & 10.8 \\
200   & 50    & 13.7  & 72.9  & 8.3   & 10.8  & 15.0  & 10.5  & 11.7  & 90.4  & 15.8  & 14.1  & 39.6  & 18.5 \\
200   & 100   & 17.5  & 100 & 14.6  & 13.4  & 86.7  & 14.5  & 17.0  & 100 & 27.5  & 20.4  & 99.7  & 32.7 \\
200   & 200   & 29.7  & 100 & 22.8  & 22.5  & 100 & 25.3  & 25.7  & 100 & 45.4  & 32.1  & 100 & 52.6 \\
\hline
\hline
	\end{tabular}%
\par
\begin{minipage}{1.05\textwidth}
	\smallskip
	\scriptsize Notes. \itshape The model has as factor error structure with  two (Part A) or three (Part B) factors. For details on all other model parameters, see Table \ref{tab:CDmoments_H1_FE}. For an explanation of the tests, see Table \ref{tab:CDW_rejections_FE}.
\end{minipage}
	\label{tab:CDW_rejections_CCE}%
\end{table}%

\section{Empirical illustration: R\& D investments}
\label{section:empirical}
\label{section:empirical_RD}
In this section we illustrate the applicability of the standard and the weighted CD statistics using the R\&D investments data of \citet{doi:10.1162/REST00272}. Information on up to twelve manufacturing industries in ten countries (Denmark, Finland, Germany, Italy, Japan, The Netherlands, Portugal, Sweden, United Kingdom, and The US) over a time period from 1980 to 2005 was used to construct the dataset. After some minor modification to the original dataset (see Section \ref{appendix::empirical} of the Supplementary Appendix) we are left with with a panel dataset of $N=82$ and $T=25$.

In this application serial correlation is important from an economic point of view. In Section \ref{appendix::empirical} of the Supplementary Appendix we outline how the CD statistic can be modified to account for serial correlation, given a set of known weights $w_{i}$.

\citet{doi:10.1162/REST00272} question whether R\&D can be estimated in a standard Griliches-type ``knowledge production function'' framework ignoring the potential presence of knowledge spillovers between cross-sectional units as well as other cross-section dependencies. Among other things they document a ``\emph{[\ldots] strong evidence for cross-sectional dependence and the presence of a common factor structure in the data, which we interpret as indicative for the presence of knowledge spillovers and additional unobserved cross-sectional dependencies}'' \citep[p.437]{doi:10.1162/REST00272}. Cross-sectional dependence was measured by means of a CD test.

In this section we will primarily revisit some of the results in Table 5 of \citet{doi:10.1162/REST00272} for pooled (static) production function estimates. Our goal is to investigate how the divergent properties of the CD test statistic might have influenced the choice between the First Difference (FD) estimator with yearly dummies and Pooled CCE (CCEP) estimators, as presented in columns 3 and 4, respectively, of the table mentioned above. Which of the two models is considered to be correctly specified has important consequences for the conclusions that can be drawn from the entire table. With regard to the coefficient of private R\&D investments, \citet{doi:10.1162/REST00272} report that a significance test for the corresponding slope coefficient cannot reject the null hypothesis in the FD model while it can in the CCE model.

\begin{table}[htbp!]
	\centering
	\footnotesize
	\caption{Cross-sectional dependence testing with R\& D data}
	\begin{tabular}{cc|c|ccc|cc}
		\hline
		\hline
		Estimator&Serial Correlation&$CD$&$q_{0.1}(CD_{W})$&$q_{0.5}(CD_{W})$&$q_{0.9}(CD_{W})$&$\overline{CD}_{W}$&$\overline{CD}_{W+}$\\
		\hline
		FD&No&-1.90&-1.41&       0.04&       0.78&-0.23&1.49\\
		&Yes&NA&-1.24&       0.04&       0.69&-0.20&1.52\\
		\hline
		CCEP&No&-2.95&-1.72&       -0.54&       1.72&-1.30&7.38\\
		&Yes&NA&-0.91&       -0.29&       0.91&-0.69&7.98\\
		\hline
		\hline
	\end{tabular}
	\begin{minipage}{0.9\textwidth}
		\caption*{\scriptsize{Notes. \emph{All results are obtained for Weighted Covariance Bias-corrected CD statistic based on $G=30$ Rademacher draws of $w_{i}$. $CD_{BC}$ denotes the analytical bias-corrected statistic. Serial Correlation with ``Yes'' option stands for the adjusted test statistic robust to the serial correlation, as described in Section \ref{section:serialcorrelation} in the Supplementary Appendix. $q_{\tau}(CD_{W})$ denotes the $\tau$'th sample quantile of $CD_{W}$ from $G=30$ draws.}}}
	\end{minipage}	
	\label{tab::resRD1}
\end{table}
As the results for $CD_{W}$ statistic allowing for serial-correlation correction, are similar to those without serial correlation, we only focus on the latter option. First of all, as we can see from Table \ref{tab::resRD1} the conclusions that we can draw from the original CD test are almost identical to those \citet{doi:10.1162/REST00272}, despite adjustments made in terms of the sample size. In particular, while the value of CD statistic based on CCE residuals imply rejection of the null hypothesis, no such conclusion is implied by FD (at least at a $5\%$ significance level). However, as given that the original test statistic suffers from the IPP, we further investigate if this conclusion also holds after proper bias-corrections. Motivated by finite sample evidence in Section \ref{section::MC}, we consider only the $CD_{W}$ statistic.

Taking our attention to the proposed $CD_{W}$ statistic based on random Rademacher weights, we can see that the corresponding values of $\overline{CD}_{W}$ indicate that both the FD and CCEP models generate residuals that are cross-sectionally uncorrelated. As we discussed in Section \ref{section::MC}, this conclusion might be partially attributed to the fact that for small values of $T$, the proposed statistic might lack power. For this reason, we also make use of the power-enhanced version $\overline{CD}_{W+}$ of our test statistic. Notice that for this dataset $2\sqrt{\ln(N)/T}\approx 0.84$, thus power-enhancement will pick up only very large values of the correlation coefficients that otherwise might be averaged away by the original CD statistic. As we can see from the corresponding column in Table \ref{tab::resRD1}, power enhancement does not alter the conclusion we draw for the FD residuals, as there are only 2 coefficients $\widehat{\rho}_{ij}$ above the threshold. However, for CCEP residuals the effect of power enhancement is non-negligible as there are 10 values of $\widehat{\rho}_{ij}$ above the threshold.

These results indicate that after appropriate adjustments, the empirical evidence presented in \citet{doi:10.1162/REST00272} that favor simple First-difference estimator, as opposed to the CCEP estimator, remain in place. As we can see from Table \ref{tab::resRD1}, this conclusion is unaltered after appropriate serial-correlation adjustments, thus serial-correlation cannot be the main factor affecting this ranking (unlike the discussion in \citealp{doi:10.1162/REST00272}). Instead, our results might indicate that the underlying CCE restrictions (rank condition and regressors with finite factor structure) might be violated, as these restrictions are irrelevant for the FD estimator with time-effects only.
\section{Conclusion}	
\label{section:conclusions}
This article documents how the estimation of common time-specific parameters using panel data causes the CD test of \citet{CDtest2004, doi:10.1080/07474938.2014.956623} to break down. Using popular additively and multiplicatively interacted specifications for individual- and time-specific components in the model errors, we show that the CD test statistic applied to residuals of correctly specified regression models is divergent under null hypothesis of cross-section independence. We can find an equivalent term under the alternative hypothesis which may balance out the leading diverging component of $CD$ and can lead to low power in small samples. The results documented in this article are interpreted as a manifestation of the incidental parameter problem (IPP) since they ultimately follow from the estimation of $\largeO(T)$ period-specific parameters. Our main theorems illustrate the pervasive nature of the IPP in this setup, given that the consequences of estimating time specific parameters do not disappear as the sample size increases.

Our proposed weighted CD test statistic achieves our primary goal of re-establishing asymptotic standard normal inference and hence constitutes an alternative to popular bias correction methods which circumvents the problems these approaches have in the present context. Our results have far reaching implications for empirical panel data analysis, where CD test has been widely used as a model selection/diagnostic tool. An illustration of how our theoretical results translate into applications is given via simulations and real datasets.

Finally, in this paper we assumed that the parameters in the linear model are estimated using a least squares objective function. If one deviates from this setup, and instead uses an (over-identified) GMM criterion function to estimate parameters, the usual GMM J-statistic is readily available for the purpose of testing residual cross-sectional correlation. Examples in a fixed-$T$ framework are given by \citet{Sarafidis2009}, \citet{AhnEtal2013}, and \citet{JuodisPseudo2014} among others. In the large $N,T$ setup average J-statistic as a model specification tool was explicitly used e.g. by \citet{JAE:JAE2338} for the CCE-GMM estimator. Hence given the scale of the problems with CD statistic documented in this paper, these alternative procedures (if applicable) are more appropriate.

\makeatletter\@input{auxsupp.tex}\makeatother

\end{document}


\bibpunct{(}{)}{,}{a}{}{,}
	\bibliographystyle{ecta}

\clearpage\newpage

\begin{center}
	\bigskip
	
	\textbf{\large Supplementary Appendix to }
	
	\bigskip
	
	{\large The Incidental Parameters Problem in Testing for Remaining Cross-section Correlation}
	
	{\large \bigskip}
	
	by
	
	\bigskip
	
	Art\={u}ras Juodis and Simon Reese
	
	\bigskip
	
	January 2021
\end{center}
\tableofcontents
\newpage
\section{Empirical application}
\label{appendix::empirical}
All of the empirical results presented assume the country-industry pair as the unit of analysis (panel group member $i$), of which we have $N = 119$. This panel is unbalanced where for Germany, Portugal, and Sweden the length of the available time-series is substantially shorter than for other countries.

Some of our theoretical results can be easily extended to cover unbalanced datasets. However in order to simplify our analysis we disregard from using the observations for Germany, Portugal, Sweden in the construction of the test statistics. Additionally, we remove two sectors from the Great Britain where observations for $t=2004, 2005$ are missing. This way we are left with $N=82$ and $T=25$.

\section{Additional MC results}
\label{sec:suppl_MC}
This section reports additional results for the Monte Carlo experiments in Section 5 of the main paper. The five tables below are equivalents to Tables 1-5 and their underlying setup differs only in that the errors in $y_{i,t}$ and $x_{i,t}$ are drawn from a standardized $\chi^2(2)$ distribution with zero expected value and unit variance.

\begin{table}[htbp]
	\centering
	\caption{Sample moments of original CD test statistic for remaining CSD under $\mathbb{H}_0$}
	\footnotesize
	\begin{tabular}{rr|rrrr|rrrr}
		\multicolumn{10}{l}{\textbf{Part A: Application of $CD$ to 2WFE residuals}}  \\
		\hline
		\hline
		& \multicolumn{1}{l|}{$c_{\sigma}$:} & 0.1   & 0.5   & 1     & 1.5   & 0.1   & 0.5   & 1     & 1.5 \\
		\multicolumn{1}{l}{N} & \multicolumn{1}{l|}{T} & \multicolumn{4}{c|}{Mean}     & \multicolumn{4}{c}{Variance $\times 100$} \\
		\hline
		25    & 25    & -3.35 & -3.31 & -3.20 & -2.97 & 1.28  & 1.77  & 3.94  & 9.92 \\
		25    & 50    & -4.91 & -4.87 & -4.72 & -4.43 & 0.54  & 0.89  & 2.70  & 8.46 \\
		25    & 100   & -7.08 & -7.02 & -6.82 & -6.43 & 0.23  & 0.60  & 2.71  & 9.99 \\
		25    & 200   & -10.11 & -10.02 & -9.75 & -9.21 & 0.10  & 0.49  & 3.35  & 14.57 \\
		50    & 25    & -3.30 & -3.26 & -3.14 & -2.91 & 1.03  & 1.38  & 2.91  & 7.33 \\
		50    & 50    & -4.86 & -4.80 & -4.65 & -4.33 & 0.38  & 0.65  & 1.99  & 6.37 \\
		50    & 100   & -7.00 & -6.94 & -6.73 & -6.32 & 0.14  & 0.36  & 1.69  & 6.41 \\
		50    & 200   & -10.00 & -9.91 & -9.63 & -9.06 & 0.06  & 0.28  & 1.94  & 8.66 \\
		100   & 25    & -3.28 & -3.24 & -3.11 & -2.86 & 0.88  & 1.16  & 2.35  & 5.87 \\
		100   & 50    & -4.83 & -4.78 & -4.62 & -4.30 & 0.26  & 0.45  & 1.38  & 4.50 \\
		100   & 100   & -6.97 & -6.90 & -6.68 & -6.26 & 0.09  & 0.23  & 1.06  & 4.22 \\
		100   & 200   & -9.95 & -9.86 & -9.57 & -8.99 & 0.03  & 0.16  & 1.07  & 4.77 \\
		200   & 25    & -3.27 & -3.23 & -3.10 & -2.85 & 0.73  & 0.96  & 1.95  & 4.90 \\
		200   & 50    & -4.82 & -4.76 & -4.60 & -4.28 & 0.21  & 0.35  & 1.01  & 3.24 \\
		200   & 100   & -6.95 & -6.88 & -6.66 & -6.23 & 0.07  & 0.17  & 0.75  & 2.95 \\
		200   & 200   & -9.92 & -9.83 & -9.54 & -8.94 & 0.02  & 0.11  & 0.71  & 3.21 \\
		\hline
		\hline
		& \multicolumn{1}{r}{} &       &       &       & \multicolumn{1}{r}{} &       &       &       &  \\
		\multicolumn{10}{l}{\textbf{Part B: Application of $CD$ to CCE residuals}} \\
		\hline
		\hline
		& \multicolumn{1}{l|}{$c_{\sigma}$:} & 0.1   & 0.5   & 1     & 1.5   & 0.1   & 0.5   & 1     & 1.5 \\
		\multicolumn{1}{l}{N} & \multicolumn{1}{l|}{T} & \multicolumn{4}{c|}{Mean}     & \multicolumn{4}{c}{Variance $\times 100$} \\
		\hline
		25    & 25    & -3.35 & -3.31 & -3.21 & -3.02 & 1.37  & 1.86  & 3.90  & 9.30 \\
		25    & 50    & -4.92 & -4.87 & -4.74 & -4.47 & 0.50  & 0.86  & 2.66  & 8.14 \\
		25    & 100   & -7.08 & -7.02 & -6.84 & -6.48 & 0.22  & 0.57  & 2.63  & 9.45 \\
		25    & 200   & -10.11 & -10.03 & -9.77 & -9.28 & 0.10  & 0.50  & 3.26  & 13.22 \\
		50    & 25    & -3.28 & -3.24 & -3.13 & -2.90 & 1.20  & 1.60  & 3.30  & 8.04 \\
		50    & 50    & -4.85 & -4.81 & -4.66 & -4.37 & 0.38  & 0.60  & 1.73  & 5.41 \\
		50    & 100   & -7.00 & -6.94 & -6.74 & -6.35 & 0.13  & 0.34  & 1.54  & 5.83 \\
		50    & 200   & -10.00 & -9.91 & -9.64 & -9.09 & 0.06  & 0.26  & 1.76  & 7.69 \\
		100   & 25    & -3.23 & -3.19 & -3.07 & -2.83 & 1.23  & 1.59  & 2.97  & 6.77 \\
		100   & 50    & -4.82 & -4.77 & -4.61 & -4.31 & 0.31  & 0.48  & 1.32  & 4.09 \\
		100   & 100   & -6.97 & -6.90 & -6.69 & -6.27 & 0.09  & 0.22  & 1.01  & 4.01 \\
		100   & 200   & -9.95 & -9.86 & -9.57 & -9.01 & 0.03  & 0.17  & 1.11  & 4.89 \\
		200   & 25    & -3.16 & -3.12 & -3.00 & -2.75 & 1.64  & 1.90  & 3.04  & 6.35 \\
		200   & 50    & -4.80 & -4.75 & -4.59 & -4.27 & 0.27  & 0.40  & 1.05  & 3.22 \\
		200   & 100   & -6.94 & -6.88 & -6.67 & -6.24 & 0.07  & 0.16  & 0.71  & 2.86 \\
		200   & 200   & -9.92 & -9.83 & -9.54 & -8.96 & 0.02  & 0.10  & 0.69  & 3.17 \\
		\hline
		\hline
	\end{tabular}%
	\par
	\begin{minipage}{0.67\textwidth}
		\smallskip
		\scriptsize Notes. \itshape $\epsi_{i,t}$ and $e_{i,t}$ are standardized $\chi^2(2)$ with zero expected value and unit variance. $\sigma _{i,y}^{2}=c_{\sigma }\left( \varsigma _{i,y}^{2}-2\right) /4+1$; $\varsigma _{i,y}^{2}$ is $\chi^2(2)$. In part B, $r=2$ and loadings $\vlambda_{i}$ are drawn from $U(0.5, 1.5)$.  $\mLambda_{i}$ has first element from $U(0.5, 1.5)$ and second from $U(-0.5,0.5)$. Factors $\vf_t$ drawn from $N(0,1)$. In part A, we restrict $\vf_t = \left[f_t^{(1)}, 1 \right]'$, $\vlambda_{i}= \left[1, \lambda_{i}^{(2)} \right]'$ and $\mLambda_{i}= \left[1, \Lambda_{i}^{(2)} \right]'$.
	\end{minipage}
	\label{tab:supp_CDmoments_H0}%
\end{table}%

\begin{table}[htbp]
	\centering
	\caption{Sample moments of original CD test statistic applied to 2WFE residuals under $\mathbb{H}_1$}
	\footnotesize
	\begin{tabular}{rrrr|rr|rr|rr}
		\hline
		\hline
		& \multicolumn{1}{r|}{} & \multicolumn{4}{c|}{Mean}     & \multicolumn{4}{c}{Variance} \\
		& \multicolumn{1}{l|}{$\vlambda_i$:} & \multicolumn{2}{c|}{symmetric} & \multicolumn{2}{c|}{skewed} & \multicolumn{2}{c|}{symmetric} & \multicolumn{2}{c}{skewed} \\
		& \multicolumn{1}{l|}{$\sigma^2_i$:} & \multicolumn{1}{l}{$\perp \vlambda_i$} & \multicolumn{1}{l|}{$f(\vlambda_i)$} & \multicolumn{1}{l}{$\perp \vlambda_i$} & \multicolumn{1}{l|}{$f(\vlambda_i)$} & \multicolumn{1}{l}{$\perp \vlambda_i$} & \multicolumn{1}{l|}{$f(\vlambda_i)$} & \multicolumn{1}{l}{$\perp \vlambda_i$} & \multicolumn{1}{l}{$f(\vlambda_i)$} \\
		
		\multicolumn{1}{l}{N} & \multicolumn{1}{l|}{T} &       &       &       &       &       &       &       &  \\ 		\hline
		25    & \multicolumn{1}{r|}{25} & -3.37 & -3.09 & -3.30 & -2.68 & 0.02  & 0.10  & 0.02  & 0.41 \\
		25    & \multicolumn{1}{r|}{50} & -4.86 & -4.50 & -4.78 & -3.89 & 0.01  & 0.13  & 0.02  & 0.70 \\
		25    & \multicolumn{1}{r|}{100} & -6.95 & -6.47 & -6.86 & -5.63 & 0.02  & 0.25  & 0.02  & 1.27 \\
		25    & \multicolumn{1}{r|}{200} & -9.88 & -9.25 & -9.76 & -8.06 & 0.02  & 0.39  & 0.02  & 2.23 \\
		50    & \multicolumn{1}{r|}{25} & -3.32 & -2.87 & -3.25 & -2.19 & 0.01  & 0.14  & 0.02  & 0.58 \\
		50    & \multicolumn{1}{r|}{50} & -4.81 & -4.23 & -4.73 & -3.26 & 0.01  & 0.18  & 0.01  & 0.94 \\
		50    & \multicolumn{1}{r|}{100} & -6.88 & -6.09 & -6.78 & -4.78 & 0.01  & 0.31  & 0.01  & 1.55 \\
		50    & \multicolumn{1}{r|}{200} & -9.78 & -8.73 & -9.65 & -6.85 & 0.01  & 0.52  & 0.02  & 2.80 \\
		100   & \multicolumn{1}{r|}{25} & -3.30 & -2.49 & -3.23 & -1.34 & 0.01  & 0.27  & 0.01  & 0.98 \\
		100   & \multicolumn{1}{r|}{50} & -4.78 & -3.76 & -4.70 & -2.13 & 0.01  & 0.32  & 0.01  & 1.42 \\
		100   & \multicolumn{1}{r|}{100} & -6.83 & -5.47 & -6.73 & -3.19 & 0.01  & 0.49  & 0.01  & 2.41 \\
		100   & \multicolumn{1}{r|}{200} & -9.73 & -7.84 & -9.59 & -4.65 & 0.01  & 0.83  & 0.01  & 4.14 \\
		200   & \multicolumn{1}{r|}{25} & -3.30 & -1.80 & -3.22 & 0.33  & 0.01  & 0.56  & 0.01  & 1.97 \\
		200   & \multicolumn{1}{r|}{50} & -4.77 & -2.88 & -4.68 & 0.07  & 0.01  & 0.60  & 0.01  & 2.68 \\
		200   & \multicolumn{1}{r|}{100} & -6.82 & -4.28 & -6.71 & -0.12 & 0.01  & 0.94  & 0.01  & 4.25 \\
		200   & 200   & -9.70 & -6.22 & -9.56 & -0.43 & 0.01  & 1.48  & 0.01  & 7.28 \\
		\hline
		\hline
	\end{tabular}%
	\par
	\begin{minipage}{0.67\textwidth}
		\smallskip
		\scriptsize Notes. \itshape The model has a factor error structure with 3 factors. $\epsi_{i,t}$ and $e_{i,t}$ are standardized $\chi^2(2)$ with zero expected value and unit variance. ``$\sigma_{i}^2$: $\perp \vlambda_{i}$'' means that $\sigma _{i,y}^{2}=\left( \varsigma _{i,y}^{2}-2\right) /4+1$ where $\varsigma _{i,y}^{2}$ is $\chi^2(2)$. For ``$\sigma_{i}^2$: $f(\vlambda_{i})$'', we let $\sigma_{i,y}^{2} = d_{\sigma} T^{-1} \sum_{t=1}^T (\vlambda_{i}' \vf_t)^2$.
		
		The case ``$\vlambda_{i}$: symmetric'' corresponds to drawing $\vlambda_{i}$ from $U(0.5, 1.5)$. ``$\vlambda_{i}$: skewed'' loadings are from a standardized $\chi ^{2}\left(2\right)$ distribution with mean and variance equal to 1.
		
		$\mLambda_{i}$ has first element from $U(0.5, 1.5)$ and all others from $U(-0.5,0.5)$. Factors $\vf_t$ are drawn from $N(0,1)$.
	\end{minipage}
	\label{tab:supp_CDmoments_H1_FE}%
\end{table}%

\begin{table}[htbp]
	\centering
	\caption{Sample moments of original CD test statistic applied to CCE residuals under $\mathbb{H}_1$}
	\footnotesize
	\begin{tabular}{rr|rr|rr|rr|rr}
		\hline
		\hline
		&       & \multicolumn{4}{c|}{Mean}     & \multicolumn{4}{c}{Variance} \\
		& \multicolumn{1}{l|}{$\vlambda_i$:} & \multicolumn{2}{c|}{symmetric} & \multicolumn{2}{c|}{skewed} & \multicolumn{2}{c|}{symmetric} & \multicolumn{2}{c}{skewed} \\
		& \multicolumn{1}{l|}{$\sigma^2_i$:} & \multicolumn{1}{l}{$\perp \vlambda_i$} & \multicolumn{1}{l|}{$f(\vlambda_i)$} & \multicolumn{1}{l}{$\perp \vlambda_i$} & \multicolumn{1}{l|}{$f(\vlambda_i)$} & \multicolumn{1}{l}{$\perp \vlambda_i$} & \multicolumn{1}{l|}{$f(\vlambda_i)$} & \multicolumn{1}{l}{$\perp \vlambda_i$} & \multicolumn{1}{l}{$f(\vlambda_i)$} \\
		\multicolumn{1}{l}{N} & \multicolumn{1}{l|}{T} &       &       &       &       &       &       &       &  \\
		\hline
		25    & 25    & -3.27 & -3.21 & -3.25 & -3.01 & 0.03  & 0.04  & 0.03  & 0.19 \\
		25    & 50    & -4.78 & -4.72 & -4.76 & -4.42 & 0.02  & 0.03  & 0.02  & 0.25 \\
		25    & 100   & -6.87 & -6.80 & -6.86 & -6.40 & 0.02  & 0.04  & 0.02  & 0.41 \\
		25    & 200   & -9.80 & -9.70 & -9.78 & -9.16 & 0.03  & 0.06  & 0.02  & 0.68 \\
		50    & 25    & -3.21 & -3.14 & -3.19 & -2.93 & 0.02  & 0.04  & 0.02  & 0.15 \\
		50    & 50    & -4.71 & -4.65 & -4.70 & -4.37 & 0.02  & 0.02  & 0.01  & 0.17 \\
		50    & 100   & -6.78 & -6.70 & -6.78 & -6.32 & 0.01  & 0.03  & 0.01  & 0.26 \\
		50    & 200   & -9.68 & -9.58 & -9.67 & -9.04 & 0.02  & 0.04  & 0.01  & 0.42 \\
		100   & 25    & -3.16 & -3.10 & -3.14 & -2.90 & 0.02  & 0.03  & 0.02  & 0.10 \\
		100   & 50    & -4.67 & -4.61 & -4.66 & -4.34 & 0.01  & 0.02  & 0.01  & 0.11 \\
		100   & 100   & -6.73 & -6.66 & -6.73 & -6.28 & 0.01  & 0.02  & 0.01  & 0.16 \\
		100   & 200   & -9.62 & -9.52 & -9.61 & -8.99 & 0.01  & 0.03  & 0.01  & 0.27 \\
		200   & 25    & -3.11 & -3.05 & -3.09 & -2.84 & 0.02  & 0.03  & 0.02  & 0.09 \\
		200   & 50    & -4.66 & -4.58 & -4.64 & -4.31 & 0.01  & 0.02  & 0.01  & 0.09 \\
		200   & 100   & -6.71 & -6.64 & -6.71 & -6.26 & 0.01  & 0.02  & 0.01  & 0.11 \\
		200   & 200   & -9.59 & -9.48 & -9.58 & -8.96 & 0.01  & 0.02  & 0.01  & 0.21 \\
		\hline
		\hline
	\end{tabular}%
	\par
	\begin{minipage}{0.7\textwidth}
		\smallskip
		\scriptsize Notes. \itshape See Table \ref{tab:supp_CDmoments_H1_FE}.
	\end{minipage}
	\label{tab:supp_CDmoments_H1_CCE}%
\end{table}%

\begin{table}[htbp]
	\centering
	\caption{Rejection rates for weighted CD test statistic when applied to 2WFE residuals}
	\footnotesize
	\begin{tabular}{rr|rrr|rrr|rrr|rrr}
		\multicolumn{14}{l}{\textbf{Part A: Size}} \\
		\hline
		\hline
		& \multicolumn{1}{l|}{$\vlambda_i$:} & \multicolumn{6}{c|}{symmetric}                & \multicolumn{6}{c}{skewed} \\
		& \multicolumn{1}{l|}{$\sigma^2_i$:} & \multicolumn{3}{c|}{$\perp \vlambda_i$} & \multicolumn{3}{c|}{$f(\vlambda_i)$} & \multicolumn{3}{c|}{$\perp \vlambda_i$} & \multicolumn{3}{c}{$f(\vlambda_i)$} \\
		\multicolumn{1}{l}{N} & \multicolumn{1}{l|}{T} & \multicolumn{1}{l}{$CD_{W}$} & \multicolumn{1}{l}{$CD_{W+}$} & \multicolumn{1}{l|}{$CD_{BC}$} & \multicolumn{1}{l}{$CD_{W}$} & \multicolumn{1}{l}{$CD_{W+}$} & \multicolumn{1}{l|}{$CD_{BC}$} & \multicolumn{1}{l}{$CD_{W}$} & \multicolumn{1}{l}{$CD_{W+}$} & \multicolumn{1}{l|}{$CD_{BC}$} & \multicolumn{1}{l}{$CD_{W}$} & \multicolumn{1}{l}{$CD_{W+}$} & \multicolumn{1}{l}{$CD_{BC}$} \\
		\hline
		25    & 25    & 4.9   & 5.4   & 2.5   & 5.1   & 5.6   & 2.9   & 5.0   & 5.9   & 2.5   & 6.1   & 6.6   & 2.7 \\
		25    & 50    & 6.5   & 7.0   & 3.9   & 6.8   & 7.4   & 3.9   & 5.7   & 6.5   & 3.5   & 6.2   & 6.8   & 3.3 \\
		25    & 100   & 8.8   & 9.0   & 5.0   & 8.0   & 8.2   & 4.6   & 8.4   & 8.8   & 5.1   & 8.0   & 8.1   & 3.9 \\
		25    & 200   & 10.8  & 10.9  & 5.7   & 9.7   & 9.8   & 6.1   & 11.0  & 11.0  & 7.0   & 10.9  & 11.1  & 4.3 \\
		50    & 25    & 5.1   & 5.9   & 2.8   & 5.7   & 6.2   & 3.6   & 5.7   & 6.7   & 2.3   & 5.1   & 5.7   & 2.9 \\
		50    & 50    & 5.0   & 6.6   & 4.3   & 4.6   & 5.9   & 4.0   & 6.5   & 7.6   & 3.4   & 5.4   & 6.8   & 3.1 \\
		50    & 100   & 6.6   & 7.1   & 3.9   & 5.9   & 6.7   & 3.2   & 6.6   & 7.0   & 4.6   & 5.9   & 6.5   & 3.5 \\
		50    & 200   & 6.6   & 6.8   & 4.7   & 5.7   & 5.9   & 4.4   & 5.8   & 6.2   & 4.4   & 6.1   & 6.3   & 3.5 \\
		100   & 25    & 5.3   & 5.8   & 2.8   & 5.3   & 5.5   & 3.1   & 4.6   & 4.7   & 3.5   & 5.0   & 5.4   & 3.7 \\
		100   & 50    & 5.9   & 8.7   & 5.1   & 4.5   & 6.7   & 4.1   & 4.6   & 7.4   & 3.9   & 4.8   & 7.4   & 4.0 \\
		100   & 100   & 5.6   & 6.6   & 4.5   & 5.2   & 7.0   & 4.3   & 4.9   & 6.1   & 5.0   & 4.2   & 5.9   & 2.8 \\
		100   & 200   & 5.8   & 6.3   & 4.0   & 5.5   & 6.2   & 4.8   & 5.5   & 5.9   & 3.1   & 5.9   & 6.3   & 4.1 \\
		200   & 25    & 4.4   & 4.5   & 2.9   & 4.8   & 4.8   & 3.5   & 5.0   & 5.1   & 3.2   & 4.8   & 4.9   & 3.1 \\
		200   & 50    & 4.6   & 7.4   & 3.6   & 4.8   & 8.6   & 3.8   & 4.5   & 8.0   & 3.7   & 5.8   & 8.7   & 4.0 \\
		200   & 100   & 4.7   & 8.3   & 4.2   & 5.5   & 9.7   & 4.7   & 5.2   & 8.6   & 4.1   & 5.7   & 9.4   & 4.1 \\
		200   & 200   & 4.2   & 5.4   & 5.0   & 5.4   & 7.4   & 5.0   & 5.9   & 6.8   & 5.2   & 4.9   & 6.6   & 3.6 \\
		\hline
		\hline
		& \multicolumn{1}{r}{} &       &       & \multicolumn{1}{r}{} &       &       & \multicolumn{1}{r}{} &       &       & \multicolumn{1}{r}{} &       &       &  \\
		\multicolumn{14}{l}{\textbf{Part B: Power}} \\
		\hline
		\hline
		& \multicolumn{1}{l|}{$\vlambda_i$:} & \multicolumn{6}{c|}{symmetric}                & \multicolumn{6}{c}{skewed} \\
		& \multicolumn{1}{l|}{$\sigma^2_i$:} & \multicolumn{3}{c|}{$\perp \vlambda_i$} & \multicolumn{3}{c|}{$f(\vlambda_i)$} & \multicolumn{3}{c|}{$\perp \vlambda_i$} & \multicolumn{3}{c}{$f(\vlambda_i)$} \\
		\multicolumn{1}{l}{N} & \multicolumn{1}{l|}{T} & \multicolumn{1}{l}{$CD_{W}$} & \multicolumn{1}{l}{$CD_{W+}$} & \multicolumn{1}{l|}{$CD_{BC}$} & \multicolumn{1}{l}{$CD_{W}$} & \multicolumn{1}{l}{$CD_{W+}$} & \multicolumn{1}{l|}{$CD_{BC}$} & \multicolumn{1}{l}{$CD_{W}$} & \multicolumn{1}{l}{$CD_{W+}$} & \multicolumn{1}{l|}{$CD_{BC}$} & \multicolumn{1}{l}{$CD_{W}$} & \multicolumn{1}{l}{$CD_{W+}$} & \multicolumn{1}{l}{$CD_{BC}$} \\
		\hline
		25    & 25    & 12.6  & 26.6  & 8.0   & 10.7  & 14.0  & 5.8   & 12.3  & 19.6  & 14.7  & 14.3  & 20.0  & 24.0 \\
		25    & 50    & 21.2  & 87.4  & 14.8  & 17.5  & 46.1  & 12.7  & 18.0  & 59.6  & 32.9  & 23.2  & 63.7  & 48.4 \\
		25    & 100   & 34.6  & 99.9  & 26.4  & 27.3  & 95.2  & 20.9  & 32.3  & 90.8  & 50.1  & 36.3  & 96.7  & 66.6 \\
		25    & 200   & 50.4  & 100 & 39.7  & 43.6  & 100 & 34.7  & 44.1  & 99.3  & 64.4  & 51.3  & 100 & 80.4 \\
		50    & 25    & 11.7  & 21.1  & 8.1   & 10.7  & 11.7  & 5.7   & 11.7  & 20.9  & 35.2  & 12.6  & 16.0  & 58.9 \\
		50    & 50    & 19.6  & 98.3  & 15.7  & 17.5  & 56.0  & 12.8  & 18.7  & 85.5  & 59.5  & 23.0  & 79.7  & 86.0 \\
		50    & 100   & 35.4  & 100 & 27.0  & 26.3  & 99.7  & 22.2  & 29.7  & 99.8  & 80.2  & 37.5  & 100 & 95.4 \\
		50    & 200   & 49.3  & 100 & 40.6  & 43.5  & 100 & 37.4  & 45.3  & 100 & 90.7  & 53.2  & 100 & 99.1 \\
		100   & 25    & 11.1  & 14.4  & 6.5   & 10.6  & 10.9  & 5.1   & 10.4  & 19.1  & 73.4  & 12.0  & 12.9  & 93.5 \\
		100   & 50    & 19.9  & 99.9  & 15.0  & 16.3  & 66.4  & 13.3  & 18.1  & 98.9  & 92.9  & 22.3  & 94.6  & 99.7 \\
		100   & 100   & 33.3  & 100 & 25.8  & 26.2  & 100 & 23.8  & 30.8  & 100 & 98.8  & 36.7  & 100 & 100 \\
		100   & 200   & 48.7  & 100 & 41.3  & 41.9  & 100 & 37.0  & 46.2  & 100 & 100 & 52.3  & 100 & 100 \\
		200   & 25    & 12.7  & 12.7  & 8.2   & 10.8  & 10.8  & 6.6   & 11.9  & 13.9  & 96.5  & 12.9  & 13.0  & 100 \\
		200   & 50    & 18.9  & 100 & 16.6  & 17.9  & 73.1  & 13.3  & 18.8  & 100 & 100 & 22.1  & 98.2  & 100 \\
		200   & 100   & 33.9  & 100 & 27.4  & 26.3  & 100 & 21.5  & 32.8  & 100 & 100 & 36.9  & 100 & 100 \\
		200   & 200   & 50.9  & 100 & 41.9  & 43.6  & 100 & 37.2  & 46.7  & 100 & 100 & 54.6  & 100 & 100 \\
		\hline
		\hline
	\end{tabular}%
	\par
	\begin{minipage}{1.05\textwidth}
		\smallskip
		\scriptsize Notes. \itshape In Part A the model has two factors which are restricted as noted in Table \ref{tab:CDmoments_H0}. Part B corresponds to a model with factor error structure and 3 factors. For details on all other model parameters, see Table \ref{tab:supp_CDmoments_H1_FE}.
		
		$CD_W$ is the weighted CD test statistic introduced in Theorem \ref{thm:CD_W_combined}. $CD_{W+}$ is its power-enhanced refinement. $CD_{BC}$ is a CD test statistic with analytic bias correction.
	\end{minipage}
	\label{tab:supp_CDW_rejections_FE}
\end{table}%

\begin{table}[htbp]
	\centering
	\caption{Rejection rates for weighted CD test statistic when applied to CCE residuals}
	\footnotesize
	\begin{tabular}{rr|rrr|rrr|rrr|rrr}
		\multicolumn{14}{l}{\textbf{Part A: Size}} \\
		\hline
		\hline
		& \multicolumn{1}{l|}{$\vlambda_i$:} & \multicolumn{6}{c|}{symmetric}                & \multicolumn{6}{c}{skewed} \\
		& \multicolumn{1}{l|}{$\sigma^2_i$:} & \multicolumn{3}{c|}{$\perp \vlambda_i$} & \multicolumn{3}{c|}{$f(\vlambda_i)$} & \multicolumn{3}{c|}{$\perp \vlambda_i$} & \multicolumn{3}{c}{$f(\vlambda_i)$} \\
		\multicolumn{1}{l}{N} & \multicolumn{1}{l|}{T} & \multicolumn{1}{l}{$CD_{W}$} & \multicolumn{1}{l}{$CD_{W+}$} & \multicolumn{1}{l|}{$CD_{BC}$} & \multicolumn{1}{l}{$CD_{W}$} & \multicolumn{1}{l}{$CD_{W+}$} & \multicolumn{1}{l|}{$CD_{BC}$} & \multicolumn{1}{l}{$CD_{W}$} & \multicolumn{1}{l}{$CD_{W+}$} & \multicolumn{1}{l|}{$CD_{BC}$} & \multicolumn{1}{l}{$CD_{W}$} & \multicolumn{1}{l}{$CD_{W+}$} & \multicolumn{1}{l}{$CD_{BC}$} \\
		\hline
		25    & 25    & 6.2   & 6.8   & 4.6 & 5.7   & 6.4   & 4.4   & 5.6   & 6.1   & 4.7   & 5.6   & 6.6   & 5.9 \\
		25    & 50    & 7.4   & 8.1   & 5.7 & 5.8   & 6.6   & 5.9   & 7.0   & 7.9   & 6.1   & 6.7   & 7.3   & 8.9 \\
		25    & 100   & 8.3   & 8.5   & 6.6 & 8.0   & 7.9   & 7.8   & 9.0   & 9.7   & 6.4   & 8.2   & 8.8   & 14.5 \\
		25    & 200   & 9.2   & 10.1  & 7.4 & 11.3  & 11.6  & 8.8   & 10.9  & 12.2  & 9.5   & 13.0  & 14.0  & 24.2 \\
		50    & 25    & 5.1   & 5.8   & 4.1 & 4.4   & 5.0   & 4.8   & 5.9   & 6.6   & 4.9   & 5.5   & 6.7   & 8.0 \\
		50    & 50    & 5.1   & 6.5   & 5.1 & 6.1   & 7.0   & 4.2   & 5.5   & 7.0   & 5.8   & 7.2   & 8.0   & 8.1 \\
		50    & 100   & 5.5   & 6.0   & 4.7 & 5.2   & 5.5   & 6.3   & 6.3   & 6.6   & 4.7   & 6.6   & 7.2   & 11.8 \\
		50    & 200   & 6.4   & 6.5   & 5.9 & 6.9   & 6.9   & 6.7   & 6.0   & 6.6   & 6.9   & 7.4   & 7.9   & 21.9 \\
		100   & 25    & 5.2   & 5.4   & 5.9 & 5.5   & 5.7   & 6.7   & 5.4   & 5.8   & 6.6   & 5.5   & 5.7   & 8.9 \\
		100   & 50    & 5.0   & 7.4   & 4.3 & 5.3   & 8.3   & 5.9   & 5.3   & 7.0   & 4.5   & 5.1   & 6.9   & 8.0 \\
		100   & 100   & 5.5   & 6.6   & 5.1 & 5.1   & 7.0   & 5.2   & 5.8   & 7.0   & 5.4   & 5.1   & 6.9   & 11.3 \\
		100   & 200   & 5.0   & 5.5   & 4.8 & 6.0   & 6.2   & 5.1   & 5.2   & 5.9   & 5.7   & 5.8   & 6.1   & 19.2 \\
		200   & 25    & 5.8   & 5.8   & 11.7 & 5.1   & 5.1   & 12.7  & 6.3   & 6.3   & 11.7  & 4.9   & 5.0   & 16.3 \\
		200   & 50    & 6.3   & 9.4   & 4.5 & 5.6   & 9.7   & 5.3   & 5.0   & 9.1   & 5.0   & 5.6   & 9.8   & 9.5 \\
		200   & 100   & 5.4   & 8.5   & 5.5 & 5.1   & 8.3   & 5.8   & 4.9   & 8.3   & 4.9   & 5.1   & 8.4   & 10.8 \\
		200   & 200   & 5.9   & 7.0   & 5.3 & 4.8   & 5.9   & 5.9   & 4.8   & 6.7   & 4.7   & 5.0   & 7.1   & 15.8 \\
		\hline
		\hline
		\multicolumn{14}{r}{}  \\
		\multicolumn{14}{l}{\textbf{Part B: Power}} \\
		\hline
		\hline
		& \multicolumn{1}{l|}{$\vlambda_i$:} & \multicolumn{6}{c|}{symmetric}                & \multicolumn{6}{c}{skewed} \\
		& \multicolumn{1}{l|}{$\sigma^2_i$:} & \multicolumn{3}{c|}{$\perp \vlambda_i$} & \multicolumn{3}{c|}{$f(\vlambda_i)$} & \multicolumn{3}{c|}{$\perp \vlambda_i$} & \multicolumn{3}{c}{$f(\vlambda_i)$} \\
		\multicolumn{1}{l}{N} & \multicolumn{1}{l|}{T} & \multicolumn{1}{l}{$CD_{W}$} & \multicolumn{1}{l}{$CD_{W+}$} & \multicolumn{1}{l|}{$CD_{BC}$} & \multicolumn{1}{l}{$CD_{W}$} & \multicolumn{1}{l}{$CD_{W+}$} & \multicolumn{1}{l|}{$CD_{BC}$} & \multicolumn{1}{l}{$CD_{W}$} & \multicolumn{1}{l}{$CD_{W+}$} & \multicolumn{1}{l|}{$CD_{BC}$} & \multicolumn{1}{l}{$CD_{W}$} & \multicolumn{1}{l}{$CD_{W+}$} & \multicolumn{1}{l}{$CD_{BC}$} \\
		\hline
		25    & 25    & 10.4  & 16.3  & 5.7   & 9.2   & 10.7  & 5.6   & 8.5   & 13.4  & 6.2   & 9.7   & 12.8  & 7.1 \\
		25    & 50    & 15.0  & 47.6  & 9.9   & 10.2  & 19.8  & 8.9   & 12.5  & 37.4  & 11.8  & 15.2  & 31.4  & 14.5 \\
		25    & 100   & 22.6  & 83.9  & 14.0  & 17.6  & 50.6  & 13.2  & 20.5  & 64.4  & 18.7  & 23.3  & 67.0  & 24.8 \\
		25    & 200   & 33.6  & 97.6  & 22.5  & 26.0  & 85.7  & 22.5  & 28.2  & 86.6  & 28.8  & 34.4  & 91.1  & 42.4 \\
		50    & 25    & 9.7   & 12.6  & 5.7   & 7.8   & 8.5   & 4.9   & 8.7   & 15.0  & 5.9   & 8.6   & 10.3  & 7.4 \\
		50    & 50    & 11.5  & 67.0  & 8.6   & 11.5  & 23.5  & 8.2   & 10.9  & 58.8  & 9.4   & 13.5  & 42.7  & 14.9 \\
		50    & 100   & 20.6  & 98.4  & 14.7  & 14.7  & 68.5  & 13.9  & 18.1  & 91.6  & 21.6  & 21.2  & 90.4  & 28.1 \\
		50    & 200   & 28.0  & 100 & 23.5  & 24.3  & 98.6  & 22.9  & 26.6  & 99.0  & 32.1  & 33.2  & 99.6  & 43.0 \\
		100   & 25    & 9.3   & 9.9   & 4.8   & 7.5   & 7.8   & 5.5   & 9.0   & 13.0  & 6.2   & 9.0   & 9.3   & 8.5 \\
		100   & 50    & 12.6  & 80.4  & 6.4   & 9.3   & 25.4  & 7.6   & 12.2  & 82.9  & 11.9  & 13.1  & 54.9  & 15.9 \\
		100   & 100   & 18.2  & 100 & 12.5  & 15.2  & 85.8  & 12.8  & 17.6  & 99.7  & 21.8  & 22.2  & 98.9  & 29.5 \\
		100   & 200   & 29.4  & 100 & 21.9  & 23.6  & 99.9  & 24.3  & 28.6  & 100 & 36.2  & 29.7  & 100 & 46.7 \\
		200   & 25    & 9.4   & 9.7   & 7.9   & 8.1   & 8.1   & 6.6   & 7.5   & 8.1   & 7.6   & 10.7  & 10.8  & 8.6 \\
		200   & 50    & 12.1  & 90.8  & 9.4   & 10.9  & 33.4  & 7.2   & 11.9  & 95.2  & 11.2  & 12.8  & 64.5  & 17.2 \\
		200   & 100   & 20.0  & 100 & 13.0  & 14.8  & 96.3  & 13.0  & 16.6  & 100 & 21.8  & 21.1  & 99.8  & 29.8 \\
		200   & 200   & 27.2  & 100 & 22.6  & 23.5  & 100 & 21.7  & 26.6  & 100 & 42.4  & 31.0  & 100 & 49.0 \\
		\hline
		\hline
	\end{tabular}%
	\par
	\begin{minipage}{1.05\textwidth}
		\smallskip
		\scriptsize Notes. \itshape In Part A the model has two factors which are restricted as noted in Table \ref{tab:supp_CDmoments_H0}. Part B corresponds to a model with factor error structure and 3 factors. For details on all other model parameters, see Table \ref{tab:supp_CDmoments_H1_FE}.
		
		$CD_W$ is the weighted CD test statistic introduced in Theorem \ref{thm:CD_W_combined}. $CD_{W+}$ is its power-enhanced refinement. $CD_{BC}$ is a CD test statistic with analytic bias correction.
	\end{minipage}
	\label{tab:addlabel}%
\end{table}
\clearpage
\section{Additional theoretical contributions}
\label{sec:suppl_theory}

\subsection{Power analysis: Additive}
\label{section:power}

To simplify the discussion, we disregard from the effect of covariates on the variable of interest $y_{i,t}$ and assume that the true model is given by a pure static factor model, amounting to model \eqref{eq:themodelFactor} with $\vbeta=\vzeros$. Extending our results to the case of general $\vbeta$ is possible, although formally cumbersome.\footnote{To appreciate this point note that, even if the unobserved heterogeneity driving the data is misspecified, least squares estimates of $\vbeta$ are consistent for the pseudo-true parameter vector $\vdelta =\underset{N,T \to \infty}{\plim} \left[\left( \sum_{i=1}^N \breve{\mX}_i' \breve{\mX}_i \right)^{-1} \left( \sum_{i=1}^N \breve{\mX}_i' \breve{\vy}_i \right)\right]$ with $\breve{\mX}_i$ and $\breve{\vy}_i$ being orthogonal to estimates of the assumed sources of unobserved heterogeneity.}
In a first instance, suppose that a researcher erroneously assumes cross-section dependence to stem from time fixed effects. The corresponding mis-specified model is formally given by
\begin{align}
\label{eq:2WFE_H1}
	y_{i,t} = \mu_{y,t} + \nu_{i,t},
\end{align}
where $\mu_{y,t} = \vf_t' \E \left[ \vlambda_{i} \right]$ and $\nu_{i,t} = \vf_t' \left(\vlambda_{i} -\E\left[ \vlambda_{i} \right]\right) + \tilde{\vepsi}_{i,t}$.
Deviations of the data from their cross-section averages can accordingly be written as
\begin{align}
	\label{eq:themodel_2wfe_H1}
	\widehat{\nu}_{i,t} &= y_{i,t} - N^{-1}\sum_{i=1}^N y_{i,t} \\
	&= \vf_t' \tilde{\vlambda}_i + \tilde{\vepsi}_{i,t} \notag
\end{align}
where $\tilde{\vlambda}_{i} = \vlambda_{i} - N^{-1}\sum_{i=1}^N \vlambda_{i}$ and $\tilde{\vepsi}_{i,t} = \epsi_{i,t}- N^{-1}\sum_{i=1}^N \epsi_{i,t}$. These deviations are sample equivalents of the composite error term $\nu_{i,t}$ in the mis-specified time fixed effects model \eqref{eq:2WFE_H1}. Let the variance of this composite error be defined as $\varsigma_{\nu,i}^2 = \E\left[ \nu_{i,t}^2 ~|~  \sigma_{i}, \vlambda_{i} \right] = \left( \vlambda_{i}-\E\left[\vlambda_{i}\right] \right) ^{\prime }\mSigma_{F}\left( \vlambda_{i}-\E\left[ \vlambda
_{i}\right] \right) +\sigma _{i}^{2}$ and let $\overline{\varsigma_{\nu}^k}= N^{-1} \sum_{i=1}^N \varsigma_{\nu,i}^k, \; k \in \mathbb{Z} $ be shorthand notation for cross-section averages of $\varsigma_{\nu,i}^k$. Finally, define the stacked vector of the residual error term as $ \vnu_{i}=\left[ \nu_{i,1},\nu_{i,2}, \ldots,\nu_{i,T} \right]'$.

For the sake of simplicity, we  assume that $\varsigma_{\nu,i}^2$ is known and used to standardize regression residuals when constructing the CD test statistic. This greatly simplifies the proofs of the following proposition while leaving the main results of this section qualitatively unaffected (as we formally show for results under $\mathbb{H}_{0}$ in the main text). We confirm this conjecture with corresponding simulations in Section \ref{section::MC}, allowing for both unknown error variances and unknown, general slope coefficients $\vbeta$.
Given our current setup, let the CD test statistic constructed from $\widehat{\nu}_{i,t}$ be expressed by
\begin{align}
\label{eq:CD_H1}
	CD_{\mathbb{H}_1} = \sqrt{\frac{2}{TN(N-1)}}\sum_{i=2}^N\sum_{j=1}^N \sum_{t=1}^T \varsigma_{\nu,i}^{-1} \widehat{\nu}_{i,t} \widehat{\nu}_{j,t}\varsigma_{\nu,j}^{-1}.
\end{align}
The properties of this test statistic are characterized as follows:

\begin{proposition}
	\label{prop:CD_2WFE_H1}Suppose that the true model is given by \eqref{eq:2WFE_H1} and that its components satisfy Assumptions
	\ref{ass:errors}--\ref{ass:rank} and $\varsigma_{\nu,i}\in [\delta;M]$. Let the CD test statistic $CD_{\mathbb{H}_1}$ be constructed from the cross-sectionally demeaned data \eqref{eq:themodel_2wfe_H1}. Then,
	\begin{align}
	\label{eq:CDH1_2WFE}
		CD_{\mathbb{H}_{1}} &= \sqrt{\frac{2}{TN\left( N-1\right)}} \sum_{i=2}^{N}\sum_{j=1}^{i-1} \left(\varsigma_{\nu,i}^{-1} - \E\left[\varsigma_{\nu,i}^{-1}\right] \right) \vnu_{i}^{\prime}\vnu_{j} \left(\varsigma_{\nu,j}^{-1} - \E\left[\varsigma_{\nu,j}^{-1}\right] \right) + \sqrt{T}\Xi_{\mathbb{H}_{1}} \notag \\
		&+ \largeO_{P}\left( N^{-1}\sqrt{T}\right),
	\end{align}%
	where
	\begin{align*}
		\Xi_{\mathbb{H}_1} &= \sqrt{\frac{N}{2\left(N-1\right)}} \frac{1}{NT} \sum_{i=1}^N \vnu_{i}' \vnu_{i} \left[\left(\overline{\varsigma_{\nu}^{-1}}\right)^2 - 2\overline{\varsigma_{\nu}^{-1}}\varsigma_{\nu,i}^{-1} \right].
	\end{align*}%
	Furthermore, the leading stochastic term $\sqrt{\frac{1%
		}{2TN\left( N-1\right) }}\sum_{i=2}^{N}\sum_{j=1}^{i-1}\left(
	\varsigma _{\nu,i}^{-1}-\E\left[ \varsigma _{\nu,i}^{-1}\right] \right)
	\vnu_{i}^{\prime}\vnu_{j}\left( \varsigma _{\nu,j}^{-1}-\E\left[ \varsigma
	_{v,j}^{-1}\right] \right) $\ has expected value $N\sqrt{T}\cov
	\left[ \varsigma _{v,1}^{-1},\vlambda_{1}^{\prime }\right]
	\mSigma_{F}\cov\left[ \vlambda_{1},\varsigma _{v,1}^{-1}\right]$. In the special case $\cov\left[ \vlambda_{1},\varsigma_{v,1}^{-1}\right]
	=\vzeros $ the order of this leading term is determined by its variance which diverges at rate $T$.
\end{proposition}

Proposition \ref{prop:CD_2WFE_H1} establishes a decomposition of the CD test statistic that is almost identical to that obtained under the null
hypothesis: Analogous to Theorem \ref{theorem::additiveHetero}, $CD_{\mathbb{H}_{1}}$ consists of a leading stochastic component which
reflects a CD test statistic with incorrect normalization as well as an equivalent to the bias term in Eq. \eqref{eq::theorem_additive_hetero} which we denote by $\sqrt{T}\Xi_{\mathbb{H}_{1}}$. The key difference is that both components are functions of the composite model error $\nu_{i,t}=\vf_{t}^{\prime}\left( \vlambda_{i}-\E\left[\vlambda_{i}\right] \right) +\epsi_{i,t}$ rather than the true errors $\epsi_{i,t}$.

Noticeable differences with regards to Theorem \ref{theorem::additiveHetero} are given as far as the leading stochastic component of $CD_{\mathbb{H}_{1}}$, the first term in Eq. \eqref{eq:CDH1_2WFE}, is concerned. Its mean is a
positive-definite quadratic form and will therefore generally diverge to $+\infty$ at rate $N\sqrt{T}$. This result is interesting in the context of \citet{Sarafidis2009} who conjecture ``\emph{[...] that the CD test will have poor power properties when it is applied to a regression with time dummies or on cross-sectionally demeaned data}''. First, the cited statement does not acknowledge the bias term $\Xi_{\mathbb{H}_{1}}$ which generally leads $CD_{\mathbb{H}_{1}}$ to diverge. However, it is reasonable not to consider this term as a source of power towards the alternative hypothesis since it does not involve sample estimates of cross-section covariances (see Remark \ref{rem:Prop2Thm1_similar} below). Second, Proposition \ref{prop:CD_2WFE_H1} indicates that the CD test indeed has power since remaining sources of co-movements across cross-sections generally affect the mean of its leading stochastic component.
The failure to appreciate this property is rooted in a convention to treat error variances as fixed parameters, thereby ruling out any correlation between them and factor loadings. Extending this understanding of error variances to the variance of the composite error term $\nu_{i,t}$ (i.e. $\varsigma _{\nu,i}^2$) leads to the special case $\cov\left[ \vlambda_{1},\varsigma_{\nu,1}^{-1}\right]
	=\vzeros $ in which the power of $CD$ is indeed considerably reduced.

Still, while the rate at which $CD_{\mathbb{H}_1}$ diverges is in general unaffected by the inclusion of time fixed effects, power in small samples may be compromised because of the presence of two diverging components with opposite signs in Eq. \eqref{eq:CDH1_2WFE}. While the mean of the leading stochastic component dominates in large samples, it may be cancelled out by $\Xi_{\mathbb{H}_{1}}$ if the number of observations available to the researcher is rather small.

\begin{remark}
	It is straightforward to extend the results of Proposition \ref{prop:CD_2WFE_H1} to a model specification without time fixed effects and to
	address a common concern about the power of $CD$ that has repeatedly been made in the literature. As argued in \citet{PUY2008} and \citet{Sarafidis2009}, the power of this test is reduced substantially if unaccounted sources of cross-section correlation average out, as for example
	in a factor model with zero mean factor loadings. This property is given by	construction for two-way fixed effects or time fixed effects residuals by
	the mere fact that the within transformation involves cross-sectionally	demeaning the data. While the bias term in Proposition \ref{prop:CD_2WFE_H1}
	and the scaling effect on the leading stochastic component in $CD_{\mathbb{H}_{1}}$ are exclusively a consequence of accounting for time fixed effects,
	the defining role of $\cov\left[ \vlambda_{1},\varsigma_{\nu,1}^{-1}\right]$ for the rate at which $CD_{\mathbb{H}_{1}}$ diverges is a general property in models with zero mean loadings. It follows that existing claims about the power losses of $CD$ in the presence of factor with zero mean loadings merely address the special case with $\cov\left[ \vlambda_{1},\varsigma_{\nu,1}^{-1}\right]
	=\vzeros $. In fact, Monte Carlo results that have been reported to support these claims, as e.g. in \citet{CDtest2004,doi:10.1080/07474938.2014.956623}, \citet{Sarafidis2009} and \citet{BALTAGI2012164}, all use specifications which entail a zero covariance between factor loadings and the inverse error variances.\footnote{This follows because all cited articles draw $\vlambda_{i} $ and $\sigma_{i}$ independently of each other and assume a symmetric distribution for $\vlambda_{i}$. Using the integer representation of the expected value it is possible to show that $\cov\left[\vlambda_{1},\varsigma_{\nu,1}^{-1}\right]=\vzeros$ in this case.} The CD test statistic is very likely to have much better power properties in simulation experiments that enforce a non-zero correlation between factor loadings and error variances. A simple way of achieving this would be to impose a constant ratio between the magnitude of common variation and that of idiosyncratic variation which holds for all cross-sections. See \citet[Eq. (27)]{parkersul2016} or Section \ref{section::MC} in this study for a formal definition.
\end{remark}

\begin{remark}
\label{rem:Prop2Thm1_similar}
	An additional similarity of Proposition \ref{prop:CD_2WFE_H1} with Theorem \ref{theorem::CCE_sigma} is that the term $\Xi _{\mathbb{H}_{1}}$ does not involve sample estimates of cross-section covariances, implying that this term is not indicative of the degree of cross-section co-movement in the data. Moreover, it can be shown that $
	\Xi_{\mathbb{H}_{1}}\overset{p}{\rightarrow }\left( \E\left[\varsigma _{\nu,i}^{-1}\right] \right) ^{2}\E\left[ \varsigma
	_{\nu,i}^{2}\right] -2 \E\left[ \varsigma _{\nu,i}^{-1}\right] \E\left[ \varsigma
	_{\nu,i}\right] $, so that the leading term in $\sqrt{T}\Xi _{\mathbb{H}_{1}}$ would approach $-\sqrt{T/2}$ as $ \varsigma _{\nu,i} \to\varsigma _{\nu}$, for some constant $\varsigma _{\nu}>0$.
\end{remark}
\subsection{Power analysis: Multifactor}
\label{appendix::CCE_power}

A mis-specified latent common factor model can be characterized by equation \eqref{eq:themodelFactor} such that Assumptions \ref{ass:errors}--\ref{ass:indep} as well as the following Assumption \ref{ass:failrankcond} hold.
\begin{assumption}
	\label{ass:failrankcond}
	$\rk\left( \left[\vmu_{\vlambda} , \vmu_{\mLambda} \right]\right) = m+1 < r$.
\end{assumption}
Assumption \ref{ass:failrankcond} enforces a failure of the rank condition set up by \citet{ECTA:ECTA692} to ensure that the space spanned by factor estimates can be consistent for the space spanned by the true factors. Under failure of the rank condition, a fraction of the common variation affecting the dependent variable $y_{i,t}$ remains asymptotically unaccounted for. As proved in Lemma \ref{lem:facpartition}, we can assume without loss of generality that accounted and unaccounted sources of cross-section dependence are due to two uncorrelated sets of unobserved factors. Formally, we decompose
\begin{equation}
\mF\vlambda_{i}=\mF^{(1)} \vlambda_{i}^{(1)}+\mF^{(2)} \vlambda_{i}^{(2)},
\end{equation}
where a rotation of the $m+1$ factors $\mF^{(1)}$ is consistently estimated by cross-section averages $\widehat{\mF}=\left[\overline{\vy}, \overline{\mX} \right]$ whereas the remaining $r-m-1$ factors $\mF^{(2)}$ are asymptotically orthogonal to $\widehat{\mF}$. An analogous decomposition can be applied to the matrix product $\mF \mLambda_{i}$, allowing us to split the loadings matrix into two blocks $\mLambda_{i}^{(1)}$ and $\mLambda_{i}^{(2)}$. Most importantly, this decomposition implies that $\E\left[\vlambda_{i}^{(2)}\right]=\vzeros$ and $\E\left[\mLambda_{i}^{(2)}\right] = \mZeros$.

We can accordingly express the defactored data as
\begin{align}
	 \label{eq:CCE_resH1}
	 \widehat{\vnu}_{i} &= \mM_{\widehat{\mF}}\vy_i \notag \\
	 &=\mM_{\widehat{\mF}}\mF^{(1)} \vlambda_{i}^{(1)} + \mM_{\widehat{\mF}}\mF^{(2)} \vlambda_{i}^{(2)}  + \mM_{\widehat{\mF}}\vepsi_i.
\end{align}
where Assumption \ref{ass:failrankcond} and Lemma \ref{lem:facpartition} ensure that  $\mM_{\widehat{\mF}}\mF^{(2)} \vlambda_{i}^{(2)}$ does not vanish as the sample size increases. Again, we denote the variance of the composite error term $\nu_{i,t} = \vf_t^{(2)\prime}\vlambda_{i}^{(2)} +\epsi_{i,t}$ by $\varsigma _{\nu,i}^{2} = \E\left[ \nu _{i,t}^{2} ~ | ~ \vlambda_{i}^{\left( 2\right)}, \sigma_{i}^{2}\right] = \vlambda_{i}^{\left( 2\right) \prime}
\mSigma_{F,22} \vlambda_{i}^{\left( 2\right)
} +\sigma_{i}^{2}$, where the $r-m-1$-dimensional square matrix $\mSigma_{F,22}$ is the lower-right block of $\mSigma_{F}$. The order in probability of the CD test statistic can then be characterized as follows.

\begin{proposition}
	\label{prop:CD_CCE_H1}
	Suppose that the true model is given by \eqref{eq:2WFE_H1} and that its components satisfy Assumptions
	\ref{ass:errors}--\ref{ass:indep} and \ref{ass:failrankcond}, and $\varsigma_{\nu,i}\in [\delta;M]$. 
Let the CD test statistic $CD_{\mathbb{H}_{1}}$, defined in \eqref{eq:CD_H1}, be constructed from defactored data \eqref{eq:CCE_resH1}. Then,
	\begin{equation}
	\label{eq:CDH1_CCE}
	CD_{\mathbb{H}_{1}}=\sqrt{\frac{2}{TN\left( N-1\right) }}\sum_{i=2}^{N}
	\sum_{j=1}^{i-1}\vvarpi_{i}^{\prime }\vvarpi_{j}+%
	\sqrt{T}\left( \Phi _{1,\mathbb{H}_{1}}-2\Phi _{2,\mathbb{H}_{1}}\right)
	+\largeO_{P}\left(N^{-1/2} \sqrt{T}\right) 	+\largeO_{P}\left(T^{-1/2}\right)
	\end{equation}%
	where
	\begin{align*}
		\vvarpi_{i} &= \left[ \varsigma_{\nu,i}^{-1} \vnu_{i} - \mD_{i} \left( \overline{\mC}^{\left( 1\right) }\right)^{-1} \left( N^{-1}\sum_{\ell=1}^{N}\vlambda_{\ell}^{\left(1\right) }\varsigma_{\nu,\ell}^{-1}\right) \right], \\
		\Phi _{1,\mathbb{H}_{1}} &= \left( N^{-1}\sum_{\ell=1}^{N} \varsigma_{\nu,\ell}^{-1} \vlambda_{\ell}^{\left( 1\right) \prime }\right)
		\left( \overline{\mC}^{\left( 1\right) \prime }\right)^{-1} \left(
		N^{-1}\sum_{i=1}^{N}\mD_{i}^{\prime }\mD_{i}\right) \left(\overline{\mC}^{\left( 1\right) }\right)^{-1} \left( N^{-1}\sum_{\ell=1}^{N}\vlambda_{\ell}^{\left( 1\right) } \varsigma_{\nu,\ell}^{-1}\right) , \\
		\Phi _{2,\mathbb{H}_{1}} &= N^{-1}\sum_{i=1}^{N}\varsigma _{\nu,i}^{-1}
		\vnu_{i}' \mD_{i} \left( \overline{\mC}^{\left(1\right) }\right)^{-1} \left( N^{-1}\sum_{\ell=1}^{N}\vlambda_{\ell }^{\left( 1\right) }\varsigma _{\nu,\ell}^{-1}\right),
	\end{align*}
	where $\overline{\mC}^{(1)} = \begin{bmatrix}
	\overline{\vlambda}^{(1)}, ~ \overline{\mLambda}^{(1)}
	\end{bmatrix}$ and $\mD_i = \mF \begin{bmatrix}
	\vlambda_{i}^{(2)}, ~ \mLambda_{i}^{(2)}
	\end{bmatrix} + \begin{bmatrix}
	\vepsi_{i}, ~ \ve_{i}
	\end{bmatrix} $ .
	Furthermore, it holds that
	\begin{equation*}
	\plim_{N,T\rightarrow \infty }N^{-1}T^{-1/2}\sqrt{\frac{2}{TN\left(
			N-1\right) }}\sum_{i=2}^{N}\sum_{j=1}^{i-1}\vvarpi_{i}'
	\vvarpi_{j}=\cov\left[ \varsigma _{\nu,1}^{-1}, ~ ~ \vlambda_{1}^{\left( 2\right) \prime }\right] \mSigma_{F,22}\cov\left[ \vlambda_{1}^{\left(
		2\right) },~ \varsigma_{\nu,1}^{-1}\right].
	\end{equation*}%
	In the special case $\cov\left[ \varsigma _{\nu,1}^{-1},~
	\vlambda_{1}^{\left( 2\right) \prime }\right] =\vzeros$ ,
	the term$\sqrt{\frac{2}{TN\left( N-1\right) }}\sum_{i=2}^{N}\sum_{j=1}^{i-1}%
	\vvarpi_{i}'\vvarpi_{j}$ is
	asymptotically centered around zero and has a variance that diverges at rate
	$T$.
\end{proposition}
The remarks made concerning Proposition \ref{prop:CD_2WFE_H1} apply in an identical fashion to Proposition \ref{prop:CD_CCE_H1}. That is, the CD test statistic allows for the same type of decomposition under the null hypothesis as it does under the alternative, differing only in that the expressions in Proposition \ref{prop:CD_CCE_H1} contain the composite error term $\nu_{i,t}$ rather than the true model errors $\epsi_{i,t}$.
\subsection{Weighted CD statistic and power enhancement}
Asymptotic unbiasedness of $CD_W$ comes at the cost of power. More specifically, our approach to bias correction centers the leading components of $CD_W$ around zero, irrespective of whether cross-section correlation in the data is completely controlled for or not. As a consequence, only increases in the variance of $CD_W$ under its alternative hypothesis lead to power against the null hypothesis of this test.

This point can be formally illustrated by a minor modification of Proposition \ref{prop:CD_2WFE_H1} which consists of replacing $\varsigma_{\nu,i}^{-1} - \E\left[ \varsigma_{\nu,i}^{-1} \right] $ with $w_i$. Given random weights, the leading stochastic component in Eq. \eqref{eq:CDH1_2WFE},
\begin{align}
\label{eq:power_leadcomp_randweight}
\sqrt{\frac{2}{TN\left( N-1\right)}} \sum_{i=2}^{N}\sum_{j=1}^{i-1} w_i \vnu_{i}^{\prime}\vnu_{j} w_j,
\end{align}
has an asymptotic mean of zero if $N^{-1}\sqrt{T}\to 0$ since independent Rademacher distributed weights ensure that $\E\left[w_i \nu_{i,t}\nu_{j,t} w_j\right] = 0$ holds for all $i\neq j$. Thus, remaining sources of cross-section co-movement in the composite error term $\nu_{i,t}$ will not shift the location of expression \eqref{eq:power_leadcomp_randweight}. By contrast, it can be shown that the variance of this term continues to diverge at rate $T$. This entails that $CD_W$ diverges at rate $\sqrt{T}$ under its alternative hypothesis, a property that  holds under the conditions of Proposition \ref{prop:CD_2WFE_H1} as well.

Using Lemma \ref{lem:sigmahat} we can establish that:
\begin{equation}
\label{eq:rho_indiv}
	\widehat{\rho}_{ij}=0+\largeO_{P}(T^{-1/2})+\largeO_{P}((NT)^{-1/2})+\largeO_{P}(N^{-1}),
\end{equation}
for all pairs $i,j$. The threshold $2\sqrt{\ln(N)/T}$ is justified by the behavior of the maximum out of $N(N-1)/2$ estimated correlation coefficients in a model with $\rho_{ij} = 0$ for all $i,j$. Under the assumption that a CLT holds for $\widehat{\rho}_{ij}$ and that $\sqrt{T}N^{-1} \to 0$ as $N,T,\to \infty$, this maximum should diverge at rate $2\sqrt{\ln(N)/T}$. Multiplication of each indicator function in the screening statistic $\sum_{i=2}^N \sum_{j=1}^{i-1}\left| \widehat{\rho}_{ij} \right| \mathbf{1}\left( \left|\widehat{\rho}_{ij} \right| > 2\sqrt{\ln(N)/T}  \right)$ with the absolute value of its corresponding correlation coefficient then ensures that the power enhancement component in \eqref{eq:CD_Wplus} converges to 0 if $\rho_{ij} = 0$ for all pairs $i,j$, subject to additional tail regularity conditions.

\section{Further discussions}
\label{section:discussion_heuristics}
\subsection{An analytically bias-corrected CD statistic}
\label{section:CDBC_details}
Section \ref{section:BC} introduced parametric bias-correction as a feasible approach to re-establish asymptotically normal inference for the CD tests statistic under its null hypothesis. This section provides details on its implementation as a benchmark CD test statistic in the Monte Carlo experiments of Section \ref{section::MC}. The analytically bias-corrected CD test statistic $CD_{BC}$ is corrected with a plug-in estimate of the bias terms of Theorems \ref{theorem::additiveHetero} and \ref{theorem::CCE}. That is, whenever $CD_{BC}$ is applied to 2WFE residuals, we construct it as
\begin{align*}
CD_{BC} = \widehat{\Omega}^{-1/2}_{FE} \left(CD - \sqrt{T}\widehat{\Xi}\right)
\end{align*}
where
\begin{align*}
\widehat{\Xi}& = \sqrt{\frac{1}{2N(N-1)}} \left[\sum_{i=1}^N\left(1 - \sqrt{\widehat{\sigma}_i^2}N^{-1}\sum_{\ell=1}^N\left(\widehat{\sigma}_{\ell}^2\right)^{-1/2}\right)^2 - N\right], \\
\widehat{\Omega}_{FE} &= \frac{2\left(N-1\right)}{N} \left(\widehat{\Xi} + 1 \right)^2
\end{align*}
with $\widehat{\sigma}_i^2 = T^{-1}\sum_{t=1}^T  \widehat{\epsi}_{i,t}^2$ for 2WFE residuals $\widehat{\epsi}_{i,t}$. The expressions above arise from rearranging the terms of Theorem \ref{theorem::additiveHetero} so that their estimates can be obtained with minimal computational burden. In this regard, it is helpful to note that the asymptotic variance of $\Omega$ can be written as a function of the bias term $\Xi$. Likewise, for application of $CD_{BC}$ to CCE residuals, we set
\begin{align*}
CD_{BC} = \widehat{\Omega}^{-1/2}_{CCE} \left[CD - \sqrt{T}\left(\widehat{\Phi}_1 -2\widehat{\Phi}_2 \right)\right],
\end{align*}
where
\begin{align*}
\widehat{\Phi}_1 &= \sqrt{\frac{1}{2N(N-1)}} \left[N^{-1} \sum_{i=1}^N \left(\widehat{\sigma}^2_i\right)^{-1/2}\right]^2 \overline{\widehat{\vlambda}}' \widehat{\mB}'
\begin{bmatrix} \sum_{i=1}^N \widehat{\sigma}_i^2, & \vzeros_{m}' \\ \vzeros_{m}, & \sum_{i=1}^N \widehat{\mSigma}_i \end{bmatrix}
\widehat{\mB} \overline{\widehat{\vlambda}}, \\
\widehat{\Phi}_2 &= \sqrt{\frac{1}{2N(N-1)}} \left[N^{-1} \sum_{i=1}^N \left(\widehat{\sigma}^2_i\right)^{-1/2}\right] \overline{\widehat{\vlambda}}'  \begin{bmatrix}  \sum_{i=1}^N \widehat{\sigma}_i^2 \\ \vzeros_{m} \end{bmatrix}, \\
\widehat{\Omega}_{CCE} &=  \left[ 1 + \sqrt{\frac{2\left(N-1\right)}{N}}\left( \widehat{\Phi}_1 - 2\widehat{\Phi}_2\right) \right]^2
\end{align*}
with
\begin{align*}
\widehat{\mB} &= \begin{bmatrix}
1, & \vzeros_{m}' \\
\widehat{\vbeta}^{CCE}, & \mI_m
\end{bmatrix}, \\
\widehat{\vbeta}^{CCE} &= \left(\frac{1}{NT}\sum_{i=1}^N\mX_i'\mM_{\widehat{\mF}}\mX_i\right)^{-1}\left(\frac{1}{NT} \sum_{i=1}^N\mX_i'\mM_{\widehat{\mF}}\vy_i \right), \\
\overline{\widehat{\vlambda}} &= \frac{1}{N} \sum_{i=1}^N \widehat{\vlambda}_i, \\
\widehat{\vlambda}_i &= \left(\widehat{\mF}'\widehat{\mF}\right)^{-1} \widehat{\mF}'\left( \vy_{i} - \mX_i \widehat{\vbeta}^{CCE} \right), \\
\widehat{\mSigma}_i &= \frac{1}{T}\mX_i'\mM_{\widehat{\mF}}\mX_i,
\end{align*}
as well as  $\widehat{\sigma}_i^2 = T^{-1} \widehat{\vepsi}_{i}'\widehat{\vepsi}_{i}$ and $\widehat{\vepsi}_{i} = \mM_{\widehat{\mF}}(\vy_{i} - \mX_{i} \widehat{\vbeta}^{CCE})$. The plug-in estimators $\widehat{\Phi}_1$ and $\widehat{\Phi}_2$ are calculated under the assumption of independence between factor loadings and error variances so that we can write $\E\left[ \vlambda_{i} \sigma_{i}^{-1} \right] = \E\left[ \vlambda_{i} \right] \E \left[ \sigma_{i}^{-1} \right]$. While this assumption allows making the estimates of both bias terms sufficiently precise in the setup of our Monte Carlo experiments, it may lead to size distortion in a DGP with dependence between factor loadings and error variances.


\subsection{Serial correlation}
\label{section:serialcorrelation}
The original CD test of \citet{CDtest2004, doi:10.1080/07474938.2014.956623} assumes that idiosyncratic error terms are serially uncorrelated, an assumption that can be difficult to justify for most economic datasets. While this problem is mostly ignored in practice, \citet{BaltagiEtAlCD2016} have very recently proposed a modification of the CD test statistic which ensures that the test statistic is asymptotically standard normal as long as $\epsi_{i,t}$ is a stationary short memory process. In particular, under this type of assumption it can be shown that the relevant asymptotic variance for the result in Theorem \ref{thm:CD_W_combined} is given by:
\begin{equation}
\Omega=\sum_{s=-\infty}^{\infty}\left[\E[(w_{i}-\E[w_{i}])^{2}\epsi_{i,t}\epsi_{i,t-s}]\right]^{2}.
\end{equation}
The above quantity is non-standard, as it not a function of the long-run variance of $(w_{i}-\E[w_{i}])\epsi_{i,t}$. In particular, one can use data which is over-differenced in the construction of the CD statistic.

In the context of our testing problem the natural plug-in estimator of $\Omega$ is given by
\begin{equation}
\widehat{\Omega}_{N}=\frac{2}{TN(N-1)}\sum_{i=2}^{N}\sum_{j<i}(\vl_{i}'\vl_{j})^{2},
\end{equation}
where $\vl_{i}=(w_{i}-\overline{w})\widehat{\ve}_{i}$ is constructed using residuals $\widehat{\ve}_{i}$. It can be expected that under reasonable regularity conditions on the memory properties of $\epsi_{i,t}$ as well as appropriate restrictions on $N,T$ estimator $\widehat{\Omega}_{N}$ is consistent in our setup. However, we do not attempt to prove this conjecture as it does not add to the main message of this paper.

Note that \citet{BaltagiEtAlCD2016} use mean adjustment variance estimate of \citet{chen2010}, which is motivated by the need to obtain an unbiased (not only consistent) estimator of $\Omega$. However, as in our setting  $\vx_{i}$ are correlated by construction, any theoretical justification for including such an adjustment term under null hypothesis is lost. Also note that a factor of $2$ is missing in \citet{BaltagiEtAlCD2016}.

\subsection{Optimal weights}
\label{ssection::optimal_weights}
Propositions \ref{prop:CD_2WFE_H1} and \ref{prop:CD_CCE_H1} stated that the original CD test statistic diverges at rate $N\sqrt{T}$ if certain conditions on the dependence between error variances and loadings associated with unaccounted common factors hold. A simple generalization of Proposition \ref{prop:CD_CCE_H1} to the case of general weights $\{w_1, \ldots, w_N\}$ suggest that the same rate of divergence can be achieved under the condition $\cov \left[w_i,~\vlambda_{i}^{(2)}\right] \neq \vzeros$.\footnote{Analogous results hold for a generalization of Proposition \ref{prop:CD_2WFE_H1}} 

Consequently, a set of weights that generally leads to high power would be given by functions of the data that are estimates of $\vlambda_{i}^{(2)}$ under $\mathbb{H}_1$, such as $w_i = T^{-1}\sum_{t=1}^{T} \widehat{\nu}_{i,t}$ or $w_{ij} = T^{-1}\sum_{t=1}^{T} \widehat{\nu}_{i,t}\widehat{\nu}_{j,t}$ for a more general statistic with index-pair specific weights. However, the latter weight suggestion directly results in a different test statistic, namely one based on \emph{squared} cross-section covariances. This new test statistic, as well as one based on the first suggestion $w_i = T^{-1}\sum_{t=1}^{T} \widehat{\nu}_{i,t}$, is closer to an existing, separate literature on LM tests for cross-section dependence than it is to the CD test and its extant modifications. The fact that a test statistic based on summing squared cross-section covariances is not centered around zero under its null hypothesis, and that it hence needs additional recentering, further emphasizes its relation to LM tests.

As an alternative to estimates of $\vlambda_{i}^{(2)}$, estimates of other model components may be used to set up a weighted CD test statistic. In this regard, the most sensible choice is an estimate of loadings associated with unaccounted factors in the covariates of our regression model of interest, i.e. the parameter matrix $\mLambda_{i}^{(2)}$ in
\begin{align*}
\mX_{i} = \mF^{(1)} \mLambda_{i}^{(1)} + \mF^{(2)} \mLambda_{i}^{(2)} + \mE_i.
\end{align*}
This choice directs the power of $CD_W$ towards alternatives under which loadings associated with unaccounted factors are correlated in the sense of \citet{Westerlund2013247} and \citet{kapetanios2019}. Following the theoretical results in these two studies, a CD test statistic that weights cross-section covariances with estimates of $\mLambda_{i}^{(2)}$ is effectively a test for inconsistency of either the two-way fixed effects or CCE estimator of $\vbeta$. It is generally possible to direct the CD test statistic towards these more specific hypothesis. However, if one is to test whether an estimate of the slope coefficients $\vbeta$ on $\mX_i$ are inconsistent, it is simpler and more sensible to do this using a test statistic that addresses the moment condition $\cov \left[ \vlambda_{i}^{(2)},~\vec\left(\mLambda_{i}^{(2)}\right) \right]=\mZeros$ directly.

\newpage

\section{Proofs}
\subsection{Notation}
Extending the paragraph on notation in introduction in the main text, we will use the following notation in this appendix.
\begin{itemize}
\item $\mI_m$ denotes an $m \times m$ identity matrix and the subscript is sometimes disregarded from for the sake of simplicity. $\vzeros$ denotes a column vector of zeros while $\mZeros$ stands for a matrix of zeros. $\vs_m$ denotes a selection vector all of whose elements are zero except for element $m$ which is one. $\viota$ is a vector entirely consisting of ones. The dimension of these latter vectors and matrices is generally suppressed for the sake of simplicity and needs to be inferred from context.
\item For a generic $m \times n$ matrix $\mA$, $\mP_{\mA} = \mA(\mA'\mA)^{-1}\mA'$ projects onto the space spanned by the columns of $\mA$ and $\mM_{\mA} = \mI_m - \mP_{\mA}$. $\dim(\mA), \operatorname{col}(\mA)$ and $\ker(\mA)$ refer to the dimension, column space and kernel of $\mA$. Moreover, $\rk(\mA)$ denotes the rank of $\mA$, $\tr(\mA)$ its trace and $\| \mA \| =\left( \tr(\mA'\mA)\right)^{1/2}$ the Frobenius norm of $\mA$.
\item For a set of $m \times n$ matrices $\{ \mA_1,\mA_{2}, \ldots, \mA_N \}$, $\overline{\mA} = N^{-1}\sum_{i=1}^N \mA_i$. Multiple sums are generally abbreviated, so that $\sum_{i,j}^N$ is shorthand notation for $\sum_{i=1}^N \sum_{j=1}^N$.
\item $\delta$ and $M$ stand for a small and large positive real number, respectively. For two real numbers $a$ and $b$, $a \vee b = \max\{ a,b\}$.
\item For some random variable $\epsi_{i,t}$, $\kappa_4\left[ \epsi_{i,t} \right]$ denotes its fourth-order cumulant. Moreover, $\largeO(\cdot)$ and $\smallO(\cdot)$ express order of magnitude relations whereas $\largeO_{P}(\cdot)$ and $\smallO_P(\cdot)$ denote stochastic order relations, see e.g. Definitions 2.5 and 2.33 in \citet{white2001}.
\item Define $k_{N,T}=\sqrt{\frac{1}{2TN(N-1)}}$.
\item Notice that we use the following vector and scalar notation for idiosyncratic components: $\vepsi_{i}=(\epsi_{i,1},\epsi_{i,2},\ldots,\epsi_{i,T})'$, $\overline{\vepsi}=N^{-1}\sum_{i=1}^{N}\vepsi_{i}$, and $\overline{\epsi}=(NT)^{-1}\sum_{i=1}^{N}\sum_{t=1}^{T}\epsi_{i,t}$.
\end{itemize}

\paragraph{Additional notation in the multiplicative model} \hfill \par
Proofs that address a model with multifactor error structure which is controlled for via cross-section averages frequently employ the rotation and rescaling matrix
\begin{align*}
	\mB = \begin{bmatrix}
		1, & \vzeros' \\ \vbeta, & \mI_m
	\end{bmatrix}
\end{align*}
which relates the $T \times (m+1)$ matrix $\left[ \vy_i , \mX_i \right]$ (or alternatively the average over all $i$) to common and idiosyncratic variation affecting each of the observed variables directly. This relation is given by
\begin{align*}
	\begin{bmatrix}
	\vy_i, & \mX_i
	\end{bmatrix}
	&=
	\left( \mF
	\begin{bmatrix}
	\vlambda_{i}, & \mLambda_i
	\end{bmatrix} +
	\begin{bmatrix}
	\vepsi_i, & \mE_i
	\end{bmatrix} \right)
	\mB \\
	&= \mF \mC_i + \mU_i
\end{align*}
In the following, it is assumed that $\vbeta$ is bounded so that $\left\|\mB \right\| < \infty$. Furthermore, $\mB$ is nonsingular by construction, implying that $\mB^{-1}$ exists. Given these properties, it trivially follows from the properties of $\vlambda_{i}$ and $\mLambda_i$ that $\| \overline{\mC}\|=\largeO_{P} \left(1\right)$ and $\| \overline{\mC}^{-1}\|= \largeO_{P}(1)$. The same holds for the matrix $\overline{\mC}^{(1)}$ which is considered in the proof of Proposition \ref{prop:CD_CCE_H1}.

\subsection{Proofs of Theorem \ref{theorem::additiveHetero} and Corollary \ref{corollary::homosigmas}}
\label{sec:proof_additiveHetero}

\begin{proof}[\textbf{Proof of Theorem \ref{theorem::additiveHetero}}]
We start the proof by focusing on the individual components of the double sum over cross-sections which constitutes the CD test statistic. Here, we apply a first-order Taylor expansion around the true error variances. Using the mean-value theorem, we can hence set up the equality
\begin{align*}
\frac{\widehat{\epsi}_{i,t}\widehat{\epsi}_{j,t}}{\sqrt{\widehat{\sigma}_{i}^{2}\widehat{\sigma}_{j}^{2}}}=\frac{\widehat{\epsi}_{i,t}\widehat{\epsi}_{j,t}}{\sqrt{\sigma_{i}^{2}\sigma_{j}^{2}}}-\frac{1}{2}\frac{\widehat{\epsi}_{i,t}\widehat{\epsi}_{j,t}}{\left(\varsigma_{i}^{2}\right)^{3/2}\left(\sigma_{j}^{2}\right)^{1/2}}\left(\widehat{\sigma}_{i}^{2}-\sigma_{i}^{2}\right)-\frac{1}{2}\frac{\widehat{\epsi}_{i,t}\widehat{\epsi}_{j,t}}{\left(\sigma_{i}^{2}\right)^{1/2}\left(\varsigma_{j}^{2}\right)^{3/2}}\left(\widehat{\sigma}_{j}^{2}-\sigma_{j}^{2}\right)
\end{align*}
for some $\left(\varsigma_{i}^{2},\varsigma_{j}^{2}\right)$ on the
line interval between $\left(\sigma_{i}^{2},\sigma_{j}^{2}\right),$
and $\left(\widehat{\sigma}_{i}^{2},\widehat{\sigma}_{j}^{2}\right).$ This
Taylor approximation has two extra terms. However, both of them can
be combined into one single extra term once we sum over all possible
combinations of $i$ and $j$. Hence, we can write 
\begin{align}
\label{eq:CD_Taylordecomp_2WFE}
CD&=2k_{N,T} \sum_{i=2}^{N}\sum_{j=1}^{i-1}\frac{\widehat{\vepsi}_{i}'\widehat{\vepsi}_{j}}{\sqrt{\sigma_{i}^{2}\sigma_{j}^{2}}} - k_{N,T}\sum_{i=1}^{N}\sum_{j\neq i}^{N}\frac{\widehat{\vepsi}_{i}'\widehat{\vepsi}_{j}}{\left(\varsigma_{i}^{2}\right)^{3/2}\left(\sigma_{j}^{2}\right)^{1/2}}\left(\widehat{\sigma}_{i}^{2}-\sigma_{i}^{2}\right) \notag \\
&= CD_{\sigma} +CD_{(\widehat{\sigma} - \sigma)}.
\end{align}
Consider now ${CD}_{\sigma}$. Since model residuals can be written out as $\widehat{\epsi}_{i,t}= \epsi_{i,t}-\overline{\epsi}_{i}-\overline{\epsi}_t +\overline{\epsi},$ we can further decompose this term into 
\begin{align}
	\label{eq:CD_FEdecomp_2WFE}
	{CD}_{\sigma} &= \sqrt{\frac{2}{TN(N-1)}}\sum_{i=2}^{N} \sum_{j=1}^{i-1}\sum_{t=1}^{T}\frac{(\epsi_{i,t}-\overline{\epsi}_{t})(\epsi_{j,t}-\overline{\epsi}_{t})}{\sigma_{i}\sigma_{j}} - Tk_{N,T} \left[ \left(\sum_{i=1}^N \frac{\overline{\epsi}_{i}-\overline{\epsi}}{\sigma_i}\right)^2 - \sum_{i=1}^N \left(\frac{\overline{\epsi}_{i}-\overline{\epsi}}{\sigma_i}\right)^2 \right] \notag \\
	&= {CD}_{\mu,\sigma} - {CD}_{(\widehat{\mu} - \mu)}.
\end{align}
As shown formally in Section \ref{sec:var_estim_2WFE}, ${CD}_{(\widehat{\mu} - \mu)} = \largeO_{P}\left(T^{-1/2}\right)$ and $CD_{(\widehat{\sigma} - \sigma)} = \largeO_{P}(\sqrt{N}T^{-1}) + \largeO_{P}(T^{-1/2}) +\largeO_{P}(N^{-1/2}) + \largeO_{P}(N^{-1}\sqrt{T})$. Accordingly, terms that capture the effect of estimating fixed effects and error variances have an asymptotically negligible effect if some weak restrictions on the relative rate of divergence of $N$ and $T$ are satisfied. For this reason, $CD_{\mu, \sigma}$ is the leading component of the CD test statistic. It can be expressed in terms of four U statistics, and two additional terms that contribute to the bias. First, let
\begin{align*}
CD_{\vepsi} &= \sqrt{\frac{2T}{N(N-1)}}\sum_{i=2}^{N} \sum_{j=1}^{i-1}\frac{1}{T}\sum_{t=1}^{T}\epsi_{i,t}\epsi_{j,t},\\
CD_{\vepsi/\sigma} &= \sqrt{\frac{2T}{N(N-1)}}\sum_{i=2}^{N} \sum_{j=1}^{i-1}\frac{1}{T}\sum_{t=1}^{T}\frac{\epsi_{i,t}\epsi_{j,t}}{\sigma_{i}\sigma_{j}},\\
CD_{\vepsi+} &= \sqrt{\frac{2T}{N(N-1)}}\sum_{i=2}^{N} \sum_{j=1}^{i-1}\frac{1}{T}\sum_{t=1}^{T}\epsi_{i,t}\epsi_{j,t}(\sigma_{i}^{-1}+\sigma_{j}^{-1}),\\
CD_{\vepsi++} &= \sqrt{\frac{2T}{N(N-1)}}\sum_{i=2}^{N} \sum_{j=1}^{i-1}\frac{1}{T}\sum_{t=1}^{T}\epsi_{i,t}\epsi_{j,t}(\sigma_{i}^{-2}+\sigma_{j}^{-2}).\\
\end{align*}


Now observe that the leading term of $CD$ can be decomposed into
\begin{align*}
{CD}_{\mu, \sigma}&= \sqrt{\frac{2}{TN(N-1)}}\sum_{i=2}^{N} \sum_{j=1}^{i-1}\sum_{t=1}^{T}\frac{(\epsi_{i,t}-\overline{\epsi}_{t})(\epsi_{j,t}-\overline{\epsi}_{t})}{\sigma_{i}\sigma_{j}}\\
&=k_{N,T}\sum_{t=1}^{T}\left(\sum_{i=1}^{N}\frac{(\epsi_{i,t}-\overline{\epsi}_{t})}{\sigma_{i}}\right)^{2}
-k_{N,T}\sum_{i=1}^{N}\sum_{t=1}^{T}\left(\frac{\epsi_{i,t}-\overline{\epsi}_{t}}{\sigma_{i}}\right)^{2}\\
&-k_{N,T}\sum_{t=1}^{T}\left(\sum_{i=1}^{N}\frac{(\overline{\epsi}_{i}-\overline{\epsi})}{\sigma_{i}}\right)^{2}+k_{N,T}\sum_{i=1}^{N}\sum_{t=1}^{T}\left(\frac{\overline{\epsi}_{i}-\overline{\epsi}}{\sigma_{i}}\right)^{2}\\
&=I-II-III+IV.
\end{align*}
Let us now consider each term separately. First,
\begin{align*}
I&=k_{N,T}\sum_{t=1}^{T}\left(\sum_{i=1}^{N}\frac{\epsi_{i,t}}{\sigma_{i}}\right)^{2}
-2k_{N,T}\overline{\sigma^{-1}}\sum_{t=1}^{T}\left[\left(\sum_{i=1}^{N}\frac{\epsi_{i,t}}{\sigma_{i}}\right)\left(\sum_{i=1}^{N}\epsi_{i,t}\right)\right]
+k_{N,T}\left(\overline{\sigma^{-1}}\right)^{2}\sum_{t=1}^{T}\left(\sum_{i=1}^{N}\epsi_{i,t}\right)^{2}\\
&=CD_{\vepsi/\sigma}-\overline{\sigma^{-1}}CD_{\vepsi+}+\left(\overline{\sigma^{-1}}\right)^{2}CD_{\vepsi}\\
&+k_{N,T}\sum_{i=1}^{N}\sum_{t=1}^{T}\left(\frac{\epsi_{i,t}}{\sigma_{i}}\right)^{2} - 2k_{N,T}\overline{\sigma^{-1}}\sum_{i=1}^{N}\sum_{t=1}^{T}\sigma_{i}\left(\frac{\epsi_{i,t}}{\sigma_{i}}\right)^{2} + k_{N,T}\left(\overline{\sigma^{-1}}\right)^2\sum_{i=1}^{N}\sum_{t=1}^{T}\epsi_{i,t}^{2}.
\end{align*}
Similarly, for the second term we have
\begin{align*}
II&=k_{N,T}\sum_{t=1}^{T}\sum_{i=1}^{N}\left(\frac{\epsi_{i,t}}{\sigma_{i}}\right)^{2}
-2 k_{N,T}N^{-1}\sum_{t=1}^{T}\left[\left(\sum_{i=1}^{N}\frac{\epsi_{i,t}^2}{\sigma_{i}^2}\right)\left(\sum_{i=1}^{N}\epsi_{i,t}\right)\right]
+k_{N,T}N^{-1}\overline{\sigma^{-2}}\sum_{t=1}^{T}\left(\sum_{i=1}^{N}\epsi_{i,t}\right)^{2}\\
&=-N^{-1}CD_{\vepsi++} + N^{-1}\overline{\sigma^{-2}}CD_{\vepsi}\\
&+k_{N,T}\sum_{t=1}^{T}\sum_{i=1}^{N}\left(\frac{\epsi_{i,t}}{\sigma_{i}}\right)^{2}
-2 k_{N,T}N^{-1}\sum_{i=1}^{N}\sum_{t=1}^{T}\left(\frac{\epsi_{i,t}}{\sigma_{i}}\right)^{2}
+k_{N,T}N^{-1}\overline{\sigma^{-2}}\sum_{i=1}^{N}\sum_{t=1}^{T}\left(\epsi_{i,t}\right)^{2}.
\end{align*}
As for the third component $III$,
\begin{align*}
III&=T k_{N,T}\left(\sum_{i=1}^{N}\frac{(\overline{\epsi}_{i}-\overline{\epsi})}{\sigma_{i}}\right)^{2}\\
&=N k_{N,T}\left(\frac{1}{\sqrt{NT}}\sum_{i=1}^{N}\sum_{t=1}^{T}\frac{\epsi_{i,t}}{\sigma_{i}}-\frac{1}{\sqrt{NT}}\sum_{i=1}^{N}\sum_{t=1}^{T}\epsi_{i,t}\right)^{2}\\
&=\sqrt{\frac{N}{2T(N-1)}}\left(\largeO_{P}(1)-\largeO_{P}(1)\right)^{2}\\
&=\largeO_{P}(T^{-1/2}).
\end{align*}
Here the third line follows from Theorem 3 in \citet{phillipsMoon1999} applied to sequences $\epsi_{i,t}$ and $\epsi_{i,t}/\sigma_{i}$.

As for the fourth component $IV$, notice that:
\begin{align*}
IV&=T k_{N,T}\sum_{i=1}^{N}\left(\frac{\overline{\epsi}_{i}-\overline{\epsi}}{\sigma_{i}}\right)^{2}\\
&=k_{N,T}\sum_{i=1}^{N}\left(\frac{1}{\sqrt{T}}\sum_{t=1}^{T}\frac{\epsi_{i,t}}{\sigma_{i}}\right)^{2}-2 k_{N,T}\left(\frac{1}{\sqrt{NT}}\sum_{i=1}^{N}\sum_{t=1}^{T}\frac{\epsi_{i,t}}{\sigma_{i}}\right)\left(\frac{1}{\sqrt{NT}}\sum_{i=1}^{N}\sum_{t=1}^{T}\epsi_{i,t}\right)\\
&+k_{N,T}\left(\frac{1}{N}\sum_{i=1}^{N}\frac{1}{\sigma_{i}^{2}}\right)\left(\sqrt{NT} \overline{\epsi}\right)^{2}\\
&=k_{N,T}\sum_{i=1}^{N}\left(\frac{1}{\sqrt{T}}\sum_{t=1}^{T}\frac{\epsi_{i,t}}{\sigma_{i}}\right)^{2}-\largeO_{P}(N^{-1}T^{-1/2}).
\end{align*}
Here in the final use we use the same CLT from Theorem 3 in \citet{phillipsMoon1999}. For the final component notice that:
\begin{align*}
k_{N,T}\sum_{i=1}^{N}\left(\frac{1}{\sqrt{T}}\sum_{t=1}^{T}\frac{\epsi_{i,t}}{\sigma_{i}}\right)^{2}&=N k_{N,T}+\sqrt{N}k_{N,T}\frac{1}{\sqrt{N}}\sum_{i=1}^{N}\left(\left(\frac{1}{\sqrt{T}}\sum_{t=1}^{T}\frac{\epsi_{i,t}}{\sigma_{i}}\right)^{2}-1\right)\\
&=\largeO(T^{-1/2})+\largeO_{P}(N^{-1/2}T^{-1/2}).
\end{align*}
It follows that $IV=\largeO_{P}(T^{-1/2})$.

Combining these expressions yields the expression in terms of four U-statistics and two bias terms referred to above:
\begin{align}
\label{eq:CD_Ustat_decomp}
{CD}_{\mu, \sigma}&=I-II -III+IV\notag\\
&=CD_{\vepsi/\sigma} - 2\overline{\sigma^{-1}}CD_{\vepsi+} - 2N^{-1}CD_{\vepsi++} +  \left(\left(\overline{\sigma^{-1}}\right)^{2}-N^{-1}\overline{\sigma^{-2}}\right)CD_{\vepsi} \notag\\
&+k_{N,T}\sum_{i=1}^{N}\sum_{t=1}^{T} \epsi_{i,t}^{2}\left( \left( \overline{\sigma^{-1}}\right)^{2} - 2\overline{\sigma^{-1}}\sigma_{i}^{-1}  \right)
- N^{-1} k_{N,T} \sum_{i=1}^{N}\sum_{t=1}^{T} \epsi_{i,t}^{2} \left(  \left( \overline{\sigma^{-1}}\right)^{2} - 2\sigma^{-2}_i \right)+\largeO_{P}(T^{-1/2}).
\end{align}
Concerning the second bias term in the last line above, we note that $(NT)^{-1} \sum_{i=1}^{N}\sum_{t=1}^{T} \epsi_{i,t}^{2}=\largeO_{P}(1)$ and $(NT)^{-1} \sum_{i=1}^{N}\sum_{t=1}^{T} \epsi_{i,t}^{2}\sigma_{i}^{-2}=\largeO_{P}(1)$ by application of Markov's inequality. It follows that
\begin{align*}
N^{-1} k_{N,T} \sum_{i=1}^{N}\sum_{t=1}^{T} \epsi_{i,t}^{2} \left(  \left( \overline{\sigma^{-1}}\right)^{2} - 2\sigma^{-2}_i \right)=\largeO_{P}\left( N^{-1} \sqrt{T} \right).
\end{align*}
Concerning the four U-statistics, note that under Assumption \ref{ass:errorsAddHetero} the error variances $\sigma_{i}^2$ are bounded, implying that all averages involving $\sigma_{i}$ are in general of order $\largeO_{P}(1)$ by Markov's inequality. Lemma \ref{lem:U_stat} can hence be used to show that all four U-statistics $ CD_{\vepsi},  CD_{\vepsi/\sigma},  CD_{\vepsi+}$ and $CD_{\vepsi++}$ are of order $\largeO_{P}(1)$.\footnote{Strictly speaking, Lemma \ref{lem:U_stat} does not apply to $CD_{\vepsi+}$ and $CD_{\vepsi++}$. However, it can be straightforwardly extended to accommodate the specific structure of the U-statistics in both expressions .} As a consequence, all terms in the second line of Eq.  \eqref{eq:CD_Ustat_decomp} involving $N^{-1}$ are $\largeO_{P}(N^{-1})$. Additionally, a standard (cross-sectional) Lindeberg-Levi CLT implies that
\begin{equation*}
\overline{\sigma^{-1}}=\E[\sigma_{i}^{-1}]+\largeO_{P}(N^{-1/2}).
\end{equation*}
This allows us to write
\begin{align*}
&CD_{\vepsi/\sigma} - 2\overline{\sigma^{-1}}CD_{\vepsi+} + \left(\overline{\sigma^{-1}}\right)^{2}CD_{\vepsi} \\
&= 2k_{N,T} \sum_{i=2}^{N} \sum_{j=1}^{i-1}\sum_{t=1}^{T}\epsi_{i,t}\epsi_{j,t}\left(\sigma_{i}^{-1}-\E[\sigma_{i}^{-1}]\right)\left(\sigma_{j}^{-1}-\E[\sigma_{i}^{-1}]\right) + \largeO_{P}(N^{-1/2}).
\end{align*}
Combining our results on the components of equation \eqref{eq:CD_Ustat_decomp}, we conclude that
\begin{align*}
{CD}_{\mu, \sigma} &= \sqrt{\frac{2}{TN(N-1)}}\sum_{i=2}^{N} \sum_{j=1}^{i-1}\sum_{t=1}^{T}\epsi_{i,t}\epsi_{j,t}\left(\sigma_{i}^{-1}-\E[\sigma_{i}^{-1}]\right)\left(\sigma_{j}^{-1}-\E[\sigma_{i}^{-1}]\right) \\
&+ \sqrt{T}\Xi + \largeO_{P}(N^{-1/2})+ \largeO_{P}(T^{-1/2}),
\end{align*}
where
\begin{align*}
\Xi&=\sqrt{\frac{N}{2(N-1)}}\left(NT\right)^{-1}\sum_{i=1}^{N}\sum_{t=1}^{T} \epsi_{i,t}^{2}\left( \left( \overline{\sigma^{-1}}\right)^{2} - 2\overline{\sigma^{-1}}\sigma_{i}^{-1}  \right).
\end{align*}
Adding the remaining two components of the decomposition of $CD$ in Eqs. \eqref{eq:CD_FEdecomp_2WFE} and \eqref{eq:CD_Taylordecomp_2WFE} merely results in the additional order terms $\largeO_{P}\left(T^{-1}\sqrt{N}\right) + \largeO_{P}\left(N^{-1}\sqrt{T}\right)$.
It remains to prove weak convergence of the leading stochastic component in $CD$. For this purpose, set $a_{i,t}=\epsi_{i,t}(\sigma_{i}^{-1}-\E[\sigma_{i}^{-1}])$ and $q_{i}=\sigma_{i}(\sigma_{i}^{-1}-\E[\sigma_{i}^{-1}])$. We can then apply Lemma \ref{lem:U_stat} and conclude that
\begin{equation}
CD-\sqrt{T}\Xi \dto N(0,\Omega),
\end{equation}
as $N,T \to \infty$ subject to $\sqrt{N}T^{-1} \to 0$ and $\sqrt{T}N^{-1} \to 0$, where $\Omega=(\E[q_{i}^{2}])^{2}=\left(\E[(1-\sigma_{i}\E[\sigma_{i}^{-1}])^{2}]\right)^{2}$.
\end{proof}	

\bigskip

\begin{proof}[\textbf{Proof of Corollary \ref{corollary::homosigmas}}]

Using the decomposition in Eqs. \eqref{eq:CD_Taylordecomp_2WFE}- \eqref{eq:CD_FEdecomp_2WFE} for a general value of $\sigma^{2}$, we have
\begin{align}
CD 	&= CD_{\mu,\sigma}+CD_{(\widehat{\sigma}-\sigma)}-CD_{(\widehat{\mu}-\mu)},
\end{align}
where
\begin{equation}
CD_{\mu,\sigma}=\frac{1}{\sigma^{2}}\sqrt{\frac{1}{2TN(N-1)}}\sum_{t=1}^{T} \left(\sum_{i=1}^{N}\widehat{\epsi}_{i,t}\right)^2 - \frac{1}{\sigma^{2}}\sqrt{\frac{1}{2TN(N-1)}}\sum_{t=1}^{T}\sum_{i=1}^{N}\widehat{\epsi}_{i,t}^2.
\end{equation}
The first component on the right-hand side above is $0$ by construction. The second can be rewritten as
\begin{align*}
CD_{\mu,\sigma} 	&=\frac{1}{N}CD_{\epsi/\sigma}-\frac{1}{\sigma^{2}}\sqrt{\frac{2}{TN(N-1)}}\frac{N-1}{2N}\sum_{t=1}^{T}\left(\sum_{i=1}^{N}\epsi_{i,t}^{2}\right)\\
&=\frac{1}{N}CD_{\epsi/\sigma}-\sqrt{\frac{N-1}{2N^{2}}}\frac{1}{\sqrt{NT}}\sum_{t=1}^{T}\sum_{i=1}^{N}\left(\left(\frac{\epsi_{i,t}}{\sigma}\right)^{2}-1\right)-\sqrt{\left(T-\frac{T}{N}\right)/2},\\
&=-\sqrt{\left(T-\frac{T}{N}\right)/2}-\sqrt{\frac{1}{2N}}\frac{1}{\sqrt{NT}}\sum_{t=1}^{T}\sum_{i=1}^{N}\left(\left(\frac{\epsi_{i,t}}{\sigma}\right)^{2}-1\right)+\small\largeO_{P}(N^{-1/2}).
\end{align*}
Concerning the first expression in the last line, we can use the restriction $\sqrt{T}N\to 0$

Here, $CD_{\epsi/\sigma}$ is the usual CD test statistic based on raw error terms $\epsi_{i,t}/\sigma$. Lemma \ref{lem:U_stat} allows us to conclude that $CD_{\epsi/\sigma}=\largeO_{P}(1)$. The final result follows after observing that by double-index CLT:
\begin{equation}
\frac{1}{\sqrt{NT}}\sum_{t=1}^{T}\sum_{i=1}^{N}\left(\left(\frac{\epsi_{i,t}}{\sigma}\right)^{2}-1\right)=\frac{1}{\sqrt{NT}}\sum_{t=1}^{T}\sum_{i=1}^{N}\left(\eta_{i,t}^{2}-1\right)=\largeO_{P}(1),
\end{equation}
given $\E[|\eta_{i,t}|^{8}]<\infty$.
\end{proof}	

\newpage
\subsection{Proofs of Theorem \ref{theorem::CCE} and Proposition \ref{theorem::CCE_sigma}}
\label{sec:proof_CCE}

\begin{proof}[\textbf{Proof of Proposition \ref{theorem::CCE_sigma}}]\ \\
Analogous to the proof of Theorem \ref{theorem::additiveHetero}, we decompose $CD$ into
\begin{align}
\label{eq:CD_Taylordecomp_CCE}
CD&=2k_{N,T} \sum_{i=2}^{N}\sum_{j=1}^{i-1}\frac{\widehat{\vepsi}_{i}'\widehat{\vepsi}_{j}}{\sqrt{\sigma_{i}^{2}\sigma_{j}^{2}}}-k_{N,T}\sum_{i=1}^{N}\sum_{j\neq i}^{N}\frac{\widehat{\vepsi}_{i}'\widehat{\vepsi}_{j}}{\left(\varsigma_{i}^{2}\right)^{3/2}\left(\sigma_{j}^{2}\right)^{1/2}}\left(\widehat{\sigma}_{i}^{2}-\sigma_{i}^{2}\right) \nonumber \\
&= CD_{\sigma} +CD_{(\widehat{\sigma} - \sigma)}.
\end{align}
Under the assumptions of this proposition, the leading term of the CD test statistic is given by
	\begin{align*}
		CD_{\sigma} = \frac{1}{\sigma^{2}}2k_{N,T}\sum_{i=2}^{N}\sum_{j=1}^{i-1}\widehat{\vepsi}_{i}^{\prime }\widehat{\vepsi}_{j} &=
		\frac{1}{\sigma^{2}}k_{N,T}\sum_{i,j}^{N}\widehat{\vepsi}_{i}^{\prime }\widehat{\vepsi}_{j}- \frac{1}{\sigma^{2}}k_{N,T}\sum_{i=1}^{N}\widehat{\vepsi}_{i}^{\prime }\widehat{\vepsi}_{i} \\
		&=I-II.
	\end{align*}
	We begin by proving that $I=0$. Note that
	\begin{align*}
		\sum_{i,j}^{N}\widehat{\vepsi}_{i}^{\prime }\widehat{\vepsi}_{j}
		=\sum_{t=1}^{T}\left(\sum_{i=1}^{N}\widehat{\epsi}_{i,t}\right)^{2}.
	\end{align*}
	Furthermore, for each time period $t$ the sum of residuals can be expressed as
	\begin{equation}
		\sum_{i=1}^{N}\widehat{\epsi}_{i,t}=N\left(\overline{y}_{t}-\overline{\widehat{\vlambda}}^{\prime}\widehat{\vf}_{t}\right),
	\end{equation}
	where the average estimated loading $\overline{\widehat{\vlambda}}$ is implicitly defined by projection off the space spanned by $\widehat{\vf}_{t}$. That is, it can be written
	\begin{equation}
		\overline{\widehat{\vlambda}}=\left(\sum_{t=1}^{T}\widehat{\vf}_{t}\widehat{\vf}_{t}'\right)^{-1}\sum_{t=1}^{T}\widehat{\vf}_{t}\overline{y}_{t}=\left(\sum_{t=1}^{T}\widehat{\vf}_{t}\widehat{\vf}_{t}'\right)^{-1}\sum_{t=1}^{T}\widehat{\vf}_{t}\widehat{\vf}_{t}'\left[1,\vzeros_{m}'\right]'=\left[1,\vzeros_{m}'\right]',
	\end{equation}
	such that $\overline{\widehat{\vlambda}}^{\prime}\widehat{\vf}_{t}=\overline{y}_{t}$. Combining these results we conclude:
	\begin{equation}
		\sum_{i=1}^{N}\widehat{\epsi}_{i,t}=N\left(\overline{y}_{t}-\overline{\widehat{\vlambda}}^{\prime}\widehat{\vf}_{t}\right)=0,
	\end{equation}
	therefore also $I=0$. Next, we show that $II=\largeO_{P}(\sqrt{T})$. Using corresponding results from the proof of Theorem \ref{theorem::CCE} with $w_i = \sigma^{-1} \; \forall i$, we can write
	\begin{align*}
		II&=k_{N,T}\frac{1}{\sigma^{2}}\sum_{i=1}^{N}\sum_{t=1}^{T}\epsi_{i,t}^{2}+\largeO_{P}(N^{-1/2}\vee T^{-1/2})+\largeO_{P}(N^{-1}\sqrt{T})\\
		&=NT\sqrt{\frac{1}{2TN(N-1)}}+\frac{1}{\sqrt{2(N-1)}}\frac{1}{\sqrt{NT}}\left(\sum_{i=1}^{N}\sum_{t=1}^{T}\eta_{i,t}^{2}-1\right)+\largeO_{P}(N^{-1/2}\vee T^{-1/2})+\largeO_{P}(N^{-1}\sqrt{T}),\\
		&=\sqrt{\frac{TN}{2(N-1)}}+\largeO_{P}(N^{-1/2}\vee T^{-1/2})+\largeO_{P}(N^{-1}\sqrt{T}),\\
		&=\sqrt{\frac{T-T/N}{2}}+\largeO_{P}(N^{-1/2}\vee T^{-1/2})+\largeO_{P}(N^{-1}\sqrt{T}).
	\end{align*}
	Here the second line follows from double-index CLT applied to the sequence $\eta_{i,t}^{2}-1$. Combining our results with those on expressions $I$ and $II$ we conclude
	\begin{equation}
		CD=-\sqrt{\left(T-\frac{T}{N}\right)/2} + \largeO_{P}(N^{-1/2})+\largeO_{P}(T^{-1/2})+\largeO_{P}(N^{-1}\sqrt{T}).
	\end{equation}
\end{proof}

\bigskip

\begin{proof}[\textbf{Proof of Theorem \ref{theorem::CCE}}]
As in Proposition \ref{theorem::CCE_sigma}:

\begin{align}
\label{eq:CD_Taylordecomp_CCE}
CD&=2k_{N,T} \sum_{i=2}^{N}\sum_{j=1}^{i-1}\frac{\widehat{\vepsi}_{i}'\widehat{\vepsi}_{j}}{\sqrt{\sigma_{i}^{2}\sigma_{j}^{2}}}-k_{N,T}\sum_{i=1}^{N}\sum_{j\neq i}^{N}\frac{\widehat{\vepsi}_{i}'\widehat{\vepsi}_{j}}{\left(\varsigma_{i}^{2}\right)^{3/2}\left(\sigma_{j}^{2}\right)^{1/2}}\left(\widehat{\sigma}_{i}^{2}-\sigma_{i}^{2}\right) \nonumber \\
&= CD_{\sigma} +CD_{(\widehat{\sigma} - \sigma)}.
\end{align}
As shown in Section \ref{sec:var_estim_CCE}, $CD_{(\widehat{\sigma} - \sigma)}=\largeO_{P}\left( \sqrt{N}T^{-1} \right)+ \largeO_{P}\left( T^{-1/2} \right) + \largeO_{P}\left( N^{-1/2} \right) + \largeO_{P}\left( N^{-1} \sqrt{T} \right)$, implying that the estimation of unit-specific error variances is asymptotically negligible under weak assumptions on the relative rate of divergence of $N$  and $T$. 

Consider now $CD_{\sigma}$. To simplify notation and to emphasize the generality of our result, we derive its asymptotic properties for a generic set of random weights $\{ w_{i}\}_{i=1}^{N}$ which has the same properties as spelled out in Assumption \ref{ass:errors}.
Taking some conflicting notation with respect to Theorem \ref{thm:CD_W_combined} into account, let the equivalent to  ${CD}_{\sigma}$ based on these generic weights be denoted
\begin{align*}
CD_w = 2k_{N,T}\sum_{i=2}^{N}\sum_{j=1}^{i-1}\widehat{\vepsi}%
_{i}^{\prime }\widehat{\vepsi}_{j}w _{i}w _{j} &=
k_{N,T}\sum_{i,j}^{N}\widehat{\vepsi}_{i}^{\prime }\widehat{\vepsi}_{j}w_{i}w_{j}- k_{N,T}\sum_{i=1}^{N}\widehat{\vepsi}_{i}^{\prime }\widehat{\vepsi}_{i}w_{i}^{2} \\
&=I-II,
\end{align*}
where
\begin{align*}
\widehat{\vepsi}_{i} &=\mM_{\widehat{\mF}}\left(
\mF\vlambda_{i}+\vepsi_{i}\right) \\
&=\mM_{\widehat{\mF}}\vepsi_{i}-\mM_{%
	\widehat{\mF}}\left( \widehat{\mF}-\mF\overline{\mC}\right) \left(
\overline{\mC}^{-1}\right) \vlambda_{i} \\
&=\mM_{\widehat{\mF}}\vepsi_{i}-\mM_{%
	\widehat{\mF}}\overline{\mU}\left( \overline{\mC}^{-1}\right) \vlambda_{i}.
\end{align*}
Now let	
\begin{align*}
I &=k_{N,T}\sum_{i,j}^{N}\left(\vepsi_{i}-\overline{\mU}%
\left( \overline{\mC}^{-1}\right) \vlambda_{i}\right) ^{\prime
}\left(\vepsi_{j}-\overline{\mU}\left( \overline{\mC}%
^{-}\right) \vlambda_{j}\right) w_{i}w_{j} \\
&-k_{N,T}\sum_{i,j}^{N}\left(\vepsi_{i}-\overline{\mU}%
\left( \overline{\mC}^{-1}\right) \vlambda_{i}\right) ^{\prime }%
\mP_{\widehat{\mF}}\left(\vepsi_{j}-{\overline{\mU}%
}\left( \overline{\mC}^{-1}\right) \vlambda_{j}\right) w_{i}w_{j} \\
&=I_{1}-I_{2}.
\end{align*}
Here, analogous to how we proceeded in the time fixed effects model, we write
\begin{align*}
\sum_{i=1}^{N}\left(\vepsi_{i}-\overline{\mU}\left(
\overline{\mC}^{-1}\right) \vlambda_{i}\right) w_{i}=\sum_{i=1}^{N}\vepsi_{i}w_{i}-\sum_{i=1}^{N}%
\mU_{i}\left( \overline{\mC}^{-1}\right) \left( N^{-1}\sum_{\ell
	=1}^{N}\vlambda_{\ell }w_{\ell }\right) .
\end{align*}	
Consequently, letting $\overline{\vlambda}_{w }=N^{-1}\sum_{\ell=1}^{N}%
\vlambda_{\ell }w _{\ell }$\ and $\xi_{i}=%
\vepsi_{i}w _{i}-\mU_{i}\left( \overline{\mC}%
^{-}\right) \overline{\vlambda}_{w }$
\begin{align}
\label{eq:I1_arranged}
I_{1} &=k_{N,T}\sum_{i,j}^{N}\xi_{i}^{\prime }\xi%
_{j} \notag \\
&=2k_{N,T}\sum_{i=2}^{N}\sum_{j=1}^{i-1}\xi_{i}^{\prime }%
\xi_{j}+k_{N,T}\sum_{i=1}^{N}\xi_{i}^{\prime }%
\xi_{i},
\end{align}
where
\begin{align*}
k_{N,T}\sum_{i=1}^{N}\xi_{i}^{\prime }%
\xi_{i} &= k_{N,T}\sum_{i=1}^{N}\vepsi_{i}^{\prime }\vepsi_{i}w_{i}^{2} + k_{N,T}\overline{\vlambda}_{w }^{\prime }\left( \overline{\mC}%
^{-1}\right) ^{\prime }\left( \sum_{i=1}^{N}\mU_{i}^{\prime }\mU%
_{i}\right) \left( \overline{\mC}^{-1}\right) \overline{\vlambda}_{w
} \\
&-2k_{N,T}\overline{\vlambda}_{w }^{\prime }\left( \overline{\mC}%
^{-1}\right) ^{\prime }\left( \sum_{i=1}^{N}\mU_{i}^{\prime }\vepsi_{i}w _{i}\right)
\end{align*}
is a linear function of asymptotically non-negligible terms.
Lemma \ref{lem:I2} implies that $I_2 = \smallO_P(1)$ if both $\sqrt{T}/N \to 0$ and $T\to\infty$. Hence, proceed to
\begin{align*}
II &=k_{N,T}\sum_{i=1}^{N}\left(\vepsi_{i}-\overline{\mU}%
\left( \overline{\mC}^{-1}\right) \vlambda_{i}\right) ^{\prime
}\left(\vepsi_{i}-\overline{\mU}\left( \overline{\mC}%
^{-1}\right) \vlambda_{i}\right) w
_{i}^{2} \\
&-k_{N,T}\sum_{i=1}^{N}\left(\vepsi_{i}-\overline{\mU}%
\left( \overline{\mC}^{-1}\right) \vlambda_{i}\right) ^{\prime }%
\mP_{\widehat{\mF}}\left(\vepsi_{i}-\overline{\mU}\left( \overline{\mC}^{-1}\right) \vlambda_{i}\right) w
_{i}^{2} \\
&=II_{1}-II_{2}.
\end{align*}
As shown in Lemma \ref{lem:II2}, $II_{2} $ is $\smallO_P(1)$, if $N^{-1}\sqrt{T} \to 0$ and $T\to\infty$. For this reason, consider the first term above instead. We can write
\begin{align*}
II_{1}=k_{N,T}\sum_{i=1}^{N}\vepsi_{i}^{\prime }\vepsi _{i}w _{i}^{2}+k_{N,T}\sum_{i=1}^{N} \left\| \overline{\mU}\left( \overline{\mC}^{-1}\right) \vlambda_{i} \right\|^2 w _{i}^{2}-2k_{N,T}\sum_{i=1}^{N}\vlambda_{i}^{\prime
}\left( \overline{\mC}^{-1}\right) ^{\prime }\overline{\mU}^{\prime }%
\vepsi_{i}w _{i}^{2},
\end{align*}
where
\begin{align*}
&k_{N,T}\sum_{i=1}^{N}\left\| \overline{\mU}\left( \overline{\mC}^{-1}\right) \vlambda_{i} \right\|^2 w _{i}^{2} \\
&\leq k_{N,T} \frac{T}{N}\sum_{i=1}^{N} w _{i}^{2}\vlambda_{i}^{\prime }\vlambda_{i} \tr\left(  (N/T) \overline{\mU}%
^{\prime }\overline{\mU} \right) \left\| \overline{\mC}^{-1} \right\|^2  \\
&= \largeO_{P}(N^{-1}\sqrt{T})
\end{align*}	
by using result \eqref{eq:UbarUbar} and Markov's inequality on $\frac{1}{N} \sum_{i=1}^{N}w _{i}^{2}\vlambda_{i}^{\prime }\vlambda_{i}$. Next,
\begin{align*}
2k_{N,T} \left|\sum_{i=1}^{N}\vlambda_{i}^{\prime
}\left( \overline{\mC}^{-1}\right)^{\prime }\overline{\mU}^{\prime }%
\vepsi_{i}w _{i}^{2}\right| &\leq 2k_{N,T}\sqrt{N} \left( N^{-1} \sum_{i=1}^{N} \left\| \vlambda_{i}^{\prime}\left( \overline{\mC}^{-1}\right)^{\prime } \right\|^2 \right)^{1/2} \left( \sum_{i=1}^{N} \left\| \overline{\mU}^{\prime }%
\vepsi_{i}w _{i}^{2} \right\|^2 \right)^{1/2}.
\end{align*}
As noted above, $N^{-1} \sum_{i=1}^{N} \left\| \vlambda_{i}^{\prime}\left( \overline{\mC}^{-1}\right)^{\prime }  \right\|^{2} =\largeO_{P}(1)$. In order to proceed in  with the second term above, we use the identity $\overline{\mU} = \left[
\overline{\vepsi}, \; \overline{\mE}
\right]\mB$ and write
\begin{align*}
\E \left[ \sum_{i=1}^{N} \left\| \overline{\mU}^{\prime }%
\vepsi_{i}w _{i}^{2} \right\|^{2} \right] &\leq \left\| \mB \right\|^2 N^{-2}\sum_{i,i',i''}^{N} \sum_{t,t'}^T \E\left[ \epsi_{i,t} \epsi_{i',t}\epsi_{i'',t'} \epsi_{i,t'} \right] \E\left[w_{i}^{4}\right] \\
&+ \left\| \mB \right\|^2 N^{-2}\sum_{i,i',i''}^{N} \sum_{t,t'}^T \E\left[ \epsi_{i,t} \epsi_{i,t'} \right] \E \left[ \ve'_{i',t}\ve_{i'',t'}  \right]  \E\left[w_{i}^{4}\right] \\
&= \largeO(T) + \largeO(N^{-1}T^2),
\end{align*}
which follows from combining indexes such that only nonzero expectations remain. Together with uniform boundedness of $w_{i}^{4}$ and boundedness of $\left\| \mB \right\|^2$, we have
\begin{align}
\label{eq:II1_crossterm}
2k_{N,T} \left|\sum_{i=1}^{N}\vlambda_{i}^{\prime
}\left( \overline{\mC}^{-1}\right)^{\prime }\overline{\mU}^{\prime }%
\vepsi_{i}w _{i}^{2}\right| &= \sqrt{\frac{2}{TN(N-1)}} \sqrt{N} \largeO_{P}(1) \left(\largeO_{P}(\sqrt{T}) + \largeO_{P}(N^{-1/2}T) \right) \notag \\
&= \largeO_{P}(N^{-1/2}) + \largeO_{P}(N^{-1}\sqrt{T}).
\end{align}
Hence, we can conclude that
\begin{align}
\label{eq:II1_limit}
II_1 = k_{N,T}\sum_{i=1}^{N}\vepsi_{i}^{\prime }		\vepsi_{i}w_{i}^{2} + \largeO_{P}(N^{-1/2})+ \largeO_{P}(T^{-1/2}) + \largeO_{P}(N^{-1}\sqrt{T})
\end{align}
Combining the results on $I_1$, $I_2$, $II_1$ and $II_2$, we have
\begin{align*}
CD_w & = I_1 - I_2 - II_1 + II_2  \\
&= 2k_{N,T}\sum_{i=2}^{N}\sum_{j=1}^{i-1}\xi_{i}^{\prime }%
\xi_{j} + k_{N,T}\overline{\vlambda}_{w }^{\prime }\left( \overline{\mC}%
^{-1}\right) ^{\prime }\left( \sum_{i=1}^{N}\mU_{i}^{\prime }\mU%
_{i}\right) \left( \overline{\mC}^{-1}\right) \overline{\vlambda}_{w
} \\
&- 2k_{N,T}\overline{\vlambda}_{w }^{\prime }\left( \overline{\mC}%
^{-1}\right) ^{\prime }\left( \sum_{i=1}^{N}\mU_{i}^{\prime }\vepsi_{i} w_{i}\right) + \largeO_{P}(T^{-1/2}) + \largeO_{P}(N^{-1/2}) + \largeO_{P}(N^{-1}\sqrt{T}).
\end{align*}
This result transforms into one about ${CD}_{\sigma}$, the leading component of $CD$ in Eq. \eqref{eq:CD_Taylordecomp_CCE}, if we set $w_i = \sigma_i^{-1}$,  and define the influence function as $\vpsi_{i,t} = \left( \overline{\mC}'\right)^{-1}\vu_{i,t}$. Additionally using the results of Section \ref{sec:var_estim_CCE}, the main result of this theorem follows. 
\end{proof}	

\newpage
\subsection{Proof of Theorem \ref{thm:CD_W_combined} }
\label{sec:proof_CD_W_combined}	

\paragraph{Additive part.}\ \\ 
Consider the $CD_{W}$ statistic without variance scaling, and denote it by $\ddot{CD}_{W}$, i.e:
\begin{align*}
\ddot{CD}_{W}&=\sqrt{\frac{2}{TN(N-1)}}  \sum_{i=2}^N \sum_{j=1}^{i-1} \widehat{\vepsi}_{i}' \widehat{\vepsi}_{j}w_i w_j \\
&=CD_{\mu,w}-CD_{(\widehat{\mu}-\mu)}.
\end{align*}
From the results in Section \ref{sec:var_estim_2WFE} we notice that $CD_{(\widehat{\mu}-\mu)}=\largeO_{P}(T^{-1/2})$ irrespective if one considers $\sigma_{i}^{-1}$ or any $w_{i}$ as weights. Moreover, the decomposition of $CD$ provided in Theorem \ref{theorem::additiveHetero} is purely algebraic and holds even with general weights $w_i$. 
From Theorem \ref{theorem::additiveHetero} :
	\begin{align*}
			CD_{\mu,w}= 2k_{N,T}  \sum_{i=2}^N \sum_{j=1}^{i-1} \vepsi_{i}' \vepsi_{j} \left(w_i - \overline{w} \right) \left(w_j - \overline{w} \right) + \sqrt{T}\left(\Xi_{W,1} + \Xi_{W,2}\right),
	\end{align*}
where
\begin{align*}
	\Xi_{W,1} &= \sqrt{\frac{N}{2(N-1)}} \frac{1}{NT}\sum_{i=1}^{N}\sum_{t=1}^{T} \epsi_{i,t}^{2}\left(\left(\overline{w}\right)^{2} - 2\overline{w}w_{i} \right) \\
	\Xi_{W,2} &= \sqrt{\frac{1}{2N(N-1)}} \frac{1}{NT} \sum_{i=1}^{N}\sum_{t=1}^{T} \epsi_{i,t}^{2} \left(  \left( \overline{w}\right)^{2} - 2w_i^2 \right).
\end{align*}
Here, we consider the leading bias term $\Xi_{1,W}$ where $(NT)^{-1} \sum_{i=1}^N\sum_{t=1}^T \epsi_{i,t}^2 = \largeO_{P}(1)$,  $\overline{w}= \largeO_{P}\left(N^{-1/2}\right)$ and $N^{-1/2}T^{-1}\sum_{i=1}^{N}\sum_{t=1}^{T} \epsi_{i,t}^{2}w_i^2 = \largeO_{P}\left(1\right)$ straightforwardly hold by Markov's and Chebyshev's inequalities and uniform boundedness of the error variances. It follows that $\Xi_{1,W} = \largeO_{P}\left(N^{-1}\right)$. For the second remainder term $\Xi_{W,2}$, we note that $ \frac{1}{NT} \sum_{i=1}^{N}\sum_{t=1}^{T} \epsi_{i,t}^{2}w_i^2=\largeO_{P}(1)$ which together with the previous intermediary results implies $\Xi_{W,2} = \largeO_{P}(N^{-2})$.
	
We continue with the leading stochastic term in $CD_W$, given by
\begin{align*}
& 2k_{N,T}  \sum_{i=2}^N \sum_{j=1}^{i-1} \sum_{t=1}^T \epsi_{i,t} \epsi_{j,t}\left(w_i - \overline{w} \right) \left(w_j - \overline{w} \right)  \\
&=2k_{N,T}  \sum_{i=2}^N \sum_{j=1}^{i-1} \sum_{t=1}^T  \epsi_{i,t} \epsi_{j,t}w_i w_j - (\overline{w}-N^{-1})\left(2k_{N,T} \sum_{i=2}^N \sum_{j=1}^{i-1} \sum_{t=1}^T \epsi_{i,t} \epsi_{j,t}(w_i+w_j)\right)\\
&+(\overline{w}^{2}-N^{-1}\overline{w^{2}})\left(2k_{N,T} \sum_{i=2}^N \sum_{j=1}^{i-1} \sum_{t=1}^T \epsi_{i,t} \epsi_{j,t}\right).
\end{align*}
Given that $\E[w_{i}]=0$, the second and third components are of orders $\largeO_{P}(N^{-1/2})$ and $\largeO_{P}(N^{-1})$, respectively. This follows from the application of Lemma \ref{lem:U_stat} on corresponding U-statistics. The remaining non-negligible term is of the form:
	\begin{align*}
	\sqrt{\frac{2}{TN(N-1)}}  \sum_{i=2}^N \sum_{j=1}^{i-1} \sum_{t=1}^T \epsi_{i,t} \epsi_{j,t}w_i w_j.
	\end{align*}
	From Lemma \ref{lem:U_stat} it follows that:
	\begin{align}
	\label{eq:eps_CLT}
	\sqrt{\frac{2}{TN(N-1)}}\sum_{i=2}^N \sum_{j=1}^{i-1} \sum_{t=1}^T\epsi_{i,t}\epsi_{j,t}w_{i}w_{j} \dto N(0,\Omega)
	\end{align}
	as $N,T \to \infty$ jointly, where $\Omega=\E[\sigma_{i}^{2}w_{i}^{2}]^{2}$
	The last step in this proof is to show consistency of the standard deviation estimator $\widehat{\Omega} = (NT)^{-1} \sum_{i=1}^N \sum_{t=1}^{T} \widehat{\epsi}_{i,t}^{2} w_i^2$ so that convergence in law of $CD_W$ to a standard normal distribution follows by Slutsky's Theorem. Note here that an equivalent expression is given by term $II$ in the proof of Theorem  \ref{theorem::additiveHetero}. Drawing from results in Section \ref{sec:var_estim_2WFE}, we conclude that
	\begin{align*}
	\frac{1}{NT} \sum_{i=1}^N \sum_{t=1}^{T} \widehat{\epsi}_{i,t}^2 w_i^2 =\frac{1}{N} \sum_{i=1}^N \widehat{\sigma}_{i}^{2} w_i^2= \frac{1}{N} \sum_{i=1}^N \sigma_{i}^{2} w_i^2 + \largeO_{P}((NT)^{-1/2})+\largeO_{P}(T^{-1})+\largeO_{P}(N^{-1}).
	\end{align*}
	Consistency follow from the Chebyshev's inequality. It follows that $CD_W \dto N(0,1)$.

\paragraph{Multifactor error part.}\ \\
	Proceeding from the decomposition given in Theorem \ref{theorem::CCE}, we consider the two bias terms $\Phi_{W,1}$ and $\Phi_{W,2}$. For both terms, it is instructive to note that
	\begin{align*}
		\E\left[ \left\| N^{-1}\sum_{i=1}^N w_i \vlambda_{i} \right\|^2 \right] &= N^{-2}\sum_{i,j}^N \E\left[w_iw_j\right] \E \left[ \vlambda_{i}' \vlambda_j\right]
		= N^{-2}\sum_{i=1}^N \var\left[w_i\right] \E \left[ \vlambda_{i}' \vlambda_{i}\right]
		= \largeO\left(N^{-1}\right),
	\end{align*}
so that
\begin{align}
\label{eq:sumN_wilami}
	\left\| N^{-1}\sum_{i=1}^N w_i \vlambda_{i} \right\| = \largeO_{P}\left( N^{-1/2}\right).
\end{align}
 Concerning $\Phi_{W,1}$, we also need to take into account
\begin{align*}
	\left\| \left(NT\right)^{-1} \sum_{i=1}^N \sum_{t=1}^T \vpsi_{i,t} \vpsi_{i,t}'\right\| &\leq \left(NT\right)^{-1} \sum_{i=1}^N \sum_{t=1}^T \vu_{i,t}'\vu_{i,t} \left\| \overline{\mC}^{-1} \right\|^2 \\
	&= \largeO_{P}\left(1\right),
\end{align*}
which follows from Markov's inequality. Using this result, we have
\begin{align*}
\Phi_{1,W} &\leq \sqrt{\frac{N}{2(N-1)}} \left\| \left(NT\right)^{-1} \sum_{i=1}^N \sum_{t=1}^T \vpsi_{i,t} \vpsi_{i,t}'\right\|  \left\| N^{-1}\sum_{i=1}^N w_i \vlambda_{i} \right\|^2 \\
&= \largeO_{P}(N^{-1}).
\end{align*}
Next, consider $\widetilde{\Phi}_{2,W}$. By isolating $\overline{\mC}^{-1}$ from the definition of $\vpsi_{i,t}$ and by using result \eqref{eq:sumNT_uit_epsit_wi} on the remaining part:
\begin{align*}
	\left\| \left( NT \right)^{-1}  \sum_{i=1}^N \sum_{t=1}^T \vpsi_{i,t} w_i \epsi_{i,t} \right\| = \largeO_{P}\left( N^{-1/2} \right).
\end{align*}
Consequently,
\begin{align*}
	\left| \Phi_{W,2} \right| &\leq \sqrt{\frac{N}{2(N-1)}} \left\| N^{-1}\sum_{i=1}^N w_i \vlambda_{i} \right\| \left\| \left( NT \right)^{-1}  \sum_{i=1}^N \sum_{t=1}^T \vpsi_{i,t} w_i \epsi_{i,t} \right\| \\
	&= \largeO_{P}\left( N^{-1}	\right).
\end{align*}
Thus, the two bias terms affecting $CD_W$ in Theorem \ref{theorem::CCE} are negligible for weights satisfying Assumption \ref{ass:weights} as long as $N^{-1}\sqrt{T} \to 0$. We continue with the leading stochastic term of the weighted CD test statistic, given by
\begin{align*}
	2k_{N,T} \sum_{i=2}^N \sum_{j=1}^{i-1} \sum_{t=1}^T  \xi_{iN,t} \xi_{jN,t} 
\end{align*}
with $ \xi_{iN,t} = \left[w_i \epsi_{i,t} - \left( N^{-1} \sum_{\ell=1}^N w_{\ell} \vlambda_{\ell}' \right) \vpsi_{i,t} \right]$  and $\vpsi_{i,t} =\overline{\mC}^{-1} \vu_{i,t}$. Analogous to the time fixed effects part of this proof, we can apply Lemma \ref{lem:U_stat} with $a_{i,t} = w_i\epsi_{i,t}$ and $q_{i} = w_i\sigma_{i}$ to obtain
\begin{align*}
\sqrt{\frac{2}{TN(N-1)}}\sum_{i=2}^N \sum_{j=1}^{i-1} \sum_{t=1}^T \epsi_{i,t}\epsi_{j,t} w_i w_j \dto N(0,   (\sigma_w^2 \E \left[\sigma_{i}^2\right])^2)
\end{align*}
as $N,T \to \infty$ subject to $N^{-1}\sqrt{T} \to 0$. Now note that by Eq. \eqref{eq:II1_limit} in the proof of Theorem \ref{theorem::CCE},
\begin{align*}
	(NT)^{-1} \sum_{i=1}^N \sum_{t=1}^T \widehat{\epsi}_{i,t}^2 w_i^2 = (NT)^{-1}\sum_{i=1}^{N}\vepsi_{i}^{\prime }\vepsi_{i}w_{i}^{2} + \largeO_{P}\left[(NT)^{-1/2}\right] + \largeO_{P}(N^{-1})+\largeO_{P}(T^{-1})
\end{align*}
where $(NT)^{-1}\sum_{i=1}^{N}\vepsi_{i}^{\prime }\vepsi_{i}w_{i}^{2} \pto \sigma_w^2 \E \left[\sigma_{i}^2\right]$ as argued in the time fixed effects part of the proof of this theorem. By application of Slutsky's Theorem, it then follows that
\begin{align*}
\left[ (NT)^{-1} \sum_{i=1}^N \sum_{t=1}^T \widehat{\epsi}_{i,t}^2 w_i^2\right]^{-1}\sqrt{\frac{2}{TN(N-1)}}\sum_{i=2}^N \sum_{j=1}^{i-1} \sum_{t=1}^T\epsi_{i,t}\epsi_{j,t}w_i w_j \dto N(0,1).
\end{align*}
Next, consider the sum
\begin{align}
\label{eq:sumNNT_uit_ujt}
	2k_{N,T}  \sum_{i=2}^N \sum_{j=1}^{i-1} \sum_{t=1}^T \vu_{i,t}  \vu_{j,t}' = \mB' \left(2k_{N,T} \sum_{i=2}^N \sum_{j=1}^{i-1} \sum_{t=1}^T
	 (\epsi_{i,t},\ve_{i,t}')'(\epsi_{j,t},\ve_{j,t}') \right)  \mB.
\end{align}
Given that $\ve_{i,t}$ has properties similar to $\epsi_{i,t}$, it can be shown analogously the reasoning leading to \eqref{eq:eps_CLT} that a CLT holds for every element of the $(m+1) \times (m+1)$ matrix in the parentheses of Eq.  \eqref{eq:sumNNT_uit_ujt} by Lemma \ref{lem:U_stat}. This allows us to conclude that the elements of \eqref{eq:sumNNT_uit_ujt} are stochastically bounded. Thus, recalling result \eqref{eq:sumN_wilami},
\begin{align*}
	\left( N^{-1} \sum_{\ell=1}^N w_{\ell} \vlambda_{\ell}' \right) \left(\sqrt{\frac{2}{TN(N-1)}}  \sum_{i=2}^N \sum_{j=1}^{i-1} \sum_{t=1}^T \vpsi_{i,t} \vpsi_{j,t}'\right) \left( N^{-1} \sum_{\ell=1}^N \vlambda_{\ell}  w_{\ell} \right) &= \largeO_{P}(N^{-1}), \\
	\left(\sqrt{\frac{2}{TN(N-1)}}  \sum_{i=2}^N \sum_{j=1}^{i-1} \sum_{t=1}^T \epsi_{i,t} \vpsi_{j,t}'\right) \left( N^{-1} \sum_{\ell=1}^N \vlambda_{\ell}  w_{\ell} \right)&=\largeO_{P}(N^{-1/2}), \\
	\left( N^{-1} \sum_{\ell=1}^N w_{\ell} \vlambda_{\ell}' \right) \left(\sqrt{\frac{2}{TN(N-1)}}  \sum_{i=2}^N \sum_{j=1}^{i-1} \sum_{t=1}^T \vpsi_{i,t}  \epsi_{i,t}\right) &= \largeO_{P}(N^{-1/2}),
\end{align*}
so that
\begin{align*}
 \left( (NT)^{-1} \sum_{i=1}^N \sum_{t=1}^T \widehat{\epsi}_{i,t}^2 w_i^2\right)^{-1} \sqrt{\frac{2}{TN(N-1)}}  \sum_{i=2}^N \sum_{j=1}^{i-1} \sum_{t=1}^T \xi_{iN,t} \xi_{jN,t} \dto N(0,1),
\end{align*}
which leads to the final result of this part of the theorem.

\subsection{Auxiliary lemmas for Sections \ref{sec:proof_additiveHetero}--\ref{sec:proof_CD_W_combined}}
\label{sec:suppl_lemmas}
\begin{lemma}
	\label{lem:U_stat}
	Let $\{a_{i,t}\}_{t=1}^{T}$ be a scalar and $\vc_{i}$ $L-$dimensional sequences of random variables for $i=1,\ldots,N$ such that:
	\begin{itemize}
		\item $a_{i,t},a_{j,s}$ are independent for all $i\neq j$ and all $s,t$.
		\item $a_{i,t},a_{i,s}$ are iid conditionally on $\vc_{i}$ for all $s,t$.
		\item $\E[|a_{i,t}|^{8}|\vc_{i}]<M<\infty$. $\vc_{i}$ are iid over $i$.
		\item $q_{i}^{k}\equiv\E[|a_{i,t}|^{k}|\vc_{i}]$ and $\E[a_{i,t}|\vc_{i}]=0$ for $k\leq 8$.
		\item $\E[|q_{i}^{k}|]<\infty$ for $k\leq 8$.
	\end{itemize}
	Then:
	\begin{equation}
		U=\sqrt{\frac{2}{NT(N-1)}}\sum_{i=2}^{N}\sum_{j=1}^{i-1}\sum_{t=1}^{T}a_{i,t}a_{j,t}\dto N(0,\E[q_{i}^{2}]^{2}),
	\end{equation}
	jointly as $N,T\to\infty$.
\end{lemma}
\begin{proof}[Proof of Lemma \ref{lem:U_stat}]\ \\
	The prove this lemma we use Theorem 3.2 of \citet{HallHeyde1980}. In particular, express $U$ as:
	\begin{equation}
		U=\sum_{t=1}^{T}\xi_{t,N,T},\quad \xi_{t,N,T}=\sqrt{\frac{2}{NT(N-1)}}\sum_{i=2}^{N}\sum_{j=1}^{i-1}a_{i,t}a_{j,t}.
	\end{equation}
	Denote by $\calC=\sigma(\{\vc_{j}\}_{j=1}^{N})$ the $\sigma$-algebra generated by $\{\vc_{j}\}_{j=1}^{N}$ and let $\calF_{t-1,N,T}=\sigma(\calC\vee \{\xi_{s,N,T}\}_{s=1}^{t-1})$ be the $\sigma$-algebra generated by $\calC$ and $\{\xi_{s,N,T}\}_{s=1}^{t-1}$. It is easy to see that $\{\xi_{t,N,T},\calF_{t-1,N,T}:t=1,\ldots,T\}$ is a Martingale Difference Array.
	
	At first we establish the limiting variance of this MDS array. From Corollary 3.1 in \citet{HallHeyde1980} the variance is determined from:
	\begin{equation}
		V_{T}=\sum_{t=1}^{T}\E[\xi_{t,N,T}^{2}|\calF_{t-1,N,T}]\pto \eta^{2}
	\end{equation}
	By conditional independence of $a_{i,t}$, the above result simplifies:
	\begin{align*}
		\sum_{t=1}^{T}\E[\xi_{t,N,T}^{2}|\calF_{t-1,N,T}]&=\sum_{t=1}^{T}\E[\xi_{t,N,T}^{2}|\calC]\\
		&=\frac{2}{N(N-1)}\E\left[\left(\sum_{i=2}^{N}\sum_{j=1}^{i-1}a_{i,t}a_{j,t}\right)^{2}|\calC\right]\\
		&=\frac{2}{N(N-1)}\E\left[\left(\sum_{i=2}^{N}a_{i,t}A_{i-1,t}\right)^{2}|\calC\right]\\
		&=\frac{2}{N(N-1)}\E\left[\sum_{i=2}^{N}a_{i,t}^{2}A_{i-1,t}^{2}|\calC\right]\\
		&=\frac{2}{N(N-1)}\sum_{i=2}^{N}q_{i}^{2}\sum_{j=1}^{i-1}q_{j}^{2},
	\end{align*}
	where in the third line we defined the ``integrated'' variable $A_{i-1,t}=\sum_{j=1}^{i-1}a_{j,t}$, the fourth line uses the fact that $a_{i,t}A_{i-1,t}$ is a Martingale Difference Sequence. The last line uses the definition of $q_{i}^{2}$. After re-arranging elements:
	\begin{align*}
		V_{T}&=\frac{2}{N(N-1)}\sum_{i=2}^{N}q_{i}^{2}\sum_{j=1}^{i-1}q_{j}^{2}\\
		&=\frac{1}{N(N-1)}\left(N^{2}\left(\frac{1}{N}\sum_{i=1}^{N}q_{i}^{2}\right)^{2}-N\left(\frac{1}{N}\sum_{i=1}^{N}q_{i}^{4}\right)\right)\\
		&=\left(\frac{1}{N}\sum_{i=1}^{N}q_{i}^{2}\right)^{2}+\largeO_{P}(N^{-1})\\
		&=\E[q_{i}^{2}]^{2}+\largeO_{P}(N^{-1}).
	\end{align*}
	Here the third and the fourth lines follow from the application of the Kolmogorov's SLLN to iid sequences $q_{i}^{2}$ and $q_{i}^{4}$. Thus we can expect that
	\begin{equation}
		U=\sqrt{\frac{2}{NT(N-1)}}\sum_{i=2}^{N}\sum_{j=1}^{i-1}\sum_{t=1}^{T}a_{i,t}a_{j,t}\dto N(0,\E[q_{i}^{2}]^{2}).
	\end{equation}
	It remains to prove that the sufficient condition for Corollary 3.1 in \citet{HallHeyde1980} is satisfied. In particular, it is sufficient to show that $\xi_{t,N,T}$ is a (conditionally) uniformly integrable sequence:
	\begin{equation}
		\sum_{t=1}^{T}\E[\xi_{t,N,T}^{2}I(|\xi_{t,N,T}|>\epsi)|\calF_{t-1,N,T}]\pto 0.
	\end{equation}
	Instead of proving uniform integrability, we borrow the idea from  \citet{kao2012asymptotics} and instead show that conditional Lyapunov's condition:
	\begin{equation}
		B_{T}=\sum_{t=1}^{T}\E[|\xi_{t,N,T}|^{2+\delta}|\calF_{t-1,N,T}]\pto 0,\quad \delta>0,
	\end{equation}
	is satisfied under the maintained assumptions. In what follows we prove that the above condition is satisfied for $\delta=2$. 
		Observe how:
	\begin{align*}
		B_{T}&=T\E[\xi_{t,N,T}^{4}|\calC]=T^{-1}\frac{N}{4(N-1)}\E\left[\left(N\overline{a}_{t}^{2}-\overline{a^{2}}_{t}\right)^{4}|\calC\right],
	\end{align*}
	where $\overline{a}_{t}=N^{-1}\sum_{i=1}^{N}a_{i,t}$ and similarly for $\overline{a^{2}}_{t}$. By Minkowski's inequality:
		\begin{align*}
		B_{T}&\leq T^{-1}\frac{N}{4(N-1)}\left(\sqrt[4]{\E\left[\left(N\overline{a}_{t}^{2}\right)^{4}|\calC\right]}+\sqrt[4]{
		\E\left[\left(\overline{a^{2}}_{t}\right)^{4}|\calC\right]}\right)^{4}.
	\end{align*}
As by assumption $\E[|a_{i,t}|^{8}|\vc_{i}]<M$ we can use the (conditional) Rosenthal's inequality to conclude that:
\begin{equation}
\E\left[\left(N\overline{a}_{t}^{2}\right)^{4}|\calC\right]=\largeO_{P}(1).
\end{equation}
Moreover, another application of Minkowski's inequality yields:
\begin{equation}
\sqrt[4]{\E\left[\left(\overline{a^{2}}_{t}\right)^{4}|\calC\right]}\leq \sqrt[4]{
		\E\left[\left(\frac{1}{N}\sum_{i=1}^{N}(a_{i,t}^{2}-q_{i}^{2})\right)^{4}|\calC\right]}+\sqrt[4]{
		\left(\frac{1}{N}\sum_{i=1}^{N}(q_{i}^{2})\right)^{4}}.
\end{equation}
Further application of the (conditional) Rosenthal's inequality yields:
\begin{equation}
\E\left[\left(\frac{1}{N}\sum_{i=1}^{N}(a_{i,t}^{2}-q_{i}^{2})\right)^{4}|\calC\right]=\largeO_{P}(N^{-2}),
\end{equation}
which together with the fact that $N^{-1}\sum_{i=1}^{N}(q_{i}^{2})=\largeO_{P}(1)$ implies that:
\begin{align*}
		B_{T}&\leq T^{-1}\frac{N}{4(N-1)}\largeO_{P}(1)=\smallO_{P}(1).
\end{align*}	
As required. This completes the proof.
\end{proof}

\bigskip

\begin{lemma}
	\label{lem:CCEbasicresults}
	Let $w_i$ denote any stochastic weights that are independent across $i$ and satisfy $\E\left[w_i^2\right] < M$ for all $i$. Furthermore, let $w_i$ be independent of $\vf_t$, $\epsi_{i,t}$ and $\ve_{i,t}$ for all $i$ and $t$. Under Assumptions \ref{ass:errors}-\ref{ass:rank},
	\begin{align}
		\left( T^{-1}\widehat{\mF}^{\prime }\widehat{\mF}\right) ^{-1} &= \largeO_{P}\left( 1\right) \label{eq:FhatFhat} \\
		\sqrt{N}T^{-1/2} \left\|N^{-1} \sum_{i=1}^N w_i \mU_i^{\prime } \right\| &= \largeO_{P}\left( 1 \right) \label{eq:UbarUbar} \\
		\sqrt{N}T^{-1/2} \left\| \left(N^{-1}\sum_{i=1}^Nw_i\mU_{i}^{\prime}\right)\mF \right\| & = \largeO_{P}\left( 1 \right) \label{eq:UbarF}	
	\end{align}
	Additionally, under Assumptions \ref{ass:errors}-\ref{ass:rank} and given weights $w_i$ that satisfy Assumption \ref{ass:weights}, we have
	\begin{align}
		\left\|(NT)^{-1} \sum_{i=1}^N \mU_i'\vepsi_{i} w_{i} \right\|	&= \largeO_{P}(N^{-1/2}) \label{eq:sumNT_uit_epsit_wi}	
	\end{align}
\end{lemma}

\bigskip

\begin{proof}[Proof of Lemma \ref{lem:CCEbasicresults}]\ \\
	\hfill
	\begin{enumerate}
		\item \emph{Result \eqref{eq:FhatFhat}}: By \citet[][equation (36)]{ECTA:ECTA692} it holds that
		\begin{align*}
			T^{-1}\widehat{\mF}^{\prime }\widehat{\mF} = \overline{\mC}' \left( T^{-1}\mF' \mF \right) \overline{\mC} + \largeO_{P}\left[(NT)^{-1/2}\right] + \largeO_{P}(N^{-1}),
		\end{align*}
		where the function of true factors and true loadings itself converges to a positive definite matrix by Assumptions \ref{ass:factors} and \ref{ass:loadings}. Hence, Theorem 1 in \citet{KaraReeseWesterlund2015} applies and equation \eqref{eq:FhatFhat} follows.
		\item \emph{Result \eqref{eq:UbarUbar}}: Taking the square of the expression of interest, we can write
		\begin{align*}
			NT^{-1}\left\|N^{-1} \sum_{i=1}^N w_i \mU_i^{\prime } \right\|^2 \leq 	(NT)^{-1} \sum_{i,j}^N \sum_{t=1}^T w_i w_j \left( \epsi_{i,t}\epsi_{j,t} + \ve_{i,t}'\ve_{j,t} \right) \left\| \mB \right\|^2
		\end{align*}
		where $(NT)^{-1} \sum_{i,j}^N \sum_{t=1}^Tw_i w_j \epsi_{i,t}\epsi_{j,t} = \largeO_{P}\left( 1 \right)$ and $(NT)^{-1} \sum_{i,j}^N \sum_{t=1}^T w_i w_j \ve_{i,t}'\ve_{j,t} = \largeO_{P}(1)$ by Markov's inequality and zero correlation of idiosyncratic variation in both $\vy_i$ and $\mX_i$ across $i$. Result \eqref{eq:UbarUbar} follows accordingly.
		\item \emph{Result \eqref{eq:UbarF}}: Analogous to the proof of \eqref{eq:UbarUbar} we take the square of \eqref{eq:UbarF} and rearrange to arrive at
		\begin{align*}
			NT^{-1}\left\|N^{-1}\sum_{i=1}^Nw_i\mU_{i}^{\prime}\mF\right\|^2 \leq (NT)^{-1}\sum_{i,j}^N \sum_{t,t'}^T w_i w_j \vf_t'\vf_{t'} \left( \epsi_{i,t} \epsi_{j,t'} + \ve_{i,t}' \ve_{j,t'} \right) \left\| \mB \right\|^2.
		\end{align*}
		Now taking expectations of the non-negative expression $(NT)^{-1}\sum_{i,j}^N \sum_{t,t'}^T w_i w_j \vf_t'\vf_{t'} \epsi_{i,t} \epsi_{j,t'}$ and using uncorrelatedness of $\epsi_{i,t}$ across $i$ and $t$ as well as $\E \left(w_i^2\right) \leq M$, we obtain
		\begin{align*}
			\E \left[(NT)^{-1}\sum_{i,j}^N \sum_{t,t'}^T w_i w_j \vf_t'\vf_{t'} \epsi_{i,t} \epsi_{j,t'} \right] &\leq (NT)^{-1}\sum_{i=1}^N \sum_{t=1}^T M \: \tr( \mSigma_{\mF}) \E \left[\epsi_{i,t}^2\right] \\
			&= \largeO\left(1 \right)
		\end{align*}
		An identical result is given for the term involving $\ve_{i,t}' \ve_{j,t'} $. Result \eqref{eq:UbarF} is then obtained via Markov's inequality.		
		\item  \emph{Result \eqref{eq:sumNT_uit_epsit_wi}}: Given the definition of $\mU_i$, we can write
		\begin{align*}
			\E \left[ \left\| (NT)^{-1} \sum_{i=1}^N \mU_i'\vepsi_{i} w_{i} \right\|^2  \right] \leq (NT)^{-2} \sum_{i=1}^N \sum_{t,t'}^T \E \left[ \epsi_{i,t}\left( \epsi_{i,t} \epsi_{i,t'} + \ve_{i,t}'\ve_{i,t'}\right)\epsi_{i,t'} \right]  \var\left[ w_i \right] \left\| \mB \right\|^2
		\end{align*}
		where we directly use independence of $w_i$ across $i$ and $\E[w_{i}]=0$. Additionally,
		\begin{align*}i
			(NT)^{-2} \sum_{i=1}^N \sum_{t,t'}^T\E \left[ \epsi_{i,t} \epsi_{i,t} \epsi_{i,t'}\epsi_{i,t'} \right] &=  (NT)^{-2} \sum_{i=1}^N \sum_{t=1}^T \kappa_4 \left[ \epsi_{i,t} \right] + (NT)^{-2} \sum_{i=1}^N \sum_{t,t'}^T \E \left[ \epsi_{i,t}^2\right] \E\left[  \epsi_{i,t'}^2 \right] \\
			&= \largeO(N^{-1})
		\end{align*}
		and
		\begin{align*}
			(NT)^{-2} \sum_{i=1}^N \sum_{t,t'}^T\E \left[ \epsi_{i,t} \ve_{i,t}' \ve_{i,t'}\epsi_{i,t'} \right] &= (NT)^{-2} \sum_{i=1}^N \sum_{t=1}^T \E \left[ \epsi_{i,t}^2\right] \tr\left( \mSigma \right) \\
			&= \largeO\left[ \left( NT \right)^{-1} \right],
		\end{align*}
		from which it follows that
		\begin{align*}
			\left\| (NT)^{-1} \sum_{i=1}^N \mU_i'\vepsi_{i} w_{i} \right\|^2 = \largeO_{P}\left( N^{-1} \right)
		\end{align*}
		by Markov's inequality. Taking the square root leads to result \eqref{eq:sumNT_uit_epsit_wi}.
	\end{enumerate}
\end{proof}

\bigskip

\begin{lemma}
	\label{lem:I2} Under Assumptions \ref{ass:errors}-\ref{ass:rank} and $\E[|w_{i}|^{8}]<M$:
	
	\begin{align*}
		&\sqrt{\frac{1}{2TN(N-1)}}\sum_{i,j}^{N}\left(\vepsi_{i}-\overline{\mU}
		\left( \overline{\mC}^{-1}\right) ^{\prime }\vlambda_{i}\right)
		^{\prime }\mP_{\widehat{\mF}}\left(\vepsi_{j}-
		\overline{\mU}\left( \overline{\mC}^{-1}\right) ^{\prime }\vlambda_{j}\right) w_{i}w_{j}\\
		&=  \largeO_{P}\left( T^{-1/2}\right) + \largeO_{P}\left( N^{-1/2}\right) + \largeO_{P}\left( N^{-1}\sqrt{T}\right)
	\end{align*}	
\end{lemma}

\bigskip

\begin{proof}[Proof of Lemma \ref{lem:I2}]\ \\
	Let $  \overline{\vlambda}_{w} = \left( N^{-1}\sum_{\ell=1}^{N}\vlambda_{\ell }w _{\ell }\right) $ and note that $\sum_{i=1}^{N}\left(\vepsi_{i}-\overline{\mU}
	\left( \overline{\mC}^{-1}\right) ^{\prime }\vlambda_{i}\right)w _{i} = \sum_{i=1}^{N} \left(\vepsi_{i}w _{i}-\mU_{i}\left( \overline{\mC}^{-1}\right) \overline{\vlambda}_{w}\right)$. This allows us to write
	\begin{align*}
		& \sqrt{\frac{1}{ 2TN(N-1)}}\sum_{i,j}^{N}\left(\vepsi_{i}-\overline{\mU}
		\left( \overline{\mC}^{-1}\right) ^{\prime }\vlambda_{i}\right)
		^{\prime }\mP_{\widehat{\mF}}\left(\vepsi_{j}-
		\overline{\mU}\left( \overline{\mC}^{-1}\right) ^{\prime }\vlambda_{j}\right) w _{i}w _{j} \\
		&\leq  \sqrt{\frac{1}{ 2TN(N-1)}}  \left(\left\|\sum_{i=1}^{N}\widehat{\mF}'\vepsi_{i}w_{i} \right\|  + \left\|\sum_{i=1}^{N}\widehat{\mF}'\mU_{i} \right\| \left\| \overline{\mC}^{-1} \right\| \left\| \overline{\vlambda}_{w} \right\|\right)^2  \left\| \left( \widehat{\mF}^{\prime }\widehat{\mF}\right) ^{-1} \right\| .
	\end{align*}
	First, note that
	\begin{align}
		\label{eq:lambda_barW}	
		\left\| \overline{\vlambda}_{w} \right\|^2= N^{-2} \sum_{i,j}^N w_i w_j \vlambda_{i}'\vlambda_j = \largeO_{P}(1)
	\end{align}
	since $\E \left[ N^{-2} \sum_{i,j}^N w_i w_j \vlambda_{i}'\vlambda_j \right] \leq  \left(N^{-2} \sum_{i,j}^N  \E \left[w_i^2 w_j^2\right]\right)^{1/2} \left(N^{-2} \sum_{i,j}^N \E \left[ \|\vlambda_{i} \|^2 \|\vlambda_{j} \|^2\right]\right)^{1/2} =\largeO (1)$ by boundedness of the fourth moments of all stochastic components involved.		
	Next, recall that $\widehat{\mF}=\mF\overline{\mC}+\overline{\mU}$. Thus,
	\begin{align*}
		\left\|\sum_{i=1}^{N}\widehat{\mF}'\mU_{i} \right\| &\leq \sqrt{NT} \left\| \overline{\mC} \right\| \left\| \sqrt{N/T} \mF' \overline{\mU} \right\| + T\left\| \sqrt{N/T}\overline{\mU} \right\|^2 \\
		&= \largeO_{P}\left(\sqrt{NT} \right) + \largeO_{P}\left( T \right),
	\end{align*}
	where the last line is a consequence of results \eqref{eq:UbarUbar} and \eqref{eq:UbarF}. Furthermore, since $\vepsi_i = \mU_i \mB^{-1}\left[ 1, \vzeros' \right]'$, we can proceed analogously for $\left\|\sum_{i=1}^{N}\widehat{\mF}'\vepsi_{i}w_{i} \right\| \leq \left\| \sum_{i=1}^{N}\widehat{\mF}'\mU_{i}w_{i} \right\|  \left\| \mB^{-1} \right\|$ and obtain the same orders in probability as in the last equation above. Additionally using result \eqref{eq:FhatFhat}, we arrive at
	\begin{align*}
		&\sqrt{\frac{1}{ 2TN(N-1)}}\sum_{i,j}^{N}\left(\vepsi_{i}-\overline{\mU}
		\left( \overline{\mC}^{-1}\right) ^{\prime }\vlambda_{i}\right)
		^{\prime }\mP_{\widehat{\mF}}\left(\vepsi_{j}-
		\overline{\mU}\left( \overline{\mC}^{-1}\right) ^{\prime }\vlambda_{j}\right) w_{i} w_{j} \\
		&=  \sqrt{\frac{1}{ 2TN(N-1)}}\left[ \largeO_{P}\left( \sqrt{NT}\right) +\largeO_{P}\left( T\right) \right]
		^{2}\largeO_{P}\left( T^{-1}\right) \\
		&=\largeO_{P}\left( T^{-1/2}\right) + \largeO_{P}\left( N^{-1/2}\right) + \largeO_{P}\left( N^{-1}\sqrt{T}\right),
	\end{align*}
	which concludes this proof.
\end{proof}

\bigskip

\begin{lemma}
	\label{lem:II2} Under Assumptions \ref{ass:errors}-\ref{ass:rank} and $\E[|w_{i}|^{8}]<M$:
	\begin{align*}
		&\sqrt{\frac{1}{2TN(N-1)}}\sum_{i=1}^{N}\left(\vepsi_{i}-\overline{\mU}%
		\left( \overline{\mC}^{-1}\right) \vlambda_{i}\right) ^{\prime }%
		\mP_{\widehat{\mF}}\left(\vepsi_{i}-\mU	\left( \overline{\mC}^{-1}\right) \vlambda_{i}\right) w
		_{i}^{2} \\
		&= \largeO_{P}(N^{-2}\sqrt{T}) + \largeO_{P}(N^{-3/2}) + \largeO_{P}(T^{-1/2})
	\end{align*}
\end{lemma}

\begin{proof}[Proof of Lemma \ref{lem:II2}]\ \\
	Given that $\| \mA_1 - \mA_2\|^2 \leq 3\| \mA_1 \|^2 + 3\| \mA_2\|^2$ for two arbitrary $m \times n$ matrices  $\mA_1$ and $\mA_2$, we have
	\begin{align}
		\label{eq:lem4decomp}
		&\sqrt{\frac{1}{2TN(N-1)}}\sum_{i=1}^{N}\left(\vepsi_{i}-\overline{\mU}%
		\left( \overline{\mC}^{-1}\right) \vlambda_{i}\right) ^{\prime }%
		\mP_{\widehat{\mF}}\left(\vepsi_{i}-\mU\left( \overline{\mC}^{-1}\right) \vlambda_{i}\right) w
		_{i}^{2} \notag \\
		&\leq  3\sqrt{\frac{T}{2N(N-1)}} \sum_{i=1}^{N}\left\| T^{-1} w_i \vepsi_{i}^{\prime} \widehat{\mF} \right\|^2 \left\| \left(T^{-1} \widehat{\mF}' \widehat{\mF} \right)^{-1} \right\| \notag \\
		&+ 3\sqrt{\frac{T}{2N(N-1)}}\sum_{i=1}^N \left\| T^{-1} w_i  \vlambda_{i}^{\prime }\left( \overline{\mC}^{-1}\right)^{\prime }\overline{\mU}' \widehat{\mF}\right\|^2  \left\| \left(T^{-1} \widehat{\mF}' \widehat{\mF} \right)^{-1} \right\|.
	\end{align}
	Here,
	\begin{align*}
		&\sqrt{\frac{T}{2N(N-1)}} \sum_{i=1}^{N}\left\| T^{-1} w_i \vepsi_{i}^{\prime} \widehat{\mF} \right\|^2 \\
		&= 3\sqrt{\frac{T}{2N(N-1)}} \sum_{i=1}^{N}\left\| T^{-1} w_i \vepsi_{i}^{\prime} \mF  \right\|^2 \left\| \overline{\mC} \right\|^2 + 3\sqrt{\frac{T}{2N(N-1)}} \sum_{i=1}^{N}\left\| T^{-1} w_i \vepsi_{i}^{\prime} \overline{\mU} \right\|^2.
	\end{align*}
	Concerning the first term on the right-hand side above, we can use Chebyshev's inequality together with
	\begin{align*}
		\E\left[ \sqrt{\frac{T}{2N(N-1)}} \sum_{i=1}^{N}T^{-2} w_i^2  \vepsi_{i}^{\prime} \mF \mF' \vepsi_{i} \right] &=  \sqrt{\frac{T}{2N(N-1)}} \sum_{i=1}^{N} \sum_{t,t'}^T T^{-2} \E \left[w_i^2  \right] \E\left[\E \left[\epsi_{i,t} \epsi_{i,t'} \vert \sigma_{i} \right]  \right] \E\left[ \vf_t' \vf_{t'} \right] \\
		&= \sqrt{\frac{1}{2TN(N-1)}} \sum_{i=1}^{N} \E \left[\sigma_{i}^2\right] \E \left[w_i^2  \right] \tr(\mSigma_{\mF}) \\
		&= \largeO(T^{-1/2})
	\end{align*}
	to arrive at $\sqrt{\frac{T}{2N(N-1)}} \sum_{i=1}^{N}\left\| T^{-1} w_i \vepsi_{i}^{\prime} \mF  \right\|^2 \left\| \overline{\mC} \right\|^2 = \largeO_{P}(T^{-1/2})$. Likewise, recalling that $\vu_{i,t} = \mB' \left[ \epsi_{i,t}, \ve_{i,t}' \right]'$, we have
	\begin{align*}
		&\E\left[ \sqrt{\frac{T}{2N(N-1)}} \sum_{i=1}^{N}\left\| T^{-1} w_i \vepsi_{i}^{\prime} \overline{\mU} \right\|^2 \right] \\
		&\leq \sqrt{\frac{T}{2N(N-1)}} \sum_{i,i',i''}^{N} \sum_{t,t'}^T (NT)^{-2}\E \left[w_i^2 \right] \E\left[\epsi_{i,t}\left(\epsi_{i',t}\epsi_{i'',t'} + \ve_{i',t}'\ve_{i'',t'} \right) \epsi_{i,t'}\right] \left\| \mB \right\| \\
		&= \sqrt{\frac{T}{2N(N-1)}} \left( \largeO(N^{-1}) + \largeO(T^{-1})  \right) \largeO(1) \\
		&= \largeO(N^{-2}\sqrt{T}) + \largeO(N^{-1}T^{-1/2}).
	\end{align*}
	This result is obtained by noting that
	\begin{align*}
		(NT)^{-2} \sum_{i,i',i''}^N \sum_{t,t'}^T \E \left[\epsi_{i,t} \epsi_{i',t} \epsi_{i'',t'} \epsi_{i,t'} \right] &= (NT)^{-2} \sum_{i=1}^N \sum_{t=1}^T \kappa_4\left[  \epsi_{i,t} \right]+(NT)^{-2}  \sum_{i=1}^N \sum_{t,t'}^T \E \left[\epsi_{i,t}^2\right] \E \left[ \epsi_{i,t'}^2 \right] \\
		&+ (NT)^{-2} \sum_{i,i'}^N \sum_{t=1}^T \E \left[\epsi_{i',t}^2\right] \E\left[ \epsi_{i,t}^2 \right] \\
		&=\largeO(N^{-1})+\largeO(T^{-1})
	\end{align*}
	and
	\begin{align*}
		(NT)^{-2} \sum_{i,i',i''}^N \sum_{t,t'}^T \E \left[ \epsi_{i,t} \ve_{i',t}'\ve_{i'',t'} \epsi_{i,t'} \right] &= (NT)^{-2} \sum_{i,i'}^N \sum_{t=1}^T \E \left[ \epsi_{i,t}^2 \right]   \E\left[\ve_{i',t}'\ve_{i',t} \right] \\
		&= \largeO(T^{-1}).
	\end{align*}
	Combining the results obtained up to this point and additionally using \eqref{eq:FhatFhat}, we obtain
	\begin{align}
		\label{eq:norm_epsFhat}
		\sqrt{\frac{T}{2N(N-1)}} \sum_{i=1}^{N}\left\| T^{-1}  w_i \vepsi_{i}^{\prime} \widehat{\mF} \right\|^2 \left\| \left(T^{-1} \widehat{\mF}' \widehat{\mF} \right)^{-1} \right\| = \largeO_{P}(N^{-2}\sqrt{T}) + \largeO_{P}(T^{-1/2}).
	\end{align}
	Next, consider
	\begin{align*}
		\sqrt{\frac{T}{2N(N-1)}}\sum_{i=1}^N \left\| T^{-1} w_i \vlambda_{i}^{\prime }\left( \overline{\mC}^{-1}\right)^{\prime }\overline{\mU} \widehat{\mF}\right\|^2 &= \sqrt{\frac{NT}{2(N-1)}} \left(N^{-1} \sum_{i=1}^N w_i^2 \vlambda_{i}'\vlambda_{i} \right) \left\| \overline{\mC}^{-1}\right\|^2 \left\| T^{-1}\overline{\mU} \widehat{\mF}\right\|^2
	\end{align*}	
	where $N^{-1} \sum_{i=1}^N w_i^2 \vlambda_{i}'\vlambda_{i} = \largeO_{P}(1)$ since $\E \left[ N^{-1} \sum_{i=1}^N w_i^2 \vlambda_{i}'\vlambda_{i} \right] \leq  \left(N^{-1} \sum_{i=1}^N  \E \left[w_i^4\right]\right)^{1/2} \left(N^{-1} \sum_{i=1}^N \E \left[ \|\vlambda_{i} \|^4 \right]\right)^{1/2}$. Additionally, we can write
	\begin{align*}
		\left\| T^{-1}\overline{\mU} \widehat{\mF}\right\|^2 &\leq   \left( \left\| T^{-1} \mF' \overline{\mU} \right\| + \left\| T^{-1}  \overline{\mU}' \overline{\mU} \right\| \right)^2\\
		&= \largeO_{P}(N^{-2}) +  \largeO_{P}(N^{-3/2}T^{-1/2}) +  \largeO_{P}((NT)^{-1})
	\end{align*}
	with the last step following from \eqref{eq:UbarUbar} and \eqref{eq:UbarF}. Consequently, it holds that
	\begin{align*}
		&\sqrt{\frac{T}{2N(N-1)}}\sum_{i=1}^N \left\| T^{-1} w_i \vlambda_{i}^{\prime }\left( \overline{\mC}^{-1}\right)^{\prime }\overline{\mU} \widehat{\mF}\right\|^2 \left\| \left(T^{-1} \widehat{\mF}' \widehat{\mF} \right)^{-1} \right\| \\
		&= \largeO_{P}(N^{-2}\sqrt{T}) +  \largeO_{P}(N^{-3/2}) +  \largeO_{P}(N^{-1}T^{-1/2}).
	\end{align*}
	Summarizing the results derived so far, we have
	\begin{align*}
		&\sqrt{\frac{1}{2TN(N-1)}}\sum_{i=1}^{N}\left(\vepsi_{i}-\overline{\mU}%
		\left( \overline{\mC}^{-1}\right) \vlambda_{i}\right)^{\prime }%
		\mP_{\widehat{\mF}}\left(\vepsi_{i}-\mU\left( \overline{\mC}^{-1}\right) \vlambda_{i}\right) w
		_{i}^{2} \\
		&= \largeO_{P}(N^{-2}\sqrt{T})+ \largeO_{P}(N^{-3/2}) + \largeO_{P}(T^{-1/2})
	\end{align*}
	which concludes the proof.
\end{proof}

\newpage

\subsection{Additive model. The effect of estimating fixed effects and error variances}
\label{sec:var_estim_2WFE}
\paragraph{Estimating fixed effects}
We start at the decomposition of $CD_{\sigma}$ in Eq. \eqref{eq:CD_FEdecomp_2WFE} and focus on its second term 
\begin{align*}
	{CD}_{(\widehat{\mu} - \mu)} &= Tk_{N,T} \left(\sum_{i=1}^N \frac{\overline{\epsi}_{i}-\overline{\epsi}}{\sigma_i}\right)^2 - Tk_{N,T}\sum_{i=1}^N \left(\frac{\overline{\epsi}_{i}-\overline{\epsi}}{\sigma_i}\right)^2
\end{align*}
Consider now the first term on the right-hand side above.
\begin{align*}
	Tk_{N,T} \left(\sum_{i=1}^N \frac{\overline{\epsi}_{i}-\overline{\epsi}}{\sigma_i}\right)^2 &= Nk_{N,T}\left[ \left(NT\right)^{-1/2} \sum_{i=1}^N\sum_{t=1}^T \frac{\epsi_{i,t}}{\sigma_i} - \left(N^{-1} \sum_{i=1}^N \sigma_{i}^{-1} \right)\left(NT\right)^{-1/2} \sum_{i=1}^N\sum_{t=1}^T\epsi_{i,t} \right]^{2} \\
	&= \sqrt{\frac{N}{2T\left(N-1\right)}} \largeO_{P}\left(1\right) \\
	&= \largeO_{P}\left(T^{-1/2}\right). 
\end{align*}
Here, the second line line follows from Theorem 3 in \citet{phillipsMoon1999} applied to sequences $\epsi_{i,t}$ and $\epsi_{i,t}/\sigma_i$ as well as the lower bound of $\sigma_i$  at a value above zero. 
As for the second component above, note that
\begin{align*}
	&Tk_{N,T} \sum_{i=1}^N \left(\frac{\overline{\epsi}_{i}-\overline{\epsi}}{\sigma_i}\right)^2 \\
	&= k_{N,T} \sum_{i=1}^N \left( T^{-1/2} \sum_{t=1}^T \frac{\epsi_{i,t}}{\sigma_i} \right)^2 - 2k_{N,T} \left( (NT)^{-1/2} \sum_{i=1}^N  \sum_{t=1}^T \frac{\epsi_{i,t}}{\sigma_i} \right)\left( (NT)^{-1/2} \sum_{i=1}^N  \sum_{t=1}^T \epsi_{i,t} \right) \\
	&+k_{N,t}\left(N^{-1} \sum_{i=1}^N \sigma_{i}^{-1} \right) \left(  \sqrt{NT} \overline{\epsi}  \right) ^{2}\\
	&=  k_{N,T} \sum_{i=1}^N \left( T^{-1/2} \sum_{t=1}^T \frac{\epsi_{i,t}}{\sigma_i} \right)^2  + \largeO_{P} \left( N^{-1} T^{-1/2} \right)
\end{align*}
Again, we use Theorem 3 in \citet{phillipsMoon1999} to arrive at the last line. In order to derive a result for the remaining explicit term, we note that 
\begin{align*}
	k_{N,T} \sum_{i=1}^N \left( T^{-1/2} \sum_{t=1}^T \frac{\epsi_{i,t}}{\sigma_i} \right)^{2} &= Nk_{N,T} + \sqrt{N}k_{N,T} N^{-1/2}\sum_{i=1}^N \left[ \left( T^{-1/2} \sum_{t=1}^T \frac{\epsi_{i,t}}{\sigma_i}\right)^{2}-1\right] \\
	&= \largeO \left(T^{-1/2}\right) + \largeO_{P} \left[ \left(NT\right)^{-1/2}\right].
\end{align*}
From this result and the previous one it follows that 
\begin{align*}
	{CD}_{(\widehat{\mu} - \mu)} = \largeO_{P} \left(T^{-1/2} \right).
\end{align*}

\bigskip

\paragraph{Estimating error variances}
We proceed from equation \eqref{eq:CD_Taylordecomp_2WFE} in the proof of Theorem \ref{theorem::additiveHetero}, and show below that 
\begin{align*}
{CD}_{(\widehat{\sigma}-\sigma)} &=k_{N,T}\sum_{i=1}^{N}\sum_{j\neq i}^{N}\frac{\widehat{\vepsi}_{i}'\widehat{\vepsi}_{j}}{\left(\varsigma_{i}^{2}\right)^{3/2}\left(\sigma_{j}^{2}\right)^{1/2}}\left(\widehat{\sigma}_{i}^{2}-\sigma_{i}^{2}\right) \\
&= \largeO_{P}\left(N^{-1}\sqrt{T}\right)+\largeO_{P}\left(N^{-1/2}\right)+\largeO_{P}\left(T^{-1/2}\right)+\largeO_{P}\left(T^{-1}\sqrt{N}\right).
\end{align*}
 The expression in the first line above depends on an unknown
component $\varsigma_{i}^{2}$ which is bounded by $\sigma_{i}^{2}$ and $\widehat{\sigma}_{i}^{2}$. A simple expansion allows us to write
\begin{align}
	k_{N,T}\sum_{i=1}^{N}\sum_{j\neq i}^{N}\frac{\widehat{\vepsi}_{i}'\widehat{\vepsi}_{j}}{\left(\varsigma_{i}^{2}\right)^{3/2}\left(\sigma_{j}^{2}\right)^{1/2}}\left(\widehat{\sigma}_{i}^{2}-\sigma_{i}^{2}\right) & = k_{N,T}\sum_{i=1}^{N}\sum_{j\neq i}^{N}\frac{\widehat{\vepsi}_{i}'\widehat{\vepsi}_{j}}{\left(\sigma_{i}^{2}\right)^{3/2}\left(\sigma_{j}^{2}\right)^{1/2}}\left(\widehat{\sigma}_{i}^{2}-\sigma_{i}^{2}\right)\nonumber \\
	& +k_{N,T}\sum_{i=1}^{N}\sum_{j\neq i}^{N}\frac{\widehat{\vepsi}_{i}'\widehat{\vepsi}_{j}}{\left(\sigma_{j}^{2}\right)^{1/2}}\left(\widehat{\sigma}_{i}^{2}-\sigma_{i}^{2}\right)\left(\frac{1}{\left(\varsigma_{i}^{2}\right)^{3/2}}-\frac{1}{\left(\sigma_{i}^{2}\right)^{3/2}}\right)\nonumber \\
	& =a_{1}+a_{2}.\label{eq:Taylorterm}
\end{align}
The two expressions above are: 1.) A first-order expansion\emph{ at } the true error variances and 2.) the approximation error of this first-order approximation.

Consider now the first term on the right-hand side above, which is the first-order expansion \emph{at }the true error variances. It can be written 
\begin{align}
	a_{1} & =k_{N,T}\sum_{i,j}^{N}\frac{\widehat{\vepsi}_{i}'\widehat{\vepsi}_{j}}{\left(\sigma_{i}^{2}\right)^{3/2}\left(\sigma_{j}^{2}\right)^{1/2}}\left(\widehat{\sigma}_{i}^{2}-\sigma_{i}^{2}\right)-k_{N,T}\sum_{i=1}^{N}\frac{\widehat{\vepsi}_{i}'\widehat{\vepsi}_{i}}{\left(\sigma_{i}^{2}\right)^{2}}\left(\widehat{\sigma}_{i}^{2}-\sigma_{i}^{2}\right)\label{eq:Taylor_a1}\\
	& =a_{11}-a_{12}\nonumber 
\end{align}
The order of the two expressions above will be derived by expanding 
\begin{align}
	\widehat{\sigma}_{i}^{2} &  =\frac{\vepsi_{i}'\vepsi_{i}}{T}+\frac{\overline{\vepsi}'\overline{\vepsi}}{T}+\overline{\epsi}_{i}^{2}+\overline{\epsi}^{2}-2\frac{\overline{\vepsi}'\vepsi_{i}}{T}-2\overline{\epsi}_{i}\overline{\epsi} \label{eq:sigi2hat_2WFE}
\end{align}
and
\begin{equation}
	\widehat{\vepsi}_{i}'\widehat{\vepsi}_{j} = \vepsi_{i}'\vepsi_{j} + \overline{\vepsi}'\overline{\vepsi} + T\overline{\epsi}_{i}\overline{\epsi}_{j} + T\overline{\epsi}^{2} - \overline{\vepsi}'\vepsi_{i} - \overline{\vepsi}'\vepsi_{j} - T\overline{\epsi}_{i}\overline{\epsi} - T\overline{\epsi}_{j}\overline{\epsi}. \label{eq:epsi_epsj_2WFE}
\end{equation}
and by looking at the resulting terms. 

Consider the first expression, $a_{11}$, which we write
\begin{align}
	a_{11} & =k_{N,T}\sum_{i,j}^{N}\frac{\widehat{\vepsi}_{i}'\widehat{\vepsi}_{j}}{\left(\sigma_{i}^{2}\right)^{3/2}\left(\sigma_{j}^{2}\right)^{1/2}}\left(\widehat{\sigma}_{i}^{2}-\sigma_{i}^{2}\right)\nonumber \\
	& =k_{N,T}\sum_{i,j}^{N}\frac{\vepsi_{i}'\vepsi_{j}}{\left(\sigma_{i}^{2}\right)^{3/2}\left(\sigma_{j}^{2}\right)^{1/2}}\left(\widehat{\sigma}_{i}^{2}-\sigma_{i}^{2}\right)+k_{N,T}\sum_{i,j}^{N}\vepsi_{i}'\vepsi_{j}\left[N^{-1}\sum_{\ell\text{=1}}^{N}\frac{\left(\widehat{\sigma}_{\ell}^{2}-\sigma_{\ell}^{2}\right)}{\left(\sigma_{\ell}^{2}\right)^{3/2}}\right]\left[N^{-1}\sum_{\ell=1}^{N}\sigma_{\ell}^{-1}\right]\nonumber \\
	& -k_{N,T}\sum_{i,j}^{N}\frac{\vepsi_{i}'\vepsi_{j}}{\left(\sigma_{j}^{2}\right)^{1/2}}\left[N^{-1}\sum_{\ell\text{=1}}^{N}\frac{\left(\widehat{\sigma}_{\ell}^{2}-\sigma_{\ell}^{2}\right)}{\left(\sigma_{\ell}^{2}\right)^{3/2}}\right]-k_{N,T}\sum_{i,j}^{N}\frac{\vepsi_{i}'\vepsi_{j}\left(\widehat{\sigma}_{i}^{2}-\sigma_{i}^{2}\right)}{\left(\sigma_{i}^{2}\right)^{3/2}}\left[N^{-1}\sum_{\ell=1}^{N}\sigma_{\ell}^{-1}\right]\nonumber \\
	& +Tk_{N,T}\sum_{i,j}^{N}\frac{\overline{\epsi}_{i}\overline{\epsi}_{j}}{\left(\sigma_{i}^{2}\right)^{3/2}\left(\sigma_{j}^{2}\right)^{1/2}}\left(\widehat{\sigma}_{i}^{2}-\sigma_{i}^{2}\right)+Tk_{N,T}\sum_{i,j}^{N}\overline{\epsi}_{i}\overline{\epsi}_{j}\left[N^{-1}\sum_{\ell\text{=1}}^{N}\frac{\left(\widehat{\sigma}_{\ell}^{2}-\sigma_{\ell}^{2}\right)}{\left(\sigma_{\ell}^{2}\right)^{3/2}}\right]\left[N^{-1}\sum_{\ell=1}^{N}\sigma_{\ell}^{-1}\right]\nonumber \\
	& -Tk_{N,T}\sum_{i,j}^{N}\frac{\overline{\epsi}_{i}\overline{\epsi}_{j}}{\left(\sigma_{j}^{2}\right)^{1/2}}\left[N^{-1}\sum_{\ell\text{=1}}^{N}\frac{\left(\widehat{\sigma}_{\ell}^{2}-\sigma_{\ell}^{2}\right)}{\left(\sigma_{\ell}^{2}\right)^{3/2}}\right]-Tk_{N,T}\sum_{i,j}^{N}\frac{\overline{\epsi}_{i}\overline{\epsi}_{j}\left(\widehat{\sigma}_{i}^{2}-\sigma_{i}^{2}\right)}{\left(\sigma_{i}^{2}\right)^{3/2}}\left[N^{-1}\sum_{\ell=1}^{N}\sigma_{\ell}^{-1}\right]\nonumber \\
	& =a_{111}+a_{112}-a_{113}-a_{114}+a_{115}+a_{116}-a_{117}-a_{118}.\label{eq:Taylor_a11}
\end{align}
The order in probability of all terms above will now be established
by implicit application of either Markov's or Chebyshev's inequality.
Handling $\sigma_{\ell}^{k}$ for $k<0$ is straightforward
since $\sigma_{\ell}$ is bounded from below at a value above zero by
Assumption \ref{ass:errorsAddHetero}in the main paper. For this reason, and in order to
limit the notational burden, we disregard from powers of $\sigma_{j}^{2}$
that appear in the denominator of fractions involving $\epsi_{i,t}$.
For the same reason, we disregard from the average $N^{-1}\sum_{\ell=1}^{N}\sigma_{\ell}^{-1}$.
The second and third terms in \eqref{eq:Taylor_a11}, $a_{112}$ and
$a_{113}$, are differently weighted versions of 
\begin{align*}
	& 2k_{N,T}\sum_{i=2}^{N}\sum_{j=1}^{i-1}\vepsi_{i}'\vepsi_{j}\left[N^{-1}\sum_{\ell\text{=1}}^{N}\left(\widehat{\sigma}_{\ell}^{2}-\sigma_{\ell}^{2}\right)\right]+k_{N,T}\sum_{i=1}^{N}\vepsi_{i}'\vepsi_{i}\left[N^{-1}\sum_{\ell\text{=1}}^{N}\left(\widehat{\sigma}_{\ell}^{2}-\sigma_{\ell}^{2}\right)\right]\\
	= & a_{1121}+a_{1122}.
\end{align*}
The order of the middle part appearing in both terms is $\largeO_{P}\left[\left(NT\right)^{-1/2}\right]+\largeO_{P}\left(N^{-1}\right)$,
as given by result \eqref{eq:sumi_sig2hat_min_sig2}. Accordingly,
we get 
\begin{align*}
	a_{1121} & =\largeO_{P}\left[\left(NT\right)^{-1/2}\right]+\largeO_{P}\left(N^{-1}\right)\\
	a_{1122} & =\largeO_{P}\left(N^{-1/2}\right)+\largeO_{P}\left(N^{-1}\sqrt{T}\right),
\end{align*}
since the first term in $a_{1121}$ is $\largeO_{P}\left(1\right)$, whereas
the corresponding term in $a_{1122}$ is $\largeO_{P}\left(\sqrt{T}\right)$.
It follows that $a_{112}=\largeO_{P}\left(N^{-1/2}\right)+\largeO_{P}\left(N^{-1}\sqrt{T}\right)$
and $a_{113}=\largeO_{P}\left(N^{-1/2}\right)+\largeO_{P}\left(N^{-1}\sqrt{T}\right)$.
We continue with $a_{111}$ and $a_{114}$ which both are differently
weighted versions of $k_{N,T}\sum_{i,j}^{N}\vepsi_{i}'\vepsi_{j}\left(\widehat{\sigma}_{i}^{2}-\sigma_{i}^{2}\right)$.
Hence, Lemma \ref{lem:sumij_epsi_epsj_sigi2_diff} implies 
\begin{align*}
	a_{111} & =\largeO_{P}\left(T^{-1}\sqrt{N}\right)+\largeO_{P}\left(T^{-1/2}\right)+\largeO_{P}\left(N^{-1/2}\right)+\largeO_{P}\left(N^{-1}\sqrt{T}\right),
\end{align*}
with an identical result for $a_{114}.$ Concerning $a_{111}$, we
additionally decompose

\begin{align}
	a_{111} & =k_{N,T}\sum_{i,j}^{N}\vepsi_{i}'\vepsi_{j}\left(\widehat{\sigma}_{i}^{2}-\sigma_{i}^{2}\right)\nonumber \\
	& =k_{N,T} \sum_{i=1}^N \sum_{j\neq i}^{N}\vepsi_{i}'\vepsi_{j}\left(\widehat{\sigma}_{i}^{2}-\sigma_{i}^{2}\right)+k_{N,T}\sum_{i=1}^{N}\vepsi_{i}'\vepsi_{i}\left(\widehat{\sigma}_{i}^{2}-\sigma_{i}^{2}\right).\label{eq:Taylor_a111}
\end{align}
where the bound that we stated above holds for both expressions on
the right-hand side above. We emphasize this point because the second
term in the decomposition of $a_{111}$ cancels out with an identical
term in the decomposition of $a_{12}$ in \eqref{eq:Taylor_a1}. To
appreciate this point, see the discussion of expression \eqref{eq:a12}
below.

Terms $a_{115}$ and $a_{118}$ are almost identical to $a_{111}$
and $a_{114}$, being weighted versions of are functions of $\left(N\overline{\epsi}\right)Tk_{N,T}\sum_{i=1}^{N}\overline{\epsi}_{i}\left(\widehat{\sigma}_{i}^{2}-\sigma_{i}^{2}\right)$.
Noting that $\overline{\epsi}=\largeO_{P}\left[\left(NT\right)^{-1/2}\right]$,
we can use Lemma \ref{lem:sumij_epsbari_epsbarj_sig2_diff} to establish
\begin{align*}
	a_{115} & =\largeO_{P}\left(\sqrt{N}T^{-1/2}\right)\left[\largeO_{P}\left(T^{-1/2}\right)+\largeO_{P}\left(N^{-1/2}T^{-1}\right)+\largeO_{P}\left(N^{-3/2}\right)\right]\\
	& =\largeO_{P}\left(T^{-1}\sqrt{N}\right)+\largeO_{P}\left(T^{-3/2}\right)+\largeO_{P}\left(N^{-1}T^{-1/2}\right),
\end{align*}
with $a_{118}$ being of the same order. We emphasize the link between
$a_{115}$ and $a_{118}$ as well as $a_{111}$ and $a_{114}$ since
one part of $a_{115}$ cancels out with an identical expression in
$a_{12}$ below, precisely as a component of $a_{111}$ cancels out.\footnote{This point continues to hold even if we explicitly acknowledge the
	scaling with functions of $\sigma_{i}$ and $\sigma_{j}$ which we
	largely disregard from in this proof.} To be more precise, the second term in 
\[
a_{115}=\frac{Tk_{N,T}}{2} \sum_{i=1}^N \sum_{j\neq i}^{N}\frac{\overline{\epsi}_{i}\overline{\epsi}_{j}}{\left(\sigma_{i}^{2}\right)^{3/2}\left(\sigma_{j}^{2}\right)^{1/2}}\left(\widehat{\sigma}_{i}^{2}-\sigma_{i}^{2}\right)+\frac{Tk_{N,T}}{2}\sum_{i=1}^{N}\frac{\overline{\epsi}_{i}^{2}}{\left(\sigma_{i}^{2}\right)^{2}}\left(\widehat{\sigma}_{i}^{2}-\sigma_{i}^{2}\right)
\]
disappears if we combine all terms in $a_{11}$ and $a_{12}.$ The
last two terms in our decomposition of $a_{11},$ $a_{116}$ and $a_{117}$,
both depend on the term $Tk_{N,T} \left(N\overline{\epsi}\right)^2\left[N^{-1}\sum_{\ell\text{=1}}^{N}\left(\widehat{\sigma}_{\ell}^{2}-\sigma_{\ell}^{2}\right)\right]$.
Using result \eqref{eq:sumi_sig2hat_min_sig2} here, we can conclude
that 
\begin{align*}
	a_{116} & =k_{N,T}N^{2}T\largeO_{P}\left[\left(NT\right)^{-1}\right]\left\{ \largeO_{P}\left[\left(NT\right)^{-1/2}\right]+\largeO_{P}\left(N^{-1}\right)\right\} \\
	& =N\sqrt{T}\left\{ \largeO_{P}\left[\left(NT\right)^{-3/2}\right]+\largeO_{P}\left(N^{-2}T^{-1}\right)\right\} \\
	& =\largeO_{P}\left(N^{-1/2}T^{-1}\right)+\largeO_{P}\left(N^{-1}T^{-1/2}\right),
\end{align*}
and that the same order holds for $a_{117}$. 

Taking together all eight terms of $a_{11}$, as defined in expression
\eqref{eq:Taylor_a11}, we can conclude that 
\begin{equation}
	a_{11}=\largeO_{P}\left(N^{-1}\sqrt{T}\right)+\largeO_{P}\left(N^{-1/2}\right)+\largeO_{P}\left(T^{-1/2}\right)+\largeO_{P}\left(T^{-1}\sqrt{N}\right).\label{eq:Taylor_a11order}
\end{equation}

Next, we write out the second term in our expansion \eqref{eq:Taylor_a1},
which is given by
\begin{align}
	a_{12} & =k_{N,T}\sum_{t=1}^{T}\sum_{i=1}^{N}\widehat{\epsi}_{i,t}^{2}\left(\widehat{\sigma}_{i}^{2}-\sigma_{i}^{2}\right)\nonumber \\
	& k_{N,T}\sum_{t=1}^{T}\sum_{i=1}^{N}\epsi_{i,t}^{2}\left(\widehat{\sigma}_{i}^{2}-\sigma_{i}^{2}\right) 
	+ Tk_{N,T}\sum_{i=1}^{N}\left(\widehat{\sigma}_{i}^{2} 
	-\sigma_{i}^{2}\right)\left(T^{-1}\overline{\vepsi}'\overline{\vepsi}\right) 
	+ Tk_{N,T}\sum_{i=1}^{N}\overline{\epsi}_{i}^{2} \left(\widehat{\sigma}_{i}^{2}-\sigma_{i}^{2}\right)\nonumber \\
	& + \overline{\epsi}^{2}Tk_{N,T}\sum_{i=1}^{N} \left(\widehat{\sigma}_{i}^{2}-\sigma_{i}^{2}\right) 
	- 2k_{N,T}\sum_{t=1}^{T}\sum_{i=1}^{N} \epsi_{i,t}\overline{\epsi}_{t}\left(\widehat{\sigma}_{i}^{2}-\sigma_{i}^{2}\right) 
	- \overline{\epsi}\: 2Tk_{N,T}\sum_{i=1}^{N}\overline{\epsi}_{i} \left(\widehat{\sigma}_{i}^{2}-\sigma_{i}^{2}\right)\nonumber \\
	& =a_{121}+a_{122}+a_{123}+a_{124}-a_{125}-a_{126}.\label{eq:a12}
\end{align}

A term identical to $a_{121}$ is given by the second term of expression
\eqref{eq:Taylor_a111}, a component of the term $a_{11}$. Since
$a_{1}=a_{11}-a_{21}$, term $a_{121}$ and its equivalent in expression
\eqref{eq:Taylor_a111} cancel out. The same holds for $a_{123}$
where an identical term is contained in expression $a_{115}$. Among
the four remaining terms, 
\begin{align*}
	a_{122} & =\largeO_{P}\left(N^{-3/2}\right)+\largeO_{P}\left(N^{-2}\sqrt{T}\right).
\end{align*}
and analogously
\[
a_{124}=\largeO_{P}\left(N^{-3/2}T^{-1}\right)+\largeO_{P}\left(N^{-2}T^{-1/2}\right)
\]
are a consequence of result \eqref{eq:sumi_sig2hat_min_sig2} together
with $T^{-1}\overline{\vepsi}'\overline{\vepsi}=\largeO_{P}\left(N^{-1}\right)$.
Next, we consider $a_{125}$ which, up to a different weighting, corresponds
to either of expressions $a_{111}$ or $a_{114}$ times an additional
factor of $N^{-1}$. Hence, by results \eqref{eq:sumi_epsvar_epsvardev}
and \eqref{eq:sumij_epsvar_epsvardev}, we arrive at 
\[
a_{125}=\largeO_{P}\left(N^{-2}\sqrt{T}\right)+\largeO_{P}\left(N^{-3/2}\right)+\largeO_{P}\left(N^{-1}T^{-1/2}\right)+\largeO_{P}\left(N^{-1/2}T^{-1}\right).
\]
The same observation as for term $a_{125}$ can be made for $a_{126}$which
is equivalent to $a_{115}$ or $a_{118}$ times an additional factor
of $N^{-1}$. It follows that 
\[
a_{126}=\largeO_{P}\left(N^{-1/2}T^{-3}\right)+\largeO_{P}\left(N^{-1}T^{-5/2}\right)+\largeO_{P}\left(N^{-3/2}T^{-2}\right)+\largeO_{P}\left(N^{-2}T^{-3/2}\right)+\largeO_{P}\left(N^{-3}\sqrt{T}\right).
\]
This specifies the properties of all six terms in the decomposition
of $a_{12}.$ Combining the resulting orders, we arrive at 
\[
a_{12}=\largeO_{P}\left(N^{-2}\sqrt{T}\right)+\largeO_{P}\left(N^{-3/2}\right)+\largeO_{P}\left[\left(NT\right)^{-1/2}\right]
\]
where we disregard from terms $a_{121}$ and $a_{123}$ which are
eliminated completely in the difference $a_{11}-a_{12}$. Using the
results obtained for both components of this difference, we conclude
that 
\begin{align}
	a_{1} & =2k_{N,T} \sum_{i=2}^{N}\sum_{j=1}^{i-1}\frac{\widehat{\vepsi}_{i}'\widehat{\vepsi}_{j}}{\left(\sigma_{i}^{2}\right)^{3/2}\left(\sigma_{j}^{2}\right)^{1/2}}\left(\widehat{\sigma}_{i}^{2}-\sigma_{i}^{2}\right)\nonumber \\
	& =\largeO_{P}\left(N^{-1}\sqrt{T}\right)+\largeO_{P}\left(N^{-1/2}\right)+\largeO_{P}\left(T^{-1/2}\right)+\largeO_{P}\left(T^{-1}\sqrt{N}\right),\label{eq:Taylor_a1order}
\end{align}
where we include the bounded random variables $\left(\sigma_{i}^{2}\right)^{-3/2}$
and $\left(\sigma_{j}^{2}\right)^{-1/2}$ for the sake of completeness.\medskip{}

Up to this point, we have investigated the first-order approximation
in equation \eqref{eq:Taylorterm} \emph{at }the true error variances.
We proceed with results on term $a_{2}$, the error of our linear
approximation. We take absolute values and use the Cauchy-Schwarz inequality to arrive
at
\begin{align}
	\left|a_{2}\right|= & \left| k_{N,T}\sum_{i=1}^{N}\sum_{j\neq i}^{N}\frac{\widehat{\vepsi}_{i}'\widehat{\vepsi}_{j}}{\left(\sigma_{j}^{2}\right)^{1/2}}\left(\widehat{\sigma}_{i}^{2}-\sigma_{i}^{2}\right)\left[\frac{1}{\left(\varsigma_{i}^{2}\right)^{3/2}}-\frac{1}{\left(\sigma_{i}^{2}\right)^{3/2}}\right]\right|\nonumber \\
	& \leq k_{N,T}\left( \sum_{i=1}^{N}\left[\sum_{j\neq i}^{N}\frac{\widehat{\vepsi}_{i}'\widehat{\vepsi}_{j}}{\left(\sigma_{j}^{2}\right)^{1/2}}\left(\widehat{\sigma}_{i}^{2}-\sigma_{i}^{2}\right)\right]^{2}\right) ^{1/2}\left( \sum_{i=1}^{N}\left[\frac{1}{\left(\varsigma_{i}^{2}\right)^{3/2}}-\frac{1}{\left(\sigma_{i}^{2}\right)^{3/2}}\right]^{2}\right) ^{1/2}\nonumber \\
	& \leq k_{N,T}\left( \sum_{i=1}^{N}\left[\sum_{j\neq i}^{N}\frac{\widehat{\vepsi}_{i}'\widehat{\vepsi}_{j}}{\left(\sigma_{j}^{2}\right)^{1/2}}\left(\widehat{\sigma}_{i}^{2}-\sigma_{i}^{2}\right)\right]^{2}\right) ^{1/2}\left( \sum_{i=1}^{N}\left[\frac{1}{\left(\widehat{\sigma}_{i}^{2}\right)^{3/2}}-\frac{1}{\left(\sigma_{i}^{2}\right)^{3/2}}\right]^{2}\right) ^{1/2}\nonumber \\
	& = k_{N,T}\sqrt{a_{21}a_{22}}.\label{eq:Taylor_a2}
\end{align}
A further first-order Taylor approximation of $\left(\widehat{\sigma}_{i}^{2}\right)^{-3/2}$
around $\left(\sigma_{i}^{2}\right)^{-3/2}$ results in
\[
\frac{1}{\left(\widehat{\sigma}_{i}^{2}\right)^{3/2}}-\frac{1}{\left(\sigma_{i}^{2}\right)^{3/2}}=-3/2\frac{1}{\left(\varsigma_{i}^{2}\right)^{5/2}}\left(\widehat{\sigma}_{i}^{2}-\sigma_{i}^{2}\right)
\]
where we take into account a slight misuse of notation by letting
$\varsigma_{i}^{2}$ be a point in between $\widehat{\sigma}_{i}^{2}$
and $\sigma_{i}^{2}$. Substituting this expression into $a_{22}$
leads to
\begin{align}
	a_{22} & =\sum_{i=1}^{N}\left[-3/2\frac{1}{\left(\varsigma_{i}^{2}\right)^{5/2}}\left(\widehat{\sigma}_{i}^{2}-\sigma_{i}^{2}\right)\right]^{2}\nonumber \\
	& \leq9/4\left(\sup_{i}\frac{1}{\left(\varsigma_{i}^{2}\right)^{5/2}}\right)\sum_{i=1}^{N}\left(\widehat{\sigma}_{i}^{2}-\sigma_{i}^{2}\right)^{2}\nonumber \\
	& =9/4\left(\frac{1}{\left(\inf_{i}\varsigma_{i}^{2}\right)^{5/2}}\right)\left[\largeO_{P}\left(N/T\right)+\largeO_{P}\left(N^{-1}\right)+\largeO_{P}\left(T^{-1}\right)\right]\nonumber \\
	& =\largeO_{P}\left(N/T\right)+\largeO_{P}\left(N^{-1}\right)+\largeO_{P}\left(T^{-1}\right).\label{eq:Taylor_a22order}
\end{align}
where Lemma \ref{lem:inf_sighat2_2WFE} is used to ensure that the
last line holds as long as $T^{-1}\sqrt{N}\to0$ is satisfied. To arrive
this result, we use result \eqref{eq:sumi_sighat2_min_sig2_squared},
as well as Lemma \ref{lem:inf_sighat2_2WFE} to establish conditions
on the relative expansion rate of $N$ and $T$ under which $\inf_{i}\varsigma_{i}^{2}$
is stochastically bounded. The lemma is merely a generalization of
Lemma 5 in \citet{Moon2007416} and the resulting assumptions on the
divergence rates of $N$ and $T$ can be relaxed even further if one
is willing to assume the existence of higher-order moments of $\epsi_{i,t}$
beyond $\E\left[\epsi_{i,t}^{8}\right].$ In addition
to result \eqref{eq:Taylor_a22order}, which bounds $a_{22}$, the
order in probability of $a_{2}$ is determined by that of $a_{21}.$
Lemma \ref{lem:Taylor_a21order} establishes the corresponding result
on $a_{21}$ so that the order in probability of \eqref{eq:Taylor_a2}
is given by
\begin{align}
	a_{2}= & \left\{ \frac{2}{TN\left(N-1\right)}\left[\largeO_{P}\left(N^{2}\right)+\largeO_{P}\left(NT\right)+\largeO_{P}\left(N^{-1}T^{2}\right)\right]\left[\largeO_{P}\left(N/T\right)+\largeO_{P}\left(N^{-1}\right)\right]\right\} ^{1/2}\nonumber \\
	= & \largeO_{P}\left(T^{-1}\sqrt{N}\right)+\largeO_{P}\left(T^{-1/2}\right)+\largeO_{P}\left(N^{-3/2}\sqrt{T}\right).\label{eq:Taylor_a2order}
\end{align}

This result marks the last one needed to characterize the impact of
estimating error variances rather than knowing their true value. We
return to expression \eqref{eq:Taylorterm} and use our results \eqref{eq:Taylor_a1order}
and \eqref{eq:Taylor_a2order} on its two components $a_{1}$ and
$a_{2}$ to arrive at 
\begin{equation}
{CD}_{(\widehat{\sigma}-\sigma)} = \largeO_{P}\left(N^{-1}\sqrt{T}\right)+\largeO_{P}\left(N^{-1/2}\right)+\largeO_{P}\left(T^{-1/2}\right)+\largeO_{P}\left(T^{-1}\sqrt{N}\right).
\end{equation}

\subsection{Multifactor model. The effect of error variance estimation}
\label{sec:var_estim_CCE}
Similarly to Section \ref{sec:var_estim_2WFE}, we proceed from equation \eqref{eq:CD_Taylordecomp_CCE} in the proof of Theorem \ref{theorem::CCE}, and show below that 
\begin{align*}
	{CD}_{(\widehat{\sigma}-\sigma)} &=k_{N,T}\sum_{i=1}^{N}\sum_{j\neq i}^{N}\frac{\widehat{\vepsi}_{i}'\widehat{\vepsi}_{j}}{\left(\varsigma_{i}^{2}\right)^{3/2}\left(\sigma_{j}^{2}\right)^{1/2}}\left(\widehat{\sigma}_{i}^{2}-\sigma_{i}^{2}\right) \\
	&= \largeO_{P}\left(N^{-1}\sqrt{T}\right)+\largeO_{P}\left(N^{-1/2}\right)+\largeO_{P}\left(T^{-1/2}\right)+\largeO_{P}\left(T^{-1}\sqrt{N}\right).
\end{align*}

The expression depends on an unknown component $\varsigma_{i}^{2}$ which is bounded by $\sigma_{i}^{2}$
and $\widehat{\sigma}_{i}^{2}$. A simple expansion allows us to write
\begin{align}
	k_{N,T}\sum_{i=1}^{N}\sum_{j\neq i}^{N}\frac{\widehat{\vepsi}_{i}'\widehat{\vepsi}_{j}}{\left(\varsigma_{i}^{2}\right)^{3/2}\left(\sigma_{j}^{2}\right)^{1/2}}\left(\widehat{\sigma}_{i}^{2}-\sigma_{i}^{2}\right) & =k_{N,T}\sum_{i=1}^{N}\sum_{j\neq i}^{N}\frac{\widehat{\vepsi}_{i}'\widehat{\vepsi}_{j}}{\left(\sigma_{i}^{2}\right)^{3/2}\left(\sigma_{j}^{2}\right)^{1/2}}\left(\widehat{\sigma}_{i}^{2}-\sigma_{i}^{2}\right)\nonumber \\
	& +k_{N,T}\sum_{i=1}^{N}\sum_{j\neq i}^{N}\frac{\widehat{\vepsi}_{i}'\widehat{\vepsi}_{j}}{\left(\sigma_{j}^{2}\right)^{1/2}}\left(\widehat{\sigma}_{i}^{2}-\sigma_{i}^{2}\right)\left(\frac{1}{\left(\varsigma_{i}^{2}\right)^{3/2}}-\frac{1}{\left(\sigma_{i}^{2}\right)^{3/2}}\right)\nonumber \\
	& =a_{1}^{*}+a_{2}^{*}.\label{eq:Taylorterm_CCE}
\end{align}
The two expressions above are: 1.) A first-order expansion\emph{ at
}the true error variances and 2.) the approximation error of this
first-order approximation. 

Consider now the first term on the right-hand side above, which is
the first-order expansion \emph{at }the true error variances. It can
be written 
\begin{align}
	a_{1}^{*} & =k_{N,T}\sum_{i,j}^{N}\frac{\widehat{\vepsi}_{i}'\widehat{\vepsi}_{j}}{\left(\sigma_{i}^{2}\right)^{3/2}\left(\sigma_{j}^{2}\right)^{1/2}}\left(\widehat{\sigma}_{i}^{2}-\sigma_{i}^{2}\right)-k_{N,T}\sum_{i=1}^{N}\frac{\widehat{\vepsi}_{i}'\widehat{\vepsi}_{i}}{\left(\sigma_{i}^{2}\right)^{1/2}}\left(\widehat{\sigma}_{i}^{2}-\sigma_{i}^{2}\right)\label{eq:Taylor_a1*}\\
	& =a_{11}^{*}-a_{12}^{*}\nonumber 
\end{align}
Note that the scalings $\left[\left(\sigma_{i}^{2}\right)^{3/2}\left(\sigma_{j}^{2}\right)^{1/2}\right]^{-1}$
and $\left(\sigma_{i}^{2}\right)^{-1/2}$ in both expressions are
immaterial since all orders in probability of subterms arising from
either $a_{11}^{*}$ or $a_{12}^{*}$ will be derived via either Markov's
or Chebyshev's inequality. Since error variances have a lower bound
above zero, we can always bound their inverse from above. For this
reason, and to simplify notation in the following proofs, we disregard
from the denominators in the fractions that constitute $a_{11}^{*}$
and $a_{12}^{*}$. 
The order of the two expressions above will further be derived by expanding 
\begin{align}
\widehat{\sigma}_{i}^{2} &  =\frac{\vepsi_{i}'\vepsi_{i}}{T}+\frac{\vlambda_{i}'\left(\overline{\mC}^{-1}\right)'\overline{\mU}'\mM_{\widehat{F}}\overline{\mU}\left(\overline{\mC}^{-1}\right)\vlambda_{i}}{T}-\frac{\vepsi_{i}'\mP_{\widehat{\mF}}\vepsi_{i}}{T}-2\frac{\vlambda_{i}'\left(\overline{\mC}^{-1}\right)'\overline{\mU}'\mM_{\widehat{\mF}}\vepsi_{i}}{T}.\label{eq:sigi2hat_CCE}
\end{align}
and
\begin{equation}
\widehat{\vepsi}_{i}'\widehat{\vepsi}_{j} = \vepsi_{i}'\vepsi_{j}+\vlambda_{i}'\left(\overline{\mC}^{-1}\right)'\overline{\mU}'\mM_{\widehat{\mF}}\overline{\mU}\left(\overline{\mC}^{-1}\right)\vlambda_{j}-\vepsi_{i}'\mP_{\widehat{\mF}}\vepsi_{j}-\vlambda_{i}'\left(\overline{\mC}^{-1}\right)'\overline{\mU}'\mM_{\widehat{\mF}}\vepsi_{j}-\vepsi_{i}'\mM_{\widehat{\mF}}\overline{\mU}\left(\overline{\mC}^{-1}\right)\vlambda_{j}
 \label{eq:epsi_epsj_CCE}
\end{equation}
and by looking at the resulting terms. 

Consider the first expression, $a_{11}^{*}$, where we use \eqref{eq:epsi_epsj_CCE} to write
\begin{align}
	a_{11}^{*}= & k_{N,T}\sum_{i,j}^{N}\widehat{\vepsi}_{i}'\widehat{\vepsi}_{j}\left(\widehat{\sigma}_{i}^{2}-\sigma_{i}^{2}\right)\nonumber \\
	= & k_{N,T}\sum_{i,j}^{N}\vepsi_{i}'\vepsi_{j}\left(\widehat{\sigma}_{i}^{2}-\sigma_{i}^{2}\right)+k_{N,T}\sum_{i,j}^{N}\vlambda_{i}'\left(\overline{\mC}^{-1}\right)'\overline{\mU}'\mM_{\widehat{\mF}}\overline{\mU}\left(\overline{\mC}^{-1}\right)\vlambda_{j}\left(\widehat{\sigma}_{i}^{2}-\sigma_{i}^{2}\right)\nonumber \\
	& -k_{N,T}\sum_{i,j}^{N}\vepsi_{i}'\mP_{\widehat{\mF}}\vepsi_{j}\left(\widehat{\sigma}_{i}^{2}-\sigma_{i}^{2}\right)-k_{N,T}\sum_{i,j}^{N}\vlambda_{i}'\left(\overline{\mC}^{-1}\right)'\overline{\mU}'\mM_{\widehat{\mF}}\vepsi_{j}\left(\widehat{\sigma}_{i}^{2}-\sigma_{i}^{2}\right)\nonumber \\
	& -k_{N,T}\sum_{i,j}^{N}\vepsi_{i}'\mM_{\widehat{\mF}}\overline{\mU}\left(\overline{\mC}^{-1}\right)\vlambda_{j}\left(\widehat{\sigma}_{i}^{2}-\sigma_{i}^{2}\right)\nonumber \\
	& =a_{111}^{*}+a_{112}^{*}-a_{113}^{*}-a_{114}^{*}-a_{115}^{*}. \label{eq:Taylor_a11*}
\end{align}
Concerning the second term in \eqref{eq:Taylor_a11*}, $a_{112}^{*}$,
we can write 
\begin{align}
	\left|a_{112}^{*}\right| & =\left|k_{N,T}\sum_{i,j}^{N}\vlambda_{i}'\left(\overline{\mC}^{-1}\right)'\overline{\mU}'\mM_{\widehat{\mF}}\overline{\mU}\left(\overline{\mC}^{-1}\right)\vlambda_{j}\left(\widehat{\sigma}_{i}^{2}-\sigma_{i}^{2}\right)\right|\nonumber \\
	& \leq Tk_{N,T}\left[ \sum_{i=1}^{N}\left\Vert \left(\widehat{\sigma}_{i}^{2}-\sigma_{i}^{2}\right)\vlambda_{i}'\right\Vert^2 \right]^{1/2} \left\Vert T^{-1}\overline{\mU}'\mM_{\widehat{\mF}}\overline{\mU}\right\Vert \left\Vert \left(\overline{\mC}^{-1}\right)\right\Vert ^{2}\left\Vert \sum_{j=1}^{N}\vlambda_{j}\right\Vert \nonumber \\
	& =\sqrt{\frac{T}{2N\left(N-1\right)}}\left[\largeO_{P}\left(\sqrt{N}T^{-1/2}\right)+\largeO_{P}\left(1\right)\right] \nonumber  \left\{ \largeO_{P}\left(N^{-1}\right)+\largeO_{P}\left[\left(NT\right)^{-1/2}\right]\right\} \largeO_{P}\left(N\right)\nonumber \\
	& =\largeO_{P}\left(T^{-1/2}\right)+\largeO_{P}\left(N^{-1/2}\right)+\largeO_{P}\left(N^{-1}\sqrt{T}\right).\label{eq:a112*_order}
\end{align}
Here, the third line results from Lemma \ref{lem:sumi_sig2diff_sq_CCE}
as well as result \eqref{eq:Ubar_MF_Ubar}. The third term in \eqref{eq:Taylor_a11*},
$a_{113}^{*}$, can be bounded from above in absolute value by 
\begin{align}
	\left|a_{113}^{*}\right| & =\left|k_{N,T}\sum_{i,j}^{N}\vepsi_{i}'\mP_{\widehat{\mF}}\vepsi_{j}\left(\widehat{\sigma}_{i}^{2}-\sigma_{i}^{2}\right)\right|\nonumber \\
	& \leq Tk_{N,T}\left\Vert T^{-1}\sum_{i=1}^{N}\vepsi_{i}'\widehat{\mF}\left(\widehat{\sigma}_{i}^{2}-\sigma_{i}^{2}\right)\right\Vert \left\Vert T^{-1}\sum_{j=1}^{N}\widehat{\mF}'\vepsi_{j}\right\Vert \left\Vert \left(T^{-1}\widehat{\mF}'\widehat{\mF}\right)^{-1}\right\Vert \nonumber \\
	& =Tk_{N,T}\left[\largeO_{P}\left(NT^{-3/2}\right)+\largeO_{P}\left(T^{-1}\sqrt{N}\right)+\largeO_{P}\left(T^{-1/2}\right)+\largeO_{P}\left(N^{-1}\right)\right] \nonumber \\
	&\times \left[\largeO_{P}\left(N^{-1/2}\right)+\largeO_{P}\left(\sqrt{N}T^{-1/2}\right)\right]\nonumber \\
	&= \sqrt{\frac{T}{2N(N-1)}}\Big[ \largeO_{P}\left(N^{-3/2}\right)+\largeO_{P}\left[\left(NT\right)^{-1/2}\right]+\largeO_{P}\left(N^{3/2}T^{-2}\right) \nonumber\\
	&+\largeO_{P}\left(NT^{-3/2}\right)+\largeO_{P}\left(T^{-1}\sqrt{N}\right)\Big] \nonumber \\
	& =\largeO_{P}\left(\sqrt{T}N^{-5/2}\right)+\largeO_{P}\left(N^{-3/2}\right)+\largeO_{P}\left(\sqrt{N}T^{-3/2}\right)+\largeO_{P}\left(T^{-1}\right)+\largeO_{P}\left[\left(NT\right)^{-1/2}\right].\label{eq:a113*_order}
\end{align}
Arriving at the order in probability stated here requires us to use
Lemma \ref{lem:sumi_epsi_Fhat_sigi2_diff} as well as results \eqref{eq:sig2hat_min_sig2}
and \eqref{eq:sumi_eps_Fhat}. We proceed to the fourth term, where
we write 
\begin{align}
	\left|a_{114}^{*}\right| & \leq Tk_{N,T}\left[\sum_{i=1}^{N}\left\Vert\left(\widehat{\sigma}_{i}^{2}-\sigma_{i}^{2}\right) \vlambda_{i}'  \right\Vert^2 \right]^{1/2} \left\Vert \overline{\mC}^{-1}\right\Vert \left\Vert T^{-1}\sum_{j=1}^{N}\overline{\mU}'\mM_{\widehat{\mF}}\vepsi_{j}\right\Vert \nonumber \\
	& =\sqrt{\frac{T}{2N(N-1)}}\left[\largeO_{P}\left(\sqrt{N}T^{-1/2}\right)+\largeO_{P}\left(1\right)\right]\left[\largeO_{P}\left(1\right)+\largeO_{P}\left(\sqrt{N}T^{-1/2}\right)\right]\nonumber \\
	& =\largeO_{P}\left(T^{-1/2}\right) + \largeO_{P}\left(N^{-1/2}\right) + \largeO_{P}\left(N^{-1}\sqrt{T}\right).\label{eq:a114*_order}
\end{align}
In addition to Lemma \ref{lem:sumi_sig2diff_sq_CCE}, we can use
result \eqref{eq:Ubar_MF_Ubar} to derive the upper bound for $\left|a_{114}^{*}\right|$
stated above. The relevance of result \eqref{eq:Ubar_MF_Ubar} can
be seen by noting that $\left\Vert T^{-1}\sum_{j=1}^{N}\overline{\mU}'\mM_{\widehat{\mF}}\vepsi_{j}\right\Vert =N\left\Vert T^{-1}\overline{\mU}'\mM_{\widehat{\mF}}\overline{\vepsi}\right\Vert $
and by recalling that $\left\Vert T^{-1}\overline{\mU}'\mM_{\widehat{\mF}}\overline{\mU}\right\Vert $
includes all terms in $\left\Vert T^{-1}\overline{\mU}'\mM_{\widehat{\mF}}\overline{\vepsi}\right\Vert $. 

We proceed with the fifth term in \eqref{eq:Taylor_a11*}, $a_{115}^{*}$.
After having taken absolute values, we can write
\begin{align*}
	\left|a_{115}^{*}\right| & \leq Tk_{N,T}\left\Vert T^{-1}\sum_{i=1}^{N}\left(\widehat{\sigma}_{i}^{2}-\sigma_{i}^{2}\right)\vepsi_{i}'\overline{\mU}\right\Vert \left\Vert \overline{\mC}^{-1}\right\Vert \left\Vert \sum_{j=1}^{N}\vlambda_{j}\right\Vert \\
	& +Tk_{N,T}\left\Vert T^{-1}\sum_{i=1}^{N}\left(\widehat{\sigma}_{i}^{2}-\sigma_{i}^{2}\right)\vepsi_{i}'\mP_{\widehat{\mF}}\overline{\mU}\right\Vert \left\Vert \overline{\mC}^{-1}\right\Vert \left\Vert \sum_{j=1}^{N}\vlambda_{j}\right\Vert \\
	& =\left(a_{1151}^{*}+a_{1152}^{*}\right)\left\Vert \overline{\mC}^{-1}\right\Vert \left\Vert \sum_{j=1}^{N}\vlambda_{j}\right\Vert 
\end{align*}
Consider now $a_{1152}^{*}$. The order in probability of this expression
is bounded by Lemma \ref{lem:sumi_epsi_Fhat_sigi2_diff} as well as
results \eqref{eq:sig2hat_min_sig2} and \eqref{eq:normsq_Ubar_Fhat}. More specifically,
we can write 
\begin{align*}
	a_{1152}^{*} & \leq Tk_{N,T}\left\Vert T^{-1}\sum_{i=1}^{N}\left(\widehat{\sigma}_{i}^{2}-\sigma_{i}^{2}\right)\vepsi_{i}'\widehat{\mF}\right\Vert \left\Vert \left(T^{-1}\widehat{\mF}'\widehat{\mF}\right)^{-1}\right\Vert \left\Vert T^{-1}\widehat{\mF}'\overline{\mU}\right\Vert \\
	& =\sqrt{\frac{T}{2N(N-1)}}\left[\largeO_{P}\left(NT^{-3/2}\right)+\largeO_{P}\left(T^{-1}\sqrt{N}\right)+\largeO_{P}\left(N^{-1}\right)+\largeO_{P}\left(T^{-1/2}\right)\right] \nonumber \\
	&\times \left\{ \largeO_{P}\left[\left(NT\right)^{-1/2}\right]+\largeO_{P}\left(N^{-1}\right)\right\} \\
	& =\largeO_{P}\left(N^{-1/2}T^{-3/2}\right)+\largeO_{P}\left[\left(NT\right)^{-1}\right]+\largeO_{P}\left(N^{-3/2}T^{-1/2}\right)+\largeO_{P}\left(N^{-2}\right)+\largeO_{P}\left(N^{-3}\sqrt{T}\right).
\end{align*}
The a statement for the remaining term $a_{1151}^{*}$ can be made
by invoking Lemma \ref{lem:sumi_epsi_Fhat_sigi2_diff}. This is possible
since $a_{1151}^{*}$ is completely contained in the expression $T^{-1}\sum_{i=1}^{N}\left(\widehat{\sigma}_{i}^{2}-\sigma_{i}^{2}\right)\vepsi_{i}'\widehat{\mF}$,
which is the expression handled by this lemma. Accordingly, we have
\[
a_{1151}^{*}=\largeO_{P}\left(T^{-1}\right)+\largeO_{P}\left[\left(NT\right)^{-1/2}\right]+\largeO_{P}\left(N^{-1}\right)+\largeO_{P}\left(\sqrt{T}N^{-2}\right).
\]
Using this result on $a_{1151}^{*}$ together with the previous one
on $a_{1152}^{*}$ we arrive at
\begin{equation}
	a_{115}^{*}=\largeO_{P}\left(T^{-1}\right)+\largeO_{P}\left[\left(NT\right)^{-1/2}\right]+\largeO_{P}\left(N^{-1}\right)+\largeO_{P}\left(\sqrt{T}N^{-2}\right).\label{eq:a115*_order}
\end{equation}
Lastly, we consider term $a_{111}^{*}$ which is decomposed into 
\begin{align*}
	a_{111}^{*} & =k_{N,T} \sum_{i=1}^N \sum_{j\neq i}^{N}\vepsi_{i}'\vepsi_{j}\left(\widehat{\sigma}_{i}^{2}-\sigma_{i}^{2}\right) 
	+ k_{N,T}\sum_{i=1}^{N}\vepsi_{i}'\vepsi_{i} \left(\widehat{\sigma}_{i}^{2}-\sigma_{i}^{2}\right)\\
	& =a_{1111}^{*}+a_{1112}^{*}.
\end{align*}
By Lemma \ref{lem:sumij_epsi_epsj-sigi2diff}, it holds that 
\[
a_{1111}^{*}=\largeO_p\left(\sqrt{N}T^{-1}\right)+\largeO_{P}\left(T^{-1/2}\right)+\largeO_{P}\left(N^{-1/2}\right).
\]
As for expression $a_{1112}^{*}$, we will emphasize below that an
identical expression exists in the decomposition of $a_{12}^{*}$
which eliminates this term. However, for now, we simply conclude from
results \eqref{eq:a112*_order}-\eqref{eq:a115*_order} and our discussion
here that 
\begin{equation}
	a_{11}^{*}=k_{N,T} \sum_{i=1}^{N}\vepsi_{i}'\vepsi_{i}\left(\widehat{\sigma}_{i}^{2}-\sigma_{i}^{2}\right)+\largeO_{P}\left(\sqrt{N}T^{-1}\right)+\largeO_{P}\left(T^{-1/2}\right)+\largeO_{P}\left(N^{-1/2}\right)+\largeO_{P}\left(N^{-1}\sqrt{T}\right)\label{eq:a11*_order}
\end{equation}

Having investigated $a_{11}^{*}$ we proceed with the second major
component of the linear approximation term at the true error variances.
We can decompose this second term as 
\begin{align*}
	a_{12}^{*} & = k_{N,T} \sum_{i=1}^{N}\vepsi_{i}'\vepsi_{i}\left(\widehat{\sigma}_{i}^{2}-\sigma_{i}^{2}\right)
	+ k_{N,T}\sum_{i=1}^{N}\vlambda_{i}'\left(\overline{\mC}^{-1}\right)'\overline{\mU}'\mM_{\widehat{\mF}}\overline{\mU}\left(\overline{\mC}^{-1}\right)\vlambda_{i}\left(\widehat{\sigma}_{i}^{2}-\sigma_{i}^{2}\right)\\
	& - k_{N,T} \sum_{i=1}^{N}\vepsi_{i}'\mP_{\widehat{\mF}}\vepsi_{i}\left(\widehat{\sigma}_{i}^{2}-\sigma_{i}^{2}\right)- 2k_{N,T}\sum_{i=1}^{N}\vlambda_{i}'\left(\overline{\mC}^{-1}\right)'\overline{\mU}'\mM_{\widehat{\mF}}\vepsi_{i}\left(\widehat{\sigma}_{i}^{2}-\sigma_{i}^{2}\right)\\
	& =a_{121}^{*}+a_{122}^{*}-a_{123}^{*}-a_{124}^{*}
\end{align*}
Term $a_{121}^{*}$ is identical to term $a_{1112}^{*}$ in $a_{11}^{*}$
and hence cancels out. Concerning $a_{122}^{*}$, we write 
\begin{align*}
	& \left|k_{N,T}\sum_{i=1}^{N}\vlambda_{i}'\left(\overline{\mC}^{-1}\right)'\overline{\mU}'\mM_{\widehat{\mF}}\overline{\mU}\left(\overline{\mC}^{-1}\right)\vlambda_{i}\left(\widehat{\sigma}_{i}^{2}-\sigma_{i}^{2}\right)\right|\\
	\leq & Tk_{N,T}\left( \sum_{i=1}^{N}\left\Vert \vlambda_{i}\right\Vert ^{2}\right) ^{1/2}\left\Vert \overline{\mC}^{-1}\right\Vert ^{2}\left\Vert T^{-1}\overline{\mU}'\mM_{\widehat{\mF}}\overline{\mU}\right\Vert \left[\sum_{i=1}^{N}\left(\widehat{\sigma}_{i}^{2}-\sigma_{i}^{2}\right)^{2}\right]^{1/2}\\
	= & \sqrt{\frac{T}{2N(N-1)}}\sqrt{N}\left\{ \largeO_{P}\left(N^{-1}\right)+\largeO_{P}\left[\left(NT\right)^{-1/2}\right]\right\} \left[\largeO_{P}\left(\sqrt{N}T^{-1/2}\right)+\largeO_{P}\left(N^{-1/2}\right)\right]\\
	= & \largeO_{P}\left[\left(NT\right)^{-1/2}\right]+\largeO_{P}\left(N^{-1}\right)+\largeO_{P}\left(N^{-2}\sqrt{T}\right)
\end{align*}
where we use result \eqref{eq:Ubar_MF_Ubar} and Lemma \ref{lem:sumi_sig2diff_sq_CCE}.
The same lemma, together with result \eqref{eq:sumi_eps_PF_eps} is
employed to show that the next term, $a_{123}^{*}$, is 
\begin{align*}
	a_{123}^{*} & \leq Tk_{N,T}\left[T^{-1}\sum_{i=1}^{N}\left(\vepsi_{i}'\mP_{\widehat{\mF}}\vepsi_{i}\right)^{2}\right]^{1/2}\left[\sum_{i=1}^{N}\left(\widehat{\sigma}_{i}^{2}-\sigma_{i}^{2}\right)^{2}\right]^{1/2}\\
	& =\sqrt{\frac{T}{2N(N-1)}}\left[\largeO_{P}\left(T^{-1}\sqrt{N}\right)+\largeO_{P}\left(N^{-1}\right)\right]\left[\largeO_{P}\left(\sqrt{N}T^{-1/2}\right)+\largeO_{P}\left(N^{-1/2}\right)\right]\\
	& =\largeO_{P}\left(T^{-1}\right)+\largeO_{P}\left(N^{-1}T^{-1/2}\right)+\largeO_{P}\left(N^{-3/2}\right)+\largeO_{P}\left(N^{-5/2}\sqrt{T}\right)
\end{align*}
Lastly, we use result \eqref{eq:sumi_lami_Ubar_MF_epsi} to obtain
\begin{align*}
	\left|a_{124}^{*}\right| & \leq 2Tk_{N,T}\left[\sum_{i=1}^{N}\left(T^{-1}\vlambda_{i}'\left(\overline{\mC}^{-1}\right)'\overline{\mU}'\mM_{\widehat{\mF}}\vepsi_{i}\right)^{2}\right]^{1/2}\left[\sum_{i=1}^{N}\left(\widehat{\sigma}_{i}^{2}-\sigma_{i}^{2}\right)^{2}\right]^{1/2}\\
	& =\sqrt{\frac{2T}{N(N-1)}}\left[\largeO_{P}\left(N^{-1/2}\right)+\largeO_{P}\left(T^{-1/2}\right)\right]\left[\largeO_{P}\left(\sqrt{N}T^{-1/2}\right)+\largeO_{P}\left(N^{-1/2}\right)\right]\\
	& =\largeO_{P}\left[\left(NT\right)^{-1/2}\right]+\largeO_{P}\left(N^{-1}\right)+\largeO_{P}\left(N^{-2}\sqrt{T}\right).
\end{align*}
We can now summarize the four last intermediary results by concluding
that 
\[
a_{12}^{*} = k_{N,T} \sum_{i=1}^{N}\frac{\vepsi_{i}'\vepsi_{i}}{\left(\sigma_{i}^{2}\right)^{2}}\left(\widehat{\sigma}_{i}^{2}-\sigma_{i}^{2}\right) 
+ \largeO_{P}\left(T^{-1}\right) + \largeO_{P}\left[\left(NT\right)^{-1/2}\right] + \largeO_{P}\left(N^{-1}\right) + \largeO_{P}\left(N^{-2}\sqrt{T}\right),
\]
which together with \eqref{eq:a11*_order} leads to 
\begin{align}
	a_{1}^{*} & =a_{11}^{*}-a_{12}^{*}\nonumber \\
	& =\largeO_{P}\left(T^{-1}\sqrt{N}\right)+\largeO_{P}\left(T^{-1/2}\right)+\largeO_{P}\left(N^{-1/2}\right)+\largeO_{P}\left(N^{-1}\sqrt{T}\right).\label{eq:Taylor_a1*order}
\end{align}
\medskip{}

So far we, our focus has been on a linear approximation of estimated
error variances, evaluated at their true values. In the remainder
of this proof, we focus on its approximation error which we denoted
$a_{2}^{*}$ in expression \eqref{eq:Taylorterm_CCE}. We take absolute
values and use the Cauchy-Schwarz inequality to arrive at
\begin{align}
	\left|a_{2}^{*}\right|= & \left|k_{N,T}\sum_{i=1}^{N}\sum_{j\neq i}^{N}\frac{\widehat{\vepsi}_{i}'\widehat{\vepsi}_{j}}{\left(\sigma_{j}^{2}\right)^{1/2}}\left(\widehat{\sigma}_{i}^{2}-\sigma_{i}^{2}\right)\left[\frac{1}{\left(\varsigma_{i}^{2}\right)^{3/2}}-\frac{1}{\left(\sigma_{i}^{2}\right)^{3/2}}\right]\right|\nonumber \\
	& \leq Tk_{N,T}\left(\sum_{i=1}^{N}\left[T^{-1}\sum_{j\neq i}^{N}\frac{\widehat{\vepsi}_{i}'\widehat{\vepsi}_{j}}{\left(\sigma_{j}^{2}\right)^{1/2}}\left(\widehat{\sigma}_{i}^{2}-\sigma_{i}^{2}\right)\right]^{2}\right)^{1/2}\left(\sum_{i=1}^{N}\left[\frac{1}{\left(\varsigma_{i}^{2}\right)^{3/2}}-\frac{1}{\left(\sigma_{i}^{2}\right)^{3/2}}\right]^{2}\right)^{1/2}\nonumber \\
	& \leq Tk_{N,T}\left(\sum_{i=1}^{N}\left[T^{-1}\sum_{j\neq i}^{N}\frac{\widehat{\vepsi}_{i}'\widehat{\vepsi}_{j}}{\left(\sigma_{j}^{2}\right)^{1/2}}\left(\widehat{\sigma}_{i}^{2}-\sigma_{i}^{2}\right)\right]^{2}\right)^{1/2}\left(\sum_{i=1}^{N}\left[\frac{1}{\left(\widehat{\sigma}_{i}^{2}\right)^{3/2}}-\frac{1}{\left(\sigma_{i}^{2}\right)^{3/2}}\right]^{2}\right)^{1/2}\nonumber \\
	& = \sqrt{\frac{T}{2N(N-1)}} \sqrt{a_{21}^{*}a_{22}^{*}.}\label{eq:a2*}
\end{align}
A further first-order Taylor approximation of $\left(\widehat{\sigma}_{i}^{2}\right)^{-3/2}$
around $\left(\sigma_{i}^{2}\right)^{-3/2}$ results in
\begin{equation*}
\frac{1}{\left(\widehat{\sigma}_{i}^{2}\right)^{3/2}}-\frac{1}{\left(\sigma_{i}^{2}\right)^{3/2}}=-3/2\frac{1}{\left(\varsigma_{i}^{2}\right)^{5/2}}\left(\widehat{\sigma}_{i}^{2}-\sigma_{i}^{2}\right)
\end{equation*}
Substituting this expression into $a_{22}^{*}$
leads to
\begin{align}
	a_{22}^{*} & =\sum_{i=1}^{N}\left[-3/2\frac{1}{\left(\varsigma_{i}^{2}\right)^{5/2}}\left(\widehat{\sigma}_{i}^{2}-\sigma_{i}^{2}\right)\right]^{2}\nonumber \\
	& \leq9/4\left(\sup_{i}\frac{1}{\left(\varsigma_{i}^{2}\right)^{5/2}}\right)\sum_{i=1}^{N}\left(\widehat{\sigma}_{i}^{2}-\sigma_{i}^{2}\right)^{2}\nonumber \\
	& =9/4\left(\frac{1}{\left(\inf_{i}\varsigma_{i}^{2}\right)^{5/2}}\right)\left[\largeO_{P}\left(N/T\right)+\largeO_{P}\left(N^{-1}\right)+\largeO_{P}\left(T^{-1}\right)\right]\nonumber \\
	& =\largeO_{P}\left(NT^{-1}\right)+\largeO_{P}\left(N^{-1}\right).\label{eq:Taylor_a22*order}
\end{align}
where result \eqref{lem:sumi_sig2diff_sq_CCE} is used in line three
and Lemma \ref{lem:inf_sighat2_2WFE} in the last line to ensure that
its result holds as long as $T^{-1}\sqrt{N}\to 0$ is satisfied. In addition
to result \eqref{eq:Taylor_a22*order}, which bounds $a_{22}^{*}$,
the order in probability of $a_{2}^{*}$ is determined by that of
$a_{21}^{*}.$ Lemma \ref{lem:Taylor_a21*order} establishes the corresponding
result on $a_{21}^{*}.$

Accordingly, we have 
\begin{align}
	\left|a_{2}^{*}\right| & =\left[\largeO_{P}\left(T^{-1/2}\right)+\largeO_{P}\left(N^{-1/2}\right)+\largeO_{P}\left(N^{-1}\sqrt{T}\right)\right]\left[\largeO_{P}\left(\sqrt{N}T^{-1/2}\right)+\largeO_{P}\left(N^{-1/2}\right)\right]\nonumber \\
	& =\largeO_{P}\left(T^{-1}\sqrt{N}\right)+\largeO_{P}\left(T^{-1/2}\right)+\largeO_{P}\left(N^{-1/2}\right)+\largeO_{P}\left(N^{-3/2}\sqrt{T}\right).\label{eq:Taylor_a2*order}
\end{align}
This result marks the last one needed to characterize the impact of
estimating error variances rather than knowing their true value. We
return to expression \eqref{eq:Taylorterm_CCE} and use our results
\eqref{eq:Taylor_a1*order} and \eqref{eq:Taylor_a2*order} on its
two components $a_{1}^{*}$ and $a_{2}^{*}$ to arrive at 
\[
k_{N,T}\sum_{i=1}^{N}\sum_{j\neq i}^{N}\frac{\widehat{\vepsi}_{i}'\widehat{\vepsi}_{i}}{\left(\varsigma_{i}^{2}\right)^{3/2}\left(\sigma_{j}^{2}\right)^{1/2}}\left(\widehat{\sigma}_{i}^{2}-\sigma_{i}^{2}\right)=\largeO_{P}\left(N^{-1}\sqrt{T}\right)+\largeO_{P}\left(N^{-1/2}\right)+\largeO_{P}\left(T^{-1/2}\right)+\largeO_{P}\left(T^{-1}\sqrt{N}\right).
\]

\subsection{Auxiliary lemmas for Section \ref{sec:var_estim_2WFE}}
\begin{lemma}
	\label{lem:sigmahat}
	Suppose that Assumption \ref{ass:errorsAddHetero} holds. Then,	
	\begin{align}
		\widehat{\sigma}_{i}^{2}-\sigma_{i}^{2} & =\largeO_{P}\left(T^{-1/2}\right)+\largeO_{P}\left(N^{-1}\right)\label{eq:sig2hat_min_sig2}\\
		N^{-1}\sum_{i\text{=1}}^{N}\left(\widehat{\sigma}_{i}^{2}-\sigma_{i}^{2}\right) & =\largeO_{P}\left[\left(NT\right)^{-1/2}\right]+\largeO_{P}\left(N^{-1}\right)\label{eq:sumi_sig2hat_min_sig2}\\
		\sum_{i=1}^{N}\left(\widehat{\sigma}_{i}^{2}-\sigma_{i}^{2}\right)^{2} & =\largeO_{P}\left(N/T\right)+\largeO_{P}\left(T^{-1}\right)+\largeO_{P}\left(N^{-1}\right)\label{eq:sumi_sighat2_min_sig2_squared}
	\end{align}
\end{lemma}
\begin{proof}[Proof of Lemma \ref{lem:sigmahat}]
	All orders in probability in this proof are derived via application
	of Markov's or Chebyshev's inequality. Details are left out for the
	sake of brevity. For the first result, we can explicitly write out
	all sums in equation \eqref{eq:sigi2hat_2WFE} in order to arrive
	at
	\begin{align*}
		\widehat{\sigma}_{i}^{2}-\sigma_{i}^{2} & =T^{-1}\sum_{t=1}^{T}\left(\epsi_{i,t}^{2}-\sigma_{i}^{2}\right)+T^{-1}\sum_{t=1}^{T}\left(\overline{\epsi}_{t}\right)^{2}+\left(T^{-1}\sum_{t=1}^{T}\epsi_{i,t}\right)^{2}+\left[\left(NT\right)^{-1}\sum_{j=1}^{N}\sum_{t=1}^{T}\epsi_{j,t}\right]^{2}\\
		& -2\left(NT\right)^{-1}\sum_{j=1}^{N}\sum_{t=1}^{T}\epsi_{i,t}\epsi_{j,t}-2\left(T^{-1}\sum_{t=1}^{T}\epsi_{i,t}\right)\left[\left(NT\right)^{-1}\sum_{j=1}^{N}\sum_{t=1}^{T}\epsi_{j,t}\right]\\
		& =\largeO_{P}\left(T^{-1/2}\right)+\largeO_{P}\left(N^{-1}\right)+\largeO_{P}\left[\left(NT\right)^{-1/2}\right].
	\end{align*}
	For the second result, note that summation over $N$ simplifies the
	decomposition in equation \eqref{eq:sigi2hat_2WFE} to 
	\begin{align*}
		N^{-1}\sum_{i\text{=1}}^{N}\left(\widehat{\sigma}_{i}^{2}-\sigma_{i}^{2}\right) &= \left(NT\right)^{-1}\sum_{i\text{=1}}^{N}\sum_{t=1}^{T}\left(\epsi_{i,t}^{2}-\sigma_{i}^{2}\right)-N^{-2}T^{-1}\sum_{i,j}^{N}\sum_{t=1}^{T}\epsi_{i,t}\epsi_{j,t}+N^{-1}T^{-2}\sum_{i=1}^{N}\sum_{t=1}^{T}\epsi_{i,t}\epsi_{i,t'}-\overline{\epsi}^{2}\\
		& =\largeO_{P}\left[\left(NT\right)^{-1/2}\right]+\largeO_{P}\left(N^{-1}\right)+\largeO_{P}\left(T^{-1}\right).
	\end{align*}
	For the third result, we borrow from the explicit expression of $\widehat{\sigma}_{i}^{2}-\sigma_{i}^{2}$
	from the first result above and describe an upper bound as 
	\begin{align*}
		\sum_{i=1}^{N}\left(\widehat{\sigma}_{i}^{2}-\sigma_{i}^{2}\right)^{2} & \leq M\sum_{i=1}^{N}\left[T^{-1}\sum_{t=1}^{T}\left(\epsi_{i,t}^{2}-\sigma_{i}^{2}\right)\right]^{2}+MN\left[T^{-1}\sum_{t=1}^{T}\left(\overline{\epsi}_{t}\right)^{2}\right]^{2}+M\sum_{i=1}^{N}\left(T^{-1}\sum_{t=1}^{T}\epsi_{i,t}\right)^{4}\\
		& +MN\overline{\epsi}^{4}+M\sum_{i=1}^{N}\left[\left(NT\right)^{-1}\sum_{j=1}^{N}\sum_{t=1}^{T}\epsi_{i,t}\epsi_{j,t}\right]^{2}+2\sum_{i=1}^{N}\left(T^{-1}\sum_{t=1}^{T}\epsi_{i,t}\right)^{2}\overline{\epsi}^{2}\\
		& =\frac{M}{T^{2}}\sum_{i=1}^{N}\sum_{t,t'}^{T}\left(\epsi_{i,t}^{2}-\sigma_{i}^{2}\right)\left(\epsi_{i,t'}^{2}-\sigma_{i}^{2}\right)+\frac{MN}{N^{4}T^{2}}\sum_{i,i',i'',i'''}^{N}\sum_{t,t'}^{T}\epsi_{i,t}\epsi_{i',t}\epsi_{i'',t'}\epsi_{i''',t'}\\ &+\frac{MN}{T^{4}}\sum_{i=1}^{N}\sum_{t,t',t'',t'''}^{T}\epsi_{i,t}\epsi_{i,t'}\epsi_{i,t''}\epsi_{i,t'''} +MN\overline{\epsi}^{4}\\
		& +\frac{M}{\left(NT\right)^{2}}\sum_{i,j,j'}^{N}\sum_{t,t'}^{T}\epsi_{i,t}\epsi_{j,t}\epsi_{j',t'}\epsi_{i,t'}+\overline{\epsi}^{2}\frac{M}{T^{2}}\sum_{i=1}^{N}\sum_{t,t'}^{T}\epsi_{i,t}\epsi_{i,t'}\\
		& =\largeO_{P}\left(N/T\right)+\largeO_{P}\left(T^{-1}\right)+\largeO_{P}\left(N^{-1}\right).
	\end{align*}
\end{proof}

\bigskip

\begin{lemma}
	\label{lem:buildingblocks_2WFE}
		Suppose that Assumption \ref{ass:errorsAddHetero} holds. Then,	
	\begin{align}
		\sum_{i=1}^{N}\left(T^{-1}\sum_{t=1}^{T}\epsi_{i,t}\overline{\epsi}_{t}\right)^{2} & =\largeO_{P}\left(N^{-1}\right)+\largeO_{P}\left(T^{-1}\right)\label{eq:sumit_eps_epsbar_sq}\\
		k_{N,T}\sum_{i=1}^{N}\vepsi_{i}'\vepsi_{i}\left(\frac{\vepsi_{i}'\vepsi_{i}}{T}-\sigma_{i}^{2}\right) & =\largeO_{P}\left(N^{-1/2}\right)+\largeO_{P}\left(T^{-1/2}\right)\label{eq:sumi_epsvar_epsvardev}\\
		k_{N,T} \sum_{i=1}^N \sum_{j\neq i}^{N}\vepsi_{i}'\vepsi_{j}\left(\frac{\vepsi_{i}'\vepsi_{i}}{T}-\sigma_{i}^{2}\right) & =\largeO_{P}\left(T^{-1/2}\right)+\largeO_{P}\left(T^{-1}\sqrt{N}\right)\label{eq:sumij_epsvar_epsvardev}\\
		\sum_{i=1}^{N}\left(\sum_{j\neq i}^{N}\vepsi_{i}'\vepsi_{j}\right)^{2}\left(\frac{\vepsi_{i}'\vepsi_{i}}{T}-\sigma_{i}^{2}\right)^{2} & =\largeO_{P}\left(N^{2}\right)\label{eq:sumij_espi_espj_epsisq-diff}\\
		\sum_{i=1}^{N}\left(\sum_{j\neq i}^{N}\vepsi_{i}'\vepsi_{j}\right)^{2}\left(\frac{\vepsi_{i}'\overline{\vepsi}}{T}\right)^{2} & =\largeO_{P}\left(N\right)+\largeO_{P}\left(T\right) \label{eq:sumij_epsi_epsj_epsi_epsbar}\\
		\sum_{i=1}^{N}\overline{\epsi}_{i}\left(\frac{\vepsi_{i}'\vepsi_{i}}{T}-\sigma_{i}^{2}\right) & =\largeO_{P}\left(NT^{-1}\right)\label{eq:sumi_epsbari_epsivardev}\\
		\sum_{i=1}^{N}\left(\sum_{j\neq i}^{N}T\overline{\epsi}_{i}\overline{\epsi}_{j}\right)^{2}\left(\frac{\vepsi_{i}'\vepsi_{i}}{T}-\sigma_{i}^{2}\right)^{2} & =\largeO_{P}\left(N^{2}T^{-1}\right)\label{eq:sumij_epsbari_epsbarj_epsisq-diff}\\
		\sum_{i=1}^{N}\left(\sum_{j\neq i}^{N}T\overline{\epsi}_{i}\overline{\epsi}_{j}\right)^{2}\left(\frac{\vepsi_{i}'\overline{\vepsi}}{T}\right)^{2} & =\largeO_{P}\left(NT^{-1}\right)+\largeO_{P}\left(1\right)\label{eq:sumij_epsbari_epsbarj_epsi_epsbar}
	\end{align}
\end{lemma}
\begin{proof}[Proof of Lemma \ref{lem:buildingblocks_2WFE}]
	Non-negativity of the first expression allows us to look at
	\begin{align*}
		\E\left[T^{-2}\sum_{i=1}^{N}\sum_{t,t'}^{T}\epsi_{i,t}\overline{\epsi}_{t}\overline{\epsi}_{t'}\epsi_{i,t'}\right] & =\left(NT\right)^{-2}\sum_{i,i',i''}^{N}\sum_{t,t'}^{T}\E\left[\epsi_{i,t}\epsi_{i',t}\epsi_{i'',t'}\epsi_{i,t'}\right]\\
		& =\largeO\left(N^{-1}\right)+\largeO\left(T^{-1}\right),
	\end{align*}
	where the last line results from eliminating terms where indices cannot
	be matched to arrive at a non-zero expectation term.
	
	Consider the second result now. We can write
	\begin{align*}
	k_{N,T}\sum_{i=1}^{N}\vepsi_{i}'\vepsi_{i}\left(\frac{\vepsi_{i}'\vepsi_{i}}{T}-\sigma_{i}^{2}\right) 
	& = Tk_{N,T} \sum_{i=1}^{N}\left(\frac{\vepsi_{i}'\vepsi_{i}}{T}-\sigma_{i}^{2}\right)^{2} 
	+ Tk_{N,T} \sum_{i=1}^{N}\sigma_{i}^{2}\left(\frac{\vepsi_{i}'\vepsi_{i}}{T}-\sigma_{i}^{2}\right).\\
	& = N^{-1}\sqrt{T} \left[\largeO_{P}\left(NT^{-1}\right)+\largeO_{P}\left(\sqrt{NT}\right)\right]\\
	& = \largeO_{P}\left(T^{-1/2}\right) + \largeO_{P}\left(N^{-1/2}\right).
	\end{align*}
	The first order term in the second line above is given by results
	in the proof of \eqref{eq:sumi_sighat2_min_sig2_squared}. For the
	second term, we use Chebyshev's inequality and boundedness of $\sigma_{i}^{2}$.
	
	We proceed with the third term. For notational simplicity we will abbreviate $\sum_{i=1}^N \sum_{j\neq i}^{N}$ by $\sum_{i \neq j}^{N}$ in this proof. Given that $i\neq j$, the expression
	\[
	\frac{k_{N,T}}{T}\sum_{i\neq j}^{N}\sum_{t,t'}^{T}\epsi_{j,t}\epsi_{i,t}\left(\epsi_{i,t'}^{2}-\sigma_{i}^{2}\right)
	\]
	has a zero expected value, we derive its order in probability from
	its variance. In the expression
	\[
	E\left\{ \left[\frac{k_{N,T}}{T}\sum_{i\neq j}^{N}\sum_{t,t'}^{T}\epsi_{j,t}\epsi_{i,t}\left(\epsi_{i,t'}^{2}-\sigma_{i}^{2}\right)\right]^{2}\right\} =\frac{k_{N,T}^{2}}{T^{2}}\sum_{i\neq j}^{N}\sum_{i'\neq j'}^{N}\sum_{t,t',t'',t'''}^{T}\epsi_{j,t}\epsi_{i,t}\left(\epsi_{it'}^{2}-\sigma_{i}^{2}\right)\epsi_{j',t''}\epsi_{i',t''}\left(\epsi_{i',t'''}^{2}-\sigma_{i'}^{2}\right)
	\]
	a number of terms cancels out because independence of errors across
	cross-sections and across time leads to a number of terms that depend
	on the zero expected value of $\epsi_{it}$. While a detailed
	account of all individual cases is omitted here for the sake of brevity,
	we can state that the variance of $k_{N,T}T^{-1}\sum_{i\neq j}^{N}\sum_{t,t'}^{T}\epsi_{j,t}\epsi_{i,t}\left(\epsi_{i,t'}^{2}-\sigma_{i}^{2}\right)$
	reduces to
	\begin{align*}
		&E\left\{ \left[\frac{k_{N,T}}{T}\sum_{i\neq j}^{N}\sum_{t,t'}^{T}\epsi_{j,t}\epsi_{i,t}\left(\epsi_{i,t'}^{2}-\sigma_{i}^{2}\right)\right]^{2}\right\}  \\
		& =\frac{1}{4N^2T^{3}}\sum_{i\neq j}^{N}\sum_{t,t'}^{T}\E\left(\epsi_{j,t}^{2}\right)\E\left(\epsi_{i,t}^{2}\right)\E\left[\left(\epsi_{i,t'}^{2}-\sigma_{i}^{2}\right)^{2}\right]\\
		& +\frac{1}{4N^2T^{3}}\sum_{i\neq i'\neq j}^{N}\sum_{t=1}^{T}\E\left(\epsi_{j,t}^{2}\right)\E\left(\epsi_{i,t}^{3}\right)\E\left(\epsi_{i',t}^{3}\right)+\smallO\left(T^{-1}\right)+\smallO\left(T^{-2}N\right)\\
		& =\largeO\left(T^{-1}\right)+\largeO\left(T^{-2}N\right).
	\end{align*}
	Result \eqref{eq:sumij_epsvar_epsvardev} follows from the equation
	above by the application of Chebyshev's inequality.
	
	Before applying Markov's inequality to establish result number four,
	note that all elements in the sum 
	\[
	\sum_{i=1}^{N}\sum_{j\neq i}^{N}\sum_{j'\neq i}^{N}\sum_{t,t',t'',t'''}^{T}\epsi_{i,t}\epsi_{j,t}\epsi_{j',t'}\epsi_{i,t'}\left(\frac{\epsi_{i,t''}^{2}}{T}-\sigma_{i}^{2}\right)\left(\frac{\epsi_{i,t'''}^{2}}{T}-\sigma_{i}^{2}\right)
	\]
	whose indexes violate either $j=j'$ or $t=t'$ will have an expected
	value of zero. This is a direct consequence of the fact that $i\neq j$
	and $i\neq j'$ generally holds. If $j=j'$ and $t=t'$ are satisfied
	we also need the index restriction $t''=t'''$ to ensure a zero expected
	value. Hence, 
	\begin{align*}
		\E\left[\sum_{i=1}^{N}\left(\sum_{j\neq i}^{N}\vepsi_{i}'\vepsi_{j}\right)^{2}\left(\frac{\vepsi_{i}'\vepsi_{i}}{T}-\sigma_{i}^{2}\right)^{2}\right]= & T^{-2}\sum_{i=1}^{N}\sum_{j\neq i}^{N}\sum_{t,t''}^{T}\E\left[\epsi_{i,t}^{2}\epsi_{j,t}^{2}\left(\epsi_{i,t''}^{2}-\sigma_{i}^{2}\right)^{2}\right]+\smallO\left(N^{2}\right)\\
		= & \largeO\left(N^{2}\right),
	\end{align*}
	which leads to result \eqref{eq:sumij_espi_espj_epsisq-diff}. 
	
	Moving to result \eqref{eq:sumij_epsi_epsj_epsi_epsbar}, we continue
	to look for restrictions on the indexes which would result in nonzero
	expected values. However, since the term under consideration involves
	more sums over cross-section indexes, we limit our attention to restrictions
	that would lead to the maximal number of sums rather than documenting
	all possible index restrictions which result in nonzero expectations.
	Taking expectations, we have 
	\begin{align*}
		\E\left[\sum_{i=1}^{N}\left[\sum_{j\neq i}^{N}\vepsi_{i}'\vepsi_{j}\right]^{2}\left(\frac{\vepsi_{i}'\overline{\vepsi}}{T}\right)^{2}\right] & =\left(NT\right)^{-2}\sum_{i=1}^{N}\sum_{j\neq i}^{N}\sum_{t,t'''}^{T}\E\left[\epsi_{i,t}\epsi_{j,t}\epsi_{j',t'}\epsi_{i,t'}\epsi_{i,t'}\epsi_{i',t''}\epsi_{i'',t'''}\epsi_{i,t'''}\right].
	\end{align*}
	The set of index restrictions maximizing the number of sums over cross-sections
	is $j=j'$, $i'=i''$, $t=t'$ and $t''=t'''$. This leads to three
	sums over cross-sections as well as two sums over time, resulting
	in a term of order $\largeO\left(N\right)$. The alternative set of index
	restrictions $j=j',i=i'=i''$ and $t=t'$ maximizes number of sums
	over time and leads to a term of order $\largeO\left(T\right)$. Summarizing
	these two results, we can state that application of Markov's inequality
	establishes that 
	\[
	\sum_{i=1}^{N}\left[\sum_{j\neq i}^{N}\vepsi_{i}'\vepsi_{j}\right]^{2}\left(\frac{\vepsi_{i}'\overline{\vepsi}}{T}\right)^{2}=\largeO_{P}\left(N\right)+\largeO_{P}\left(T\right).
	\]
	
	The last three results in Lemma \ref{lem:buildingblocks_2WFE} are
	very similar to the previous three expressions targeted above. For
	this reason, the following proofs will leave out unnecessary details.
	For result \eqref{eq:sumi_epsbari_epsivardev}, we write out the expression
	of interest as $T^{-2}\sum_{i=1}^{N}\sum_{t,t'}^{T}\epsi_{i,t}\left(\epsi_{i,t'}^{2}-\sigma_{i}^{2}\right).$
	Taking first squares and then expectations, we get 
	\begin{align*}
		& E\left\{ \left[T^{-2}\sum_{i=1}^{N}\sum_{t,t'}^{T}\epsi_{i,t}\left(\epsi_{i,t'}^{2}-\sigma_{i}^{2}\right)\right]^{2}\right\} \\
		& =T^{-4}\sum_{i,j}^{N}\sum_{t,t',t'',t'''}^{T}\E\left[\epsi_{i,t}\epsi_{j,t'}\left(\epsi_{i,t''}^{2}-\sigma_{i}^{2}\right)\left(\epsi_{j,t'''}^{2}-\sigma_{i}^{2}\right)\right]\\
		& =T^{-4}\sum_{i,j}^{N}\sum_{t,t'}^{T}\E\left(\epsi_{i,t}^{3}\right)\E\left(\epsi_{i',t''}^{3}\right)+\smallO\left(N^{2}T^{-2}\right)\\
		& =\largeO\left(N^{2}T^{-2}\right),
	\end{align*}
	from which the desired result follows.
	
	Concerning result \eqref{eq:sumij_epsbari_epsbarj_epsisq-diff}, we
	write 
	\begin{align*}
		& E\left\{ \sum_{i=1}^{N}\left[\sum_{j\neq i}^{N}T\overline{\epsi}_{i}\overline{\epsi}_{j}\right]^{2}\left(\frac{\vepsi_{i}'\vepsi_{i}}{T}-\sigma_{i}^{2}\right)^{2}\right\} \\
		= & T^{-4}\sum_{i=1}^{N}\sum_{j\neq i}^{N}\sum_{t,t',t''}^{T}\E\left(\epsi_{i,t}^{2}\right)\E\left(\epsi_{j,t'}^{2}\right)\E\left[\left(\epsi_{i,t''}^{2}-\sigma_{i}^{2}\right)^{2}\right]+T^{-4}\sum_{i=1}^{N}\sum_{j\neq i}^{N}\sum_{t,t',t''}^{T}\E\left(\epsi_{i,t}^{3}\right)\E\left(\epsi_{j,t'}^{2}\right)\E\left(\epsi_{i,t''}^{3}\right)\\
		+ & \smallO\left(N^{2}T^{-1}\right)\\
		= & \largeO\left(N^{2}T^{-1}\right).
	\end{align*}
	which leads to result \eqref{eq:sumij_epsbari_epsbarj_epsisq-diff}. 
	
	Moving to result \eqref{eq:sumij_epsbari_epsbarj_epsi_epsbar}, we
	can state that 
	\begin{align*}
		& E\left\{ \sum_{i=1}^{N}\left[\sum_{j\neq i}^{N}T\overline{\epsi}_{i}\overline{\epsi}_{j}\right]^{2}\left(\frac{\vepsi_{i}'\overline{\vepsi}}{T}\right)^{2}\right\} \\
		= & N^{-2}T^{-4}\sum_{i=1}^{N}\sum_{j\neq i}^{N}\sum_{t,t',t'',t'''}^{T}\E\left(\epsi_{i,t}^{2}\right)\E\left(\epsi_{j,t'}^{2}\right)\E\left(\epsi_{i,t''}^{2}\right)\E\left(\epsi_{i,t'''}^{2}\right)\\
		+ & N^{-2}T^{-4}\sum_{i,i'}^{N}\sum_{j\neq i}^{N}\sum_{t,t',t'''}^{T}\E\left(\epsi_{i,t}^{2}\right)\E\left(\epsi_{j,t'}^{2}\right)\E\left(\epsi_{i,t''}^{2}\right)\E\left(\epsi_{i',t''}^{2}\right)+\smallO\left(1\right)+\smallO\left(NT^{-1}\right)\\
		= & \largeO\left(1\right)+\largeO\left(NT^{-1}\right),
	\end{align*}
	so that we arrive at result \eqref{eq:sumij_epsbari_epsbarj_epsi_epsbar}.
\end{proof}

\bigskip

\begin{lemma}
	\label{lem:sumij_epsi_epsj_sigi2_diff}Suppose that Assumption \ref{ass:errorsAddHetero} holds.
	Then, 
	\[
	k_{N,T}\sum_{i,j}^{N}\vepsi_{i}'\vepsi_{j}\left(\widehat{\sigma}_{i}^{2}-\sigma_{i}^{2}\right)=\mathbf{1}\left[\E\left(\epsi_{it}^{3}\right)\neq0\right]\largeO_{P}\left(T^{-1}\sqrt{N}\right)+\largeO_{P}\left(T^{-1/2}\right)+\largeO_{P}\left(N^{-1/2}\right)+\largeO_{P}\left(N^{-1}\sqrt{T}\right)
	\]
\end{lemma}
\begin{proof}[Proof of Lemma \ref{lem:sumij_epsi_epsj_sigi2_diff}]
	We use the decomposition of $\widehat{\sigma}_{i}^{2}$ in equation \eqref{eq:sigi2hat_2WFE}
	to write 
	\begin{align*}
		k_{N,T}\sum_{i,j}^{N}\vepsi_{i}'\vepsi_{j}\left(\widehat{\sigma}_{i}^{2}-\sigma_{i}^{2}\right) & =k_{N,T}\sum_{i,j}^{N}\vepsi_{i}'\vepsi_{j}\left(\frac{\vepsi_{i}'\vepsi_{i}}{T}-\sigma_{i}^{2}\right)+k_{N,T}\frac{\overline{\vepsi}'\overline{\vepsi}}{T}\left(\sum_{i,j}^{N}\vepsi_{i}'\vepsi_{j}\right)+k_{N,T}\sum_{i,j}^{N}\vepsi_{i}'\vepsi_{j}\overline{\epsi}_{i}^{2}\\
		& +k_{N,T}\overline{\epsi}^{2}\sum_{i,j}^{N}\vepsi_{i}'\vepsi_{j}-2k_{N,T}\sum_{i,j}^{N}\vepsi_{i}'\vepsi_{j}\frac{\overline{\vepsi}'\vepsi_{i}}{T}-2k_{N,T}\sum_{i,j}^{N}\vepsi_{i}'\vepsi_{j}\overline{\epsi}_{i}\overline{\epsi}\\
		& =a_{1111}+a_{1112}+a_{1113}+a_{1114}-a_{1115}-a_{1116}.
	\end{align*}
	The leading term in the decomposition above is $a_{1111}$ for which
	results \eqref{eq:sumij_epsvar_epsvardev} and \eqref{eq:sumi_epsvar_epsvardev}
	can be used to arrive at 
	\begin{align*}
		a_{1111} & =k_{N,T}\sum_{i=1}^N \sum_{j\neq i}^{N}\vepsi_{i}'\vepsi_{j}\left(T^{-1}\vepsi_{i}'\vepsi_{i}-\sigma_{i}^{2}\right)+k_{N,T}\sum_{i=1}^{N}\vepsi_{i}'\vepsi_{i}\left(T^{-1}\vepsi_{i}'\vepsi_{i}-\sigma_{i}^{2}\right)\\
		& =\largeO_{P}\left(T^{-1}\sqrt{N}\right)+\largeO_{P}\left(T^{-1/2}\right)+\largeO_{P}\left(N^{-1/2}\right).
	\end{align*}
	Next, terms $a_{1112}$ and $a_{1114}$ can be combined and rewritten
	in order to yield
	\begin{align*}
		a_{1112}+a_{1114} & =k_{N,T}\left(\frac{\overline{\vepsi}'\overline{\vepsi}}{T}+\overline{\epsi}^{2}\right)N^{2}T\left(\frac{\overline{\vepsi}'\overline{\vepsi}}{T}\right)\\
		& =k_{N,T}\left\{ \largeO_{P}\left(N^{-1}\right)+\largeO_{P}\left[\left(NT\right)^{-1}\right]\right\} \largeO_{P}\left(NT\right)\\
		& =\largeO_{P}\left(N^{-1}\sqrt{T}\right)
	\end{align*}
	Expression $a_{1113}$ is non-negative since $\sum_{i,j}^{N}\vepsi_{i}'\vepsi_{j}$
	is and because $\overline{\epsi}_{i}^{2}$ does not change the
	sign of $\vepsi_{i}$. We can hence directly apply
	expectations without taking squares in beforehand. Hence, we get
	\begin{align*}
		\E\left(a_{1113}\right) & =k_{N,T}T^{-2}\E\left(\sum_{i,j}^{N}\sum_{t,t',t''}^{T}\epsi_{i,t}\epsi_{j,t}\epsi_{i,t'}\epsi_{i,t''}\right)\\
		& =k_{N,T}T^{-2}\sum_{i=1}^{N}\sum_{t,t'}^{T}\E\left(\epsi_{i,t}\right)\E\left(\epsi_{i,t'}\right)+\smallO\left(T^{-1/2}\right)\\
		& =\largeO\left(T^{-1/2}\right),
	\end{align*}
	which leads to $a_{1113}=\largeO_{P}\left(T^{-1/2}\right)$. The non-negativity
	holds as well for $a_{1115}=2k_{N,T}NT^{-1}\sum_{i=1}^{N}\overline{\vepsi}'\vepsi_{i}\vepsi_{i}'\overline{\vepsi}$,
	so that 
	\begin{align*}
		\E\left(a_{1115}\right) & =\frac{k_{N,T}}{NT}\sum_{i,j}^{N}\sum_{t=1}^{T}\E\left(\epsi_{i,t}^{2}\right)\E\left(\epsi_{j,t}^{2}\right)+\frac{k_{N,T}}{NT}\sum_{i=1}^{N}\sum_{t,t'}^{T}\E\left(\epsi_{i,t}^{2}\right)\E\left(\epsi_{j,t}^{2}\right)\\
		& +\smallO\left(T^{-1/2}\right)+\smallO\left(N^{-1}\sqrt{T}\right)\\
		& =\largeO\left(T^{-1/2}\right)+\largeO\left(N^{-1}\sqrt{T}\right)
	\end{align*}
	implies $a_{1115}=\largeO_{P}\left(T^{-1/2}\right)+\largeO_{P}\left(N^{-1}\sqrt{T}\right)$.
	Only in the case of $a_{1116}$ we have to take squares before applying
	expectations. Here, 
	\begin{align*}
		\E\left[\left(k_{N,T}\sum_{i,j}^{N}\vepsi_{i}'\vepsi_{j}\overline{\epsi}_{i}\right)^{2}\right] & =\frac{k_{N,T}^{2}}{N^{2}}\sum_{i,i',j}^{N}\sum_{t=1}^{T}\E\left(\epsi_{i,t}^{2}\right)\E\left(\epsi_{j,t}^{2}\right)\E\left(\epsi_{i',t}^{2}\right)+\frac{k_{N,T}^{2}}{N^{2}}\sum_{i=1}^{N}\sum_{t,t',t''}^{T}\E\left(\epsi_{i,t}^{2}\right)\E\left(\epsi_{i,t'}^{2}\right)\E\left(\epsi_{i,t''}^{2}\right)\\
		& +2\frac{k_{N,T}^{2}}{N^{2}}\sum_{i,i'}^{N}\sum_{t,t'}^{T}\E\left(\epsi_{i,t}^{2}\right)\E\left(\epsi_{i,t'}^{2}\right)\E\left(\epsi_{i',t'}^{2}\right)+\smallO\left(N^{-1}\right)+\smallO\left(N^{-3}T^{2}\right)+\smallO\left(N^{-2}T\right)\\
		& =\largeO\left(N^{-1}\right)+\largeO\left(N^{-3}T^{2}\right)+\largeO\left(N^{-2}T\right),
	\end{align*}
	so that we arrive at 
	\begin{align*}
		a_{1116} & =\largeO_{P}\left[\left(NT\right)^{-1/2}\right]\left[\largeO_{P}\left(N^{-1/2}\right)+\largeO_{P}\left(N^{-3/2}T\right)+\largeO_{P}\left(N^{-1}\sqrt{T}\right)\right]\\
		& =\largeO_{P}\left(N^{-1}T^{-1/2}\right)+\largeO_{P}\left(N^{-3/2}\right)+\largeO_{P}\left(N^{-2}\sqrt{T}\right).
	\end{align*}
	Combining our results on terms $a_{1111}$ to $a_{1116}$, we arrive
	at 
	\[
	k_{N,T}\sum_{i,j}^{N}\vepsi_{i}'\vepsi_{j}\left(\widehat{\sigma}_{i}^{2}-\sigma_{i}^{2}\right)=\mathbf{1}\left[\E\left(\epsi_{it}^{3}\right)\neq0\right]\largeO_{P}\left(T^{-1}\sqrt{N}\right)+\largeO_{P}\left(T^{-1/2}\right)+\largeO_{P}\left(N^{-1/2}\right)+\largeO_{P}\left(N^{-1}\sqrt{T}\right).
	\]
\end{proof}

\bigskip

\begin{lemma}
	\label{lem:sumij_epsbari_epsbarj_sig2_diff}Suppose that Assumption
	\ref{ass:errorsAddHetero}  holds. Then, 
	\[
	Tk_{N,T}\sum_{i=1}^{N}\overline{\epsi}_{i}\left(\widehat{\sigma}_{i}^{2}-\sigma_{i}^{2}\right)=\largeO_{P}\left(T^{-1/2}\right)+\largeO_{P}\left(N^{-1/2}T^{-1}\right)+\largeO_{P}\left(N^{-3/2}\right)
	\]
\end{lemma}
\begin{proof}[Proof of Lemma \ref{lem:sumij_epsbari_epsbarj_sig2_diff}]
	We use the decomposition of $\widehat{\sigma}_{i}^{2}$ in equation \eqref{eq:sigi2hat_2WFE}
	to write 
	\begin{align*}
		Tk_{N,T}\sum_{i=1}^{N}\overline{\epsi}_{i}\left(\widehat{\sigma}_{i}^{2}-\sigma_{i}^{2}\right) & =Tk_{N,T}\sum_{i=1}^{N}\overline{\epsi}_{i}\left(\frac{\vepsi_{i}'\vepsi_{i}}{T}-\sigma_{i}^{2}\right)+Tk_{N,T}\left(\frac{\overline{\vepsi}'\overline{\vepsi}}{T}+\overline{\epsi}^{2}\right)\sum_{i=1}^{N}\overline{\epsi}_{i}+Tk_{N,T}\sum_{i=1}^{N}\overline{\epsi}_{i}^{3}\\
		& -2Tk_{N,T}\sum_{i=1}^{N}\overline{\epsi}_{i}\frac{\overline{\vepsi}'\vepsi_{i}}{T}-2Tk_{N,T}\sum_{i=1}^{N}\overline{\epsi}_{i}^{2}\overline{\epsi}\\
		& =a_{1151}+a_{1152}+a_{1153}+a_{1154}-a_{1155}.
	\end{align*}
	Here, using result \eqref{eq:sumi_epsbari_epsivardev}, we directly
	get
	\begin{align*}
		a_{1151} & =Tk_{N,T}\sum_{i=1}^{N}\overline{\epsi}_{i}\left(\frac{\vepsi_{i}'\vepsi_{i}}{T}-\sigma_{i}^{2}\right)\\
		& = N^{-1}\sqrt{T}\largeO_{P}\left(NT^{-1}\right)\\
		& =\largeO_{P}\left(T^{-1/2}\right)
	\end{align*}
	Furthermore, since $T^{-1}\overline{\vepsi}'\overline{\vepsi}=\largeO_{P}\left(N^{-1}\right)$
	and $\overline{\epsi}=\largeO_{P}\left[\left(NT\right)^{-1/2}\right]$,
	we have 
	\begin{align*}
		a_{1152} & =Tk_{N,T}\left(\frac{\overline{\vepsi}'\overline{\vepsi}}{T}+\overline{\epsi}^{2}\right)\left(N\overline{\epsi}\right)\\
		& =N^{-1}\sqrt{T}\largeO_{P}\left(N^{-1}\right)\largeO_{P}\left(\sqrt{N}T^{-1/2}\right)\\
		& =\largeO_{P}\left(N^{-3/2}\right)
	\end{align*}
	For $a_{1153}$, we acknowledge that
	\begin{align*}
		\E\left[\left(\sum_{i=1}^{N}\overline{\epsi}_{i}^{3}\right)^{2}\right] & =T^{-6}\left[\sum_{i=1}^{N}\sum_{t=1}^{T}\E\left(\epsi_{i,t}^{3}\right)\right]^{2}+T^{-6}\sum_{i=1}^{N}\left[\sum_{t=1}^{T}\E\left(\epsi_{i,t}^{2}\right)\right]^{3}+\smallO\left(N^{2}T^{-4}\right)+\smallO\left(NT^{-3}\right)\\
		& =\largeO\left(N^{2}T^{-4}\right)+\largeO\left(NT^{-3}\right)
	\end{align*}
	in order to arrive at 
	\begin{align*}
		a_{1153} & =Tk_{N,T}\left[\largeO_{P}\left(NT^{-2}\right)+\largeO_{P}\left(\sqrt{N}T^{-3/2}\right)\right]\\
		& =\largeO_{P}\left(T^{-3/2}\right)+\largeO_{P}\left(N^{-1/2}T^{-1}\right).
	\end{align*}
	Next, we move on to $a_{1154}$ which we express as
	\begin{align*}
		a_{1155} & =2k_{N,T}\left(NT\right)^{-1}\sum_{i,i'}^{N}\sum_{t,t'}^{T}\epsi_{i,t}\epsi_{i',t'}\epsi_{i,t'}.
	\end{align*}
	Taking squares of the double sum over cross-sections and time and
	applying expectations, we get
	\begin{align*}
		\E\left[\left(\sum_{i,i'}^{N}\sum_{t,t'}^{T}\epsi_{i,t}\epsi_{i',t'}\epsi_{i,t'}\right)^{2}\right] & =\sum_{i,i',i''}^{N}\sum_{t=1}^{T}\E\left(\epsi_{i,t}^{2}\right)\E\left(\epsi_{i',t}^{2}\right)\E\left(\epsi_{i'',t}^{2}\right)+\sum_{i=1}^{N}\sum_{t,t',t''}^{T}\E\left(\epsi_{i,t}^{2}\right)\E\left(\epsi_{i,t'}^{2}\right)\E\left(\epsi_{i,t''}^{2}\right)\\
		& +\sum_{i,i'}^{N}\sum_{t,t'}^{T}\E\left(\epsi_{i,t}^{3}\right)\E\left(\epsi_{i',t'}^{3}\right)+\smallO\left(N^{3}T\right)+\smallO\left(NT^{3}\right)+\smallO\left[\left(NT\right)^{2}\right]\\
		& =\largeO\left(N^{3}T\right)+\largeO\left(NT^{3}\right)+\largeO\left[\left(NT\right)^{2}\right],
	\end{align*}
	implying that 
	\begin{align*}
		a_{1154} & =k_{N,T}\left(NT\right)^{-1}\left\{ \largeO_{P}\left(N^{3/2}\sqrt{T}\right)+\largeO_{P}\left(\sqrt{N}T^{3/2}\right)+\largeO_{P}\left(NT\right)\right\} .\\
		& =k_{N,T}\left\{ \largeO_{P}\left(\sqrt{N}T^{-1/2}\right)+\largeO_{P}\left(N^{-1/2}\sqrt{T}\right)+\largeO_{P}\left(1\right)\right\} \\
		& =\largeO_{P}\left(N^{-1/2}T^{-1}\right)+\largeO_{P}\left(N^{-3/2}\right)+\largeO_{P}\left(N^{-1}T^{-1/2}\right)
	\end{align*}
	Lastly, we have $a_{1155}=Tk_{N,T}\overline{\epsi}\sum_{i=1}^{N}\overline{\epsi}_{i}^{2}$
	where $\sum_{i=1}^{N}\overline{\epsi}_{i}^{2}=\largeO_{P}\left(NT^{-1}\right)$ by Markov's and Rosenthal's inequalities, leading us to $a_{1155}=\largeO_{P}\left(N^{-1/2}T^{-1}\right)$. Combining
	our results on terms $a_{1151}$ to $a_{1155}$, we can state that
	\begin{align*}
		Tk_{N,T}\sum_{i=1}^{N}\overline{\epsi}_{i}\left(\widehat{\sigma}_{i}^{2}-\sigma_{i}^{2}\right) & =\largeO_{P}\left(T^{-1/2}\right)+\largeO_{P}\left(N^{-1/2}T^{-1}\right)+\largeO_{P}\left(N^{-3/2}\right)
	\end{align*}
\end{proof}

\bigskip

\begin{lemma}
	\label{lem:inf_sighat2_2WFE}Suppose that Assumption \ref{ass:errorsAddHetero}  holds and assume
	that $N T^{-2}\to0$. Then there exists a constant $M>0$ such that
	$\inf_{i}\varsigma_{i}^{2}\geq M$ with probability approaching one.
\end{lemma}
\begin{proof}[Proof of Lemma \ref{lem:inf_sighat2_2WFE}]
	First, note that $\inf_{i}\varsigma_{i}^{2}\geq\inf_{i}\sigma_{i}^{2}-\sup_{i}\left|\widehat{\sigma}_{i}^{2}-\sigma_{i}^{2}\right|.$
	Here, $\sigma_{i}^{2}$ is bounded from below by Assumption \ref{ass:errorsAddHetero} . Concerning
	$\sup_{i}\left|\widehat{\sigma}_{i}^{2}-\sigma_{i}^{2}\right|$, we consider
	the tail probability
	\begin{align*}
		\Pr\left(\sup_{i}\left|\widehat{\sigma}_{i}^{2}-\sigma_{i}^{2}\right|>\epsilon\right) & \leq\sum_{i=1}^{N}\Pr\left(\left|\frac{\widehat{\vepsi}_{i}'\widehat{\vepsi}_{i}}{T}-\sigma_{i}^{2}\right|>\epsilon\right)\\
		& \leq\sum_{i=1}^{N}\Pr\left(\left|\frac{\vepsi_{i}'\vepsi_{i}}{T}-\sigma_{i}^{2}\right|>\frac{\epsilon}{6}\right)+\sum_{i=1}^{N}\Pr\left(\frac{\overline{\vepsi}'\overline{\vepsi}}{T}>\frac{\epsilon}{6}\right)+\sum_{i=1}^{N}\Pr\left(\overline{\epsi}_{i}^{2}>\frac{\epsilon}{6}\right)\\
		& +\sum_{i=1}^{N}\Pr\left(\overline{\epsi}^{2}>\frac{\epsilon}{6}\right)+\sum_{i=1}^{N}\Pr\left(\left|\frac{\vepsi_{i}'\overline{\vepsi}}{T}\right|>\frac{\epsilon}{6}\right)+\sum_{i=1}^{N}\Pr\left(\left|\overline{\epsi}_{i}\overline{\epsi}\right|>\frac{\epsilon}{6}\right)
	\end{align*}
	Sums over probabilities involving $\frac{\overline{\vepsi}'\overline{\vepsi}}{T}$,
	$\left|\frac{\vepsi_{i}'\overline{\vepsi}}{T}\right|$,$\overline{\epsi}^{2},\overline{\epsi}_{i}^{2}$
	and $\left|\overline{\epsi}_{i}\overline{\epsi}\right|$ can be shown
	to converge to zero via application of Chebyshev's inequality. For
	the remaining probability in terms of $\left|\frac{\vepsi_{i}'\vepsi_{i}}{T}-\sigma_{i}^{2}\right|$,
	we choose the power of four instead. Given that the random variable
	$\left(\epsi_{i,t}^{2}-\sigma_{i}^{2}\right)$ has zero expected
	value and is independent across cross-sections and (conditionally) time, by Rosenthals' inequality:
	\begin{equation*}
\E\left[T^{-1}\sum_{t=1}^{T}\left(\epsi_{i,t}^{2}-\sigma_{i}^{2}\right)\right]^{4}=\E(\sigma_{i}^{8})\E\left[T^{-1}\sum_{t=1}^{T}\left(\eta_{i,t}^{2}-1\right)\right]^{4}=\largeO(T^{-2})
	\end{equation*}
so that it holds that $\sum_{i=1}^{N}\left\{ \E\left[T^{-1}\sum_{t=1}^{T}\left(\epsi_{i,t}^{2}-\sigma_{i}^{2}\right)\right]^{4}\right\} =\largeO\left(NT^{-2}\right)$.
	Accordingly,
	\begin{align*}
		\sum_{i=1}^{N}\Pr\left(\left|\frac{\vepsi_{i}'\vepsi_{i}}{T}-\sigma_{i}^{2}\right|>\epsilon\right) & \leq\sum_{i=1}^{N}\frac{\E\left[\left(T^{-1}\vepsi_{i}'\vepsi_{i}-\sigma_{i}^{2}\right)^{4}\right]}{\epsilon^{4}}\\
		& =\frac{N}{T^{2}}\frac{1}{\epsilon^{4}}\frac{T^{2}}{N}\sum_{i=1}^{N}E\left\{ \left[T^{-1}\sum_{t=1}^{T}\left(\epsi_{i,t}^{2}-\sigma_{i}^{2}\right)\right]^{4}\right\} \\
		& =\largeO_{P}\left(N T^{-2}\right).
	\end{align*}
	Hence under the assumption that $N T^{-2}\to 0$, the above expression is $\smallO_{P}(1)$.  This result implies that the statement made in Lemma \ref{lem:inf_sighat2_2WFE} holds.
\end{proof}

\bigskip

\begin{lemma}
	\label{lem:Taylor_a21order}Under Assumption \ref{ass:errorsAddHetero} , 
	\[
	\sum_{i=1}^{N}\left(\sum_{j\neq i}^{N}\frac{\widehat{\vepsi}_{i}'\widehat{\vepsi}_{j}}{
	\sigma_{j}}\right)^{2}\left(\widehat{\sigma}_{i}^{2}-\sigma_{i}^{2}\right)^{2}=\largeO_{P}\left(N^{2}\right)+\largeO_{P}\left(NT\right)+\largeO_{P}\left(N^{-1}T^{2}\right)
	\]
\end{lemma}
\begin{proof}[Proof of Lemma \ref{lem:Taylor_a21order}]
	To begin with, note that $\sigma_{j}$ is bounded from below
	at a value greater than zero by Assumption \ref{ass:errorsAddHetero} . Hence, $\sigma_{j}^{-1}$
	is bounded from above by some finite number $M$. Accordingly, we
	can write the term
	\begin{align}
		& \sum_{i=1}^{N}\left(\sum_{j\neq i}^{N}\frac{\widehat{\vepsi}_{i}'\widehat{\vepsi}_{j}}{\sigma_{j}}\right)^{2}\left(\widehat{\sigma}_{i}^{2}-\sigma_{i}^{2}\right)^{2}\nonumber \\
		\leq & M\sum_{i=1}^{N}\left(\sum_{j\neq i}^{N}\vepsi_{i}'\vepsi_{j}\right)^{2}\left(\widehat{\sigma}_{i}^{2}-\sigma_{i}^{2}\right)^{2}+M\sum_{i=1}^{N}\left[\left(N-1\right)\left(\overline{\vepsi}'\overline{\vepsi}+T\overline{\epsi}^{2}\right)\right]^{2}\left(\widehat{\sigma}_{i}^{2}-\sigma_{i}^{2}\right)^{2} \nonumber \\
		+&M\sum_{i=1}^{N}\left(\sum_{j\neq i}^{N}T\overline{\epsi}_{i}\overline{\epsi}_{j}\right)^{2}\left(\widehat{\sigma}_{i}^{2}-\sigma_{i}^{2}\right)^{2} +M\sum_{i=1}^{N}\left[\left(N-1\right)\vepsi_{i}'\overline{\vepsi}\right]^{2}\left(\widehat{\sigma}_{i}^{2}-\sigma_{i}^{2}\right)^{2} \nonumber \\
		+& M\sum_{i=1}^{N}\left(\sum_{j\neq i}^{N}\overline{\vepsi}'\vepsi_{j}\right)^{2}\left(\widehat{\sigma}_{i}^{2}-\sigma_{i}^{2}\right)^{2}+M\sum_{i=1}^{N}\left[T\left(N-1\right)\overline{\epsi}_{i}\overline{\epsi}\right]^{2}\left(\widehat{\sigma}_{i}^{2}-\sigma_{i}^{2}\right)^{2}+M\sum_{i=1}^{N}\left(T\sum_{j\neq i}^{N}\overline{\epsi}_{j}\overline{\epsi}\right)^{2}\left(\widehat{\sigma}_{i}^{2}-\sigma_{i}^{2}\right)^{2}\nonumber \\
		=& a_{211}+a_{212}+a_{213}+a_{214}+a_{215}+a_{216}+a_{217}.\label{eq:a21}
	\end{align}
	The first out of seven terms in \eqref{eq:a21} can be decomposed
	into
	\begin{align}
		a_{211} & \leq M\sum_{i=1}^{N}\left(\sum_{j\neq i}^{N}\vepsi_{i}'\vepsi_{j}\right)^{2}\left(\frac{\vepsi_{i}'\vepsi_{i}}{T}-\sigma_{i}^{2}\right)^{2}+\sum_{i=1}^{N}\left(\sum_{j\neq i}^{N}\vepsi_{i}'\vepsi_{j}\right)^{2}\left(\frac{\overline{\vepsi}'\overline{\vepsi}}{T}+\overline{\epsi}^{2}\right)^{2}+\sum_{i=1}^{N}\left(\sum_{j\neq i}^{N}\vepsi_{i}'\vepsi_{j}\right)^{2}\overline{\epsi}_{i}^{4}\nonumber \\
		& +M\sum_{i=1}^{N}\left(\sum_{j\neq i}^{N}\vepsi_{i}'\vepsi_{j}\right)^{2}\left(\frac{\vepsi_{i}'\overline{\vepsi}}{T}\right)^{2}+M\sum_{i=1}^{N}\left(\sum_{j\neq i}^{N}\vepsi_{i}'\vepsi_{j}\right)^{2}\left(\overline{\epsi}_{i}\overline{\epsi}\right)^{2}\label{eq:a211}
	\end{align}
	For the first and fourth expression on the right-hand side above we
	use results \eqref{eq:sumij_espi_espj_epsisq-diff} and \eqref{eq:sumij_epsi_epsj_epsi_epsbar}.
	Concerning the second expression, we only need to find an upper bound
	for the non-negative term $\sum_{i=1}^{N}\left(\sum_{j\neq i}^{N}\vepsi_{i}'\vepsi_{j}\right)^{2}$.
	Taking expectations here, we get 
	\begin{align*}
		\E\left[\sum_{i=1}^{N}\left(\sum_{j\neq i}^{N}\vepsi_{i}'\vepsi_{j}\right)^{2}\right] & =\sum_{i=1}^{N}\sum_{j\neq i}^{N}\sum_{t=1}^{T}\E\left(\epsi_{i,t}^{2}\epsi_{j,t}^{2}\right)\\
		& =\largeO\left(N^{2}T\right).
	\end{align*}
	Given that $\frac{\overline{\vepsi}'\overline{\vepsi}}{T}=\largeO_{P}\left(N^{-1}\right)$
	and $\overline{\epsi}^{2}=\largeO_{P}\left[\left(NT\right)^{-1}\right]$,
	we get $\sum_{i=1}^{N}\left(\sum_{j\neq i}^{N}\vepsi_{i}'\vepsi_{j}\right)^{2}\left(\frac{\overline{\vepsi}'\overline{\vepsi}}{T}\right)^{2}=\largeO_{P}\left(T\right)$.
	Next, we consider 
	\[
	\sum_{i=1}^{N}\left(\sum_{j\neq i}^{N}\vepsi_{i}'\vepsi_{j}\right)^{2}\overline{\epsi}_{i}^{4}=T^{-4}\sum_{i=1}^{N}\sum_{j\neq i}^{N}\sum_{j'\neq i}^{N}\sum_{t,t',t'',t''',t'''',t''''''}^{T}\epsi_{i,t}\epsi_{j,t}\epsi_{j',t'}\epsi_{i,t'}\epsi_{i,t''}\epsi_{i,t'''}\epsi_{i,t''''}\epsi_{i,t'''''},
	\]
	whose expected value is 
	\begin{align*}
		\E\left[\sum_{i=1}^{N}\left(\sum_{j\neq i}^{N}\vepsi_{i}'\vepsi_{j}\right)^{2}\overline{\epsi}_{i}^{4}\right] & =T^{-4}\sum_{i=1}^{N}\sum_{j\neq i}^{N}\sum_{t,t'',t'''}^{T}\E\left(\epsi_{i,t}^{2}\right)\E\left(\epsi_{j,t}^{2}\right)\E\left(\epsi_{i,t''}^{2}\right)\E\left(\epsi_{i,t'''}^{2}\right)+\smallO\left(N^{2}T^{-1}\right)\\
		& =\largeO\left(N^{2}T^{-1}\right),
	\end{align*}
	which implies that $\sum_{i=1}^{N}\left(\sum_{j\neq i}^{N}\vepsi_{i}'\vepsi_{j}\right)^{2}\overline{\epsi}_{i}^{4}=\largeO_{P}\left(N^{2}T^{-1}\right)$.
	Lastly, the fifth term in \eqref{eq:a211} requires us to determine
	the order of 
	\[
	\sum_{i=1}^{N}\left(\sum_{j\neq i}^{N}\vepsi_{i}'\vepsi_{j}\right)^{2}\overline{\epsi}_{i}^{2}=T^{-2}\sum_{i=1}^{N}\sum_{j\neq i}^{N}\sum_{j'\neq i}^{N}\sum_{t,t',t'',t'''}^{T}\epsi_{i,t}\epsi_{j,t}\epsi_{j',t'}\epsi_{i,t'}\epsi_{i,t''}\epsi_{i,t'''}.
	\]
	Taking expectations, we get 
	\begin{align*}
		\E\left[\sum_{i=1}^{N}\left(\sum_{j\neq i}^{N}\vepsi_{i}'\vepsi_{j}\right)^{2}\overline{\epsi}_{i}^{2}\right] & =T^{-2}\sum_{i=1}^{N}\sum_{j\neq i}^{N}\sum_{t,t''}^{T}\E\left(\epsi_{i,t}^{2}\right)\E\left(\epsi_{j,t}^{2}\right)\E\left(\epsi_{i,t''}^{2}\right)+\smallO\left(N^{2}\right)\\
		& =\largeO\left(N^{2}\right).
	\end{align*}
	Additionally, recalling that $\overline{\epsi}^{2}=\largeO_{P}\left[\left(NT\right)^{-1}\right]$
	leads to $\sum_{i=1}^{N}\left(\sum_{j\neq i}^{N}\vepsi_{i}'\vepsi_{j}\right)^{2}\left(\overline{\epsi}_{i}\overline{\epsi}\right)^{2}=\largeO_{P}\left(NT^{-1}\right)$.
	
	Summarizing our results, we arrive at
	\begin{equation}
		a_{211}=\largeO_{P}\left(N^{2}\right)+\largeO_{P}\left(T\right)\label{eq:a211_order}
	\end{equation}
	
	The second out of seven terms in \eqref{eq:a21} is bounded from above
	by application of result \eqref{eq:sumi_sighat2_min_sig2_squared}.
	This allows us to write the term as
	\begin{align}
		a_{212} & \leq MN^{2}\left(\overline{\vepsi}'\overline{\vepsi}+T\overline{\epsi}^{2}\right)^{2}\sum_{i=1}^{N}\left(\widehat{\sigma}_{i}^{2}-\sigma_{i}^{2}\right)^{2}\nonumber \\
		& =\largeO_{P}\left(NT\right)+\largeO_{P}\left(N^{-1}T^{2}\right).\label{eq:a212_order}
	\end{align}
	
	We continue with the third out of seven terms in \eqref{eq:a21} which
	is written
	\begin{align*}
		a_{213} & \leq M\sum_{i=1}^{N}\left(\sum_{j\neq i}^{N}T\overline{\epsi}_{i}\overline{\epsi}_{j}\right)^{2}\left(\frac{\vepsi_{i}'\vepsi_{i}}{T}-\sigma_{i}^{2}\right)^{2}+\sum_{i=1}^{N}\left(\sum_{j\neq i}^{N}T\overline{\epsi}_{i}\overline{\epsi}_{j}\right)^{2}\left(\frac{\overline{\vepsi}'\overline{\vepsi}}{T}+\overline{\epsi}^{2}\right)^{2}+\sum_{i=1}^{N}\left(\sum_{j\neq i}^{N}T\overline{\epsi}_{i}\overline{\epsi}_{j}\right)^{2}\overline{\epsi}_{i}^{4}\\
		& +M\sum_{i=1}^{N}\left(\sum_{j\neq i}^{N}T\overline{\epsi}_{i}\overline{\epsi}_{j}\right)^{2}\left(\frac{\vepsi_{i}'\overline{\vepsi}}{T}\right)^{2}+M\sum_{i=1}^{N}\left(\sum_{j\neq i}^{N}T\overline{\epsi}_{i}\overline{\epsi}_{j}\right)^{2}\left(\overline{\epsi}_{i}\overline{\epsi}\right)^{2}\\
		& =a_{2131}+a_{2132}+a_{2133}+a_{2134}+a_{2135}.
	\end{align*}
	Application of results \eqref{eq:sumij_epsbari_epsbarj_epsisq-diff}
	and \eqref{eq:sumij_epsbari_epsbarj_epsi_epsbar} leads to $a_{2131}=\largeO_{P}\left(N^{2}T^{-1}\right)$
	and $a_{2134}=\largeO_{P}\left(NT^{-1}\right)+\largeO_{P}\left(1\right)$. In order to derive orders for the remaining terms in $a_{213}$, we implicitly use the fact that, for $r>2$, $\overline{\epsi}_i^r = \largeO\left( T^{-r/2} \right)$ by Rosenthal's inequality for independent random variables. For
	term $a_{2132}$, we note that independence of errors across cross-sections implies 
	\begin{align*}
		\E \left[\sum_{i=1}^{N}\left(\sum_{j\neq i}^{N}T\overline{\epsi}_{i}\overline{\epsi}_{j}\right)^{2}\right] &=   \sum_{i=1}^{N} \E \left(\overline{\epsi}_{i}^2 \right) \E\left[\left(\sum_{j\neq i}^{N} \sum_{t''=1}^T\epsi_{j,t''}\right)^{2}\right] \\
		&=\largeO_P \left(NT^{-1}\right)\largeO_P \left(NT\right) \\
		&=\largeO_{P}\left(N^2\right),
	\end{align*}
from which it follows that 
	\begin{align*}
		a_{2132} & =\largeO_{P}\left(N^2\right)\left\{ \largeO_{P}\left(N^{-1}\right)+\largeO_{P}\left[\left(NT\right)^{-1}\right]\right\} \\
		& =\largeO_{P}\left(N\right).
	\end{align*}
	By the same reasoning, we have 
	\begin{align*}
	\E\left(a_{2133}\right) &= \sum_{i=1}^{N}\E\left(\overline{\epsi}_{i}^{6}\right) \E\left[\left(\sum_{j\neq i}^{N} \sum_{t=1}^T\epsi_{j,t}\right)^{2}\right] \\
	&=\largeO \left(NT^{-3}\right)\largeO \left(NT\right) \\
	&=\largeO \left(N^2 T^{-2}\right),
	\end{align*}
	which leads to $a_{2133}=\largeO_{P}\left(N^{2}T^{-2}\right)$.
	Finally, concerning term $a_{2135}$ we can write  
	\begin{align*}
	\E \left(\sum_{i=1}^{N}\sum_{j\neq i}^{N} \overline{\epsi}_{i}^{4}\overline{\epsi}_{j}^{2}\right) &= T^2 \sum_{i=1}^{N}\E \left(\overline{\epsi}_{i}^{4}\right) \E \left[ \left(\sum_{j\neq i}^{N} \sum_{t=1}^{T}\epsi_{j,t}\right)^{2} \right] \\
	&= T^2 \largeO\left(NT^{-2}\right) \largeO\left(NT\right) \\
	&=  \largeO\left(N^2T^{-3}\right), 	
	\end{align*}
so that 
	\begin{align*}
	a_{2135} &= M T^2 \overline{\epsi}^2 \sum_{i=1}^{N} \overline{\epsi}_i^4 \left(\sum_{j\neq i}^{N}\overline{\epsi}_{j}^{2} \right)\\
		& =T^2 \largeO_{P}\left[\left( NT \right)^{-1}\right] \largeO_p\left(N^2T^{-3}\right) \\ 
		&=\largeO_{P}\left(NT^{-2}\right).
	\end{align*}
This result, together
	with the four previous ones, establishes 
	\[
	a_{213}=\largeO_{P}\left(N^{2}T^{-1}\right)+\largeO_{P}\left(N\right)
	\]
	\bigskip{}
	Concerning the fourth out of seven terms in \eqref{eq:a21}, we write 
	\begin{align*}
		a_{214}\leq & MN^{2}T^{-2}\sum_{i=1}^{N}\left(\vepsi_{i}'\overline{\vepsi}\right)^{2}\left(\vepsi_{i}'\vepsi_{i}-\sigma_{i}^{2}\right)^{2}+MN^{2}\sum_{i=1}^{N}\left(\vepsi_{i}'\overline{\vepsi}\right)^{2}\left(\frac{\overline{\vepsi}'\overline{\vepsi}}{T}+\overline{\epsi}^{2}\right)^{2}\\
		+ & MN^{2}\sum_{i=1}^{N}\left(\vepsi_{i}'\overline{\vepsi}\right)^{2}\overline{\epsi}_{i}^{4}+MN^{2}T^{-2}\sum_{i=1}^{N}\left(\vepsi_{i}'\overline{\vepsi}\right)^{4}+N^{2}\sum_{i=1}^{N}\left(\vepsi_{i}'\overline{\vepsi}\right)^{2}\overline{\epsi}_{i}^{2}\overline{\epsi}^{2}\\
		= & a_{2141}+a_{2142}+a_{2143}+a_{2144}+a_{2145}.
	\end{align*}
	Concerning $a_{2142}$, we use result \eqref{eq:sumit_eps_epsbar_sq}
	which allows us to write
	\begin{align*}
		a_{2142} & =N^{2}\left[\largeO_{P}\left(N^{-1}T^{2}\right)+\largeO_{P}\left(T\right)\right]\largeO_{P}\left(N^{-2}\right)\\
		& =\largeO_{P}\left(N^{-1}T^{2}\right)+\largeO_{P}\left(T\right).
	\end{align*}
	With regards to 
	\[
	a_{2141}=MT^{-2}\sum_{i,i',i''}^{N}\sum_{t,t',t'',t'''}^{T}\epsi_{i,t}\epsi_{i',t}\epsi_{i'',t'}\epsi_{i,t'}\left(\epsi_{i,t''}^{2}-\sigma_{i}^{2}\right)\left(\epsi_{i,t'''}^{2}-\sigma_{i}^{2}\right),
	\]
	we proceed by looking for index restrictions that result in the maximum
	number of sums over cross-sections or time, respectively. The restrictions
	$i'=i'',t=t'$ and $t''=t'''$ maximize the number of sums over $N$,
	leading to an expression of order $T^{-2}\largeO_{P}\left[\left(NT\right)^{2}\right].$
	Alternatively, the restrictions $i=i'=i''$ and $t''=t'''$ maximize
	the number of sums over $T$ and result in a term of order $T^{-2}\largeO_{P}\left(NT^{3}\right)$.
	Overall, we can hence state that. 
	\begin{align*}
		a_{2141}= & \largeO_{P}\left(N^{2}\right)+\largeO_{P}\left(NT\right).
	\end{align*}
	Next, a statement concerning 
	\[
	a_{2144}=M\left(NT\right)^{-2}\sum_{i,i',i'',i''',i'''''}^{N}\sum_{t,t',t'',t'''}^{T}\epsi_{i,t}\epsi_{i',t}\epsi_{i'',t'}\epsi_{i,t'}\epsi_{i,t''}\epsi_{i''',t''}\epsi_{i,t'''}\epsi_{i'''',t''}
	\]
	can be made by noting that the index restrictions $i'=i''$, $i'''=i''''$,
	$t=t'$ and $t''=t'''$ lead to one component with non-zero expected
	value consisting of three sums over cross-sections and two sums over
	time. The corresponding component is hence $\left(NT\right)^{-2}\largeO_{P}\left(N^{3}T^{2}\right)$.
	Additionally, the index restrictions $i=i'=i''=i'''=i''''$ lead to
	another component of order $\left(NT\right)^{-2}\largeO_{P}\left(NT^{4}\right).$
	Moreover, a third component of order $\left(NT\right)^{-2}\largeO_{P}\left(N^{2}T^{3}\right)$
	results from the index restrictions $i=i'''=i''''$, $i'=i''$ and
	$t=t'$. Hence, we can state that 
	\begin{align*}
		a_{2144}= & \largeO_{P}\left(N\right)+\largeO_{P}\left(T\right)+\largeO_{P}\left(N^{-1}T^{2}\right).
	\end{align*}
	We continue with the similar expression
	\[
	a_{2143}=MT^{-4}\sum_{i,i',i''}^{N}\sum_{t,t',t'',t''',t'''',t'''''}^{T}\epsi_{i,t}\epsi_{i',t}\epsi_{i'',t'}\epsi_{i,t'}\epsi_{i,t''}\epsi_{i,t'''}\epsi_{i,t''''}\epsi_{i,t'''''}.
	\]
	Applying expectations and looking for index restrictions which imply
	nonzero expected values, we obtain
	\begin{align*}
		\E\left(a_{2143}\right) & =MT^{-4}\sum_{i,i'}^{N}\sum_{t,t'',t'''}^{T}\E\left(\epsi_{i,t}^{2}\right)\E\left(\epsi_{i',t}^{2}\right)\E\left(\epsi_{i,t''}^{2}\right)\E\left(\epsi_{i,t'''}^{2}\right)+T^{-4}\sum_{i=1}^{N}\left[\sum_{t=1}^{T}\E\left(\epsi_{i,t}^{2}\right)\right]^{4}\\
		& +\smallO\left(N^{2}T^{-1}\right)+\smallO\left(N\right)\\
		& =\largeO\left(N^{2}T^{-1}\right)+\largeO\left(N\right),
	\end{align*}
	from which it follows that $a_{2143}=\largeO_{P}\left(N^{2}T^{-1}\right)+\largeO_{P}\left(N\right)$.
	Lastly, we focus on
	\[
	a_{2145}=\overline{\epsi}^{2}T^{-2}\sum_{i,i',i''}^{N}\sum_{t,t',t'',t'''}^{T}\epsi_{i,t}\epsi_{i',t}\epsi_{i'',t'}\epsi_{i,t'}\epsi_{i,t''}\epsi_{i,t'''}
	\]
	where 
	\begin{align*}
		&\E\left(T^{-2}\sum_{i,i',i''}^{N}\sum_{t,t',t'',t'''}^{T}\epsi_{i,t}\epsi_{i',t}\epsi_{i'',t'}\epsi_{i,t'}\epsi_{i,t''}\epsi_{i,t'''}\right)  \\ 
		 = &T^{-2}\sum_{i,i'}^{N}\sum_{t,t''}^{T}\E\left(\epsi_{i,t}^{2}\right)\E\left(\epsi_{i',t}^{2}\right)\E\left(\epsi_{i,t''}^{2}\right)\\
		 + &T^{-2}\sum_{i=1}^{N}\sum_{t,t',t''}^{T}\E\left(\epsi_{i,t}^{2}\right)\E\left(\epsi_{i,t'}^{2}\right)\E\left(\epsi_{i,t''}^{2}\right)+\smallO\left(N^{2}\right)+\smallO\left(NT\right)\\
		 = &\largeO\left(N^{2}\right)+\largeO\left(NT\right).
	\end{align*}
	Accordingly, it holds that
	\begin{align*}
		a_{2145} & =\largeO_{P}\left[\left(NT\right)^{-1}\right]\left[\largeO_{P}\left(N^{2}\right)+\largeO_{P}\left(NT\right)\right]\\
		& =\largeO_{P}\left(NT^{-1}\right)+\largeO_{P}\left(1\right).
	\end{align*}
	Using this result, we can now state that the fourth out of seven terms
	in \eqref{eq:a21} is 
	\begin{equation}
		a_{214}=\largeO_{P}\left(N^{2}\right)+\largeO_{P}\left(NT\right)+\largeO_{P}\left(N^{-1}T^{2}\right)\label{eq:Taylorapprox_term1_4order-1}
	\end{equation}
	
	We continue with the fifth out of seven terms in \eqref{eq:a21} which
	we bound from above by extending the sums over $j$ and $j'$ with
	the cases $j=i$ and $j'=i$. Accordingly, we can write
	\[
	a_{215}\leq M\left(NT\right)^{2}\left(\frac{\overline{\vepsi}'\overline{\vepsi}}{T}\right)^{2}\sum_{i=1}^{N}\left(\widehat{\sigma}_{i}^{2}-\sigma_{i}^{2}\right)^{2}
	\]
	Using result \eqref{eq:sumi_sighat2_min_sig2_squared} on the term
	$\sum_{i=1}^{N}\left(\widehat{\sigma}_{i}^{2}-\sigma_{i}^{2}\right)^{2}$,
	we get 
	\begin{align}
		a_{215}= & \left(NT\right)^{2}\largeO_{P}\left(N^{-2}\right)\left[\largeO_{P}\left(N/T\right)+\largeO_{P}\left(N^{-1}\right)\right]\nonumber \\
		= & \largeO_{P}\left(NT\right)+\largeO_{P}\left(N^{-1}T^{2}\right).\label{eq:a215_order}
	\end{align}
	The sixth out of seven terms in \eqref{eq:a21} is quite similar to
	the fourth term seen above. We can write 
	\begin{align*}
		a_{216} & \leq M\left(NT\right)^{2}\overline{\epsi}^{2}\sum_{i=1}^{N}\overline{\epsi}_{i}^{2}\left(\frac{\vepsi_{i}'\vepsi_{i}}{T}-\sigma_{i}^{2}\right)^{2}+M\left(NT\right)^{2}\left[\overline{\epsi}^{2}\left(\frac{\overline{\vepsi}'\overline{\vepsi}}{T}\right)^{2}+\overline{\epsi}^{6}\right]\sum_{i=1}^{N}\overline{\epsi}_{i}^{2}\\
		& +M\left(NT\right)^{2}\overline{\epsi}^{2}\sum_{i=1}^{N}\overline{\epsi}_{i}^{2}\left(\frac{\overline{\vepsi}'\vepsi_{i}}{T}\right)^{2}+M\left(NT\right)^{2}\overline{\epsi}^{2}\sum_{i=1}^{N}\overline{\epsi}_{i}^{6}\\
		& +M\left(NT\right)^{2}\overline{\epsi}^{4}\sum_{i=1}^{N}\overline{\epsi}_{i}^{4}\\
		& =a_{2161}+a_{2162}+a_{2163}+a_{2164}+a_{2165}.
	\end{align*}
	Concerning $a_{2161}$, we need to have a look at the expected value
	of $\sum_{i=1}^{N}\overline{\epsi}_{i}^{2}\left(\frac{\vepsi_{i}'\vepsi_{i}}{T}-\sigma_{i}^{2}\right)^{2}$.
	Here, 
	\begin{align*}
		\E\left[\sum_{i=1}^{N}\overline{\epsi}_{i}^{2}\left(\frac{\vepsi_{i}'\vepsi_{i}}{T}-\sigma_{i}^{2}\right)^{2}\right] & =2T^{-4}\sum_{i=1}^{N}\sum_{t,t'}^{T}\epsi_{i,t}^{2}\left(\epsi_{i,t'}^{2}-\sigma_{i}^{2}\right)^{2}+2T^{-4}\sum_{i=1}^{N}\sum_{t,t'}^{T}\epsi_{i,t}^{3}\epsi_{i,t'}^{3}+\largeO\left(NT^{-2}\right)\\
		& =\largeO\left(NT^{-2}\right),
	\end{align*}
	which implies $a_{2161}=\left(NT\right)^{2}\largeO_{P}\left[\left(NT\right)^{-1}\right]\largeO_{P}\left(NT^{-2}\right)=\largeO_{P}\left(N^{2}T^{-1}\right)$.
	For the next term, knowledge of $\sum_{i=1}^{N}\overline{\epsi}_{i}^{2}=\largeO_{P}\left(NT^{-1}\right)$
	allows us to arrive at $a_{2162}=\largeO_{P}\left(1\right).$ For term $a_{2163}$,
	we look at $\sum_{i=1}^{N}\overline{\epsi}_{i}^{2}\left(\frac{\overline{\vepsi}'\vepsi_{i}}{T}\right)^{2}$,
	which has expected value 
	\begin{align*}
		\E\left[\sum_{i=1}^{N}\overline{\epsi}_{i}^{2}\left(\frac{\overline{\vepsi}'\vepsi_{i}}{T}\right)^{2}\right] & =N^{-2}T^{-4}\sum_{i=1}^{N}\sum_{t,t',t''}^{T}\E\left(\epsi_{i,t}^{2}\right)\E\left(\epsi_{i,t'}^{2}\right)\E\left(\epsi_{i,t''}^{2}\right)+N^{-2}T^{-4}\sum_{i,i'}^{N}\sum_{t,t'}^{T}\E\left(\epsi_{i,t}^{2}\right)\E\left(\epsi_{i,t'}^{2}\right)\E\left(\epsi_{i',t'}^{2}\right)\\
		& +\smallO\left[\left(NT\right)^{-1}\right]+\smallO\left(T^{-2}\right)\\
		& =\largeO\left[\left(NT\right)^{-1}\right]+\largeO\left(T^{-2}\right).
	\end{align*}
	Accordingly, we have 
	\begin{align*}
		a_{2163} & =\left(NT\right)^{2}\largeO_{P}\left[\left(NT\right)^{-1}\right]\left\{ \largeO_{P}\left[\left(NT\right)^{-1}\right]+\largeO_{P}\left(T^{-2}\right)\right\} \\
		& =\largeO_{P}\left(1\right)+\largeO_{P}\left(NT^{-1}\right).
	\end{align*}
	The last two terms are very similar. Using Markov's inequality, it is straight forward to show that $\sum_{i=1}^{N}\overline{\epsi}_{i}^{4}=\largeO_{P}\left(NT^{-2}\right)$,
	and using the same reasoning we arrive at $\sum_{i=1}^{N}\overline{\epsi}_{i}^{6}=\largeO_{P}\left(NT^{-3}\right).$
	It follows that $a_{2164}=\largeO_{P}\left(N^{2}T^{-2}\right)$ and $a_{2165}=\largeO_{P}\left(NT^{-2}\right).$
	For the entire term $a_{216}$, we can hence state that 
	\[
	a_{216}=\largeO_{P}\left(N^{2}T^{-1}\right)+\largeO_{P}\left(1\right).
	\]
	The last building block needed for the central result of this lemma
	is term $a_{217}$. Similar to our treatment of term $a_{215},$we
	write 
	\begin{align}
		a_{217} & \leq\left(NT\right)^{2}M\overline{\epsi}^{4}\sum_{i=1}^{N}\left(\widehat{\sigma}_{i}^{2}-\sigma_{i}^{2}\right)^{2}\nonumber \\
		& =\largeO_{P}\left(N/T\right)+\largeO_{P}\left(N^{-1}\right),\label{eq:a217_order}
	\end{align}
	where the last line follows from result \eqref{eq:sumi_sighat2_min_sig2_squared}.
	Combining the results on terms \eqref{eq:a211_order} to \eqref{eq:a217_order},
	we arrive at the final result of this lemma, namely
	\[
	\sum_{i=1}^{N}\left(\sum_{j\neq i}^{N}\frac{\widehat{\vepsi}_{i}'\widehat{\vepsi}_{j}}{\sigma_{j}}\right)^{2}\left(\widehat{\sigma}_{i}^{2}-\sigma_{i}^{2}\right)^{2}=\largeO_{P}\left(N^{2}\right)+\largeO_{P}\left(NT\right)+\largeO_{P}\left(N^{-1}T^{2}\right).
	\]
\end{proof}

\subsection{Auxiliary lemmas for Section \ref{sec:var_estim_CCE}}
\begin{lemma}
	\label{lem:CCEterms_simple}Suppose that Assumptions \ref{ass:errors}-\ref{ass:rank} hold. Then,
	\begin{align}
		\left\Vert T^{-1}\overline{\mU}'\mM_{\widehat{\mF}}\overline{\mU}\right\Vert  & =\largeO_{P}\left(N^{-1}\right)+\largeO_{P}\left[\left(NT\right)^{-1/2}\right]\label{eq:Ubar_MF_Ubar}\\
		\sum_{i=1}^{N}\left\Vert T^{-1}\vlambda_{i}\vepsi_{i}'\overline{\mU}\right\Vert ^{2} & =\largeO_{P}\left(N^{-1}\right)+\largeO_{P}\left(T^{-1}\right)\label{eq:sumi_eps_ubar_lam}\\
		\sum_{i=1}^{N}\left\Vert T^{-1}\vepsi_{i}'\overline{\mU}\right\Vert ^{2} & =\largeO_{P}\left(N^{-1}\right)+\largeO_{P}\left(T^{-1}\right)\label{eq:sumi_eps_ubar}\\
		\sum_{i=1}^{N}\left\Vert T^{-1}\vlambda_{i}\vepsi_{i}'\widehat{\mF}\right\Vert ^{2} & =\largeO_{P}\left(NT^{-1}\right)+\largeO_{P}\left(N^{-1}\right)\label{eq:sumi_lam_eps_Fhat}\\
		\sum_{i=1}^{N}\left\Vert T^{-1}\vepsi_{i}'\widehat{\mF}\right\Vert ^{2} & =\largeO_{P}\left(NT^{-1}\right)+\largeO_{P}\left(N^{-1}\right)\label{eq:sumi_eps_Fhat}\\
		\sum_{i=1}^{N}\left\Vert T^{-1}\vepsi_{i}'\mP_{\widehat{\mF}}\vepsi_{i}\right\Vert ^{2} & =\largeO_{P}\left(NT^{-2}\right)+\largeO_{P}\left(N^{-2}\right)\label{eq:sumi_eps_PF_eps}\\
		\sum_{i=1}^{N}\left\Vert T^{-1}\overline{\mU}'\mM_{\widehat{\mF}}\vepsi_{i}\right\Vert ^{2} & =\largeO_{P}\left(N^{-1}\right)+\largeO_{P}\left(T^{-1}\right)\label{eq:sumi_Ubar_MF_epsi}\\
		\sum_{i=1}^{N}\left(T^{-1}\vlambda_{i}'\left(\overline{\mC}^{-1}\right)'\overline{\mU}'\mM_{\widehat{\mF}}\vepsi_{i}\right)^{2} & =\largeO_{P}\left(N^{-1}\right)+\largeO_{P}\left(T^{-1}\right)\label{eq:sumi_lami_Ubar_MF_epsi}\\
		T^{-2}\sum_{i,j}^{N}\left(\vlambda_{j}'\vlambda_{i}\right)^{2}\vepsi_{j}\overline{\mU}\overline{\mU}'\vepsi_{i}. & =\largeO_{P}\left(1\right)\label{eq:sq_sumi_epsi_Ubar}\\
		T^{-2}\sum_{i,j}^{N}\left(\vlambda_{j}'\vlambda_{i}\right)^{2}\vepsi'_{j}\widehat{\mF}\widehat{\mF}'\vepsi_{i} & =\largeO_{P}\left(NT^{-1}\right)+\largeO_{P}\left(1\right)\label{eq:sq_sumi_lami_epsi_Fhat}\\
		\left\Vert T^{-1}\overline{\mU}'\widehat{\mF}\right\Vert ^{2} & =\largeO_{P}\left[\left(NT\right)^{-1}\right]+\largeO_{P}\left(N^{-2}\right)\label{eq:normsq_Ubar_Fhat}\\
		\sum_{i=1}^{N}\left(\left\Vert T^{-1}\sum_{j\neq i}^{N}\vepsi_{i}'\vepsi_{j}\vepsi_{i}'\mF\right\Vert ^{2}\right) & =\largeO_{P}\left(N^{2}\right)\label{eq:sumi _sq_sumj_epsi_epsj_epsi_F}\\
		\sum_{i=1}^{N}\left(\left\Vert T^{-1}\sum_{j\neq i}^{N}\vepsi_{i}'\vepsi_{j}\vepsi_{i}'\overline{\mU}\right\Vert ^{2}\right) & =\largeO_{P}\left(N\right)+\largeO_{P}\left(T\right)\label{eq:sumi_sq_sumj_epsi_epsj_epsi_Ubar}
	\end{align}
\end{lemma}
\begin{proof}[Proof of Lemma \ref{lem:CCEterms_simple}]
	For the first result, note that 
	\begin{align*}
		\left\Vert T^{-1}\overline{\mU}'\mM_{\widehat{\mF}}\overline{\mU}\right\Vert  & \leq\left\Vert T^{-1}\overline{\mU}'\overline{\mU}\right\Vert +\left\Vert T^{-1}\overline{\mU}'\widehat{\mF}\right\Vert ^{2}\left\Vert \left(T^{-1}\widehat{\mF}'\widehat{\mF}\right)^{-1}\right\Vert \\
		& =\largeO_{P}\left(N^{-1}\right)+\largeO_{P}\left[\left(NT\right)^{-1/2}\right]
	\end{align*}
	by the results in Lemma \ref{lem:CCEbasicresults}. 
	
	In order to arrive at \eqref{eq:sumi_eps_ubar_lam}, we write
	\begin{align*}
		\sum_{i=1}^{N}\left\Vert T^{-1}\vlambda_{i}\vepsi_{i}'\overline{\mU}\right\Vert ^{2} & \leq T^{-2}\sum_{i=1}^{N}\vepsi_{i}'\overline{\mU}\overline{\mU}'\vepsi_{i}\left\Vert \vlambda_{i}\right\Vert ^{2}
	\end{align*}
	Taking expectations of this non-negative term, we obtain
	\begin{align*}
		\E\left(T^{-2}\sum_{i=1}^{N}\vepsi_{i}'\overline{\mU}\overline{\mU}'\vepsi_{i}\left\Vert \vlambda_{i}\right\Vert ^{2}\right) & =\left(NT\right)^{-2}\sum_{i,i'}^{N}\sum_{t=1}^{T}\E\left(\epsi_{i,t}^{2}\right)\E\left(\epsi_{i',t}^{2}\right)\E\left(\left\Vert \vlambda_{i}\right\Vert ^{2}\right) \\ &+\left(NT\right)^{-2}\sum_{i=1}^{N}\sum_{t,t'}^{T}\E\left(\epsi_{i,t}^{2}\right)\E\left(\epsi_{i,t'}^{2}\right)\E\left(\left\Vert \vlambda_{i}\right\Vert ^{2}\right)\\
		& +\left(NT\right)^{-2}\sum_{i,i'}^{N}\sum_{t=1}^{T}\E\left(\epsi_{i,t}^{2}\right)\E\left(\ve_{i',t}'\ve_{i',t}\right)\E\left(\left\Vert \vlambda_{i}\right\Vert ^{2}\right)+\largeO\left[\left(NT\right)^{-1}\right]\\
		& =\largeO\left(N^{-1}\right)+\largeO\left(T^{-1}\right),
	\end{align*}
	so that the required result follows by the Markov's inequality. The proof
	of \eqref{eq:sumi_eps_ubar} uses the same reasoning and is therefore
	omitted.
	
	For result \eqref{eq:sumi_lam_eps_Fhat}, we can expand the expression
	into
	\begin{align*}
		\sum_{i=1}^{N}\left\Vert T^{-1}\vlambda_{i}\vepsi_{i}'\widehat{\mF}\right\Vert ^{2} & \leq3\sum_{i=1}^{N}\left\Vert T^{-1}\frac{\vlambda_{i}\vepsi_{i}'\mF}{\left(\sigma_{i}^{2}\right)^{3/2}}\right\Vert ^{2}\left\Vert \overline{\mC}\right\Vert ^{2}+3\sum_{i=1}^{N}\left\Vert T^{-1}\frac{\vlambda_{i}\vepsi_{i}'\overline{\mU}}{\left(\sigma_{i}^{2}\right)^{3/2}}\right\Vert ^{2}.
	\end{align*}
	If we take expectations of the first part on the right-hand side above,
	we get
	\begin{align*}
		\E\left(\sum_{i=1}^{N}\left\Vert T^{-1}\frac{\vlambda_{i}\vepsi_{i}'\mF}{\left(\sigma_{i}^{2}\right)^{3/2}}\right\Vert ^{2}\right) & \leq T^{-2}\sum_{i=1}^{N}\sum_{t=1}^{T}\E\left(\left\Vert \vlambda_{i}\right\Vert ^{2}\right)\E\left(\epsi_{i,t}^{2}\right)\E\left(\mF_{t}'\mF_{t}\right)\\
		& =\largeO\left(NT^{-1}\right).
	\end{align*}
	Using this result together with Markov's inequality and result \eqref{eq:sumi_eps_ubar_lam},
	we arrive at 
	\[
	\sum_{i=1}^{N}\left\Vert T^{-1}\frac{\vlambda_{i}\vepsi_{i}'\widehat{\mF}}{\left(\sigma_{i}^{2}\right)^{3/2}}\right\Vert ^{2}=\largeO_{P}\left(NT^{-1}\right)+\largeO_{P}\left(N^{-1}\right).
	\]
	The proof of result \eqref{eq:sumi_eps_Fhat} is almost identical
	to that of \eqref{eq:sumi_lam_eps_Fhat} and therefore omitted. 
	
	We proceed with result \eqref{eq:sumi_eps_PF_eps}, where we write
	
	\[
	\sum_{i=1}^{N}\left\Vert T^{-1}\vepsi_{i}'\mP_{\widehat{\mF}}\vepsi_{i}\right\Vert ^{2}\leq\sum_{i=1}^{N}\left\Vert T^{-2}\vepsi_{i}'\widehat{\mF}\widehat{\mF}'\vepsi_{i}\right\Vert ^{2}\left\Vert \left(T^{-1}\widehat{\mF}'\widehat{\mF}\right)^{-1}\right\Vert .
	\]
	Furthermore, from Lemma \ref{lem:CCEbasicresults} $\left\Vert \left(\widehat{\mF}'\widehat{\mF}\right)^{-1}\right\Vert =\largeO_{P}(1)$. It remains to consider
	
	\begin{align*}
		\sum_{i=1}^{N}\left\Vert T^{-2}\vepsi_{i}'\widehat{\mF}\widehat{\mF}'\vepsi_{i}\right\Vert ^{2} & \leq3\sum_{i=1}^{N}\left\Vert T^{-2}\vepsi_{i}'\mF\mF'\vepsi_{i}\right\Vert ^{2}\left\Vert \overline{\mC}\right\Vert ^{4}+3\sum_{i=1}^{N}\left\Vert T^{-2}\vepsi_{i}'\overline{\mU}\overline{\mU}'\vepsi_{i}\right\Vert ^{2}.
	\end{align*}
	Here,
	\begin{align*}
		\sum_{i=1}^{N}\left\Vert T^{-2}\vepsi_{i}'\mF\mF'\vepsi_{i}\right\Vert ^{2} & =T^{-4}\sum_{i=1}^{N}\sum_{t,t',t'',t'''}^{T}\epsi_{i,t}\epsi_{i,t'}\epsi_{i,t''}\epsi_{i,t'''}\vf_{t}'\vf_{t'}\vf_{t''}'\vf_{t'''}\\
		& =\largeO_{P}\left(NT^{-2}\right)
	\end{align*}
	and 
	\begin{align*}
		\sum_{i=1}^{N}\left\Vert T^{-2}\vepsi_{i}'\overline{\mU}\overline{\mU}'\vepsi_{i}\right\Vert ^{2} & \leq\sum_{i=1}^{N}\left\Vert T^{-1/2}\vepsi_{i}\right\Vert ^{4}\left\Vert T^{-1/2}\overline{\mU}\right\Vert ^{4}\\
		& =T^{-4}\sum_{i=1}^{N}\sum_{t,t',t'',t'''}^{T}\epsi_{i,t}^{2}\epsi_{i,t'}^{2}\epsi_{i,t''}^{2}\epsi_{i,t'''}^{2}\left\Vert T^{-1/2}\overline{\mU}\right\Vert ^{4}\\
		& =\largeO_{P}\left(N^{-2}\right)
	\end{align*}
	follow from application of Markov's inequality and the elimination
	of terms with zero expected value. Accordingly, we have 
	\[
	\sum_{i=1}^{N}\left\Vert T^{-2}\vepsi_{i}'\widehat{\mF}\widehat{\mF}'\vepsi_{i}\right\Vert ^{2}=\largeO_{P}\left(NT^{-2}\right)+\largeO_{P}\left(N^{-2}\right),
	\]
	which establishes result \eqref{eq:sumi_eps_PF_eps}.
	
	Next, consider result \eqref{eq:sumi_Ubar_MF_epsi}. Here we use results
	\eqref{eq:sig2hat_min_sig2}, \eqref{eq:sumi_eps_ubar}, \eqref{eq:sumi_eps_Fhat} and \eqref{eq:normsq_Ubar_Fhat}
	to arrive at
	\begin{align*}
		\sum_{i=1}^{N}\left\Vert T^{-1}\overline{\mU}'\mM_{\widehat{\mF}}\vepsi_{i}\right\Vert ^{2} & \leq3\sum_{i=1}^{N}\left\Vert T^{-1}\overline{\mU}'\vepsi_{i}\right\Vert ^{2}+3\left\Vert \left(T^{-1}\widehat{\mF}'\widehat{\mF}\right)^{-1}\right\Vert ^{2}\left\Vert T^{-1}\overline{\mU}'\widehat{\mF}\right\Vert ^{2}\sum_{i=1}^{N}\left\Vert T^{-1}\widehat{\mF}'\vepsi_{i}\right\Vert ^{2}\\
		& =\largeO_{P}\left(N^{-1}\right)+\largeO_{P}\left(T^{-1}\right)\\
		&+\left\{ \largeO_{P}\left[\left(NT\right)^{-1}\right]+\largeO_{P}\left(N^{-2}\right)\right\} \left[\largeO_{P}\left(NT^{-1}\right)+\largeO_{P}\left(N^{-1}\right)\right]\\
		& =\largeO_{P}\left(N^{-1}\right)+\largeO_{P}\left(T^{-1}\right).
	\end{align*}
	The proof of result \eqref{eq:sumi_lami_Ubar_MF_epsi} is identical
	to that seen immediately above except for the use of results \eqref{eq:sumi_eps_ubar_lam}
	and \eqref{eq:sumi_lam_eps_Fhat} instead of \eqref{eq:sumi_eps_ubar}
	and \eqref{eq:sumi_eps_Fhat}.
	
	We move on with results that involve a double sum over cross-sections.
	To derive result \eqref{eq:sq_sumi_epsi_Ubar}, we take expectations
	in order to obtain 	
	\begin{align*}
		&T^{-2}\sum_{i,j}^{N}\E\left[\left(\vlambda_{j}'\vlambda_{i}\right)^{2}\right]\E\left[\vepsi_{j}'\overline{\mU}\overline{\mU}'\vepsi_{i}\right] \\
		 \leq & M T^{-2}\sum_{i,j}^{N}\E\left[\vepsi_{j}'\overline{\mU}\overline{\mU}'\vepsi_{i}\right]\\
		=&\left(NT\right)^{-2}\sum_{i,i',j,j'}^{N}\sum_{t,t'}^{T}\E\left(\epsi_{j,t}\epsi_{j',t}\epsi_{i',t'}\epsi_{i,t'}\right) \|\mB\|^2+\left(NT\right)^{-2}\sum_{i,i',j,j'}^{N}\sum_{t,t'}^{T}\E\left(\epsi_{j,t}\epsi_{i,t'}\right)\E\left(\ve'_{j',t}\ve_{i',t'}\right) \|\mB\|^2\\
		=&\left(NT\right)^{-2}\sum_{i,j}^{N}\sum_{t,t'}^{T}\E\left(\epsi_{j,t}^{2}\epsi_{i,t'}^{2}\right)\|\mB\|^2+\left(NT\right)^{-2}\sum_{i,i'}^{N}\sum_{t=1}^{T}\E\left(\epsi_{i,t}^{2}\right)\E\left(\ve'_{i',t}\ve_{i',t}\right)\|\mB\|^2,\\
		=&\largeO\left(1\right)+\largeO\left(T^{-1}\right),
	\end{align*}
	from which result \eqref{eq:sq_sumi_epsi_Ubar} directly follows.
	The result also constitutes one half of the required steps to prove
	result \eqref{eq:sq_sumi_lami_epsi_Fhat}. The corresponding expression
	can be written 
	\begin{align*}
		T^{-2}\sum_{i,j}^{N}\left(\vlambda_{j}'\vlambda_{i}\right)^{2}\vepsi'_{j}\widehat{\mF}\widehat{\mF}'\vepsi_{i} & \leq\frac{3}{T^{2}}\sum_{i,j}^{N}\left(\vlambda_{j}'\vlambda_{i}\right)^{2}\vepsi'_{j}\mF\mF'\vepsi_{i}\left\Vert \overline{\mC}\right\Vert ^{2}+\frac{3}{T^{2}}\sum_{i,j}^{N}\left(\vlambda_{j}'\vlambda_{i}\right)^{2}\vepsi'_{j}\overline{\mU}\overline{\mU}'\vepsi_{i},
	\end{align*}
	where we have already shown that the second term on the right-hand
	side is stochastically bounded. Taking expectations of the first term
	on the right-hand side above results in 
	\begin{align*}
		\E\left[T^{-2}\sum_{i,j}^{N}\sum_{t,t'}^{T}\left(\vlambda_{j}'\vlambda_{i}\right)^{2}\epsi_{j,t}\mF_{t}'\mF_{t'}\epsi_{i,t'}\right] & =T^{-2}\sum_{i=1}^{N}\sum_{t=1}^{T}\E\left(\left\Vert \vlambda_{i}\right\Vert ^{4}\right)\E\left(\epsi_{i,t}^{2}\right)\E\left(\mF_{t}'\mF_{t'}\right)\\
		& =\largeO\left(NT^{-1}\right),
	\end{align*}
	resulting in a corresponding order in probability for the term without
	expectations. Combining both result allows us to arrive at \eqref{eq:sq_sumi_lami_epsi_Fhat}.
	
	Result \eqref{eq:normsq_Ubar_Fhat} is proven in the same way as result
	\eqref{eq:sq_sumi_lami_epsi_Fhat}and a detailed proof is therefore
	omitted. To see the link between both terms, note that 
	\[
	\left\Vert T^{-1}\overline{\mU}'\widehat{\mF}\right\Vert ^{2}=\left(NT\right)^{-2}\sum_{i,j}^{N}\tr\left(\vu_{i}\widehat{\mF}\widehat{\mF}'\vu_{j}\right),
	\]
	where $\vu_{i}=\mB'[\epsi_{i,t},\ve_{i,t}]'$.
	
	Concerning result \eqref{eq:sumi _sq_sumj_epsi_epsj_epsi_F}, we write
	out the expectation of the norm as 
	\begin{align*}
		\E\left(\sum_{i=1}^{N}\left\Vert T^{-1}\sum_{j\neq i}^{N}\vepsi_{i}'\vepsi_{j}\vepsi_{i}'\mF\right\Vert ^{2}\right) & \leq MT^{-2}\sum_{i=1}^{N}\sum_{j\neq i}^{N}\sum_{j'\neq i}^{N}\sum_{t,t',t'',t'''}^{T}\E\left(\epsi_{i,t}\epsi_{j,t}\epsi_{i,t'}\epsi_{j',t'}\epsi_{i,t''}\epsi_{i,t'''}\right)\E\left(\vf_{t''}'\vf_{t'''}\right)\\`f
		& =MT^{-2}\sum_{i=1}^{N}\sum_{j\neq i}^{N}\sum_{t,t''}^{T}\E\left(\epsi_{i,t}^{2}\right)\E\left(\epsi_{j,t}^{2}\right)\E\left(\epsi_{i,t''}^{2}\right)\E\left(\left\Vert \vf_{t''}\right\Vert ^{2}\right)+\smallO\left(N^{2}\right)\\
		& =\largeO\left(N^{2}\right),
	\end{align*}
	from which result \eqref{eq:sumi _sq_sumj_epsi_epsj_epsi_F} directly
	follows.
	
	We proceed with result \eqref{eq:sumi_sq_sumj_epsi_epsj_epsi_Ubar},
	where we can expend the term under investigation into
	\[
	\sum_{i=1}^{N}\left\Vert T^{-1}\sum_{j\neq i}^{N}\vepsi_{i}'\vepsi_{j}\vepsi_{i}'\overline{\mU}\right\Vert ^{2} \leq M \sum_{i=1}^{N}\left\Vert T^{-1}\sum_{j\neq i}^{N}\vepsi_{i}'\vepsi_{j}\vepsi_{i}'\overline{\vepsi}\right\Vert ^{2}+ M \sum_{i=1}^{N}\left\Vert T^{-1}\sum_{j\neq i}^{N}\vepsi_{i}'\vepsi_{j}\vepsi_{i}'\overline{\ve}\right\Vert ^{2}.
	\]
	Here the order of the first result is determined by result \eqref{eq:sumij_epsi_epsj_epsi_epsbar}.
	Concerning the second result, an proof analogous to that of \eqref{eq:sumi _sq_sumj_epsi_epsj_epsi_F}
	allows us to show that $\sum_{i=1}^{N}\left\Vert T^{-1}\sum_{j\neq i}^{N}\left(\sigma_{j}^{2}\right)^{1/2}\vepsi_{i}'\vepsi_{j}\vepsi_{i}'\overline{\ve}\right\Vert ^{2}=\largeO_{P}\left(N\right)$.
	The combination of these two results establishes \eqref{eq:sumi_sq_sumj_epsi_epsj_epsi_Ubar}.
\end{proof}

\bigskip

\begin{lemma}
	\label{lem:sumi_epsi_Fhat_sigi2_diff}Suppose that Assumptions \ref{ass:errors}-\ref{ass:rank}
	hold. Then,
	\[
	\left\Vert T^{-1}\sum_{i=1}^{N}\vepsi_{i}'\widehat{\mF}\left(\widehat{\sigma}_{i}^{2}-\sigma_{i}^{2}\right)\right\Vert =\largeO_{P}\left(NT^{-3/2}\right)+\largeO_{P}\left(T^{-1}\sqrt{N}\right)+\largeO_{P}\left(N^{-1}\right)+\largeO_{P}\left(T^{-1/2}\right)
	\]
\end{lemma}
\begin{proof}[Proof of Lemma \ref{lem:sumi_epsi_Fhat_sigi2_diff}]
	We decompose the term under consideration into
	\begin{align*}
		\left\Vert T^{-1}\sum_{i=1}^{N}\left(\widehat{\sigma}_{i}^{2}-\sigma_{i}^{2}\right)\vepsi_{i}'\widehat{\mF}\right\Vert  & \leq\left\Vert T^{-2}\sum_{i=1}^{N}\left(\vepsi_{i}'\vepsi_{i}-\sigma_{i}^{2}\right)\vepsi_{i}'\widehat{\mF}\right\Vert +\left\Vert T^{-2}\sum_{i=1}^{N}\vlambda_{i}'\left(\overline{\mC}^{-1}\right)'\overline{\mU}'\mM_{\widehat{\mF}}\overline{\mU}\left(\overline{\mC}^{-1}\right)\vlambda_{i}\vepsi_{i}'\widehat{\mF}\right\Vert \\
		& +\left\Vert T^{-2}\sum_{i=1}^{N}\vepsi_{i}'\mP_{\widehat{\mF}}\vepsi_{i}\vepsi_{i}'\widehat{\mF}\right\Vert +\left\Vert 2T^{-2}\sum_{i=1}^{N}\vlambda_{i}'\left(\overline{\mC}^{-1}\right)'\overline{\mU}'\mM_{\widehat{\mF}}\vepsi_{i}\vepsi_{i}'\widehat{\mF}\right\Vert \\
		& =g_{21}+g_{22}+g_{23}+g_{24}.
	\end{align*}
	Consider now 
	\begin{align*}
		g_{21} & \leq\left\Vert T^{-2}\sum_{i=1}^{N}\left(\vepsi_{i}'\vepsi_{i}-\sigma_{i}^{2}\right)\vepsi_{i}'\mF\right\Vert \left\Vert \overline{\mC}\right\Vert +\left\Vert T^{-2}\sum_{i=1}^{N}\left(\vepsi_{i}'\vepsi_{i}-\sigma_{i}^{2}\right)\vepsi_{i}'\overline{\mU}\right\Vert \\
		& =g_{211}\left\Vert \overline{\mC}\right\Vert +g_{212}.
	\end{align*}
	We take squares of the first term in this decomposition and take expectations.
	This allows us to get 
	\begin{align*}
		\E\left(g_{211}^{2}\right) & \leq T^{-4}\sum_{i,j}^{N}\sum_{t,t',t'',t'''}^{T}\E\left[\left(\epsi_{i,t}^{2}-\sigma_{i}^{2}\right)\epsi_{i,t'}\epsi_{j,t''}\left(\epsi_{j,t'''}^{2}-\sigma_{j}^{2}\right)\right]\E\left(\vf_{t'}'\vf_{t''}\right),\\
		& =T^{-4}\sum_{i,j}^{N}\sum_{t,t''}^{T}\E\left[\left(\epsi_{i,t}^{2}-\sigma_{i}^{2}\right)\epsi_{i,t}\right]\E\left[\epsi_{j,t''}\left(\epsi_{j,t''}^{2}-\sigma_{j}^{2}\right)\right]\E\left(\vf_{t}'\vf_{t''}\right)\\
		& +T^{-4}\sum_{i=1}^{N}\sum_{t,t'}^{T}\E\left[\left(\epsi_{i,t}^{2}-\sigma_{i}^{2}\right)\left(\epsi_{j,t}^{2}-\sigma_{j}^{2}\right)\right]\E\left[\epsi_{i,t'}^{2}\right]\E\left(\vf_{t'}'\vf_{t''}\right)+\largeO\left(NT^{-2}\right)\\
		& =\largeO\left(NT^{-2}\right).
	\end{align*}
	Accordingly, $g_{211}=\largeO_{P}\left(T^{-1}\sqrt{N}\right)$ results from
	application of Chebyshev's inequality. When dealing with $g_{212},$we
	first split the entire expression using the Cauchy-Schwarz inequality and establish
	somewhat more conservative upper bounds on the resulting terms:
	\begin{align*}
		g_{212} & \leq\left(\sum_{i=1}^{N}\left\Vert T^{-1}\vepsi_{i}'\overline{\mU}\right\Vert ^{2}\right)^{1/2}\left(\sum_{i=1}^{N}\left\Vert T^{-1}\left(\vepsi_{i}'\vepsi_{i}-\sigma_{i}^{2}\right)\right\Vert ^{2}\right)^{1/2}\\
		& =\left[\largeO_{P}\left(N^{-1/2}\right)+\largeO_{P}\left(T^{-1/2}\right)\right]\largeO_{P}\left(\sqrt{N}T^{-1/2}\right)\\
		& =\largeO_{P}\left(T^{-1}\sqrt{N}\right)+\largeO_{P}\left(T^{-1/2}\right).
	\end{align*}
	Here, the second line result from \eqref{eq:sumi_eps_ubar} as well
	as a straightforward application of Markov's inequality. Given this
	result, we can state that 
	\[
	g_{21}=\largeO_{P}\left(T^{-1}\sqrt{N}\right)+\largeO_{P}\left(T^{-1/2}\right).
	\]
	We continue with the use of the Cauchy-Schwarz inequality to split up lengthy terms
	when handling the remaining terms $g_{22}$, $g_{23}$ and $g_{24}$.
	First, we write 
	\begin{align*}
		g_{22} & \leq\left(\sum_{i=1}^{N}\left\Vert T^{-1}\sum_{i=1}^{N}\vlambda_{i}\vepsi_{i}'\widehat{\mF}\right\Vert ^{2}\right)^{1/2}\left(\sum_{i=1}^{N}\left\Vert \vlambda_{i}\right\Vert ^{2}\right)^{1/2}\left\Vert \left(\overline{\mC}^{-1}\right)\right\Vert ^{2}\left\Vert T^{-1}\overline{\mU}'\mM_{\widehat{\mF}}\overline{\mU}\right\Vert \\
		& =\left[\largeO_{P}\left(\sqrt{N}T^{-1/2}\right)+\largeO_{P}\left(N^{-1/2}\right)\right]\largeO_{P}\left(\sqrt{N}\right)\left\{ \largeO_{P}\left(N^{-1}\right)+\largeO_{P}\left[\left(NT\right)^{-1/2}\right]\right\} \\
		& =\largeO_{P}\left(T^{-1}\sqrt{N}\right)+\largeO_{P}\left(T^{-1/2}\right)+\largeO_{P}\left(N^{-1}\right),
	\end{align*}
	which follows from results \eqref{eq:Ubar_MF_Ubar} and \eqref{eq:sumi_lam_eps_Fhat}.
	Next, we have
	\begin{align*}
		g_{23} & \leq\left(\sum_{i=1}^{N}\left\Vert T^{-1}\vepsi_{i}'\widehat{\mF}\right\Vert ^{2}\right)^{1/2}\left(\sum_{i=1}^{N}\left\Vert T^{-1}\vepsi_{i}'\mP_{\widehat{\mF}}\vepsi_{i}\right\Vert ^{2}\right)^{1/2}\\
		& =\left[\largeO_{P}\left(\sqrt{N}T^{-1/2}\right)+\largeO_{P}\left(N^{-1/2}\right)\right]\left[\largeO_{P}\left(T^{-1}\sqrt{N}\right)+\largeO_{P}\left(N^{-1}\right)\right]\\
		& =\largeO_{P}\left(NT^{-3/2}\right)+\largeO_{P}\left(T^{-1}\right)+\largeO_{P}\left[\left(NT\right)^{-1/2}\right]+\largeO_{P}\left(N^{-3/2}\right),
	\end{align*}
	which relies on results \eqref{eq:sumi_eps_Fhat} and \eqref{eq:sumi_eps_PF_eps}.
	Lastly, we consider 
	\begin{align*}
		g_{24} & =\left(\sum_{i=1}^{N}\left\Vert T^{-1}\vlambda_{i}\vepsi_{i}'\widehat{\mF}\right\Vert ^{2}\right)^{1/2}\left\Vert \overline{\mC}^{-1}\right\Vert \left(\sum_{i=1}^{N}\left\Vert T^{-1}\overline{\mU}'\mM_{\widehat{\mF}}\vepsi_{i}\right\Vert ^{2}\right)^{1/2}.
	\end{align*}
	The first term on the right-hand side above is $\largeO_{P}\left(\sqrt{N}T^{-1/2}\right)+\largeO_{P}\left(N^{-1/2}\right)$
	by result \eqref{eq:sumi_lam_eps_Fhat}. For the last term above we
	use result \eqref{eq:sumi_Ubar_MF_epsi} to show that its it is $\largeO_{P}\left(N^{-1/2}\right)+\largeO_{P}\left(T^{-1/2}\right).$
	Given these two intermediary results, we can conclude that 
	\begin{align*}
		g_{24} & =\left[\largeO_{P}\left(\sqrt{N}T^{-1/2}\right)+\largeO_{P}\left(N^{-1/2}\right)\right]\left[\largeO_{P}\left(N^{-1/2}\right)+\largeO_{P}\left(T^{-1/2}\right)\right]\\
		& =\largeO_{P}\left(T^{-1}\sqrt{N}\right)+\largeO_{P}\left(T^{-1/2}\right)+\largeO_{P}\left(N^{-1}\right),
	\end{align*}
	which is the last building block that we need to establish the main
	result of this lemma. Combining it with our results on $g_{21},g_{22}$,
	and $g_{23}$, we arrive at
	\[
	\left\Vert T^{-1}\sum_{i=1}^{N}\vepsi_{i}'\widehat{\mF}\left(\widehat{\sigma}_{i}^{2}-\sigma_{i}^{2}\right)\right\Vert =\largeO_{P}\left(NT^{-3/2}\right)+\largeO_{P}\left(T^{-1}\sqrt{N}\right)+\largeO_{P}\left(N^{-1}\right)+\largeO_{P}\left(T^{-1/2}\right)
	\]
\end{proof}
\begin{lemma}
	\label{lem:sumij_epsi_epsj-sigi2diff}Suppose that Assumptions \ref{ass:errors}-\ref{ass:rank}
	hold. Then,
	\[
	k_{N,T}\sum_{i=1}^N \sum_{j\neq i}^{N}\vepsi_{i}'\vepsi_{j}\left(\widehat{\sigma}_{i}^{2}-\sigma_{i}^{2}\right)=\largeO_{P}\left(\sqrt{N}T^{-1}\right)+\largeO_{P}\left(T^{-1/2}\right)+\largeO_{P}\left(N^{-1/2}\right)
	\]
\end{lemma}
\begin{proof}[Proof of Lemma \ref{lem:sumij_epsi_epsj-sigi2diff}]
	For notational simplicity, we abbreviate $\sum_{i=1}^N \sum_{j\neq i}^{N}$ by $\sum_{i \neq j}^{N}$ in this proof. We can expand the expression into 
	\begin{align*}
		& \left\Vert k_{N,T}\sum_{i\neq j}^{N}\vepsi_{i}'\vepsi_{j}\left(\widehat{\sigma}_{i}^{2}-\sigma_{i}^{2}\right)\right\Vert \\
		\leq & \left\Vert \frac{k_{N,T}}{T}\sum_{i\neq j}^{N}\vepsi_{i}'\vepsi_{j}\left(\vepsi_{i}'\vepsi_{i}-\sigma_{i}^{2}\right)\right\Vert + \left\Vert \frac{k_{N,T}}{T}\sum_{i\neq j}^{N}\vepsi_{i}'\vepsi_{j}\vlambda_{i}'\left(\overline{\mC}^{-1}\right)'\overline{\mU}'\mM_{\widehat{\mF}}\overline{\mU}\left(\overline{\mC}^{-1}\right)\vlambda_{i}\right\Vert .\\
		+ & \left\Vert \frac{k_{N,T}}{T}\sum_{i\neq j}^{N}\vepsi_{i}'\vepsi_{j}\vepsi_{i}'\mP_{\widehat{\mF}}\vepsi_{i}\right\Vert +\left\Vert \frac{2k_{N,T}}{T}\sum_{i\neq j}^{N}\vepsi_{i}'\vepsi_{j}\vlambda_{i}'\left(\overline{\mC}^{-1}\right)'\overline{\mU}'\mM_{\widehat{\mF}}\vepsi_{i}\right\Vert \\
		= & g_{31}+g_{32}+g_{33}+g_{34}
	\end{align*}
	In order to determine the order in probability of $g_{31}$ we can
	use result \eqref{eq:sumij_epsvar_epsvardev} in Lemma \ref{lem:buildingblocks_2WFE} which was established for the two-way fixed
	effects model but continues to apply. Hence
	\[
	g_{31}=\largeO_{P}\left(T^{-1/2}\right) + \largeO_{P}\left(T^{-1}\sqrt{N}\right).
	\]
	We continue with $g_{32}$ where we isolate the inner term $\left(\overline{\mC}^{-1}\right)'\overline{\mU}'\mM_{\widehat{\mF}}\overline{\mU}\left(\overline{\mC}^{-1}\right)$
	by the use of selection vectors for element $\ell$ out of $r$ element, formally denoted $\vs_{\ell_r}$:
	\begin{align*}
		\left|g_{32}\right| & \leq\frac{k_{N,T}}{T}\left(\sum_{\ell_{r},\ell'_{r}}^{r}\left\Vert \sum_{i\neq j}^{N}\vepsi_{i}'\vepsi_{j}\vlambda_{i}'\vs_{\ell_{r}}\vs_{\ell'_{r}}'\vlambda_{i}\right\Vert ^{2}\right)^{1/2}\left(\sum_{\ell_{r},\ell'_{r}}^{r}\left\Vert \vs_{\ell_{r}}'\left(\overline{\mC}^{-1}\right)'\overline{\mU}'\mM_{\widehat{\mF}}\overline{\mU}\left(\overline{\mC}^{-1}\right)\vs_{\ell'_{r}}\right\Vert ^{2}\right)^{1/2}\\
		& \leq k_{N,T} \left|\sum_{i\neq j}^{N}\vepsi_{i}'\vepsi_{j}\vlambda_{i}'\vlambda_{i}\right|\left\Vert \left(\overline{\mC}^{-1}\right)\right\Vert ^{2}\left\Vert T^{-1}\overline{\mU}'\mM_{\widehat{\mF}}\overline{\mU}\right\Vert .
	\end{align*}
	Here we square the term in absolute values and take expectations,
	resulting in
	\begin{align*}
		\E\left(\left|\sum_{i\neq j}^{N}\vepsi_{i}'\vepsi_{j}\vlambda_{i}'\vlambda_{i}\right|^{2}\right) & \leq\sum_{i,i'}^{N}\sum_{j\neq i}^{N}\sum_{j'\neq i'}^{N}\sum_{t,t'}^{T}\E\left(\epsi_{i,t}\epsi_{j,t}\epsi_{i't'}\epsi_{j',t'}\right)\E\left(\left\Vert \vlambda_{i}\right\Vert ^{4}\right)\\
		& =\sum_{i=1}^{N}\sum_{j\neq i}^{N}\sum_{t=1}^{T}\E\left(\epsi_{i,t}^{2}\right)\E\left(\epsi_{j,t}^{2}\right)\E\left(\left\Vert \vlambda_{i}\right\Vert ^{4}\right)\\
		& =\largeO\left(N^{2}T\right),
	\end{align*}
	so that $k_{N,T} \left|\sum_{i\neq j}^{N}\vepsi_{i}'\vepsi_{j}\vlambda_{i}'\vlambda_{i}\right|$ is
	stochastically bounded. Further using result \eqref{eq:Ubar_MF_Ubar},
	we arrive at 
	\[
	g_{32}=\largeO_{P}\left(N^{-1}\right)+\largeO_{P}\left[\left(NT\right)^{-1/2}\right].
	\]
	Next, consider 
	\[
	\left|g_{33}\right|\leq T k_{N,T} \left(\sum_{i=1}^{N}\left\Vert \sum_{j\neq i}^{N}T^{-1}\vepsi_{i}'\vepsi_{j}\right\Vert ^{2}\right)^{1/2}\left(\sum_{i=1}^{N}\left\Vert T^{-1}\vepsi_{i}'\mP_{\widehat{\mF}}\vepsi_{i}\right\Vert ^{2}\right)^{1/2}.
	\]
	Here we can apply result \eqref{eq:sumi_eps_PF_eps} to state that
	$\left(\sum_{i=1}^{N}\left\Vert T^{-1}\vepsi_{i}'\mP_{\widehat{\mF}}\vepsi_{i}\right\Vert ^{2}\right)^{1/2}=\largeO_{P}\left(T^{-1}\sqrt{N}\right)+\largeO_{P}\left(N^{-1}\right).$
	Additionally, we have 
	\begin{align*}
		\E\left[\sum_{i=1}^{N}\left\Vert \sum_{j\neq i}^{N}T^{-1}\vepsi_{i}'\vepsi_{j}\right\Vert ^{2}\right] & \leq T^{-2}\sum_{i=1}^{N}\sum_{j\neq i}^{N}\sum_{t=1}^{T}\E\left(\epsi_{i,t}^{2}\right)\E\left(\epsi_{j,t}^{2}\right)+0\\
		& =\largeO\left(N^{2}T^{-1}\right).
	\end{align*}
	This result allows us to arrive at 
	\begin{align*}
		\left|g_{33}\right| & =N^{-1}\sqrt{T} \largeO_{P}\left(NT^{-1/2}\right)\left[\largeO_{P}\left(T^{-1}\sqrt{N}\right)+\largeO_{P}\left(N^{-1}\right)\right]\\
		& =\largeO_{P}\left(T^{-1}\sqrt{N}\right)+\largeO_{P}\left(N^{-1}\right).
	\end{align*}
	
	Lastly, we have the term
	\[
	g_{34}\leq 2Tk_{N,T}\left[\sum_{i=1}^{N}\left\Vert T^{-1}\sum_{j\neq i}^{N}\vepsi_{i}'\vepsi_{j}\vlambda_{i}'\right\Vert ^{2}\right]^{1/2}\left\Vert \overline{\mC}^{-1}\right\Vert \left(\sum_{i=1}^{N}\left\Vert T^{-1}\overline{\mU}'\mM_{\widehat{\mF}}\vepsi_{i}\right\Vert ^{2}\right)^{1/2}.
	\]
	Concerning the right-hand side of the expression above, $\sum_{i=1}^{N}\left\Vert \sum_{j\neq i}^{N}T^{-1}\vepsi_{i}'\vepsi_{j}\vlambda_{i}'\right\Vert ^{2}=\largeO_{P}\left(N^{2}T^{-1}\right)$
	can be shown exactly as in the case of $g_{33}$ via application of
	Markov's inequality. The only extra requirement is $\E\left(\left\Vert \vlambda_{i}\right\Vert ^{2}\right)<\infty$,
	which is ensured by Assumption 4. Additionally, we use result \eqref{eq:sumi_Ubar_MF_epsi}
	to make a statement about the sum involving $\left\Vert T^{-1}\overline{\mU}'\mM_{\widehat{\mF}}\vepsi_{i}\right\Vert ^{2}$.
	Together with our previous result and the factor $k_{N,T}$ we can
	now state that
	\begin{align*}
		g_{34} & =N^{-1}\sqrt{T}\largeO_{P}\left(NT^{-1/2}\right)\left[\largeO_{P}\left(N^{-1/2}\right)+\largeO_{P}\left(T^{-1/2}\right)\right].\\
		& =\largeO_{P}\left(N^{-1/2}\right)+\largeO_{P}\left(T^{-1/2}\right)
	\end{align*}
	The result on $g_{34}$ is the last one that we require to establish
	an order in probability for the expression at the center of this lemma.
	Using the upper bounds derived for $g_{31},g_{32},g_{33}$ and $g_{34}$
	we conclude that 
	\[
	k_{N,T}\sum_{i\neq j}^{N}\frac{\vepsi_{i}'\vepsi_{j}}{\left(\sigma_{i}^{2}\right)^{3/2}\left(\sigma_{j}^{2}\right)^{1/2}}\left(\widehat{\sigma}_{i}^{2}-\sigma_{i}^{2}\right)=\largeO_{P}\left(T^{-1}\sqrt{N}\right)+\largeO_{P}\left(T^{-1/2}\right)+\largeO_{P}\left(N^{-1/2}\right)
	\]
	which was to be shown.
\end{proof}
\begin{lemma}
	\label{lem:sumi_sig2diff_sq_CCE}Suppose that Assumptions \ref{ass:errors}-\ref{ass:rank}.
	Then,
	\begin{align*}
		\sum_{i=1}^{N}\left(\widehat{\sigma}_{i}^{2}-\sigma_{i}^{2}\right)^{2} & =\largeO_{P}\left(N/T\right)+\largeO_{P}\left(N^{-1}\right)+\largeO_{P}\left(T^{-1}\right).\\
		\sum_{i=1}^{N}\left\Vert \left(\widehat{\sigma}_{i}^{2}-\sigma_{i}^{2}\right)\vlambda_{i}\right\Vert ^{2} & =\largeO_{P}\left(N/T\right)+\largeO_{P}\left(1\right)
	\end{align*}
\end{lemma}
\begin{proof}[Proof of Lemma \ref{lem:sumi_sig2diff_sq_CCE}]
	Using expression \eqref{eq:sigi2hat_CCE}, the left-hand side of the
	first result can be decomposed into 
	\begin{align*}
		\sum_{i=1}^{N}\left(\widehat{\sigma}_{i}^{2}-\sigma_{i}^{2}\right)^{2} & \leq5\sum_{i=1}^{N}\left(\frac{\vepsi_{i}'\vepsi_{i}}{T}-\sigma_{i}^{2}\right)^{2}+5T^{-2}\sum_{i=1}^{N}\left(\vlambda_{i}'\left(\overline{\mC}^{-1}\right)'\overline{\mU}'\mM_{\widehat{\mF}}\overline{\mU}\left(\overline{\mC}^{-1}\right)\vlambda_{i}\right)^{2}\\
		& +5T^{-2}\sum_{i=1}^{N}\left(\vepsi_{i}'\mP_{\widehat{\mF}}\vepsi_{i}\right)^{2}+10T^{-2}\sum_{i=1}^{N}\left(\vlambda_{i}'\left(\overline{\mC}^{-1}\right)'\overline{\mU}'\mM_{\widehat{\mF}}\vepsi_{i}\right)^{2}\\
		& =g_{41}+g_{42}+g_{43}+g_{44}.
	\end{align*}
	We begin with noting that $g_{41}=\largeO_{P}\left(NT^{-1}\right)+\largeO_{P}\left(N^{-1}\right)+\largeO_{P}\left(T^{-1}\right)$
	by Lemma \ref{lem:sigmahat} for the two-way fixed effects model. Next,
	we write term $g_{42}$ as 
	\begin{align*}
		g_{42} & \leq\left\Vert \left(\overline{\mC}^{-1}\right)\right\Vert ^{2}\left\Vert T^{-1}\overline{\mU}'\mM_{\widehat{\mF}}\overline{\mU}\right\Vert ^{2}\sum_{i=1}^{N}\left\Vert \vlambda_{i}\right\Vert ^{4}\\
		& =\largeO_{P}\left(N^{-1}\right)+\largeO_{P}\left(T^{-1}\right),
	\end{align*}
	where the order in probability holds due to result \eqref{eq:Ubar_MF_Ubar}
	and Assumption \ref{ass:loadings}. Lastly, results for the terms $g_{43}$ and $g_{44}$ are
	given by \eqref{eq:sumi_eps_PF_eps} and \eqref{eq:sumi_lami_Ubar_MF_epsi},
	respectively. Combining the results on $g_{41},g_{42},g_{43}$ and
	$g_{44}$, we arrive at the first result of this lemma. 
	
	Concerning the second result, we write
	\begin{align*}
		\sum_{i=1}^{N}\left\Vert \left(\widehat{\sigma}_{i}^{2}-\sigma_{i}^{2}\right)\vlambda_{i}\right\Vert ^{2} & \leq5\sum_{i=1}^{N}\left(\frac{\vepsi_{i}'\vepsi_{i}}{T}-\sigma_{i}^{2}\right)^{2}\left\Vert \vlambda_{i}\right\Vert ^{2}+5T^{-2}\sum_{i=1}^{N}\left(\vlambda_{i}'\left(\overline{\mC}^{-1}\right)'\overline{\mU}'\mM_{\widehat{\mF}}\overline{\mU}\left(\overline{\mC}^{-1}\right)\vlambda_{i}\right)^{2}\left\Vert \vlambda_{i}\right\Vert ^{2}\\
		& +5T^{-2}\sum_{i=1}^{N}\left(\vepsi_{i}'\mP_{\widehat{\mF}}\vepsi_{i}\right)^{2}\left\Vert \vlambda_{i}\right\Vert ^{2}+10T^{-2}\sum_{i=1}^{N}\left(\vlambda_{i}'\left(\overline{\mC}^{-1}\right)'\overline{\mU}'\mM_{\widehat{\mF}}\vepsi_{i}\right)^{2}\left\Vert \vlambda_{i}\right\Vert ^{2}\\
		& =g_{45}+g_{46}+g_{47}+g_{48}.
	\end{align*}
	The order of $g_{45}$ is the same as that of $g_{41}$. This can
	be seen by noting that \[\E\left(g_{45}\right)=\sum_{i=1}^{N}\E\left[\left(\frac{\vepsi_{i}'\vepsi_{i}}{T}-\sigma_{i}^{2}\right)^{2}\right]\left(\E\left\Vert \vlambda_{i}\right\Vert ^{2}\right),\]
	which extends the expectation that is used in the proof of Lemma \ref{lem:sigmahat} for the two-way fixed effects model simply by $\E[\left\Vert \vlambda_{i}\right\Vert ^{2}]$.
	Since the second moments of factor loadings are assumed to be finite
	by assumption, their presence does not affect the order in probability.
	Hence, $g_{45}=\largeO_{P}\left(NT^{-1}\right)+\largeO_{P}\left(N^{-1}\right)+\largeO_{P}\left(T^{-1}\right)$
	holds. The same conclusion can be drawn for terms $g_{47}$ and $g_{48}$,
	where the proofs of terms \eqref{eq:sumi_eps_PF_eps} and \eqref{eq:sumi_lami_Ubar_MF_epsi}
	are merely extended with an additional term $\vlambda_{i}$
	in the sum over cross-sections.  For term $g_{46},$proceeding in
	the same way would require us to assume finiteness of $\E\left(\left\Vert \vlambda_{i}\right\Vert ^{6}\right)$,
	so we choose a more conservative bound that is valid with finite fourth
	moments of $\vlambda_{i}$. Since $\sum_{i=1}^{N}\left\Vert \vlambda_{i}\right\Vert ^{4}\left\Vert \vlambda_{i}\right\Vert ^{2}\leq\sum_{i=1}^{N}\left\Vert \vlambda_{i}\right\Vert ^{4}\sup_{i}\left\Vert \vlambda_{i}\right\Vert ^{2}\leq\left(\sum_{i=1}^{N}\left\Vert \vlambda_{i}\right\Vert ^{4}\right)\left(\sum_{i=1}^{N}\left\Vert \vlambda_{i}\right\Vert ^{2}\right)$,
	we get 
	\begin{align*}
		g_{46} & \leq\left\Vert \left(\overline{\mC}^{-1}\right)\right\Vert \left\Vert T^{-1}\overline{\mU}'\mM_{\widehat{\mF}}\overline{\mU}\right\Vert ^{2}\sum_{i=1}^{N}\left\Vert \vlambda_{i}\right\Vert ^{4}\sum_{i=1}^{N}\left\Vert \vlambda_{i}\right\Vert ^{2}\\
		& =\largeO_{P}\left(1\right)+\largeO_{P}\left(NT^{-1}\right).
	\end{align*}
	Combining this result with the previous three leads to 
	\[
	\sum_{i=1}^{N}\left\Vert \left(\widehat{\sigma}_{i}^{2}-\sigma_{i}^{2}\right)\vlambda_{i}\right\Vert ^{2}=\largeO_{P}\left(1\right)+\largeO_{P}\left(NT^{-1}\right).
	\]
\end{proof}
\begin{lemma}
	\label{lem:sumi_sigi2diff_epsi_Ubar}Suppose that Assumptions \ref{ass:errors}-\ref{ass:rank}
	hold. Then,
	\[
	\sum_{i=1}^{N}\left\Vert T^{-1}\left(\widehat{\sigma}_{i}^{2}-\sigma_{i}^{2}\right)\vepsi_{i}'\overline{\mU}\right\Vert ^{2}=\largeO_{P}\left(NT^{-3}\right)+\largeO_{P}\left(T^{-2}\right)+\largeO_{P}\left[\left(NT\right)^{-1}\right]+\largeO_{P}\left(N^{-2}\right)
	\]
\end{lemma}
\begin{proof}[Proof of Lemma \ref{lem:sumi_sigi2diff_epsi_Ubar}]
	This term is given by 
	\begin{align*}
		\sum_{i=1}^{N}\left\Vert T^{-1}\left(\widehat{\sigma}_{i}^{2}-\sigma_{i}^{2}\right)\vepsi_{i}'\overline{\mU}\right\Vert ^{2} & \leq\sum_{i=1}^{N}\left\Vert T^{-1}\left(\frac{\vepsi_{i}'\vepsi_{i}}{T}-\sigma_{i}^{2}\right)\vepsi_{i}'\overline{\mU}\right\Vert ^{2}+\sum_{i=1}^{N}\left\Vert T^{-2}\vlambda_{i}'\left(\overline{\mC}^{-1}\right)'\overline{\mU}'\mM_{\widehat{\mF}}\overline{\mU}\left(\overline{\mC}^{-1}\right)\vlambda_{i}\vepsi_{i}'\overline{\mU}\right\Vert ^{2}\\
		& +4\sum_{i=1}^{N}\left\Vert T^{-2}\vepsi_{i}'\mP_{\widehat{\mF}}\vepsi_{i}\vepsi_{i}'\overline{\mU}\right\Vert ^{2}+\sum_{i=1}^{N}\left\Vert T^{-2}\vlambda_{i}'\left(\overline{\mC}^{-1}\right)'\overline{\mU}'\mM_{\widehat{\mF}}\vepsi_{i}\vepsi_{i}'\overline{\mU}\right\Vert ^{2}\\
		& =g_{51}+g_{52}+g_{53}+g_{54}
	\end{align*}
	Consider the first term which we can decompose into
	\begin{align*}
		g_{51} & =N^{-2}T^{-4}\sum_{i,j,j'}^{N}\sum_{t,t',t'',t'''}^{T}\left(\epsi_{i,t}^{2}-\sigma_{i}^{2}\right)\epsi_{i,t'}\epsi_{j,t'}\epsi_{j,t''}\epsi_{i,t''}\left(\epsi_{i,t'''}^{2}-\sigma_{i}^{2}\right)\\
		& +N^{-2}T^{-4}\sum_{i,j,j'}^{N}\sum_{t,t',t'',t'''}^{T}\left(\epsi_{i,t}^{2}-\sigma_{i}^{2}\right)\epsi_{i,t'}\epsi_{i,t''}\left(\epsi_{i,t'''}^{2}-\sigma_{i}^{2}\right)\ve_{j,t'}'\ve_{j',t''}\\
		& =g_{511}+g_{512}.
	\end{align*}
	Taking expectations of $g_{511}$, we can characterize it as follows
	in terms of its leading components:
	\begin{align*}
		\E\left(g_{511}\right) & = \E\left[N^{-2}T^{-4}\sum_{i,j,j'}^{N}\sum_{t,t',t'',t'''}^{T}\left(\epsi_{i,t}^{2}-\sigma_{i}^{2}\right)\epsi_{i,t'}\epsi_{j,t'}\epsi_{j',t''}\epsi_{i,t''}\left(\epsi_{i,t'''}^{2}-\sigma_{i}^{2}\right)\right]\\
		& = N^{-2}T^{-4}\sum_{i=1}^{N}\sum_{t,t',t''}^{T}\E\left[\left(\epsi_{i,t}^{2}-\sigma_{i}^{2}\right)^{2}\right]\E\left(\epsi_{i,t'}^{2}\right)\E\left(\epsi_{i,t''}^{2}\right)  \\
		& + N^{-2}T^{-4}\sum_{i,j}^{N}\sum_{t,t'}^{T}\E\left[\left(\epsi_{i,t}^{2}-\sigma_{i}^{2}\right)^{2}\right] \E\left(\epsi_{j,t'}^{2}\right)\E\left(\epsi_{i,t'}^{2}\right)+ \smallO\left[\left(NT\right)^{-1}\right]+\smallO\left(T^{-2}\right)\\
		& = \largeO\left[\left(NT\right)^{-1}\right]+\largeO\left(T^{-2}\right),
	\end{align*}
	so that $g_{511}=\largeO_{P}\left[\left(NT\right)^{-1}\right]+\largeO_{P}\left(T^{-2}\right)$.
	Using an analogous argument, it can be shown that $g_{512}=\largeO_{P}\left(T^{-2}\right)$.
	Accordingly, it holds that 
	\[
	g_{51}=\largeO_{P}\left[\left(NT\right)^{-1}\right]+\largeO_{P}\left(T^{-2}\right).
	\]
	
	We continue with $g_{52}$ to which we apply results \eqref{eq:Ubar_MF_Ubar}
	and \eqref{eq:sumi_eps_ubar_lam} to arrive at
	\begin{align*}
		g_{52} & \leq\left\Vert \overline{\mC}^{-1}\right\Vert ^{4}\left\Vert T^{-1}\overline{\mU}'\mM_{\widehat{\mF}}\overline{\mU}\right\Vert ^{2}\left(\sum_{i=1}^{N}\left\Vert \vlambda_{i}'\right\Vert ^{2}\right)\left(T^{-2}\sum_{i=1}^{N}\left\Vert \vlambda_{i}\vepsi_{i}'\overline{\mU}\right\Vert ^{2}\right)\\
		& =\left\{ \largeO_{P}\left(N^{-2}\right)+\largeO_{P}\left[\left(NT\right)^{-1}\right]\right\} \largeO_{P}\left(N\right)\left[\largeO_{P}\left(N^{-1}\right)+\largeO_{P}\left(T^{-1}\right)\right]\\
		& =\largeO_{P}\left(N^{-2}\right)+\largeO_{P}\left[\left(NT\right)^{-1}\right]+\largeO_{P}\left(T^{-2}\right).
	\end{align*}
	Concerning term $g_{53}$, we write 
	\begin{align*}
		g_{53} & \leq4\left(\sum_{i=1}^{N}\left\Vert T^{-1}\vepsi_{i}'\mP_{\widehat{\mF}}\vepsi_{i}\right\Vert ^{2}\right)\left(T^{-2}\sum_{i=1}^{N}\left\Vert \vepsi_{i}'\overline{\mU}\right\Vert ^{2}\right)\\
		& =\left[\largeO_{P}\left(NT^{-2}\right)+\largeO_{P}\left(N^{-2}\right)\right]\left[\largeO_{P}\left(N^{-1}\right)+\largeO_{P}\left(T^{-1}\right)\right]\\
		& =\largeO_{P}\left(NT^{-3}\right)+\largeO_{P}\left(T^{-2}\right)+\largeO_{P}\left(N^{-2}T^{-1}\right)+\largeO_{P}\left(N^{-3}\right),
	\end{align*}
	where results \eqref{eq:sumi_eps_ubar} and \eqref{eq:sumi_eps_PF_eps}
	are used to arrive at the second line. Lastly, via results \eqref{eq:sumi_eps_ubar_lam}
	and \eqref{eq:sumi_Ubar_MF_epsi} we arrive at
	\begin{align*}
		g_{54} & =\left\Vert \overline{\mC}^{-1}\right\Vert ^{2}\left(T^{-2}\sum_{i=1}^{N}\left\Vert \vlambda_{i}\vepsi_{i}'\overline{\mU}\right\Vert ^{2}\right)\left(\sum_{i=1}^{N}\left\Vert T^{-1}\overline{\mU}'\mM_{\widehat{\mF}}\vepsi_{i}\right\Vert ^{2}\right)\\
		& =\left[\largeO_{P}\left(N^{-1}\right)+\largeO_{P}\left(T^{-1}\right)\right]\left[\largeO_{P}\left(N^{-1}\right)+\largeO_{P}\left(T^{-1}\right)\right]\\
		& =\largeO_{P}\left(N^{-2}\right)+\largeO_{P}\left[\left(NT\right)^{-1}\right]+\largeO_{P}\left(T^{-2}\right).
	\end{align*}
	Combining the last four intermediary results, we hence arrive at 
	\[
	\sum_{i=1}^{N}\left\Vert T^{-1}\left(\widehat{\sigma}_{i}^{2}-\sigma_{i}^{2}\right)\vepsi_{i}'\overline{\mU}\right\Vert ^{2}=\largeO_{P}\left(NT^{-3}\right)+\largeO_{P}\left(T^{-2}\right)+\largeO_{P}\left[\left(NT\right)^{-1}\right]+\largeO_{P}\left(N^{-2}\right).
	\]
\end{proof}
\begin{lemma}
	\label{lem:sumi_sigi2diff_epsi_Fhat}Suppose that Assumptions \ref{ass:errors}-\ref{ass:rank}
	hold. Then,
	\[
	\sum_{i=1}^{N}\left\Vert T^{-1}\left(\widehat{\sigma}_{i}^{2}-\sigma_{i}^{2}\right)\vepsi_{i}'\widehat{\mF}\right\Vert ^{2}=\largeO_{P}\left(NT^{-2}\right)+\largeO_{P}\left(T^{-1}\right)+\largeO_{P}\left(N^{-2}\right).
	\]
\end{lemma}
\begin{proof}[Proof of Lemma \ref{lem:sumi_sigi2diff_epsi_Fhat}]
	The expression of interest can bounded from above by 
	\[
	\sum_{i=1}^{N}\left\Vert T^{-1}\left(\widehat{\sigma}_{i}^{2}-\sigma_{i}^{2}\right)\vepsi_{i}'\widehat{\mF}\right\Vert ^{2}\leq\sum_{i=1}^{N}\left\Vert T^{-1}\left(\widehat{\sigma}_{i}^{2}-\sigma_{i}^{2}\right)\vepsi_{i}'\mF\right\Vert ^{2}\left\Vert \overline{\mC}\right\Vert ^{2}+\sum_{i=1}^{N}\left\Vert T^{-1}\left(\widehat{\sigma}_{i}^{2}-\sigma_{i}^{2}\right)\vepsi_{i}'\overline{\mU}\right\Vert ^{2},
	\]
	where we can use Lemma \ref{lem:sumi_sigi2diff_epsi_Ubar} for the
	second term on the right-hand side above. It remains to establish
	a corresponding order in probability for
	\begin{align*}
		&\sum_{i=1}^{N}\left\Vert T^{-1}\left(\widehat{\sigma}_{i}^{2}-\sigma_{i}^{2}\right)\vepsi_{i}'\mF\right\Vert ^{2} \\
		\leq & T^{-2}\sum_{i=1}^{N}\left\Vert \left(\frac{\vepsi_{i}'\vepsi_{i}}{T}-\sigma_{i}^{2}\right)\vepsi_{i}'\mF\right\Vert ^{2}+T^{-2}\sum_{i=1}^{N}\left\Vert T^{-1}\vlambda_{i}'\left(\overline{\mC}^{-1}\right)'\overline{\mU}'\mM_{\widehat{\mF}}\overline{\mU}\left(\overline{\mC}^{-1}\right)\vlambda_{i}\vepsi_{i}'\mF\right\Vert ^{2}\\
		+ & 4T^{-2}\sum_{i=1}^{N}\left\Vert T^{-1}\vepsi_{i}'\mP_{\widehat{\mF}}\vepsi_{i}\vepsi_{i}'\mF\right\Vert ^{2}+T^{-2}\sum_{i=1}^{N}\left\Vert T^{-1}\vlambda_{i}'\left(\overline{\mC}^{-1}\right)'\overline{\mU}'\mM_{\widehat{\mF}}\vepsi_{i}\vepsi_{i}'\mF\right\Vert ^{2}\\
		= & g_{55}+g_{56}+g_{57}+g_{58}.
	\end{align*}
	Consider now $g_{55}$, whose expected value is given by
	\begin{align*}
		\E\left(g_{55}\right) & =T^{-4}\sum_{i=1}^{N}\sum_{t,t''}^{T}\E\left[\left(\epsi_{i,t}^{2}-\sigma_{i}^{2}\right)^{2}\right]\E\left(\epsi_{i,t''}^{2}\right)\E\left(\vf_{t''}'\vf_{t''}\right)\\
		& +T^{-4}\sum_{i=1}^{N}\sum_{t,t'}^{T}\E\left(\epsi_{i,t}^{3}\right)\E\left(\epsi_{i,t'}^{3}\right)\E\left(\vf_{t}'\vf_{t'}\right)+\largeO\left(NT^{-2}\right)\\
		& =\largeO\left(NT^{-2}\right),
	\end{align*}
	implying $g_{55}=\largeO_{P}\left(NT^{-2}\right)$. We continue with $g_{56}$
	where we can write
	\begin{align*}
		g_{56} & \leq\left\Vert \overline{\mC}^{-1}\right\Vert ^{4}\left\Vert T^{-1}\overline{\mU}'\mM_{\widehat{\mF}}\overline{\mU}\right\Vert ^{2}\left(\sum_{i=1}^{N}\left\Vert \vlambda_{i}'\right\Vert ^{2}\right)\left(\sum_{i=1}^{N}\left\Vert \frac{\vlambda_{i}\vepsi_{i}'\mF}{T}\right\Vert ^{2}\right)\\
		& =\left\{ \largeO_{P}\left(N^{-2}\right)+\largeO_{P}\left[\left(NT\right)^{-1}\right]\right\} \largeO_{P}\left(N\right)\left[\largeO_{P}\left(NT^{-1}\right)+\largeO_{P}\left(N^{-1}\right)\right]\\
		& =\largeO_{P}\left(NT^{-2}\right)+\largeO_{P}\left(T^{-1}\right)+\largeO_{P}\left(N^{-2}\right).
	\end{align*}
	The second line above is the consequence of applying results\eqref{eq:sumi_lam_eps_Fhat}
	and \eqref{eq:Ubar_MF_Ubar}. It is valid to use the former in order
	to arrive at an order in probability for $\sum_{i=1}^{N}\left\Vert \vlambda_{i}\vepsi_{i}'\mF\right\Vert ^{2}$
	since $\widehat{\mF}=\mF\overline{\mC}+\overline{\mU}$.
	We continue with the third term, $g_{57}$, which is written out as
	\begin{align*}
		g_{57} & \leq\sum_{i=1}^{N}\left\lVert T^{-1}\vepsi_{i}'\widehat{\mF}\right\rVert ^{4}\left\lVert \left(T^{-1}\widehat{\mF}'\widehat{\mF}\right)\right\rVert ^{2}\left\lVert T^{-1}\mF'\vepsi_{i}\right\rVert ^{2}\\
		& \leq M\left\lVert \left(T^{-1}\widehat{\mF}'\widehat{\mF}\right)\right\rVert ^{2}\sum_{i=1}^{N}\left\lVert T^{-1}\vepsi_{i}'\mF\right\rVert ^{6}+M\left\lVert \left(T^{-1}\widehat{\mF}'\widehat{\mF}\right)\right\rVert ^{2}\left(\sum_{i=1}^{N}\left\lVert T^{-1}\vepsi_{i}'\overline{\mU}\right\rVert ^{4}\right)\left(\sum_{i=1}^{N}\left\lVert T^{-1}\mF'\vepsi_{i}\right\rVert ^{2}\right)\\
		& =g_{571}+g_{572}.
	\end{align*}
	Here, note that 
	\[
	\sum_{i=1}^{N}\left\lVert T^{-1}\vepsi_{i}'\mF\right\rVert ^{6}=\sum_{i=1}^{N}\sum_{t,t',t'',t''',t'''',t'''''}^{T}\epsi_{i,t}\vf_{t}'\vf_{t'}\epsi_{i,t'}\epsi_{i,t''}\vf_{t''}'\vf_{t'''}\epsi_{i,t'''}\epsi_{i,t''''}\vf_{t''''}'\vf_{t'''''}\epsi_{i,t'''''}.
	\]
	The expectation of this term is given by
	\begin{align*}
		\E\left(\sum_{i=1}^{N}\left\lVert T^{-1}\vepsi_{i}'\mF\right\rVert ^{6}\right) & =T^{-6}\sum_{i=1}^{N}\left[\sum_{t=1}^{T}\E\left(\epsi_{i,t}^{2}\right)\E\left(\left\lVert \vf_{t}\right\rVert ^{2}\right)\right]^{3}+\smallO\left(NT^{-3}\right)\\
		& =\largeO\left(NT^{-3}\right),
	\end{align*}
	from which $g_{571}=\largeO_{P}\left(NT^{-3}\right)$ follows. In order
	to get a corresponding result for $g_{572}$, we need to investigate
	the expression
	\begin{align*}
		\sum_{i=1}^{N}\left\lVert T^{-1}\vepsi_{i}'\overline{\mU}\right\rVert ^{4} & \leq3T^{-4}\sum_{i=1}^{N}\left(\vepsi_{i}'\overline{\vepsi}\overline{\vepsi}'\vepsi_{i}\right)^{2}+3T^{-4}\sum_{i=1}^{N}\left(\vepsi_{i}'\overline{\ve}\overline{\ve}'\vepsi_{i}\right)^{2}\\
		& =\left(NT\right)^{-4}\sum_{i,i',i'',i''',i''''}^{N}\sum_{t,t',t'',t'''}^{T}\epsi_{i,t}\epsi_{i',t}\epsi_{i'',t'}\epsi_{i,t'}\epsi_{i,t''}\epsi_{i''',t''}\epsi_{i'''',t'''}\epsi_{i,t'''}\\
		& +\left(NT\right)^{-4}\sum_{i,i',i'',i''',i''''}^{N}\sum_{t,t',t'',t'''}^{T}\epsi_{i,t}\ve'_{i',t}\ve_{i'',t'}\epsi_{i,t'}\epsi_{i,t''}\ve'_{i''',t''}\ve_{i'''',t'''}\epsi_{i,t'''}.
	\end{align*}
	The first term in the decomposition above has expected value
	\begin{align*}
		& \E\left(T^{-4}\sum_{i=1}^{N}\left(\vepsi_{i}'\overline{\vepsi}\overline{\vepsi}'\vepsi_{i}\right)^{2}\right)\\
		= & \left(NT\right)^{-4}\sum_{i=1}^{N}\sum_{t,t',t'',t'''}^{T}\E\left(\epsi_{i,t}^{2}\right)\E\left(\epsi_{i,t'}^{2}\right)\E\left(\epsi_{i,t''}^{2}\right)\E\left(\epsi_{i,t'''}^{2}\right)+3\left(NT\right)^{-4}\sum_{i,i',i'''}^{N}\sum_{t,t''}^{T}\E\left(\epsi_{i,t}^{2}\right)\E\left(\epsi_{i',t}^{2}\right)\E\left(\epsi_{i,t''}^{2}\right)\E\left(\epsi_{i''',t''}^{2}\right)\\
		+ & 2\left(NT\right)^{-4}\sum_{i,i'''}^{N}\sum_{t,t',t''}^{T}\E\left(\epsi_{i,t}^{2}\right)\E\left(\epsi_{i,t'}^{2}\right)\E\left(\epsi_{i,t''}^{2}\right)\E\left(\epsi_{i''',t''}^{2}\right)+\smallO\left(N^{-3}\right)+\smallO\left(N^{-1}T^{-2}\right)+\smallO\left(N^{-2}T^{-1}\right)\\
		= & \largeO\left(N^{-1}T^{-2}\right)+\largeO\left(N^{-3}\right)+\largeO\left(N^{-2}T^{-1}\right).
	\end{align*}
	Analogously, we can write the expectation of the second term as 
	\begin{align*}
		& \E\left(T^{-4}\sum_{i=1}^{N}\left(\vepsi_{i}'\overline{\ve}\overline{\ve}'\vepsi_{i}\right)^{2}\right) \\  
		= & N^{-2}T^{-4}\sum_{i,i',i''}^{N}\sum_{t,t'}^{T}\E\left(\epsi_{i,t}^{2}\right)\E\left(\epsi_{i,t'}^{2}\right)\E\left(\left\lVert \ve_{i',t}\right\rVert ^{2}\right)\E\left(\left\lVert \ve'_{i'',t'}\right\rVert ^{2}\right)+\smallO\left(N^{-1}T^{-2}\right)\\
		= &\largeO\left(N^{-1}T^{-2}\right).
	\end{align*}
	Consequently, $\sum_{i=1}^{N}\left\lVert T^{-1}\vepsi_{i}'\overline{\mU}\right\rVert ^{4}=\largeO_{P}\left(N^{-1}T^{-2}\right)+\largeO_{P}\left(N^{-2}T^{-1}\right)+\largeO_{P}\left(N^{-3}\right)$.
	In order to arrive at an order in probability for $g_{572}$, we also
	need a result for $\sum_{i=1}^{N}\left\lVert \mF'\vepsi_{i}\right\rVert ^{2}$,
	a term that is completely contained in $\sum_{i=1}^{N}\left\lVert \widehat{\mF}'\vepsi_{i}\right\rVert ^{2}$.
	Accordingly, we can use result \eqref{eq:sumi_eps_Fhat} to arrive
	at the upper bound $\sum_{i=1}^{N}\left\lVert T^{-1}\mF'\vepsi_{i}\right\rVert ^{2}=\largeO_{P}\left(NT^{-1}\right)+\largeO_{P}\left(N^{-1}\right)$.
	Using this result as well as the former, we can state that
	\begin{align*}
		g_{572} & =\left[\largeO_{P}\left(N^{-1}T^{-2}\right)+\largeO_{P}\left(N^{-2}T^{-1}\right)+\largeO_{P}\left(N^{-3}\right)\right]\left[\largeO_{P}\left(NT^{-1}\right)+\largeO_{P}\left(N^{-1}\right)\right]\\
		& =\largeO_{P}\left(T^{-3}\right)+\largeO_{P}\left(N^{-2}T^{-1}\right)+\largeO_{P}\left(N^{-1}T^{-2}\right)+\largeO_{P}\left(N^{-4}\right).
	\end{align*}
	If we now add our result on $g_{571}$, we arrive at 
	\[
	g_{57}=\largeO_{P}\left(NT^{-3}\right)+\largeO_{P}\left(N^{-1}T^{-2}\right)+\largeO_{P}\left(N^{-2}T^{-1}\right)+\largeO_{P}\left(N^{-4}\right)
	\]
	
	Lastly, application of results \eqref{eq:sumi_lam_eps_Fhat} and \eqref{eq:sumi_Ubar_MF_epsi}
	leads to
	\begin{align*}
		g_{58} & =\left\Vert \overline{\mC}^{-1}\right\Vert ^{2}\left(\sum_{i=1}^{N}\left\Vert T^{-1}\vlambda_{i}\vepsi_{i}'\mF\right\Vert ^{2}\right)\left(\sum_{i=1}^{N}\left\Vert T^{-1}\overline{\mU}'\mM_{\widehat{\mF}}\vepsi_{i}\right\Vert ^{2}\right)\\
		& =\left[\largeO_{P}\left(NT^{-1}\right)+\largeO_{P}\left(N^{-1}\right)\right]\left[\largeO_{P}\left(N^{-1}\right)+\largeO_{P}\left(T^{-1}\right)\right]\\
		& =\largeO_{P}\left(NT^{-2}\right)+\largeO_{P}\left(T^{-1}\right)+\largeO_{P}\left(N^{-2}\right).
	\end{align*}
	Combining our results on terms $g_{55}$, $g_{56}$, $g_{57}$ and
	$g_{58}$, we arrive at 
	\[
	\sum_{i=1}^{N}\left\Vert T^{-1}\left(\widehat{\sigma}_{i}^{2}-\sigma_{i}^{2}\right)\vepsi_{i}'\mF\right\Vert ^{2}=\largeO_{P}\left(NT^{-2}\right)+\largeO_{P}\left(T^{-1}\right)+\largeO_{P}\left(N^{-2}\right).
	\]
	Together with Lemma \ref{lem:sumi_sigi2diff_epsi_Ubar}, this result
	allows us to conclude that
	\[
	\sum_{i=1}^{N}\left\Vert T^{-1}\left(\widehat{\sigma}_{i}^{2}-\sigma_{i}^{2}\right)\vepsi_{i}'\widehat{\mF}\right\Vert ^{2}=\largeO_{P}\left(NT^{-2}\right)+\largeO_{P}\left(T^{-1}\right)+\largeO_{P}\left(N^{-2}\right).
	\]
\end{proof}
\begin{lemma}
	\label{lem:Taylor_a21*order}Suppose that Assumptions \ref{ass:errors}-\ref{ass:rank} hold. Then,
	\[
	\sum_{i=1}^{N}\left[T^{-1}\sum_{j\neq i}^{N}\frac{\widehat{\vepsi}_{i}'\widehat{\vepsi}_{j}}{\sigma_{j}}\left(\widehat{\sigma}_{i}^{2}-\sigma_{i}^{2}\right)\right]^{2}=\largeO_{P}\left(N^{2}T^{-2}\right)+\largeO_{P}\left(NT^{-1}\right)+\largeO_{P}\left(1\right)+\largeO_{P}\left(N^{-1}T\right)+\largeO_{P}\left(N^{-2}T^{2}\right)
	\]
\end{lemma}
\begin{proof}
	Analogous to our treatment of term $a_{1}^{*}$ we disregard from
	the scaling $\sigma_{i}^{-1}$ since its values
	have a lower bound above zero. When using a form of Markov's inequality
	to establish orders in probability, we can find an upper bound for
	any component of the expression of interest for this lemma. 
	
	We decompose the expression of interest into
	\begin{align*}
		& \sum_{i=1}^{N}\left[T^{-1}\sum_{j\neq i}^{N}\widehat{\vepsi}_{i}'\widehat{\vepsi}_{j}\left(\widehat{\sigma}_{i}^{2}-\sigma_{i}^{2}\right)\right]^{2}\\
		= & T^{-2}\sum_{i=1}^{N}\left[\sum_{j\neq i}^{N}\vepsi_{i}'\vepsi_{j}\left(\widehat{\sigma}_{i}^{2}-\sigma_{i}^{2}\right)\right]^{2}+T^{-2}\sum_{i=1}^{N}\left[\sum_{j\neq i}^{N}\vlambda_{i}'\left(\overline{\mC}^{-1}\right)'\overline{\mU}'\mM_{\widehat{\mF}}\overline{\mU}\left(\overline{\mC}^{-1}\right)\vlambda_{j}\left(\widehat{\sigma}_{i}^{2}-\sigma_{i}^{2}\right)\right]^{2}\\
		+ & T^{-2}\sum_{i=1}^{N}\left[\sum_{j\neq i}^{N}\vepsi_{i}'\mP_{\widehat{\mF}}\vepsi_{j}\left(\widehat{\sigma}_{i}^{2}-\sigma_{i}^{2}\right)\right]^{2}+T^{-2}\sum_{i=1}^{N}\left[\sum_{j\neq i}^{N}\vlambda_{i}'\left(\overline{\mC}^{-1}\right)'\overline{\mU}'\mM_{\widehat{\mF}}\vepsi_{j}\left(\widehat{\sigma}_{i}^{2}-\sigma_{i}^{2}\right)\right]^{2}\\
		+T^{-2} & \sum_{i=1}^{N}\left[\sum_{j\neq i}^{N}\vepsi_{i}'\mM_{\widehat{\mF}}\overline{\mU}\left(\overline{\mC}^{-1}\right)\vlambda_{j}\left(\widehat{\sigma}_{i}^{2}-\sigma_{i}^{2}\right)\right]^{2}\\
		= & g_{61}+g_{62}+g_{63}+g_{64}+g_{65}.
	\end{align*}
	Consider expression $g_{62}$ which can be written 
	\[
	g_{62}\leq\sum_{i=1}^{N} \left( \left\Vert \left(\widehat{\sigma}_{i}^{2}-\sigma_{i}^{2}\right)\vlambda_{i}'\right\Vert ^{2} \left\Vert \sum_{j\neq i}^{N}\vlambda_{j}\right\Vert ^{2} \right)  \left\Vert T^{-1}\overline{\mU}'\mM_{\widehat{\mF}}\overline{\mU}\right\Vert^{2} \left\Vert \overline{\mC}^{-1}\right\Vert^{4}.
	\]
The order in probability of  $\left\Vert T^{-1}\overline{\mU}'\mM_{\widehat{\mF}}\overline{\mU}\right\Vert^{2}$ is determined by result \eqref{eq:Ubar_MF_Ubar}. For the sum over $i$, we can use Lemma  \ref{lem:sumi_sig2diff_sq_CCE} despite the presence of the additional term $\left\Vert \sum_{j\neq i}^{N}\vlambda_{j}\right\Vert ^{2}$. To appreciate that this additional term does not change the results of Lemma \ref{lem:sumi_sig2diff_sq_CCE}, note that the proof of Lemma \eqref{eq:Ubar_MF_Ubar} mainly relies on Markov's inequality. Given independence of errors and loadings as well as independence of factor loadings across cross-sections, the expectation of $\left\Vert \sum_{j\neq i}^{N}\vlambda_{j}\right\Vert ^{2}$ can hence always be separated from the expectation of terms that are indexed by $i$. For this reason, we can conclude that 
	\begin{align}
		g_{62} & =\left[\largeO_{P}\left(1\right)+\largeO_{P}\left(NT^{-1}\right)\right]\left\{ \largeO_{P}\left(N^{-2}\right)+\largeO_{P}\left[\left(NT\right)^{-1}\right]\right\} \largeO_{P}\left(N\right)\nonumber \\
		& =\largeO_{P}\left(N^{-1}\right)+\largeO_{P}\left(T^{-1}\right)+\largeO_{P}\left(NT^{-2}\right).\label{eq:g62_order}
	\end{align}
	Next, we focus on expression $g_{64}$ which can be bounded in a very
	similar way. We write
	\begin{align}
		g_{64} & \leq\sum_{i=1}^{N}\left\Vert \left(\widehat{\sigma}_{i}^{2}-\sigma_{i}^{2}\right)\vlambda_{i}'\right\Vert ^{2}\left\Vert \overline{\mC}^{-1}\right\Vert ^{2}\left\Vert T^{-1}\sum_{j\neq i}^{N}\overline{\mU}'\mM_{\widehat{\mF}}\vepsi_{j}\right\Vert ^{2}\nonumber \\
		& =\left[\largeO_{P}\left(1\right)+\largeO_{P}\left(NT^{-1}\right)\right]\left[\largeO_{P}\left(1\right)+\largeO_{P}\left(NT^{-1}\right)\right]\nonumber \\
		& =\largeO_{P}\left(1\right)+\largeO_{P}\left(NT^{-1}\right)+\largeO_{P}\left(N^{2}T^{-2}\right),\label{eq:g64_order}
	\end{align}
	where the result $\left\Vert T^{-1}\sum_{j\neq i}^{N}\overline{\mU}'\mM_{\widehat{\mF}}\vepsi_{j}\right\Vert ^{2}=\largeO_{P}\left(1\right)+\largeO_{P}\left(NT^{-1}\right)$
	follows from \eqref{eq:Ubar_MF_Ubar} since the former term is entirely
	contained in $N^{2}\left\Vert T^{-1}\overline{\mU}'\mM_{\widehat{\mF}}\overline{\mU}\right\Vert ^{2}$
	and since the error variances $\sigma_{j}^{2}$ have lower bound above
	zero.
	
	We continue with expression 
	\begin{align*}
		g_{61} & \leq5T^{-2}\sum_{i=1}^{N}\left[\sum_{j\neq i}^{N}\vepsi_{i}'\vepsi_{j}\left(\frac{\vepsi_{i}'\vepsi_{i}}{T}-\sigma_{i}^{2}\right)\right]^{2}+5T^{-2}\sum_{i=1}^{N}\left[\sum_{j\neq i}^{N}\vepsi_{i}'\vepsi_{j}\vlambda_{i}'\left(\overline{\mC}^{-1}\right)'\left(T^{-1}\overline{\mU}'\mM_{\widehat{\mF}}\overline{\mU}\right)\left(\overline{\mC}^{-1}\right)\vlambda_{i}\right]^{2}\\
		& 5T^{-2}\sum_{i=1}^{N}\left[\sum_{j\neq i}^{N}\vepsi_{i}'\vepsi_{j}\left(T^{-1}\vepsi_{i}'\mP_{\widehat{\mF}}\vepsi_{i}\right)\right]^{2}+5T^{-2}\sum_{i=1}^{N}\left[\sum_{j\neq i}^{N}\vepsi_{i}'\vepsi_{j}\vlambda_{i}'\left(\overline{\mC}^{-1}\right)'\left(T^{-1}\overline{\mU}'\mM_{\widehat{\mF}}\vepsi_{i}\right)\right]^{2}\\
		& =g_{611}+g_{612}+g_{613}+g_{614}.
	\end{align*}
	Here, $g_{611}=\largeO_{P}\left(N^{2}T^{-2}\right)$ follows from result
	\eqref{eq:sumij_espi_espj_epsisq-diff} in the time fixed effects
	case. When considering $g_{612}$, we first note that 
	\begin{align*}
		\E\left[\sum_{i=1}^{N}\left(\sum_{j\neq i}^{N}\vepsi_{i}'\vepsi_{j}\vlambda_{i}'\vlambda_{i}\right)^{2}\right] & \leq M\sum_{i=1}^{N}\sum_{t,t'}^{T}\E\left(\epsi_{i,t}^{2}\right)\E\left(\epsi_{i,t'}^{2}\right)\\
		& +M\sum_{i=1}^{N}\sum_{j\neq i}^{N}\sum_{t=1}^{T}\E\left(\epsi_{i,t}^{2}\right)\E\left(\epsi_{j,t}^{2}\right)+\smallO\left(NT^{2}\right)+\smallO\left(N^{2}T\right)\\
		& =\largeO\left(NT^{2}\right)+\largeO\left(N^{2}T\right),
	\end{align*}
	from which it follows that
	\begin{equation}
		T^{-2}\sum_{i=1}^{N}\left(\sum_{j\neq i}^{N}\vepsi_{i}'\vepsi_{j}\vlambda_{i}'\vlambda_{i}\right)^{2}=\largeO_{P}\left(N\right)+\largeO_{P}\left(N^{2}T^{-1}\right).\label{eq:sumij_epsi_epsj_lami_sq}
	\end{equation}
	This result, together with result \eqref{eq:Ubar_MF_Ubar}, is then
	used to arrive at
	\begin{align*}
		g_{612} & \leq MT^{-2}\sum_{i=1}^{N}\left[\sum_{j\neq i}^{N}\vepsi_{i}'\vepsi_{j}\left(\vlambda_{i}'\vlambda_{i}\right)\right]^{2}\left\Vert \overline{\mC}^{-1}\right\Vert ^{4}\left\Vert T^{-1}\overline{\mU}'\mM_{\widehat{\mF}}\overline{\mU}\right\Vert ^{2}\\
		& =\left[\largeO_{P}\left(N\right)+\largeO_{P}\left(N^{2}T^{-1}\right)\right]\left\{ \largeO_{P}\left(N^{-2}\right)+\largeO_{P}\left[\left(NT\right)^{-1}\right]\right\} \\
		& =\largeO_{P}\left(N^{-1}\right)+\largeO_{P}\left(T^{-1}\right)+\largeO_{P}\left(NT^{-2}\right)
	\end{align*}
	We continue with
	\begin{align*}
		g_{613} & \leq T^{-2}\left\Vert T^{-1}\left(\widehat{\mF}'\widehat{\mF}\right)^{-1}\right\Vert ^{2}\left\Vert \overline{\mC}\right\Vert ^{4}\sum_{i=1}^{N}\left(\left\Vert T^{-1}\sum_{j\neq i}^{N}\vepsi_{i}'\vepsi_{j}\vepsi_{i}'\mF\right\Vert ^{2}\left\Vert T^{-1}\mF'\vepsi_{i}\right\Vert ^{2}\right)\\
		& +T^{-2}\left\Vert T^{-1}\left(\widehat{\mF}'\widehat{\mF}\right)^{-1}\right\Vert ^{2}\sum_{i=1}^{N}\left(\left\Vert T^{-1}\sum_{j\neq i}^{N}\vepsi_{i}'\vepsi_{j}\vepsi_{i}'\overline{\mU}\right\Vert ^{2}\left\Vert T^{-1}\widehat{\mF}'\vepsi_{i}\right\Vert ^{2}\right)\\
		& +T^{-2}\left\Vert T^{-1}\left(\widehat{\mF}'\widehat{\mF}\right)^{-1}\right\Vert ^{2}\left\Vert \overline{\mC}\right\Vert ^{2}\sum_{i=1}^{N}\left(\left[\left\Vert T^{-1}\sum_{j\neq i}^{N}\vepsi_{i}'\vepsi_{j}\vepsi_{i}'\mF\right\Vert ^{2}\right]\left\Vert T^{-1}\overline{\mU}'\vepsi_{i}\right\Vert ^{2}\right)\\
		& g_{6131}+g_{6132}+g_{6133}
	\end{align*}
	Here,
	\begin{align*}
		& \E\left(T^{-2}\sum_{i=1}^{N}\left\Vert T^{-1}\sum_{j\neq i}^{N}\vepsi_{i}'\vepsi_{j}\vepsi_{i}'\mF\right\Vert ^{2}\left\Vert T^{-1}\mF'\vepsi_{i}\right\Vert ^{2}\right)\\
		& =T^{-6}\sum_{i=1}^{N}\sum_{j\neq i}^{N}\sum_{j'\neq i}^{N}\sum_{t,t',t'',t''',t'''',t'''''}^{T}\E\left[\epsi_{i,t}\epsi_{j,t}\epsi_{i,t'}\epsi_{j',t'}\epsi_{i,t''}\epsi_{i,t'''}\epsi_{i,t''''}\epsi_{i,t'''''}\right]\E\left[\vf_{t''}'\vf_{t'''}\vf_{t''''}'\vf_{t'''''}\right]\\
		& =T^{-6}\sum_{i=1}^{N}\sum_{j\neq i}^{N}\sum_{t,t'',t''''}^{T}\E\left[\epsi_{i,t}^{2}\right]\E\left[\epsi_{j,t}^{2}\right]\E\left[\epsi_{i,t''}^{2}\right]\E\left[\epsi_{i,t''''}^{2}\right]\E\left[\left\Vert \vf_{t''}\right\Vert ^{2}\left\Vert \vf_{t''''}\right\Vert ^{2}\right]+\smallO\left(N^{2}T^{-3}\right)\\
		& =\largeO\left(N^{2}T^{-3}\right),
	\end{align*}
	which implies that $g_{6131}=\largeO_{P}\left(N^{2}T^{-3}\right).$ Concerning
	the next term, $g_{6132}$, we note that 
	\begin{align*}
		T^{-2}\sum_{i=1}^{N}\left(\left\Vert T^{-1}\sum_{j\neq i}^{N}\vepsi_{i}'\vepsi_{j}\vepsi_{i}'\overline{\mU}\right\Vert ^{2}\left\Vert T^{-1}\widehat{\mF}'\vepsi_{i}\right\Vert ^{2}\right) & \leq T^{-2}\sum_{i=1}^{N}\left\Vert T^{-1}\sum_{j\neq i}^{N}\vepsi_{i}'\vepsi_{j}\vepsi_{i}'\overline{\mU}\right\Vert ^{2}\sum_{i=1}^{N}\left\Vert T^{-1}\widehat{\mF}'\vepsi_{i}\right\Vert ^{2}\\
		& =T^{-2}\left[\largeO_{P}\left(N\right)+\largeO_{P}\left(T\right)\right]\left[\largeO_{P}\left(NT^{-1}\right)+\largeO_{P}\left(N^{-1}\right)\right]\\
		& =\largeO_{P}\left(N^{2}T^{-3}\right)+\largeO_{P}\left(NT^{-2}\right)+\largeO_{P}\left[\left(NT\right)^{-1}\right],
	\end{align*}
	which follows from results \eqref{eq:sumi_eps_Fhat} and \eqref{eq:sumi_sq_sumj_epsi_epsj_epsi_Ubar}.
	Accordingly, $g_{6132}=\largeO_{P}\left(N^{2}T^{-3}\right)+\largeO_{P}\left(NT^{-2}\right)+\largeO_{P}\left[\left(NT\right)^{-1}\right]$.
	In the last component of $g_{613},$term $g_{6133}$, we can use results
	\eqref{eq:sumi_eps_ubar} and \eqref{eq:sumi _sq_sumj_epsi_epsj_epsi_F}
	to write
	\begin{align*}
		T^{-2}\sum_{i=1}^{N}\left(\left[\left\Vert T^{-1}\sum_{j\neq i}^{N}\vepsi_{i}'\vepsi_{j}\vepsi_{i}'\mF\right\Vert ^{2}\right]\left\Vert T^{-1}\overline{\mU}'\vepsi_{i}\right\Vert ^{2}\right) & \leq T^{-2}\sum_{i=1}^{N}\left[\left\Vert T^{-1}\sum_{j\neq i}^{N}\vepsi_{i}'\vepsi_{j}\vepsi_{i}'\mF\right\Vert ^{2}\right]\sum_{i=1}^{N}\left\Vert T^{-1}\overline{\mU}'\vepsi_{i}\right\Vert ^{2}\\
		& =T^{-2}\largeO_{P}\left(N^{2}\right)\left[\largeO_{P}\left(N^{-1}\right)+\largeO_{P}\left(T^{-1}\right)\right]\\
		& =\largeO_{P}\left(NT^{-2}\right)+\largeO_{P}\left(N^{2}T^{-3}\right),
	\end{align*}
	which also determines the overall order in probability of $g_{6133}$.
	Combining our results on terms $g_{6131},g_{6132}$ and $g_{6133}$,
	we arrive at 
	\[
	g_{613}=\largeO_{P}\left(N^{2}T^{-3}\right)+\largeO_{P}\left(NT^{-2}\right)+\largeO_{P}\left[\left(NT\right)^{-1}\right].
	\]
	
	Lastly, we investigate 
	\begin{align*}
		g_{614} & \leq5\sum_{i=1}^{N}\left(\left\lVert T^{-1}\sum_{j\neq i}^{N}\vepsi_{i}'\vepsi_{j}\vlambda_{i}'\right\rVert \left\lVert T^{-1}\overline{\mU}'\vepsi_{i}\right\rVert \right)^{2}\left\Vert \overline{\mC}^{-1}\right\Vert ^{2}\\
		& +5\left\lVert \overline{\mC}^{-1}\right\rVert ^{2}\sum_{i=1}^{N}\left\lVert T^{-1}\sum_{j\neq i}^{N}\vepsi_{i}'\vepsi_{j}\vlambda_{i}'\right\rVert ^{2}\left\lVert T^{-1}\overline{\mU}'\widehat{\mF}\right\rVert ^{2}\left\lVert \left(T^{-1}\widehat{\mF}'\widehat{\mF}\right)^{-1}\right\rVert ^{2}\sum_{i=1}^{N}\left\lVert T^{-1}\widehat{\mF}'\vepsi_{i}\right\rVert ^{2}\\
		& =g_{6141}+g_{6142}
	\end{align*}
	Here, the order of $g_{6141}$ is determined by the expression 
	\begin{align*}
		& \sum_{i=1}^{N}\left[T^{-1}\left\lVert \sum_{j\neq i}^{N}\vepsi_{i}'\vepsi_{j}\vlambda_{i}'\right\rVert \left\lVert T^{-1}\overline{\mU}'\vepsi_{i}\right\rVert \right]^{2}\\
		= & N^{-2}T^{-4}\sum_{i,i',i''}^{N}\sum_{j\neq i}^{N}\sum_{j'\neq i}^{N}\sum_{t,t',t'',t'''}^{T}\epsi_{i,t}\epsi_{j,t}\epsi_{j',t'}\epsi_{i,t'}\epsi_{i,t''}\epsi_{i',t''}\epsi_{i'',t'''}\epsi_{i,t'''}\vlambda_{i}'\vlambda_{i}\\
		+ & N^{-2}T^{-4}\sum_{i,i',i''}^{N}\sum_{j\neq i}^{N}\sum_{j'\neq i}^{N}\sum_{t,t',t'',t'''}^{T}\epsi_{i,t}\epsi_{j,t}\epsi_{j',t'}\epsi_{i,t'}\epsi_{i,t''}\ve'_{i',t''}\ve_{i'',t'''}\epsi_{i,t'''}\vlambda_{i}'\vlambda_{i}.
	\end{align*}
	Now take the first term in the decomposition above. Its expected value
	can be expressed as 
	\begin{align*}
		& \E\left(N^{-2}T^{-4}\sum_{i,i',i'',j,j'}^{N}\sum_{t,t',t'',t'''}^{T}\epsi_{i,t}\epsi_{j,t}\epsi_{j',t'}\epsi_{i,t'}\epsi_{i,t''}\epsi_{i',t''}\epsi_{i'',t'''}\epsi_{i,t'''}\vlambda_{i}'\vlambda_{i}\right)\\
		\leq & \frac{M}{N^{2}T^{4}}\sum_{i,i',j}^{N}\sum_{t,t''}^{T}\E\left(\epsi_{i,t}^{2}\right)\E\left(\epsi_{j,t}^{2}\right)\E\left(\epsi_{i,t''}^{2}\right)\E\left(\epsi_{i',t''}^{2}\right)\\
		& +\frac{M}{N^{2}T^{4}}\sum_{i,i'}^{N}\sum_{t,t',t''}^{T}\E\left(\epsi_{i,t}^{2}\right)\E\left(\epsi_{i,t'}^{2}\right)\E\left(\epsi_{i,t''}^{2}\right)\E\left(\epsi_{i',t''}^{2}\right)+\smallO\left(NT^{-2}\right)+\smallO\left(T^{-1}\right)\\
		& =\largeO\left(NT^{-2}\right)+\largeO\left(T^{-1}\right).
	\end{align*}
	Analogously, the expectation of the second term is given by
	\begin{align*}
		& \E\left(N^{-2}T^{-4}\sum_{i,i',i''}^{N}\sum_{j\neq i}^{N}\sum_{j'\neq i}^{N}\sum_{t,t',t'',t'''}^{T}\epsi_{i,t}\epsi_{j,t}\epsi_{j',t'}\epsi_{i,t'}\epsi_{i,t''}\ve'_{i',t''}\ve_{i'',t'''}\epsi_{i,t'''}\vlambda_{i}'\vlambda_{i}\right)\\
		\leq & \frac{M}{N^{2}T^{4}}\sum_{i,i''}^{N}\sum_{j\neq i}^{N}\sum_{t',t''}^{T}\E\left(\epsi_{i,t}^{2}\right)\E\left(\epsi_{j,t}^{2}\right)\E\left(\epsi_{i,t''}^{2}\right)\E\left(\ve'_{i',t''}\ve_{i',t''}\right)+\smallO\left(NT^{-2}\right)\\
		= & \largeO\left(NT^{-2}\right).
	\end{align*}
	These results suggest that $g_{6141}=\largeO_{P}\left(NT^{-2}\right)+\largeO_{P}\left(T^{-1}\right)$.
	We continue with $g_{6142}$ where results \eqref{eq:sumi_eps_Fhat},
	\eqref{eq:normsq_Ubar_Fhat} and \eqref{eq:sumij_epsi_epsj_lami_sq}
	lead to 
	\begin{align*}
		g_{6142} & =5\left\lVert \overline{\mC}^{-1}\right\rVert ^{2}\sum_{i=1}^{N}\left\lVert \sum_{j\neq i}^{N}\vepsi_{i}'\vepsi_{j}\vlambda_{i}'\right\rVert ^{2}\left\lVert T^{-1}\overline{\mU}'\widehat{\mF}\right\rVert ^{2}\left\lVert \left(T^{-1}\widehat{\mF}'\widehat{\mF}\right)^{-1}\right\rVert ^{2}\sum_{i=1}^{N}\left\lVert T^{-1}\widehat{\mF}'\vepsi_{i}\right\rVert ^{2}\\
		& =\left[\largeO_{P}\left(N\right)+\largeO_{P}\left(N^{2}T^{-1}\right)\right]\left\{ \largeO_{P}\left(N^{-2}\right)+\largeO_{P}\left[\left(NT\right)^{-1}\right]\right\} \left[\largeO_{P}\left(NT^{-1}\right)+\largeO_{P}\left(N^{-1}\right)\right]\\
		& =\left[\largeO_{P}\left(N^{-1}\right)+\largeO_{P}\left(T^{-1}\right)+\largeO_{P}\left(NT^{-2}\right)\right]\left[\largeO_{P}\left(NT^{-1}\right)+\largeO_{P}\left(N^{-1}\right)\right]\\
		& =\largeO_{P}\left(N^{2}T^{-3}\right)+\largeO_{P}\left(NT^{-2}\right)+\largeO_{P}\left(T^{-1}\right)+\largeO_{P}\left(N^{-2}\right).
	\end{align*}
	Combining our results on $g_{6141}$ and $g_{6142}$, we get 
	\[
	g_{614}=\largeO_{P}\left(N^{2}T^{-3}\right)+\largeO_{P}\left(NT^{-2}\right)+\largeO_{P}\left(T^{-1}\right)+\largeO_{P}\left(N^{-2}\right)
	\]
	
	Our result on $g_{614}$ is the last building block required to make
	a statement about $g_{61}$. Additionally using our results on $g_{611},g_{612}$
	and $g_{613}$, we arrive at 
	\begin{align}
		g_{61} & =\largeO_{P}\left(N^{2}T^{-2}\right)+\largeO_{P}\left(N^{-1}\right)+\largeO_{P}\left(T^{-1}\right).\label{eq:g61_order}
	\end{align}
	Now consider the term 
	\begin{align*}
		g_{65} & \leq\left\Vert \overline{\mC}^{-1}\right\Vert ^{2}\left\Vert \sum_{j\neq i}^{N}\vlambda_{j}\right\Vert ^{2}\sum_{i=1}^{N}\left\Vert T^{-1}\left(\widehat{\sigma}_{i}^{2}-\sigma_{i}^{2}\right)\vepsi_{i}'\mM_{\widehat{\mF}}\overline{\mU}\right\Vert ^{2}\\
		& \leq3\left\Vert \overline{\mC}^{-1}\right\Vert ^{2}\left\Vert \sum_{j\neq i}^{N}\vlambda_{j}\right\Vert ^{2}\sum_{i=1}^{N}\left\Vert T^{-1}\left(\widehat{\sigma}_{i}^{2}-\sigma_{i}^{2}\right)\vepsi_{i}'\overline{\mU}\right\Vert ^{2}\\
		& +3\left\Vert \overline{\mC}^{-1}\right\Vert ^{2}\left\Vert \sum_{j\neq i}^{N}\vlambda_{j}\right\Vert ^{2}\sum_{i=1}^{N}\left\Vert T^{-1}\left(\widehat{\sigma}_{i}^{2}-\sigma_{i}^{2}\right)\vepsi_{i}'\widehat{\mF}\right\Vert ^{2}\left\Vert \left(T^{-1}\widehat{\mF}'\widehat{\mF}\right)^{-1}\right\Vert ^{2}+\largeO_{P}\left(NT^{-1}\right)\left\Vert T^{-1}\widehat{\mF}\overline{\mU}\right\Vert ^{2}\\
		& =3\left\Vert \overline{\mC}^{-1}\right\Vert ^{2}\left\Vert \sum_{j\neq i}^{N}\vlambda_{j}\right\Vert ^{2}\left(g_{651}+g_{652}\right).
	\end{align*}
	The two resulting expressions $g_{651}$ and $g_{652}$ require a
	further investigation of the terms $\sum_{i=1}^{N}\left\Vert T^{-1}\left(\widehat{\sigma}_{i}^{2}-\sigma_{i}^{2}\right)\vepsi_{i}'\overline{\mU}\right\Vert ^{2}$
	and $\sum_{i=1}^{N}\left\Vert T^{-1}\left(\widehat{\sigma}_{i}^{2}-\sigma_{i}^{2}\right)\vepsi_{i}'\widehat{\mF}\right\Vert ^{2}$,
	a task that is carried out in Lemmas \ref{lem:sumi_sigi2diff_epsi_Ubar}
	and \ref{lem:sumi_sigi2diff_epsi_Fhat} as well as their respective
	proofs. Using the results of these two lemmas, as well as result \eqref{eq:normsq_Ubar_Fhat},
	we arrive at 
	\begin{align*}
		g_{651} & =\largeO_{P}\left(NT^{-3}\right)+\largeO_{P}\left(T^{-2}\right)+\largeO_{P}\left[\left(NT\right)^{-1}\right]+\largeO_{P}\left(N^{-2}\right)\\
		g_{652} & =\left[\largeO_{P}\left(NT^{-2}\right)+\largeO_{P}\left(T^{-1}\right)+\largeO_{P}\left(N^{-2}\right)\right]\left[\largeO_{P}\left[\left(NT\right)^{-1}\right]+\largeO_{P}\left(N^{-2}\right)\right]\\
		& =\largeO_{P}\left(T^{-3}\right)+\largeO_{P}\left(N^{-1}T^{-2}\right)+\largeO_{P}\left(N^{-2}T^{-1}\right)+\largeO_{P}\left(N^{-4}\right),
	\end{align*}
	from which it follows that
	\begin{equation}
		g_{65}=\largeO_{P}\left(NT^{-3}\right)+\largeO_{P}\left(T^{-2}\right)+\largeO_{P}\left[\left(NT\right)^{-1}\right]+\largeO_{P}\left(N^{-2}\right).\label{eq:g65_order}
	\end{equation}
	The last term to consider is $g_{53}$ where we write
	\begin{align}
		g_{63} & \leq\sum_{i=1}^{N}\left\Vert T^{-1}\left(\widehat{\sigma}_{i}^{2}-\sigma_{i}^{2}\right)\vepsi_{i}'\widehat{\mF}\right\Vert ^{2}\left\Vert \left(T^{-1}\widehat{\mF}'\widehat{\mF}\right)^{-1}\right\Vert ^{2}\left\Vert T^{-1}\widehat{\mF}'\left(\sum_{j\neq i}^{N}\vepsi_{j}\right)\right\Vert ^{2}\nonumber \\
		& =\left[\largeO_{P}\left(NT^{-2}\right)+\largeO_{P}\left(T^{-1}\right)+\largeO_{P}\left(N^{-2}\right)\right]\left[\largeO_{P}\left(NT^{-1}\right)+\largeO_{P}\left(1\right)\right]\nonumber \\
		& =\largeO_{P}\left(N^{2}T^{-3}\right)+\largeO_{P}\left(NT^{-2}\right)+\largeO_{P}\left(T^{-1}\right)+\largeO_{P}\left(N^{-2}\right)\label{eq:g63_order}
	\end{align}
	In order to arrive at the second line above, we use Lemma \ref{lem:sumi_sigi2diff_epsi_Fhat}
	together with result \eqref{eq:normsq_Ubar_Fhat}. Application of
	the second result is valid since $\sigma_{j}^{2}$ has a lower bound
	above zero and since $\left\Vert T^{-1}\widehat{\mF}'\sum_{j\neq i}^{N}\vepsi_{j}\right\Vert ^{2}$
	is completely contained in $N^{2}\left\Vert T^{-1}\widehat{\mF}\overline{\mU}\right\Vert ^{2}$. 
	
	Results \eqref{eq:g62_order}--\eqref{eq:g63_order} address all
	five required terms to arrive at the central result of this lemma:
	\[
	T^{-2}\sum_{i=1}^{N}\left[\sum_{j\neq i}^{N}\frac{\widehat{\vepsi}_{i}'\widehat{\vepsi}_{j}}{\left(\sigma_{j}^{2}\right)^{1/2}}\left(\widehat{\sigma}_{i}^{2}-\sigma_{i}^{2}\right)\right]^{2}=\largeO_{P}\left(N^{2}T^{-2}\right)+\largeO_{P}\left(NT^{-1}\right)+\largeO_{P}\left(1\right)
	\]
	
\end{proof}

\section{Proofs for additional theoretical contributions}
\subsection{Proofs of Propositions  \ref{prop:CD_2WFE_H1} and  \ref{prop:CD_CCE_H1}}
\label{section::proofs_power}
\begin{proof}[\textbf{Proof of Proposition \ref{prop:CD_2WFE_H1}}]\ \\
	The same steps as in the proof of Theorem \ref{theorem::additiveHetero} allow us to arrive at the decomposition
	\begin{align*}
		CD_{\mathbb{H}_{1}} &= \sqrt{\frac{1}{2TN\left( N-1\right) }} \sum_{i=2}^{N}\sum_{j=1}^{i-1} \left(\varsigma_{\nu,i}^{-1} - \E\left[\varsigma_{\nu,i}^{-1}\right]\right) \vnu_{i}^{\prime}\vnu_{j}\left( \varsigma_{\nu,j}^{-1} - \E\left[ \varsigma_{\nu,j}^{-1}\right]\right) +\sqrt{T}\Xi_{\mathbb{H}_{1}}  \\
		&+ \frac{\sqrt{T}}{N}\sqrt{\frac{N}{2\left( N-1\right) }}\left( \left(NT\right)^{-1} \sum_{i,j}^{N}\frac{\vnu_{i}^{\prime}\vnu_{j}}{\varsigma_{\nu,i}^{2}} - \overline{\varsigma_{\nu}^{-2}} \left( NT\right)^{-1} \sum_{i,j}^{N} \vnu_{i}^{\prime}\vnu_{j}\right) ,
	\end{align*}%
	where%
	\begin{align*}
		\Xi_{\mathbb{H}_1} &= \sqrt{\frac{N}{2\left(N-1\right)}} \left(NT\right)^{-1} \sum_{i=1}^N \vnu_{i}' \vnu_{i} \left[\left(\overline{\varsigma_{\nu}^{-1}}\right)^2 - 2\overline{\varsigma_{\nu}^{-1}}\varsigma_{\nu,i} \right].
	\end{align*}%
	The terms in parentheses in the second line are stochastically bounded,
	entailing that the entire second line is $\largeO_{P}\left( N^{-1}\sqrt{T}\right) $. Next, consider the leading stochastic component of $CD_{\mathbb{H}_{1}}$, namely $\sqrt{\frac{1}{2TN\left(N-1\right) }}\sum_{i=2}^{N}\sum_{j=1}^{i-1} \left( \varsigma_{\nu,i}^{-1}-\E\left[ \varsigma_{\nu,i}^{-1}\right]\right) \vnu_{i}^{\prime}\vnu_{j}\left( \varsigma_{\nu,j}^{-1}-\E\left[ \varsigma_{\nu,j}^{-1}\right] \right) $. Here, we have
	\begin{align*}
		& \E\left[ \left( \varsigma _{\nu,i}^{-1}-\E\left[ \varsigma_{\nu,i}^{-1}\right] \right) \nu_{i,t}\nu_{j,t}\left( \varsigma_{\nu,j}^{-1}-\E\left[ \varsigma _{\nu,j}^{-1}\right] \right) \right]  \\
		&=\E\left[ \left( \varsigma_{\nu,i}^{-1}-\E\left[ \varsigma_{\nu,i}^{-1}\right] \right) \left( {\vlambda }_{i}-\E\left[	\vlambda_{i}\right] \right)^{\prime }\right] \mSigma_{F} \E\left[ \left( \vlambda_{j}-\E%
		\left[ \vlambda_{j}\right] \right) \left( \varsigma_{\nu,j}^{-1}-%
		\E\left[ \varsigma_{\nu,j}^{-1}\right] \right) \right] +0 \\
		&=\cov\left[ \varsigma_{\nu,i}^{-1},\vlambda%
		_{i}^{\prime }\right] \mSigma_{F}\cov\left[
		\vlambda_{j},\varsigma_{\nu,j}^{-1}\right]  \\
		&=\cov\left[ \varsigma_{\nu,1}^{-1},\vlambda%
		_{1}^{\prime }\right] \mSigma_{F}\cov\left[
		\vlambda_{1},\varsigma_{\nu,1}^{-1}\right] ,
	\end{align*}%
	where we use $i.i.d.$-ness of $\vlambda_{i}$ and $\sigma
	_{i}^{2}$ to conclude that the covariance $\cov\left[
	\varsigma_{\nu,i}^{-1},\vlambda_{i}^{\prime }\right] $ is the
	same across cross-sections.
	
	Lastly, consider the variance of $\sqrt{\frac{1}{2TN\left( N-1\right) }}\sum_{i=2}^{N}\sum_{j=1}^{i-1}\left( \varsigma_{\nu,i}^{-1}-\E\left[\varsigma_{v}^{-1}\right] \right) \vnu_{i}^{\prime }\vnu_{j}\left( \varsigma_{\nu,j}^{-1}-\E\left[ \varsigma_{v}^{-1}\right]
	\right) $ in the case where $\cov\left[ \varsigma_{\nu,i}^{-1},\vlambda_{i}^{\prime }\right] =0$. We have
	\begin{align*}
		&\var\left[ \sum_{i=2}^{N}\sum_{j=1}^{i-1}\left( \varsigma_{\nu,i}^{-1}-\E\left[ \varsigma_{\nu,i}^{-1}\right] \right) \vnu_{i}^{\prime}\vnu_{j}\left( \varsigma_{\nu,j}^{-1}-\E\left[\varsigma_{\nu,j}^{-1}\right] \right) \right] \\
		&=\E\left\{ \left[\sum_{i=2}^{N}\sum_{j=1}^{i-1}\left( \varsigma_{\nu,i}^{-1}-\E\left[\varsigma_{\nu,i}^{-1}\right] \right) \vnu_{i}^{\prime }\vnu_{j}\left( \varsigma_{\nu,j}^{-1}-\E\left[ \varsigma_{\nu,j}^{-1}\right] \right) \right] ^{2}\right\} -0 \\
		&=\sum_{i=2}^{N}\sum_{j=1}^{i-1}\E\left\{ \left[ \left( \varsigma_{\nu,i}^{-1}-\E\left[ \varsigma_{\nu,i}^{-1}\right] \right) \vnu_{i}^{\prime}\vnu_{j}\left( \varsigma_{\nu,j}^{-1}-\E\left[\varsigma_{\nu,j}^{-1}\right] \right) \right] ^{2}\right\} .
	\end{align*}%
	Here, $\left( \vnu_{i}^{\prime }\vnu_{j}\right) ^{2}=\left[\left( \vlambda_{i}^{\prime }\mF^{\prime } + {\vepsi}_{i}^{\prime }\right) \left(\mF \vlambda_{j} + \vepsi_{j}\right) \right]^{2}$, where we only consider the leading term $\tr\left( \vlambda_{i}\vlambda_{i}^{\prime}\mF^{\prime}\mF \vlambda_{j}\vlambda_{j}^{\prime} \mF^{\prime}\mF\right)$ further. Here, it is possible to show that
	\begin{align*}
		\E \left[ tr\left( \vlambda_{i}\vlambda_{i}' \mF'\mF \vlambda_{j}\vlambda_{j}' \mF'\mF \right) \: | \: \vlambda_{i}, \vlambda_j, \sigma_{i}^2, \sigma_j^2 \right] \leq T^2 M \: \tr\left( \vlambda_{i}\vlambda_{i}' \vlambda_{j}\vlambda_{j}' \right),
	\end{align*}
	since the fourth-order moments of $\vf_{t}$ are bounded. By application of the Law of Iterated expectations (LIE), it then follows that
	\begin{align*}	
		&\sum_{i=2}^{N}\sum_{j=1}^{i-1}\E\left\{ \left( \varsigma_{\nu,i}^{-1} - \E\left[ \varsigma_{\nu,i}^{-1}\right] \right)^{2} tr\left(\vlambda_{i}\vlambda_{i}^{\prime}\mF^{\prime}\mF\vlambda_{j}\vlambda_{j}^{\prime }\mF^{\prime}\mF\right) \left( \varsigma_{\nu,j}^{-1}-\E\left[ \varsigma_{\nu,j}^{-1}\right] \right)^{2}\right\}  \\
		&=M \:T^{2}\sum_{i=2}^{N}\sum_{j=1}^{i-1}\E\left\{ \left( \varsigma_{\nu,i}^{-1}-\E\left[ \varsigma_{\nu,i}^{-1}\right] \right)^{2}tr\left(\vlambda_{i}\vlambda_{i}^{\prime}\vlambda_{j}\vlambda_{j}^{\prime }\right) \left( \varsigma_{\nu,j}^{-1}-\E\left[ \varsigma_{\nu,j}^{-1}\right] \right) ^{2}\right\}  \\
		&=\largeO\left( N^{2}T^{2}\right) ,
	\end{align*}%
	assuming that the higher-order moment $\E\left\{ \left( \varsigma_{\nu,i}^{-1} - \E\left[ \varsigma_{\nu,i}^{-1}\right] \right)^{2}tr\left(\vlambda_{i}\vlambda_{i}^{\prime}\right) \right\} $ exists. Involving the square of the scaling $\sqrt{\frac{1}{2TN\left( N-1\right)}}$, we can conclude that the variance of
	\begin{equation}
		\sqrt{\frac{1}{2TN\left( N-1\right) }}\sum_{i=2}^{N}\sum_{j=1}^{i-1}\left(\varsigma_{\nu,i}^{-1}-\E\left[\varsigma_{v}^{-1}\right] \right)\vnu_{i}^{\prime }\vnu_{j}\left(\varsigma_{\nu,j}^{-1}-\E\left[ \varsigma_{v}^{-1}\right] \right)
	\end{equation}
	diverges at rate $T$.
\end{proof}
\bigskip
\begin{proof}[\textbf{Proof of Proposition \ref{prop:CD_CCE_H1}}] \ \\
	Let $k_{N,T}=\sqrt{\frac{1}{2TN\left( N-1\right) }}$ and $\vvarpi_{i}=\left[ \varsigma_{\nu,i}^{-1}\vnu_{i}-\mD_{i}\left(
	\overline{\mC}^{\left( 1\right) }\right) ^{-1}\left( N^{-1}\sum_{\ell
		=1}^{N}\vlambda_{\ell }^{\left( 1\right) }\varsigma _{\nu,\ell}^{-1}\right) \right]$, where $\vnu_{i}$ and $\varsigma_{\nu,i}^2$ are defined in the discussion following equation \eqref{eq:CCE_resH1} and where
	\begin{align*}
		\mD_i &= \mF \begin{bmatrix}
			\vlambda_i^{(2)}, & \mLambda_i^{(2)}
		\end{bmatrix}
		+ \begin{bmatrix}
			\vepsi_{i}, & \ve_i
		\end{bmatrix}.
	\end{align*}
	Using the same steps as in the Proof of Theorem
	\ref{theorem::CCE}, we can decompose the CD test statistic as
	\begin{equation*}
		CD_{\mathbb{H}_{1}}=I-II
	\end{equation*}%
	where
	\begin{eqnarray*}
		I &=&2k_{N,T}\sum_{i=2}^{N}\sum_{j=1}^{i-1}\vvarpi_{i}^{\prime }%
		\vvarpi_{j}+k_{N,T}\sum_{i=1}^{N}\vvarpi%
		_{i}^{\prime }\vvarpi_{i}-k_{N,T}\sum_{i,j}^{N}\vvarpi_{i}^{\prime }\mP_{\widehat{\mF}}\vvarpi_{j}
		\\
		II &=&k_{N,T}\sum_{i=1}^{N}\varsigma _{\nu,i}^{-2}\left( \vnu	_{i}^{\prime }-\vlambda_{i}^{\left( 1\right) \prime }\left(
		\overline{\mC}^{\left( 1\right) \prime }\right) ^{-1}\overline{\mD}%
		^{\prime }\right) \mM_{\widehat{\mF}}\left( \vnu_{i}-%
		\overline{\mD}\left( \overline{\mC}^{\left( 1\right) }\right) ^{-1}%
		\vlambda_{i}^{\left( 1\right) }\right) .
	\end{eqnarray*}%
	As shown in Lemmas \ref{lem:sumij_varpii_PF_varpij} and \ref{lem:II_power}, it holds that
	\begin{align*}
		k_{N,T}\sum_{i,j}^{N}\vvarpi_{i}^{\prime }\mP_{\widehat{\mF}}\vvarpi_{j} &= \largeO_{P}\left[\max\left(N^{-1},T^{-1}\right)\left(k_{N,T}\sum_{i,j}^{N}\vvarpi_{i}^{\prime}\vvarpi_{j}\right)\right]
	\end{align*}
	and
	\begin{align*}
		II &=  k_{N,T}\sum_{i=1}^{N}\varsigma_{\nu}^{-2}\vnu_{i}^{\prime}\vnu_{i}+\largeO_{P}\left(N^{-1/2}\sqrt{T}\right)+\largeO_{P}\left(T^{-1/2}\right).
	\end{align*}
	Furthermore, the expression $k_{N,T}\sum_{i=1}^{N}\vvarpi%
	_{i}^{\prime }\vvarpi_{i}$ is written out
	\begin{eqnarray*}
		k_{N,T}\sum_{i=1}^{N}\vvarpi_{i}^{\prime }\vvarpi%
		_{i}, &=&k_{N,T}\sum_{i=1}^{N}\varsigma _{\nu,i}^{-2}\vnu%
		_{i}^{\prime }\vnu_{i}\\
		&&+\sqrt{\frac{TN}{\left( N-1\right) }}%
		\left( N^{-1}\sum_{\ell =1}^{N}\varsigma _{\nu,\ell}^{-1}\vlambda%
		_{\ell }^{\left( 1\right) \prime }\right) \left( \overline{\mC}^{\left(
			1\right) \prime }\right) ^{-1}\left( \left( NT\right) ^{-1}\sum_{i=1}^{N}%
		\mD_{i}^{\prime }\mD_{i}\right)\\
		&&\times\left( \overline{\mC}%
		^{\left( 1\right) }\right) ^{-1}\left( N^{-1}\sum_{\ell =1}^{N}\vlambda_{\ell }^{\left( 1\right) }\varsigma _{\nu,\ell}^{-1}\right) \\
		&&-2\sqrt{\frac{TN}{\left( N-1\right) }}\left( NT\right)
		^{-1}\sum_{i=1}^{N}\varsigma _{\nu,i}^{-1}\vnu_{i}^{\prime }%
		\mD_{j}\left( \overline{\mC}^{\left( 1\right) }\right) ^{-1}\left(
		N^{-1}\sum_{\ell =1}^{N}\vlambda_{\ell }^{\left( 1\right)
		}\varsigma _{\nu,\ell}^{-1}\right) .
	\end{eqnarray*}%
	The first term on the right-hand side above is cancelled out by the leading term in $II$. The two remaining terms constitute expressions $\Phi _{1,\mathbb{H}_{1}}$ and $\Phi _{2,\mathbb{H}_{1}}$, respectively. Next,
	consider the properties of $2k_{N,T}\sum_{i=2}^{N}\sum_{j=1}^{i-1}%
	\vvarpi_{i}^{\prime }\vvarpi_{j}$, which we write	as
	\begin{align}
		\label{eq:lead_stoch_IFE}
		& 2k_{N,T}\sum_{i=2}^{N}\sum_{j=1}^{i-1}\vvarpi_{i}^{\prime }%
		\vvarpi_{j} \notag \\
		=&2k_{N,T}\sum_{i=2}^{N}\sum_{j=1}^{i-1}\varsigma
		_{\nu,i}^{-1}\vnu_{i}^{\prime }\vnu_{j}\varsigma
		_{\nu,j}^{-1}-4k_{N,T}\sum_{i=2}^{N}\sum_{j=1}^{i-1}\varsigma _{\nu,i}^{-1}%
		\vnu_{i}^{\prime }\mD_{j}\left( \overline{\mC}^{\left(
			1\right) }\right) ^{-1}\left( N^{-1}\sum_{\ell =1}^{N}\vlambda%
		_{\ell }^{\left( 1\right) }\varsigma _{\nu,\ell}^{-1}\right) \notag
		 \\
		+&\left( N^{-1}\sum_{\ell =1}^{N}\varsigma _{\nu,\ell}^{-1}\vlambda_{\ell }^{\left( 1\right) \prime }\right) \left( \overline{\mC}%
		^{\left( 1\right) \prime }\right) ^{-1}\left(
		2k_{N,T}\sum_{i=2}^{N}\sum_{j=1}^{i-1}\mD_{i}^{\prime }\mD%
		_{j}\right) \left( \overline{\mC}^{\left( 1\right) }\right) ^{-1}\left(
		N^{-1}\sum_{\ell =1}^{N}\vlambda_{\ell }^{\left( 1\right)
		}\varsigma _{\nu,\ell}^{-1}\right) .  
	\end{align}%
	The order in probability of all three terms above is determined by the
	expected values of the elements within the double sum $\sum_{i=2}^{N}%
	\sum_{j=1}^{i-1}$. Conditioning on the values taken by factor loadings and
	error variances, we have
	\begin{eqnarray*}
		\E\left[ \nu _{i,t}\nu _{j,t}~|~\vlambda_{i},%
		\vlambda_{j},\sigma _{i}^{2},\sigma _{j}^{2}\right] &=&%
		\vlambda_{i}^{\left( 2\right) \prime }\mSigma_{F,22}\vlambda_{j}^{\left( 2\right) }, \\
		\E\left[ \mD_{i,t}\mD_{j,t}^{\prime }~|~\vlambda_{i},\vlambda_{j},\mLambda_{i},\mLambda_{j},\sigma _{i}^{2},\sigma _{j}^{2}%
		\right] &=&\left[
		\begin{array}{c}
			\vlambda_{i}^{\left( 2\right) \prime } \\
			\mLambda_{i}^{\left( 2\right) \prime }%
		\end{array}%
		\right] \mSigma_{F,22}\left[ \vlambda_{j}^{\left(
			2\right) },~\mLambda_{j}^{\left( 2\right) }\right] , \\
		\E\left[ \nu _{i,t}\mD_{jt}^{\prime }|~\vlambda%
		_{i},\vlambda_{j},\mLambda_{j},\sigma _{i}^{2},\sigma _{j}^{2}\right] &=&%
		\vlambda_{i}^{\left( 2\right) \prime }\mSigma_{F,22}\left[ \vlambda_{j}^{\left( 2\right) },~\mLambda_{j}^{\left( 2\right) }\right] ,
	\end{eqnarray*}%
	for $i\neq j$. The zero mean property of $\vlambda_{i}^{\left(
		2\right) }$ and $\mLambda_{i}^{\left( 2\right) }$ implies that $%
	\E\left[ \mD_{i,t}\mD_{j,t}^{\prime }\right] =\mZeros$. This result carries over to $\E\left[ \varsigma _{\nu,i}^{-1}%
	\vnu_{i}^{\prime }\mD_{j}\right] $ since independence of
	$\vlambda_{i}$ and $\sigma _{i}^{2}$ over cross-sections
	implies that we can write $\E\left[ \varsigma _{\nu,i}^{-1}\vnu_{i}^{\prime }\mD_{j}\right] =$ $\E\left[ \varsigma
	_{\nu,i}^{-1}\vlambda_{i}^{\left( 2\right) \prime }\right]
	\mSigma_{F,22}\E\left\{ \left[ \vlambda%
	_{j}^{\left( 2\right) },~\mLambda_{j}^{\left( 2\right) }\right]
	\right\} =\vzeros$. Consequently, the last two terms on the right-hand
	side of \eqref{eq:lead_stoch_IFE} are asymptotically centered around zero
	since their components $\overline{\mC}^{\left( 1\right) \prime }$ and $%
	N^{-1}\sum_{\ell =1}^{N}\varsigma _{\nu,\ell}^{-1}\vlambda_{\ell
	}^{\left( 1\right) \prime }$ converge to nonstochastic limiting expressions.
	Lastly, by the LIE, we have%
	\begin{equation*}
		\E\left[ \varsigma _{\nu,i}^{-1}\nu _{i,t}\nu _{j,t}\varsigma
		_{\nu,j}^{-1}\right] =\cov\left[ \varsigma _{\nu,i}^{-1},~%
		\vlambda_{i}^{\left( 2\right) \prime }\right] \mSigma_{F,22}\cov\left[ \vlambda_{j}^{\left(
			2\right) },~\varsigma _{\nu,j}^{-1}\right] ,
	\end{equation*}%
	and consequently
	\begin{equation*}
		\E\left[ 2k_{N,T}\sum_{i=2}^{N}\sum_{j=1}^{i-1}\varsigma _{\nu,i}^{-1}%
		\vnu_{i}^{\prime }\vnu_{j}\varsigma _{\nu,j}^{-1}%
		\right] =N\sqrt{T}\cov\left[ \varsigma _{\nu,1}^{-1},~%
		\vlambda_{1}^{\left( 2\right) \prime }\right] \mSigma_{F,22}\cov\left[ \vlambda_{1}^{\left(
			2\right) },~\varsigma _{\nu,1}^{-1}\right] ,
	\end{equation*}%
	which parallels results obtained in the proof of Proposition \ref{prop:CD_2WFE_H1}. It will be shown below that variation around this expected
	value is of lower order in probability. Assuming this for now, we state that%
	\begin{equation*}
		\plim_{N,T\rightarrow \infty }N^{-1}T^{-1/2}\sqrt{\frac{2}{TN\left(
				N-1\right) }}\sum_{i=2}^{N}\sum_{j=1}^{i-1}\vvarpi_{i}^{\prime }%
		\vvarpi_{j}=\cov\left[ \varsigma _{\nu,i}^{-1},~%
		\vlambda_{i}^{\left( 2\right) \prime }\right] \mSigma_{F,22}\cov\left[ \vlambda_{j}^{\left(
			2\right) },~\varsigma _{\nu,j}^{-1}\right] .
	\end{equation*}%
	A consequence of this result, as well as those on the second and third terms
	in Eq. \eqref{eq:lead_stoch_IFE} is that $2k_{N,T}\sum_{i=2}^{N}%
	\sum_{j=1}^{i-1}\vvarpi_{i}^{\prime }\vvarpi_{j}$
	is asymptotically centered around $0$ if $\varsigma _{\nu,i}^{-1}$\ and $%
	\vlambda_{i}^{\left( 2\right) \prime }$ are uncorrelated. In
	this special case, it is required to investigate the rate at which the
	variance of $2k_{N,T}\sum_{i=2}^{N}\sum_{j=1}^{i-1}\vvarpi%
	_{i}^{\prime }\vvarpi_{j}$ diverges. First, consider
	\begin{equation*}
		\var\left[ 2k_{N,T}\sum_{i=2}^{N}\sum_{j=1}^{i-1}\varsigma
		_{\nu,i}^{-1}\vnu_{i}^{\prime }\vnu_{j}\varsigma
		_{\nu,j}^{-1}\right] =\E\left[ 4k_{N,T}^{2}\left(
		\sum_{i=2}^{N}\sum_{j=1}^{i-1}\varsigma _{\nu,i}^{-1}\vnu%
		_{i}^{\prime }\vnu_{j}\varsigma _{\nu,j}^{-1}\right) ^{2}\right] .
	\end{equation*}%
	Apart from a superscript $^{\left( 2\right) }$ on factors and factor
	loadings as well as the absence of the expected value of $\varsigma
	_{\nu,i}^{-1}$, the term $\varsigma _{\nu,i}^{-1}\vnu_{i}^{\prime }%
	\vnu_{j}\varsigma _{\nu,j}^{-1}$ is identical to the one occurred
	in the proof of Proposition \ref{prop:CD_2WFE_H1}. Hence, we can follow the same reasoning to
	conclude that
	\begin{equation*}
		\var\left[ \sqrt{\frac{2}{TN\left( N-1\right) }}%
		\sum_{i=2}^{N}\sum_{j=1}^{i-1}\varsigma _{\nu,i}^{-1}\vnu%
		_{i}^{\prime }\vnu_{j}\varsigma _{\nu,j}^{-1}\right] =\largeO\left(
		T\right) .
	\end{equation*}%
	Next, consider the middle term of the third expression in Eq. \eqref{eq:lead_stoch_IFE}. After vectorizing this matrix, we can write
	\begin{eqnarray*}
		\var\left[ \vec\left( \sum_{i=2}^{N}\sum_{j=1}^{i-1}\mD_{i}^{\prime }\mD_{j}\right) \right] &=&\E\left[ \vec\left(
		\sum_{i=2}^{N}\sum_{j=1}^{i-1}\mD_{i}^{\prime }\mD_{j}\right)
		\vec\left( \sum_{i^{\prime }=2}^{N}\sum_{j^{\prime }=1}^{i^{\prime }-1}%
		\mD_{i^{\prime }}^{\prime }\mD_{j^{\prime }}\right) ^{\prime }%
		\right] \\
		&=&\sum_{i=2}^{N}\sum_{j=1}^{i-1}\E\left[ \vec\left( \mD%
		_{i}^{\prime }\mD_{j}\right) \vec\left( \mD_{i}^{\prime }%
		\mD_{j}\right) ^{\prime }\right] .
	\end{eqnarray*}%
	Recall here that
	\begin{eqnarray*}
		\mD_{i}^{\prime }\mD_{j} &=&\left[
		\begin{array}{c}
			\vlambda_{i}^{\left( 2\right) } \\
			\mLambda_{i}^{\left( 2\right) \prime }%
		\end{array}%
		\right] \mF^{\left( 2\right) ^{\prime }}\mF^{\left( 2\right) }%
		\left[ \vlambda_{j}^{\left( 2\right) },\mLambda%
		_{j}^{\left( 2\right) }\right] +\left[
		\begin{array}{c}
			\vepsi_{i}^{\prime } \\
			\mE_{i}^{\prime }%
		\end{array}%
		\right] \left[ \vepsi_{j},~\mE_{j}\right] \\
		&&+\left[
		\begin{array}{c}
			\vlambda_{i}^{\left( 2\right) } \\
			\mLambda_{i}^{\left( 2\right) \prime }%
		\end{array}%
		\right] \mF^{\left( 2\right) ^{\prime }}\left[ \vepsi_{j},~\mE_{j}\right] +\left[
		\begin{array}{c}
			\vepsi_{i}^{\prime } \\
			\mE_{i}^{\prime }%
		\end{array}%
		\right] \mF^{\left( 2\right) }\left[ \vlambda%
		_{j}^{\left( 2\right) },~\mLambda_{j}^{\left( 2\right) }\right]
		,
	\end{eqnarray*}%
	whose leading term is the first on the right-hand side of the first line and
	where we recall that $\mC_{i}^{\left( 2\right) }=\left[ \vlambda_{j}^{\left( 2\right) },~\mLambda_{j}^{\left( 2\right)
	}\right] $. Hence, the leading term of $\vec\left( \mD_{i}^{\prime }%
	\mD_{j}\right) \vec\left( \mD_{i}^{\prime }\mD%
	_{j}\right) ^{\prime }$ is written%
	\begin{equation*}
		\vec\left( \mD_{i}^{\prime }\mD_{j}\right) \vec\left( \mD%
		_{i}^{\prime }\mD_{j}\right) ^{\prime }=\left( \mC_{i}^{\left(
			2\right) }\otimes \mC_{j}^{\left( 2\right) }\right)^{\prime}
		\vec\left( \mF^{\left( 2\right) ^{\prime }}\mF^{\left( 2\right)
		}\right) \vec\left( \mF^{\left( 2\right) ^{\prime }}\mF^{\left(
			2\right) }\right) \left( \mC_{i}^{\left( 2\right) }\otimes \mC%
		_{j}^{\left( 2\right) }\right) ,
	\end{equation*}%
	which is nothing more than a multivariate version of the leading term in $%
	\varsigma _{\nu,i}^{-1}\vnu_{i}^{\prime }\vnu%
	_{j}\varsigma _{\nu,j}^{-1}$. The current case is even simpler due to the
	absence of the inverse error standard deviation $\varsigma _{\nu,i}^{-1}$.
	Consequently, finite fourth moments for factors and finite second moments
	for loadings together with the reasoning applied for the variance of $%
	2k_{N,T}\sum_{i=2}^{N}\sum_{j=1}^{i-1}\varsigma _{\nu,i}^{-1}\vnu%
	_{i}^{\prime }\vnu_{j}\varsigma _{\nu,j}^{-1}$ allow us to
	conclude that%
	\begin{equation*}
		\var\left[ \sqrt{\frac{2}{TN\left( N-1\right) }}\vec\left(
		\sum_{i=2}^{N}\sum_{j=1}^{i-1}\mD_{i}^{\prime }\mD_{j}\right) %
		\right] =\largeO_{P}\left( T\right) .
	\end{equation*}%
	Since the outer terms $\left( N^{-1}\sum_{\ell =1}^{N}\varsigma _{\ell
		,v}^{-1}\vlambda_{\ell }^{\left( 1\right) \prime }\right) $ and
	$\left( \overline{\mC}^{\left( 1\right) \prime }\right) ^{-1}$ converge to
	finite, non-stochastic components, this rate of divergence carries over to
	the variance of the entire term $\left( N^{-1}\sum_{\ell =1}^{N}\varsigma
	_{\nu,\ell}^{-1}\vlambda_{\ell }^{\left( 1\right) \prime
	}\right) \left( \overline{\mC}^{\left( 1\right) \prime }\right)
	^{-1}\left( 2k_{N,T}\sum_{i=2}^{N}\sum_{j=1}^{i-1}\mD_{i}^{\prime }%
	\mD_{j}\right) \left( \overline{\mC}^{\left( 1\right) }\right)
	^{-1}\left( N^{-1}\sum_{\ell =1}^{N}\vlambda_{\ell }^{\left(
		1\right) }\varsigma _{\nu,\ell}^{-1}\right) $.
	
	Lastly, consider the second term in \eqref{eq:lead_stoch_IFE}, whose middle
	term is of particular interest. The variance of this term is given by%
	\begin{eqnarray*}
		\var\left[ \sum_{i=2}^{N}\sum_{j=1}^{i-1}\varsigma
		_{\nu,i}^{-1}\vnu_{i}^{\prime }\mD_{j}\right] &=&\E%
		\left[ \sum_{i=2}^{N}\sum_{j=1}^{i-1}\sum_{i^{\prime }=2}^{N}\sum_{j^{\prime
			}=1}^{i^{\prime }-1}\mD_{j}^{\prime }\vnu_{i}
		\vnu_{i^{\prime }}^{\prime }\mD_{j^{\prime }}\varsigma
		_{\nu,i}^{-1}\varsigma _{i^{\prime },v}^{-1}\right] \\
		&=&\sum_{i=2}^{N}\sum_{j=1}^{i-1}\sum_{t,t^{\prime }}^{T}\E\left[
		\mD_{j,t}\nu _{i,t}\nu _{i,t^{\prime }}\mD_{j,t^{\prime
		}}^{\prime }\varsigma _{\nu,i}^{-2}\right] .
	\end{eqnarray*}%
	Similar to the previous two terms, the leading expression in the term above
	is given by
	\begin{equation*}
		\E\left[ \mC_{j}^{\left( 2\right) \prime }\mF_{t}^{\left( 2\right)}
		\mF_{t}^{\left( 2\right)\prime }\vlambda_{i}^{\left( 2\right)}\vlambda{i}^{\left( 2\right)\prime}\mF_{t}^{\left( 2\right)}\mF_{t}^{\left( 2\right)\prime }\mC_{j}^{\left(
			2\right) }\varsigma _{\nu,i}^{-2}\right] \leq M~\E\left[ \mC%
		_{j}^{\left( 2\right) \prime }\vlambda_{i}^{\left( 2\right)}\varsigma _{\nu,i}^{-2}%
		\vlambda_{i}^{\left( 2\right)\prime}\mC_{j}^{\left( 2\right) }\right]
	\end{equation*}%
	where the upper bound on the right-hand side above is established by the
	LIE\ and boundedness of the fourth-order moments of $\mF_{t}$.
	Additionally assuming that $\E\left[ \vlambda_{i}^{\left( 2\right)}\varsigma _{\nu,i}^{-2}\vlambda_{i}^{\left( 2\right)\prime }\right] $ is
	bounded and conditioning on $\mC_{j}^{\left( 2\right) }$, we arrive at
	
	\begin{equation*}
		\sum_{i=2}^{N}\sum_{j=1}^{i-1}\sum_{t,t^{\prime }}^{T}\E\left[
		\mD_{j,t}\nu _{i,t}\nu _{i,t^{\prime }}\mD_{j,t^{\prime
		}}^{\prime }\varsigma _{\nu,i}^{-2}\right] =\largeO\left[ \left( NT\right)
		^{2}\right] ,
	\end{equation*}%
	implying that%
	\begin{equation*}
		\var\left[ \sqrt{\frac{2}{TN\left( N-1\right) }}%
		\sum_{i=2}^{N}\sum_{j=1}^{i-1}\varsigma _{\nu,i}^{-1}\vnu%
		_{i}^{\prime }\mD_{j}\right] =\largeO\left( T\right) ,
	\end{equation*}%
	which carries over to the variance of $k_{N,T}\sum_{i=2}^{N}\sum_{j=1}^{i-1}%
	\varsigma _{\nu,i}^{-1}\vnu_{i}^{\prime }\mD_{j}\left(
	\overline{\mC}^{\left( 1\right) }\right) ^{-1}\left( N^{-1}\sum_{\ell
		=1}^{N}\vlambda_{\ell }^{\left( 1\right) }\varsigma _{\nu,\ell}^{-1}\right) $.
\end{proof}

\clearpage

\subsection{Auxiliary lemmas for Section \ref{section::proofs_power}}

\begin{lemma}
	\label{lem:facpartition}
	Suppose that the true model is \eqref{eq:themodelFactor} and that Assumptions \ref{ass:errors}--\ref{ass:indep} and \ref{ass:failrankcond} hold. Also, assume that the factor estimator $\widehat{\mF}$ is given by  $\widehat{\mF} = \left[ \overline{\vy}, \overline{\mX} \right]$. Then we can without loss of generality assume that the true factors and loadings can be partitioned into
	\begin{align*}
		\vf_t &= \begin{bmatrix}
			\vf_t^{(1)} \\ \vf_t^{(2)}
		\end{bmatrix}& \text{ and } & & \vlambda_{i} = \begin{bmatrix}
			\vlambda_{i}^{(1)} \\ \vlambda_{i}^{(2)}
		\end{bmatrix} & & \text{ and } & & \mLambda_i = \begin{bmatrix}
			\mLambda_i^{(1)} \\ \mLambda_i^{(2)}
		\end{bmatrix}
	\end{align*}
	such that $\E\left[\vf_t^{(1)} \vf_t^{(2)\prime} \right] = \mZeros$, $\E \left[ \vlambda_{i}^{(2)} \right] = \vzeros$ and  $\E \left[ \mLambda_i^{(2)} \right] = \mZeros$. It holds that $\vf^{(1)}$ and $\vf^{(2)}$ are $(m+1) \times 1$ and $(r-m-1) \times 1$, respectively. $\vlambda_{i}^{(1)}$ and $\vlambda_{i}^{(2)}$ are $(m+1) \times 1$ and $(r-m-1) \times 1$. Lastly, $\mLambda_i^{(1)}$ and $\mLambda_i^{(2)}$ are $(m+1) \times m$ and $(r-m-1) \times m$.
\end{lemma}
\begin{proof}[Proof of Lemma \ref{lem:facpartition}]\ \\
	Under a violation of the rank condition, cross-section averages of all
	variables will only identify the $m+1$-dimensional linear combination $%
	\vf_{t}^{\left( 1\right) }=\mC_{0}^{\prime }\vf_{t}$ of
	the $r$ true common factors. The $\ r\times \left( m+1\right) $ matrix $%
	\mC_{0}$ is given by%
	\begin{align*}
		\mC_{0}
		&=\left[
		\E\left[ \vlambda_{i}\right],\E\left[
		\mLambda_{i}\right]\right]\mB.
	\end{align*}%
	Now define the $r \times (m+1)$ matrix $\mPhi _{1}=\E\left[ \vf_{t}\vf%
	_{t}^{\left( 1\right) \prime }\right] \left( \E\left[ \vf%
	_{t}^{\left( 1\right) }\vf_{t}^{\left( 1\right) \prime }\right]
	\right) ^{-1}$ and the $r \times 1$ vector $\vf_{t}^{\bot }=\vf_{t}-\mPhi_{1}\vf%
	_{t}^{\left( 1\right) }$ and note that by construction it holds that $%
	\E\left[ \vf_{t}^{\bot }\vf_{t}^{\left( 1\right)
		\prime }\right] =\mZeros$. Using these two expressions, we can decompose
	\begin{equation}
		\vlambda_{i}^{\prime }\vf_{t}=\vlambda%
		_{i}^{\prime }\mPhi_{1}\vf_{t}^{\left( 1\right) }+\vlambda_{i}^{\prime }\vf_{t}^{\bot }\text{.}  \label{eq:commoncomp_decomp}
	\end{equation}%
	The expression above decomposes the common component of a random variable
	that has a factor structure into the variation due to the $\left( m+1\right)
	$-dimensional linear combination of the $r$ true factors that is identified
	by a cross-section average-based factor estimator and residual common
	variation that is not picked up. This residual source of common variation is
	bound to be part of the error term in a misspecified factor model.
	
	Next, consider the product of $\mC_{0}^{\prime }$ and $\vf_{t}^{\bot }$
	
	\begin{eqnarray*}
		\mC_{0}^{\prime }\vf_{t}^{\bot } &=&\mC_{0}^{\prime }%
		\vf_{t}-\mC_{0}^{\prime }\E\left[ \vf_{t}%
		\vf_{t}^{\left( 1\right) \prime }\right] \left( \E\left[
		\vf_{t}^{\left( 1\right) }\vf_{t}^{\left( 1\right) \prime }%
		\right] \right) ^{-1}\vf_{t}^{\left( 1\right) } \\
		&=&\vf_{t}^{\left( 1\right) }-\E\left[ \mC_{0}^{\prime
		}\vf_{t}\vf_{t}^{\left( 1\right) \prime }\right] \left(
		\E\left[ \vf_{t}^{\left( 1\right) }\vf_{t}^{\left(
			1\right) \prime }\right] \right) ^{-1}\vf_{t}^{\left( 1\right) } \\
		&=&\vf_{t}^{\left( 1\right) }-\vf_{t}^{\left( 1\right) } \\
		&=&\vzeros\text{,}
	\end{eqnarray*}%
	In order to appreciate the implications of this result, let $\mathrm{col}%
	\left( \mC_{0}\right) $ denote the column space of $\mC_{0}$
	and $\ker \left( \mC_{0}^{\prime }\right) $ the kernel of $\mC%
	_{0}^{\prime }$. The result reported in the last sequence of equations above
	states that $\vf_{t}^{\bot }\in \ker \left( \mC_{0}^{\prime
	}\right) $. Note now that $\dim \left( \mathrm{col}\left( \mC%
	_{0}\right) \right) +\dim \left( \ker \left( \mC_{0}^{\prime }\right)
	\right) =r$ \citep[see e.g.][Exercise 4.3(c)]{abadir2005matrix} and that $%
	\dim \left( \mathrm{col}\left( \mC_{0}\right) \right) =\dim \left(
	\mathrm{col}\left( \mC_{0}^{\prime }\right) \right) $ \citep[Exercise 4.5(a)]{abadir2005matrix}. As a result of the rank condition on $%
	\mC_{0}$, we have $\dim \left( \mathrm{col}\left( \mC%
	_{0}^{\prime }\right) \right) =m+1$ and hence $\dim \left( \ker \left(
	\mC_{0}^{\prime }\right) \right) =$ $r-m-1$. Consequently, $\vf%
	_{t}^{\bot }$ is an element in the $s=$ $r-m-1$-dimensional vector space
	orthogonal to $\mC_{0}$. This implies that, given an arbitrary basis
	for $\ker \left( \mC_{0}^{\prime }\right) $, here denoted by the $%
	r\times s$-dimensional matrix $\mPhi _{2}$, we can express $\vf%
	_{t}^{\bot }$ as $\vf_{t}^{\bot }=\mPhi _{2}\vf_{t}^{\left(
		2\right) }$ where $\vf_{t}^{\left( 2\right) }$ is some $s\times 1$
	vector of coordinates. Plugging this expression into \eqref{eq:commoncomp_decomp} results in
	\begin{eqnarray*}
		\vlambda_{i}^{\prime }\vf_{t} &=&\vlambda%
		_{i}^{\prime }\mPhi _{1}\vf_{t}^{\left( 1\right) }+\vlambda_{i}^{\prime }\mPhi _{2}\vf_{t}^{\left( 2\right) } \\
		&=&\vlambda_{i}^{\prime }\left[
		\begin{array}{cc}
			\mPhi _{1}, & \mPhi _{2}%
		\end{array}%
		\right] \left[
		\begin{array}{c}
			\vf_{t}^{\left( 1\right) } \\
			\vf_{t}^{\left( 2\right) }%
		\end{array}%
		\right] \\
		&=&\left[
		\begin{array}{cc}
			\vlambda_{i}^{\left( 1\right) \prime }, & \vlambda%
			_{i}^{\left( 2\right) \prime }%
		\end{array}%
		\right] \left[
		\begin{array}{c}
			\vf_{t}^{\left( 1\right) } \\
			\vf_{t}^{\left( 2\right) }%
		\end{array}%
		\right]
	\end{eqnarray*}
	
	From the previously mentioned result $\E\left[ \vf_{t}^{\bot }%
	\vf_{t}^{\left( 1\right) \prime }\right] =\mZeros$ it follows that
	$\E\left[ \vf_{t}^{\left( 2\right) }\vf_{t}^{\left(
		1\right) \prime }\right] =\mZeros$. Furthermore, given that $\E%
	\left[ \vlambda_{i}\right] =\mC_{0}\mB^{-1}\left[1,\vzeros'\right]^{\prime }$ and that $\vlambda_{i}^{\left( 2\right)
		\prime }=$ $\vlambda_{i}^{\prime }\mPhi _{2}$ where $\mPhi _{2}$
	is a basis for $\ker \left( \mC_{0}^{\prime }\right) $,
	\begin{eqnarray*}
		\E\left[ \vlambda_{i}^{\left( 2\right) }\right]
		&=&\mPhi _{2}^{\prime }\mC_{0}\mB^{-1}\left[1,\vzeros'\right]^{\prime }\\
		&=&\vzeros.
	\end{eqnarray*}
	Analogously, we can define $\left[\mLambda_{i}^{( 1)\prime}, \mLambda_{i}^{(2)\prime} \right] = \mLambda_{i}^{\prime }\left[\mPhi _{1}, \mPhi _{2}
	\right]$ where $\E\left[ \mLambda_{i}^{\left( 2\right) }\right] = \mZeros$ by the same reasoning as above.
	
	Consequently, whenever the factor estimates include cross-section
	averages of the dependent variable $\mY$, we can rewrite the common
	component as being generated by a linear combination of the true factors and
	the true loadings such that the two moment conditions above are satisfied.
	The matrix $\mPhi =\left( \mPhi _{1},\mPhi _{2}\right) $ which links the
	loadings $\vlambda_{i}^{\prime }$ ($\mLambda_{i}^{\prime }$) and $\left[\vlambda_{i}^{\left( 1\right) \prime }, \vlambda_{i}^{\left( 2\right) \prime}\right] $ ($\left[\mLambda_{i}^{\left( 1\right) \prime }, \mLambda_{i}^{\left( 2\right) \prime}\right] $) is of full rank by the rank assumptions on $\E\left[
	\vf_{t}\vf_{t}^{\prime }\right] $ as well as $\mC_{0}$.
	Hence, $\mPhi $ is an invertible rotation and we can write $\left[
	\vlambda_{i}^{\left( 1\right) \prime },\vlambda_{i}^{\left(
		2\right) \prime }\right] \left[ \vf_{t}^{\left( 1\right) \prime },
	\vf_{t}^{\left( 2\right) \prime }\right] ^{\prime }=\vlambda_{i}^{\prime }\mPhi \mPhi ^{-1}\vf_{t}$. Now note that in
	factor models it is only the common component $\vlambda%
	_{i}^{\prime }\vf_{t}$ that is identified whereas the loadings $%
	\vlambda_{i}$ and factors $\vf_{t}$ are identified only
	up to an invertible rotation. As argued above, the matrix $\mPhi $ is
	invertible. Hence, irrespective of the true values of $\vlambda	_{i}$, $\mLambda_{i}$ and $\vf_{t}$ (subject to the assumptions we make on them), we
	can assume that the data are generated by $\mPhi ^{\prime}	\vlambda_{i}$,  $\mPhi ^{\prime}	\mLambda_{i}$ and $\mPhi^{-1}\vf_{t}$ instead.
\end{proof}

\bigskip

\begin{lemma}
	\label{lem:sumij_varpii_PF_varpij}
	Under Assumptions \ref{ass:errors}-\ref{ass:rank} and $\E[|\varsigma_{\nu,i}^{-1}|^{4}]<M$,
	\begin{align}
		k_{N,T}\sum_{i,j}^{N}\vvarpi_{i}^{\prime}\mP_{\widehat{\mF}}\vvarpi_{j} = \largeO_{P}\left[\max\left(N^{-1},T^{-1}\right)\left(k_{N,T}\sum_{i,j}^{N}\vvarpi_{i}^{\prime}\vvarpi_{j}\right)\right]
	\end{align}
	
\end{lemma}

\begin{proof}[Proof of Lemma \ref{lem:sumij_varpii_PF_varpij}.]
	Note that
	\[
	\vvarpi_{i}=\left[\varsigma_{\nu,i}^{-1}\vnu_{i}-\mD_{i}\left(\overline{\mC}^{\left(1\right)}\right)^{-1}\left(N^{-1}\sum_{\ell=1}^{N}\vlambda_{\ell}^{\left(1\right)}\varsigma_{\nu,\ell}^{-1}\right)\right]
	\]
	where
	\begin{align*}
		\vnu_{i} & =\vepsi_{i}+\mF^{\left(2\right)}\vlambda_{i}^{\left(2\right)},\\
		\mD_{i} & =\mU_{i}+\mF^{\left(2\right)}\mC_{i}^{\left(2\right)}.
	\end{align*}
	Then, we can rewrite
	\begin{align*}
		\vvarpi_{i} & =\left[\vepsi_{i}\varsigma_{\nu,i}^{-1}-\mU_{i}\left(\overline{\mC}^{\left(1\right)}\right)^{-1}\left(N^{-1}\sum_{\ell=1}^{N}\vlambda_{\ell}^{\left(1\right)}\varsigma_{\nu,\ell}^{-1}\right)\right]\\
		& +\mF^{\left(2\right)}\left[\vlambda_{i}^{\left(2\right)}\varsigma_{\nu,i}^{-1}-\mC_{i}^{\left(2\right)}\left(\overline{\mC}^{\left(1\right)}\right)^{-1}\left(N^{-1}\sum_{\ell=1}^{N}\vlambda_{\ell}^{\left(1\right)}\varsigma_{\nu,\ell}^{-1}\right)\right].
	\end{align*}
	Consequently, we can write out
	\begin{align}
		& k_{N,T}\sum_{i,j}^{N}\vvarpi_{i}^{\prime}\mP_{\widehat{\mF}}\textnormal{\ensuremath{\vvarpi_{j}}}\nonumber \\
		& \leq 2\sqrt{\frac{T}{2N\left(N-1\right)}}\left\lVert \left(T^{-1}\widehat{\mF}^{\prime}\widehat{\mF}\right)^{-1}\right\rVert \left\lVert T^{-1}\widehat{\mF}^{\prime}\mF^{\left(2\right)}\right\rVert^{2} \nonumber \\
		& \times \left\lVert \sum_{i=1}^{N}\vlambda_{i}^{\left(2\right)}\varsigma_{\nu,i}^{-1}-\sum_{i=1}^{N}\mC_{i}^{\left(2\right)}\left(\overline{\mC}^{\left(1\right)}\right)^{-1}\left(N^{-1}\sum_{\ell=1}^{N}\vlambda_{\ell}^{\left(1\right)}\varsigma_{\nu,\ell}^{-1}\right)\right\rVert ^{2}\nonumber \\
		& +2k_{N,T}\sum_{i,j}^{N}\left(\vepsi_{i}\varsigma_{\nu,i}^{-1}-\mU_{i}\left(\overline{\mC}^{\left(1\right)}\right)^{-1}\left(N^{-1}\sum_{\ell=1}^{N}\vlambda_{\ell}^{\left(1\right)}\varsigma_{\nu,\ell}^{-1}\right)\right)^{\prime}\mP_{\widehat{\mF}} \nonumber \\
		& \times \mP_{\widehat{\mF}} \left(\vepsi_{j}\varsigma_{\nu,j}^{-1}-\mU_{j}\left(\overline{\mC}^{\left(1\right)}\right)^{-1}\left(N^{-1}\sum_{\ell=1}^{N}\vlambda_{\ell}^{\left(1\right)}\varsigma_{\nu,\ell}^{-1}\right)\right),\label{eq:sumij_varpii_PF_varpij}
	\end{align}
	where the second line is of order $\largeO_{P}\left(T^{-1/2}\right)+\largeO_{P}\left(N^{-1/2}\right)+\largeO_{P}\left[N^{-1}\sqrt{T}\right]$
	by Lemma \ref{lem:I2} with $w_{i}=\varsigma_{\nu,i}^{-1}$. Concerning the
	first line:
	\begin{equation}
		\sqrt{\frac{T}{2N\left(N-1\right)}}\left\lVert \sum_{i=1}^{N}\vlambda_{i}^{\left(2\right)}\varsigma_{\nu,i}^{-1}-\sum_{i=1}^{N}\mC_{i}^{\left(2\right)}\left(\overline{\mC}^{\left(1\right)}\right)^{-1}\left(N^{-1}\sum_{\ell=1}^{N}\vlambda_{\ell}^{\left(1\right)}\varsigma_{\nu,\ell}^{-1}\right)\right\rVert ^{2}=\largeO_{P}\left(k_{N,T}\sum_{i,j}^{N}\vvarpi_{i}^{\prime}\textnormal{\ensuremath{\vvarpi_{j}}}\right).\label{eq:auxlem_order}
	\end{equation}
	To appreciate this result, we write
	\begin{align*}
		& \left\lVert \sum_{i=1}^{N}\vlambda_{i}^{\left(2\right)}\varsigma_{\nu,i}^{-1}-\sum_{i=1}^{N}\mC_{i}^{\left(2\right)}\left(\overline{\mC}^{\left(1\right)}\right)^{-1}\left(N^{-1}\sum_{\ell=1}^{N}\vlambda_{\ell}^{\left(1\right)}\varsigma_{\nu,\ell}^{-1}\right)\right\rVert ^{2}\\
		= & 3\left\lVert \sum_{i=1}^{N}\vlambda_{i}^{\left(2\right)}\varsigma_{\nu,i}^{-1}\right\rVert ^{2}+3\left\lVert \sum_{i=1}^{N}\mC_{i}^{\left(2\right)}\right\rVert ^{2}\left\lVert \left(\overline{\mC}^{\left(1\right)}\right)^{-1}\right\rVert ^{2}\left\lVert N^{-1}\sum_{\ell=1}^{N}\vlambda_{\ell}^{\left(1\right)}\varsigma_{\nu,\ell}^{-1}\right\rVert ^{2}
	\end{align*}
	and note that $N^{-2}\left\lVert \sum_{i=1}^{N}\vlambda_{i}^{\left(1\right)}\varsigma_{\nu,i}^{-1}\right\rVert ^{2}=\largeO_{P}\left(1\right)$
	by the same reasoning leading to \eqref{eq:lambda_barW}. Furthermore, $\left\lVert \sum_{i=1}^{N}\mC_{i}^{\left(2\right)}\right\rVert ^{2}=\largeO_{P}\left(N\right)$
	holds by the zero mean property of $\vlambda_{i}^{\left(2\right)}$
	and $\mLambda_{i}^{\left(2\right)}$. Consequently,
	\[
	\left\lVert \sum_{i=1}^{N}\mC_{i}^{\left(2\right)}\right\rVert ^{2}\left\lVert \left(\overline{\mC}^{\left(1\right)}\right)^{-1}\right\rVert ^{2}\left\lVert N^{-1}\sum_{\ell=1}^{N}\vlambda_{\ell}^{\left(1\right)}\varsigma_{\nu,\ell}^{-1}\right\rVert ^{2}=\largeO_{P}\left(N\right),
	\]
	which never has a higher order in probability than$\left\lVert \sum_{i=1}^{N}\vlambda_{i}^{\left(2\right)}\varsigma_{\nu,i}^{-1}\right\rVert ^{2}.$
	We prove this claim by noting that the expected value of $\left\lVert \sum_{i=1}^{N}\vlambda_{i}^{\left(2\right)}\varsigma_{\nu,i}^{-1}\right\rVert ^{2}$
	is given by
	\[
	\E\left[\left\lVert \sum_{i=1}^{N}\vlambda_{i}^{\left(2\right)}\varsigma_{\nu,i}^{-1}\right\rVert ^{2}\right]=\sum_{i=1}^{N}\E\left[\varsigma_{\nu,i}^{-2}\vlambda_{i}^{\left(2\right)\prime}\vlambda_{i}^{\left(2\right)}\right]+\sum_{i=1}^N \sum_{j\neq i}^{N}\E\left[\varsigma_{\nu,i}^{-1}\vlambda_{i}^{\left(2\right)\prime}\right]\E\left[\vlambda_{j}^{\left(2\right)}\varsigma_{\nu,j}^{-1}\right].
	\]
	It follows straightforwardly from Markov's inequality that
	\[
	\left\lVert \sum_{i=1}^{N}\vlambda_{i}^{\left(2\right)}\varsigma_{\nu,i}^{-1}\right\rVert ^{2}=\begin{cases}
		\largeO_{P}\left(N\right) & \text{if }\cov\left[\vlambda_{j}^{\left(2\right)};\varsigma_{\nu,j}^{-1}\right]=\vzeros\\
		\largeO_{P}\left(N^{2}\right) & otherwise
	\end{cases}.
	\]
	The reduction in the order in probability due to a zero correlation
	between $\vlambda_{i}^{\left(2\right)}$ and $\varsigma_{\nu,i}^{-1}$
	is the same as observed in the case of $k_{N,T}\sum_{i,j}^{N}\vvarpi_{i}^{\prime}\vvarpi_{j}$.
	Additionally using the scaling $\sqrt{\frac{T}{2N\left(N-1\right)}}$
	and using our result on the order of $k_{N,T}\sum_{i,j}^{N}\vvarpi_{i}^{\prime}\vvarpi_{j}$
	in the proof of Proposition 3, we arrive at
	\[
	\sqrt{\frac{T}{2N\left(N-1\right)}}\left\lVert \sum_{i=1}^{N}\vlambda_{i}^{\left(2\right)}\varsigma_{\nu,i}^{-1}\right\rVert ^{2}=\largeO_{P}\left(k_{N,T}\sum_{i,j}^{N}\vvarpi_{i}^{\prime}\vvarpi_{j}\right),
	\]
	which leads to result \eqref{eq:auxlem_order}. Now returning to the
	first line of expression \eqref{eq:sumij_varpii_PF_varpij}, we consider
	\begin{align}
		\label{eq:Fhat_Ftwo}
		&\left\lVert T^{-1}\widehat{\mF}^{\prime}\mF^{\left(2\right)}\right\rVert ^{2} \notag \\
		\leq & 3\left\lVert \overline{\mC}^{\left(1\right)\prime}\right\rVert ^{2}\left\lVert T^{-1}\mF^{\left(1\right)\prime}\mF^{\left(2\right)}\right\rVert ^{2}+9\left\lVert \overline{\mC}^{\left(2\right)\prime}\right\rVert ^{2}\left\lVert T^{-1}\mF^{\left(2\right)\prime}\mF^{\left(2\right)}\right\rVert ^{2}+k_{N,T}\sum_{i,j}^{N}\vvarpi_{i}^{\prime}\vvarpi_{j}\left\lVert T^{-1}\overline{\mU}^{\prime}\mF^{\left(2\right)}\right\rVert ^{2} \notag \\
		= & \largeO_{P}\left(T^{-1}\right)+\largeO_{P}\left(N^{-1}\right)+\largeO_{P}\left(\left(NT\right)^{-1}\right),
	\end{align}
	which holds by zero correlation between $\mF^{\left(1\right)}$
	and $\mF^{\left(2\right)}$, the zero expected value of $\overline{\mC}^{\left(2\right)}$
	and result \eqref{eq:UbarF}. Consequently, an upper bound for the order in probability
	of \eqref{eq:sumij_varpii_PF_varpij} is given by
	\[
	k_{N,T}\sum_{i,j}^{N}\vvarpi_{i}^{\prime}\mP_{\widehat{\mF}}\vvarpi_{j}=\largeO_{P}\left[\max\left[N^{-1},T^{-1}\right]\left(k_{N,T}\sum_{i,j}^{N}\vvarpi_{i}^{\prime}\vvarpi_{j}\right)\right].
	\]
\end{proof}

\bigskip

\begin{lemma}
	\label{lem:II_power}
	Under Assumptions \ref{ass:errors}-\ref{ass:rank} and $\E[|\varsigma_{\nu,i}^{-1}|^{4}]<M$,
	\begin{align*}
		&k_{N,T}\sum_{i=1}^{N}\varsigma_{\nu,i}^{-2}\left(\vnu_{i}^{\prime}-\vlambda_{i}^{\left(1\right)\prime}\left(\overline{\mC}^{\left(1\right)\prime}\right)^{-1}\overline{\mD}^{\prime}\right)\mM_{\widehat{\mF}}\left(\vnu_{i}-\overline{\mD}\left(\overline{\mC}^{\left(1\right)}\right)^{-1}\vlambda_{i}^{\left(1\right)}\right)\\
		&=  k_{N,T}\sum_{i=1}^{N}\varsigma_{\nu,i}^{-2}\vnu_{i}^{\prime}\vnu_{i}+\largeO_{P}\left(N^{-1/2}\sqrt{T}\right)+\largeO_{P}\left(T^{-1/2}\right)
	\end{align*}
\end{lemma}

\begin{proof}[Proof of Lemma \ref{lem:II_power}]
	The term that we are interested in can be expanded into
	\begin{align*}
		& k_{N,T}\sum_{i=1}^{N}\varsigma_{\nu,i}^{-2}\left(\vnu_{i}^{\prime}-\vlambda_{i}^{\left(1\right)\prime}\left(\overline{\mC}^{\left(1\right)\prime}\right)^{-1}\overline{\mD}^{\prime}\right)\mM_{\widehat{\mF}}\left(\vnu_{i}-\overline{\mD}\left(\overline{\mC}^{\left(1\right)}\right)^{-1}\vlambda_{i}^{\left(1\right)}\right)\\
		& = k_{N,T}\sum_{i=1}^{N}\varsigma_{\nu,i}^{-2}\vnu_{i}^{\prime}\vnu_{i}-2k_{N,T}\sum_{i=1}^{N}\varsigma_{\nu,i}^{-2}\vnu_{i}^{\prime}\overline{\mD}\left(\overline{\mC}^{\left(1\right)}\right)^{-1}\vlambda_{i}^{\left(1\right)}\\
		& + k_{N,T}\sum_{i=1}^{N}\varsigma_{\nu,i}^{-2}\vlambda_{i}^{\left(1\right)\prime}\left(\overline{\mC}^{\left(1\right)\prime}\right)^{-1}\overline{\mD}^{\prime}\overline{\mD}\left(\overline{\mC}^{\left(1\right)}\right)^{-1}\vlambda_{i}^{\left(1\right)}\\
		& + k_{N,T}\sum_{i=1}^{N}\varsigma_{\nu,i}^{-2}\left(\vnu_{i}^{\prime}-\vlambda_{i}^{\left(1\right)\prime}\left(\overline{\mC}^{\left(1\right)\prime}\right)^{-1}\overline{\mD}^{\prime}\right)\mP_{\widehat{\mF}}\left(\vnu_{i}-\overline{\mD}\left(\overline{\mC}^{\left(1\right)}\right)^{-1}\vlambda_{i}^{\left(1\right)}\right)\\
		& = k_{N,T}\sum_{i=1}^{N}\varsigma_{\nu,i}^{-2}\vnu_{i}^{\prime}\vnu_{i}-2a_{1}+a_{2}+a_{3},
	\end{align*}
	where we recall that $\vnu_{i} =\vepsi_{i}+\mF^{\left(2\right)}\vlambda_{i}^{\left(2\right)}$  and $\mD_{i} =\mU_{i}+\mF^{\left(2\right)}\mC_{i}^{\left(2\right)}$.
	Now consider $a_{3}$. Similar to the proof of Lemma \ref{lem:sumij_varpii_PF_varpij}, we can write
	\begin{align*}
		a_{3}& \leq 2k_{N,T}\sum_{i=1}^{N}\varsigma_{\nu,i}^{-2}\left(\vlambda_{i}^{\left(2\right)}-\overline{\mC}^{\left(2\right)}\left(\overline{\mC}^{\left(1\right)}\right)^{-1}\vlambda_{i}^{\left(1\right)}\right)^{\prime}\mF^{\left(2\right)\prime}\mP_{\widehat{\mF}}\mF^{\left(2\right)}\left(\vlambda_{i}^{\left(2\right)}-\overline{\mC}^{\left(2\right)}\left(\overline{\mC}^{\left(1\right)}\right)^{-1}\vlambda_{i}^{\left(1\right)}\right)\\
		& +  2k_{N,T}\sum_{i=1}^{N}\varsigma_{\nu,i}^{-2}\left(\vepsi_{i}-\overline{\mU}\left(\overline{\mC}^{\left(1\right)}\right)^{-1}\vlambda_{i}^{\left(1\right)}\right)^{\prime}\mP_{\widehat{\mF}}\left(\vepsi_{i}-\overline{\mU}\left(\overline{\mC}^{\left(1\right)}\right)^{-1}\vlambda_{i}^{\left(1\right)}\right).
	\end{align*}
	By Lemma \ref{lem:II2}, the last line in the expression above is $\largeO_{P}\left(N^{-2}\sqrt{T}\right)+\largeO_{P}\left(N^{-3/2}\right)+\largeO_{P}\left(T^{-1/2}\right)$.
	Concerning the remaining term, we can write
	\begin{align*}
		& k_{N,T}\sum_{i=1}^{N}\varsigma_{\nu,i}^{-2}\left(\vlambda_{i}^{\left(2\right)}-\overline{\mC}^{\left(2\right)}\left(\overline{\mC}^{\left(1\right)}\right)^{-1}\vlambda_{i}^{\left(1\right)}\right)^{\prime}\mF^{\left(2\right)\prime}\mP_{\widehat{\mF}}\mF^{\left(2\right)}\left(\vlambda_{i}^{\left(2\right)}-\overline{\mC}^{\left(2\right)}\left(\overline{\mC}^{\left(1\right)}\right)^{-1}\vlambda_{i}^{\left(1\right)}\right)\\
		& \leq Tk_{N,T}\sum_{i=1}^{N}\varsigma_{\nu,i}^{-2}\left\lVert \vlambda_{i}^{\left(2\right)}-\overline{\mC}^{\left(2\right)}\left(\overline{\mC}^{\left(1\right)}\right)^{-1}\vlambda_{i}^{\left(1\right)}\right\rVert ^{2}\left\lVert T^{-1}\widehat{\mF}^{\prime}\mF^{\left(2\right)}\right\rVert ^{2}\left\lVert \left(T^{-1}\widehat{\mF}^{\prime}\widehat{\mF}\right)^{-1}\right\rVert
	\end{align*}
	where we focus on the term
	\begin{align*}
		& Tk_{N,T}\sum_{i=1}^{N}\varsigma_{\nu,i}^{-2}\left\lVert \vlambda_{i}^{\left(2\right)}-\overline{\mC}^{\left(2\right)}\left(\overline{\mC}^{\left(1\right)}\right)^{-1}\vlambda_{i}^{\left(1\right)}\right\rVert ^{2}\\
		& \leq \left(Tk_{N,T}3\sum_{i=1}^{N}\varsigma_{\nu,i}^{-2}\vlambda_{i}^{\left(2\right)\prime}\vlambda_{i}^{\left(2\right)}\right)+\left(Tk_{N,T}\sum_{i=1}^{N}\varsigma_{\nu,i}^{-2}\vlambda_{i}^{\left(1\right)\prime}\vlambda_{i}^{\left(1\right)}\right)\left\lVert \overline{\mC}^{\left(2\right)}\right\rVert ^{2}\left\lVert \left(\overline{\mC}^{\left(1\right)}\right)^{-1}\right\rVert ^{2}\\
		& = \largeO_{P}\left(\sqrt{T}\right).
	\end{align*}
	This result follows from applying the reasoning given below equation
	\eqref{eq:norm_epsFhat} to both weighted sums over $i$ in the expression above. Combining
	it with result \eqref{eq:Fhat_Ftwo}, we arrive at
	\begin{align*}
		& k_{N,T}\sum_{i=1}^{N}\varsigma_{\nu,i}^{-2}\left(\vlambda_{i}^{\left(2\right)}-\overline{\mC}^{\left(2\right)}\left(\overline{\mC}^{\left(1\right)}\right)^{-1}\vlambda_{i}^{\left(1\right)}\right)^{\prime}\mF^{\left(2\right)\prime}\mP_{\widehat{\mF}}\mF^{\left(2\right)}\left(\vlambda_{i}^{\left(2\right)}-\overline{\mC}^{\left(2\right)}\left(\overline{\mC}^{\left(1\right)}\right)^{-1}\vlambda_{i}^{\left(1\right)}\right)\\
		& = \largeO_{P}\left(N^{-1}\sqrt{T}\right)+\largeO_{P}\left(T^{-1/2}\right).
	\end{align*}
	which, together with Lemma \ref{lem:II2}, leads to
	\begin{equation}
		a_{3}=\largeO_{P}\left(N^{-1}\sqrt{T}\right)+\largeO_{P}\left(T^{-1/2}\right).\label{eq:a3}
	\end{equation}
	Next, we investigate $a_{2}$ and write
	\begin{align}
		a_{2} & \leq k_{N,T}\sum_{i=1}^{N}\left\lVert \varsigma_{\nu,i}^{-1}\vlambda_{i}^{\left(1\right)\prime}\right\rVert ^{2}\left\lVert \left(\overline{\mC}^{\left(1\right)\prime}\right)^{-1}\right\rVert ^{2}3\left(\left\lVert \overline{\mC}^{\left(2\right)}\right\rVert ^{2}\left\lVert \mF\right\rVert ^{2}+\left\lVert \overline{\mU}\right\rVert ^{2}\right)\nonumber \\
		& =\largeO_{P}\left(N^{-1}\sqrt{T}\right),\label{eq:a2}
	\end{align}
	where the last line follows from result \eqref{eq:UbarUbar}, the zero mean of $\overline{\mC}^{\left(2\right)}$
	and the result on $\sum_{i=1}^{N}\varsigma_{\nu,i}^{-2}\vlambda_{i}^{\left(1\right)\prime}\vlambda_{i}^{\left(1\right)}$
	referred to above. The last term to consider is given by
	\begin{align*}
		a_{1} & =k_{N,T}\sum_{i=1}^{N}\varsigma_{\nu,i}^{-2}\vepsi^{\prime}\overline{\mU}\left(\overline{\mC}^{\left(1\right)}\right)^{-1}\vlambda_{i}^{\left(1\right)}+k_{N,T}\sum_{i=1}^{N}\varsigma_{\nu,i}^{-2}\vlambda_{i}^{\left(2\right)}{}^{\prime}\mF^{\left(2\right)\prime}\overline{\mU}\left(\overline{\mC}^{\left(1\right)}\right)^{-1}\vlambda_{i}^{\left(1\right)}\\
		& +k_{N,T}\sum_{i=1}^{N}\varsigma_{\nu,i}^{-2}\vepsi_{i}^{\prime}\mF^{\left(2\right)}\overline{\mC}^{\left(2\right)}\left(\overline{\mC}^{\left(1\right)}\right)^{-1}\vlambda_{i}^{\left(1\right)}\\
		& +k_{N,T}\sum_{i=1}^{N}\varsigma_{\nu,i}^{-2}\vlambda_{i}^{\left(2\right)}{}^{\prime}\mF^{\left(2\right)\prime}\mF^{\left(2\right)}\overline{\mC}^{\left(2\right)}\left(\overline{\mC}^{\left(1\right)}\right)^{-1}\vlambda_{i}^{\left(1\right)}\\
		& =a_{11}+a_{12}+a_{13}+a_{14}.
	\end{align*}
	Here, $a_{11}=\largeO_{P}\left(N^{-1/2}\right)+\largeO_{P}\left(N^{-1}\sqrt{T}\right)$
	results from proceeding almost exactly as in the steps leading to
	equation \eqref{eq:II1_crossterm} in the proof of Theorem \ref{theorem::CCE}. Concerning the second term,
	we can write
	\begin{align*}
		\left|a_{12}\right| & \leq k_{N,T}\sum_{i=1}^{N}\left\lVert \varsigma_{\nu,i}^{-1}\vlambda_{i}^{\left(2\right)}\right\rVert \left\lVert \mF^{\left(2\right)\prime}\overline{\mU}\left(\overline{\mC}^{\left(1\right)}\right)^{-1}\right\rVert \left\lVert \varsigma_{\nu,i}^{-1}\vlambda_{i}^{\left(1\right)}\right\rVert \\
		& \leq k_{N,T}N\sqrt{\frac{T}{N}}\left(N^{-1}\sum_{i=1}^{N}\left\lVert \varsigma_{\nu,i}^{-1}\vlambda_{i}^{\left(2\right)}\right\rVert ^{2}\right)^{1/2}\left(N^{-1}\sum_{i=1}^{N}\left\lVert \varsigma_{\nu,i}^{-1}\vlambda_{i}^{\left(1\right)}\right\rVert ^{2}\right)^{1/2}\left\lVert \sqrt{\frac{N}{T}}\mF^{\left(2\right)\prime}\overline{\mU}\right\rVert \left\lVert \left(\overline{\mC}^{\left(1\right)}\right)^{-1}\right\rVert \\
		& =\largeO_{P}\left(N^{-1/2}\right),
	\end{align*}
	where we make use of \eqref{eq:UbarF} with $w_{i}=1$.With regards to the
	third term, we have
	\begin{align*}
		\left|a_{13}\right| & \leq k_{N,T}\sqrt{NT}\left(\left(NT\right)^{-1}\sum_{i=1}^{N}\left\lVert \vepsi_{i}^{\prime}\mF^{\left(2\right)}\right\rVert ^{2}\right)^{1/2}\left(N^{-1}\sum_{i=1}^{N}\left\lVert \vlambda_{i}^{\left(1\right)}\varsigma_{\nu,i}^{-2}\right\rVert ^{2}\right)^{1/2}\left\lVert \sqrt{N}\overline{\mC}^{\left(2\right)}\right\rVert \left\lVert \left(\overline{\mC}^{\left(1\right)}\right)^{-1}\right\rVert \\
		& =\largeO_{P}\left(N^{-1/2}\right).
	\end{align*}
	Here, we can use results from the proof of Lemma \ref{lem:II2} to obtain $\left(NT\right)^{-1}\sum_{i=1}^{N}\left\lVert \vepsi_{i}^{\prime}\mF^{\left(2\right)}\right\rVert ^{2}=\largeO_{P}\left(1\right)$.
	Lastly, we can apply the Cauchy-Schwarz inequality to $\left|a_{14}\right|$ in order
	to arrive at
	\begin{align*}
		\left|a_{14}\right| & =k_{N,T}\sqrt{N}T\left(N^{-1}\sum_{i=1}^{N}\left\lVert \varsigma_{\nu,i}^{-1}\vlambda_{i}^{\left(2\right)}\right\rVert ^{2}\right)^{1/2}\left(N^{-1}\sum_{i=1}^{N}\left\lVert \varsigma_{\nu,i}^{-1}\vlambda_{i}^{\left(1\right)}\right\rVert ^{2}\right)^{1/2} \\
		&\times \left\lVert T^{-1/2}\mF^{\left(2\right)}\right\rVert ^{2}\left\lVert \sqrt{N}\overline{\mC}^{\left(2\right)}\right\rVert \left\lVert \left(\overline{\mC}^{\left(1\right)}\right)^{-1}\right\rVert \\
		& = \largeO_{P}\left(N^{-1/2}\sqrt{T}\right).
	\end{align*}
	Using the last four intermediary results allows us to conclude that
	\[
	a_{1}=\largeO_{P}\left(N^{-1/2}\sqrt{T}\right).
	\]
	This result, together with \eqref{eq:a3} and \eqref{eq:a2} establishes
	\begin{align*}
		& k_{N,T}\sum_{i=1}^{N}\varsigma_{\nu,i}^{-2}\left(\vnu_{i}^{\prime}-\vlambda_{i}^{\left(1\right)\prime}\left(\overline{\mC}^{\left(1\right)\prime}\right)^{-1}\overline{\mD}^{\prime}\right)\mM_{\widehat{\mF}}\left(\vnu_{i}^{\prime}-\overline{\mD}\left(\overline{\mC}^{\left(1\right)}\right)^{-1}\vlambda_{i}^{\left(1\right)}\right)\\
		& = k_{N,T}\sum_{i=1}^{N}\varsigma_{\nu,i}^{-2}\vnu_{i}^{\prime}\vnu_{i}+\largeO_{P}\left(N^{-1/2}\sqrt{T}\right)+\largeO_{P}\left(T^{-1/2}\right),
	\end{align*}
	which was to be shown.
\end{proof}


\makeatletter\@input{auxmain.tex}\makeatother